%% file: main.tex
\newif\ifsubmit     
\newif\ifllncs      
\newif\ifexabs      
\newif\ifblind      
\newif\ifdisplayextendedabstract
\newif\ifdisplayclarification

\submittrue
\displayextendedabstractfalse

\ifllncs
  \documentclass[runningheads,a4paper]{llncs}

\else
  \documentclass[letterpaper,11pt]{article}
  \usepackage[in]{fullpage}
  \usepackage{setspace}
  \usepackage[margin=1in]{geometry}
\fi
\usepackage{fontawesome5}
\usepackage{cmap}

\usepackage{iftex}
\ifPDFTeX
  \usepackage[utf8]{inputenc}
  \usepackage[noTeX]{mmap}
  \usepackage[T1]{fontenc}
\fi
\ifLuaTeX
  \usepackage{luatex85}
  \usepackage[noTeX]{mmap}
\fi

\usepackage{mdframed}
\usepackage{amsmath}
\usepackage{amsfonts}
\usepackage{amssymb}
\usepackage{amsthm}
\usepackage{color}
\usepackage{xspace}

\allowdisplaybreaks[3]

\usepackage{appendix}
\usepackage{algorithm}
\usepackage{algpseudocode}
\usepackage{braket}
\usepackage{comment}
\usepackage{url}
\usepackage{bbm}
\usepackage{multicol}
\usepackage{mdframed}
\usepackage{multirow}
\usepackage{marvosym}
\usepackage[bookmarks]{hyperref}
\usepackage[nameinlink]{cleveref}
\hypersetup{
  colorlinks,
  allcolors=blue
}

\ifllncs

  \spnewtheorem{claim}{Claim}{\bfseries}{\rmfamily}
  \crefname{claim}{claim}{claims}
  \Crefname{claim}{Claim}{Claims}
\else
  \newtheorem{theorem}{Theorem}[section]
  \newtheorem{definition}[theorem]{Definition}
  
  \newtheorem{remark}[theorem]{Remark}
  \newtheorem{lemma}[theorem]{Lemma}
  \newtheorem{corollary}[theorem]{Corollary}

  \newtheorem{claim}[theorem]{Claim}
  \newtheorem*{remark*}{Remark}
  \newtheorem{construction}[theorem]{Construction}
  \newtheorem{fact}[theorem]{Fact}

  \newtheorem*{theorem*}{Theorem}
  \newtheorem*{lemma*}{Lemma}

\usepackage[style=alphabetic,minalphanames=3,maxalphanames=4,maxnames=99,backref=true]{biblatex}

  \DeclareFieldFormat{eprint:iacr}{Cryptology ePrint Archive: \href{https://ia.cr/#1}{\texttt{#1}}}
  \DeclareFieldFormat{eprint:iacrarchive}{Cryptology ePrint Archive: \href{https://eprint.iacr.org/archive/#1}{\texttt{#1}}}
  \AtEveryBibitem{
    \clearlist{address}
    \clearfield{date}
    \clearfield{isbn}
    \clearfield{issn}
    \clearlist{location}
    \clearfield{month}
    \clearfield{series}
 
    \ifentrytype{book}{}{
      \clearlist{publisher}
      \clearname{editor}
    }
  }
\fi

\usepackage{appendix}
\usepackage{algorithm}
\usepackage{algpseudocode}
\usepackage{braket}
\usepackage{comment}
\usepackage{url}
\usepackage{multicol}
\usepackage{tikz}
\usetikzlibrary{calc,arrows,automata,positioning,fit,decorations.pathreplacing,decorations.pathmorphing}
\usepackage{caption}
\usepackage[caption=false]{subfig}
\usepackage[font=small,labelfont=bf]{caption}
\usepackage{comment}
\usepackage{pifont}
\usepackage{fontawesome5}
\usetikzlibrary{patterns}
\usetikzlibrary{arrows.meta}

\usepackage{dsfont}

\usepackage{enumitem}
\setlist[description]{noitemsep}
\setlist[enumerate]{noitemsep}
\setlist[itemize]{noitemsep}

\usepackage{soul, xcolor, xparse}
\makeatletter
  \ExplSyntaxOn
    \cs_new:Npn \white_text:n #1
    {
      \fp_set:Nn \l_tmpa_fp {#1 * .01}
      \llap{\textcolor{white}{\the\SOUL@syllable}\hspace{\fp_to_decimal:N \l_tmpa_fp em}}
      \llap{\textcolor{white}{\the\SOUL@syllable}\hspace{-\fp_to_decimal:N \l_tmpa_fp em}}
    }
    \NewDocumentCommand{\whiten}{ m }
    {
      \int_step_function:nnnN {1}{1}{#1} \white_text:n
    }
  \ExplSyntaxOff
  
  \NewDocumentCommand{ \varul }{ D<>{5} O{0.2ex} O{0.1ex} +m } {%
    \begingroup
    \setul{#2}{#3}%
    \def\SOUL@uleverysyllable{%
      \setbox0=\hbox{\the\SOUL@syllable}%
      \ifdim\dp0>\z@
      \SOUL@ulunderline{\phantom{\the\SOUL@syllable}}%
      \whiten{#1}%
      \llap{%
        \the\SOUL@syllable
        \SOUL@setkern\SOUL@charkern
      }%
      \else
      \SOUL@ulunderline{%
        \the\SOUL@syllable
        \SOUL@setkern\SOUL@charkern
      }%
      \fi}%
    \ul{#4}%
    \endgroup
  }
\makeatother

\input{header}

\title{
	How To Track Qubits Through Space and Time\\
     {\large{(Or: Sailing in a Quantum Boat)}}
}

\ifllncs
\author{James Bartusek}
\institute{Columbia University}
\author{Zikuan Huang}
\institute{University of Cambridge}
\author{Leo Orshansky}
\institute{Columbia University}
\author{Henry Yuen}
\institute{Columbia University}
\else
\ifblind
\author{}
\else
\author{James Bartusek  \\ \small{Columbia} \and Zikuan Huang  \\ \small{University of Cambridge} \and Leo Orshansky \\ \small{Columbia} \and Henry Yuen  \\ \small{Columbia}}
\fi
\fi

\date{}

\begin{document}

\maketitle

\ifllncs 
\begin{abstract}
\end{abstract}
\else
\begin{abstract}
\input{abstract}

\end{abstract}
\fi

\newpage

\tableofcontents

\ifdisplayextendedabstract
\newpage
\input{extended_abstract}

\else
\fi

\ifdisplayclarification
\newpage
\input{clarification}
\else
\fi

\newpage

\input{introduction}
\input{technical_overview}
\input{preliminary}
\input{quantum_anchor}
\input{collapsing_oracles_versus_non-collapsing_oracles}
\input{monogamy-of-entanglement_games_with_abort}
\input{entanglement_extraction_with_abort}
\input{warmup_entanglement_localization}
\input{localizing_quantum_information_sequential_repetition}
\input{applications}
\printbibliography
\appendix

\end{document}

%% file: header.tex
\ifblind
\else
    \newcommand{\james}[1]{{\color{olive} James: #1}}
    \newcommand{\zikuan}[1]{{\color{blue} Zikuan: #1}}
    \newcommand{\henry}[1]{{\color{orange} Henry: #1}}
    \newcommand{\leo}[1]{{\color{magenta} Leo: #1}}
\fi

\newcommand{\E}{\mathop{\mathbb{E}}}

\newcommand{\abs}[1]{\left|#1\right|}
\newcommand{\norm}[1]{\left\lVert#1\right\rVert}

\newcommand{\As}{\mathcal{A}}

\newcommand{\Sim}{{\sf Sim}}

\newcommand{\cA}{\mathcal{A}}

\newcommand{\yb}{{\sf YB}}

\newcommand{\enc}{{\sf Enc}}
\newcommand{\dec}{{\sf Dec}}

\newcommand{\genchain}{{\sf GenAnchorState}}

\newcommand{\sto}{{\sf StO}}

\newcommand{\csto}{{\sf CStO}}

\newcommand{\stddecomp}{{\sf StdDecomp}}

\newcommand{\F}{\mathbb{F}}

\newcommand{\N}{\mathbb{N}}

\newcommand{\cF}{\mathcal{F}}

\newcommand{\cB}{{\mathcal{B}}}
\newcommand{\cC}{{\mathcal{C}}}
\newcommand{\cD}{{\mathcal{D}}}

\newcommand{\cE}{{\mathcal{E}}}

\newcommand{\cO}{{\mathcal{O}}}

\newcommand{\epr}{{\sf EPR}}
\newcommand{\eps}{\varepsilon}

\newcommand{\keygen}{{\sf KeyGen}}

\newcommand{\Sign}{\mathsf{Sign}}
\newcommand{\Ver}{\mathsf{Ver}}
\newcommand{\QSIO}{\mathsf{QSO}}

\newcommand{\register}{\mathsf{register}}

\newcommand{\sparam}{\mathsf{sp}}

\newcommand{\ek}{{\sf ek}}
\newcommand{\sk}{{\sf sk}}
\newcommand{\vk}{{\sf vk}}

\newcommand{\rand}{\overset{\$}{\gets}}

\renewcommand{\kappa}{\ell}

\newcommand{\poly}{{\sf poly}}
\newcommand{\negl}{{\sf negl}}

\DeclareMathOperator{\Tr}{Tr}

\newcommand{\Eval}{{\sf Eval}}

\newcommand{\p}{\mathbf{p}}

\renewcommand{\ket}[1]{\left|#1\right\rangle}
\renewcommand{\bra}[1]{\left\langle#1\right|}
\newcommand{\ketbra}[2]{\ket{#1}\!\bra{#2}}
\renewcommand{\braket}[2]{\left\langle #1\vert #2\right\rangle}
\renewcommand{\cal}[1]{\mathcal{#1}}
\newcommand{\R}{\mathbb{R}}

\newcommand{\brackets}[1]{\left(#1 \right)}

\newcommand{\ipd}[1]{\left \langle {#1} \right \rangle}

\newcommand{\Enc}{\mathsf{Enc}}

\DeclareMathOperator{\TD}{TD}
\newcommand{\TraceDist}[2]{{\TD\left(#1, #2\right)}}

\newcommand{\secparam}{\lambda}

\newcommand{\Setup}{\ensuremath{{\sf Setup}}\xspace}

\DeclareMathAlphabet\mathbfcal{OMS}{cmsy}{b}{n}

\newcommand{\Protect}{{\sf{Protect}}}
\newcommand{\secp}{{\sf{\lambda}}}
\newcommand{\Localize}{\ensuremath{\mathsf{Localize}}\xspace}
\newcommand{\param}{{\sf{pp}}}
\newcommand{\CS}{{\sf CS}}
\newcommand{\Can}{{\sf Can}}
\newcommand{\QSO}{{\sf QSO}}
\newcommand{\Authenticate}{{\sf Auth}}

\newcommand{\tr}{\mathsf{Tr}}

\newcommand{\prob}[1]{{\sf{Pr}}\left[#1\right]}
\newcommand{\probs}[2]{\underset{#1}{\sf{Pr}}\left[#2\right]}

\DeclareFontFamily{U}{skulls}{}
\DeclareFontShape{U}{skulls}{m}{n}{ <-> skull }{}

\DeclareRobustCommand{\substack}[1]{\subarray{c}#1\endsubarray}

\renewcommand{\set}[1]{\left\{#1\right\}}
\newcommand{\expect}[1]{\E\left[#1\right]}

\newcommand{\anchor}{\text{\faAnchor}}
\newcommand{\vessel}{\text{\faShip}}
\newcommand{\key}{\text{\faKey}}
\newcommand{\chain}{\Psi}

\newcommand{\expectc}[2]{\mathop{\mathbb{E}}_{#1}\left[#2\right]}

\newenvironment{boxfig}[2]{\begin{figure}[#1]\fbox{\begin{minipage}{0.97\linewidth}
                        \vspace{0.2em}
                        \makebox[0.025\linewidth]{}
                        \begin{minipage}{0.95\linewidth}
            {{
                        #2 }}
                        \end{minipage}
                        \vspace{0.2em}
                        \end{minipage}}
                        }
                        {\end{figure}}

\newcommand{\areg}{\mathbf{A}}
\newcommand{\breg}{\mathbf{B}}

\newcommand{\dreg}{\mathbf{D}}

\newcommand{\lreg}{\mathbf{L}}
\newcommand{\mreg}{\mathbf{M}}

\newcommand{\rreg}{\mathbf{R}}
\newcommand{\sreg}{\mathbf{S}}

\newcommand{\ureg}{\mathbf{U}}
\newcommand{\vreg}{\mathbf{V}}
\newcommand{\wreg}{\mathbf{W}}
\newcommand{\xreg}{\mathbf{X}}

\newcommand{\zreg}{\mathbf{Z}}
\newcommand{\trace}[1]{\Tr\brackets{#1}}
\newcommand{\chal}{{\sf Chal}}

%% file: abstract.tex

While quantum position verification aims to certify a prover's location using quantum information, existing security definitions only guarantee that \emph{part} of the successful adversarial party is in the claimed location. This leaves open the possibility that a distributed team of adversaries can jointly simulate a prover in a way that defeats the intended meaning of ``being at a location'' in position-based cryptography.

We introduce stronger notions of position verification that we call \emph{quantum localization}, which requires that there is a specified, unclonable state at the verified spacetime point -- and that this state can be found nowhere else. We show that quantum localization leads naturally to a meaningful notion of \emph{trajectory} verification, in which quantum information is verifiably tracked through space and time. We construct quantum localization and trajectory verification protocols using \emph{quantum anchor states}, which generalize coset states from unclonable cryptography. The security of our schemes is proven in the classical oracle (i.e. ideal obfuscation) model, which can be heuristically instantiated in the plain model using post-quantum indistinguishability obfuscation.

We also introduce and instantiate the concept of \emph{functionality localization}, which guarantees that the adversary has the ability to compute a secret function at the verified spacetime point, and this function cannot be computed anywhere else. This raises the intriguing possibility of localizing computational capabilities in space and time. 

More broadly, we believe our notions of quantum localization and subsequent feasibility results provide stronger foundations for position-based cryptography. At a conceptual level, our results also explore the tight link between two fundamental notions in physics -- location and information -- thus serving as a natural foundation for cryptographic capabilities tied to physical spacetime.

%% file: extended_abstract.tex
\begin{center}
    {\Large\bfseries How to Track Qubits Through Space and Time}\\[0.3em]
    {\large (Or: Sailing in a Quantum Boat)}
\end{center}

Position-based cryptography (\cite{chandran2009position}) asks whether physical location can serve as a cryptographic credential: rather than establishing authority through possession of a secret key, one establishes authority by being at a particular location. This viewpoint motivates a broad family of primitives, including position-based authentication, encryption, signatures, commitments, and proofs of location. A natural foundational primitive for this program is quantum position verification (QPV), which asks whether a prover can convince verifiers that an interaction took place at a prescribed spacetime point. QPV protocols have been proven secure against unentangled adversaries, primarily relying on the analysis of certain monogamy-of-entanglement games \cite{Tomamichel_2013}.


However, all previously studied QPV protocols share an important downside: they only prove that \emph{some} part of a successful prover strategy (which might be highly distributed!) took place at the claimed location. To show the limitations imposed by this issue on downstream tasks in position-based cryptography, we motivate our work the following intriguing question: \textbf{\emph{Are there protocols to certify that something is physically moving?}} A naive attempt might be to run QPV repeatedly at a number of checkpoints along the supposed trajectory, but this fails to distinguish the motion of a single prover from several provers standing still at different checkpoints. Thus we approach a solution from first principles, where a natural idea is to utilize \emph{monogamy of entanglement}: suppose the verifier holds a register $\areg$ maximally entangled with a register $\breg$ given to the prover. Monogamy says that no other system can be maximally entangled with $\areg$. Thus, if we can certify that the register entangled with $\areg$ is at $x=0$ at time $t=0$, and later at $x=1$ at time $t=1$, we can meaningfully conclude something about the path traveled by this entanglement.

We call a protocol that certifies such a statement (for now, at a given single point in spacetime) an \emph{entanglement localization scheme}. Informally, an entanglement localization scheme guarantees that any prover strategy that succeeds with non-negligible probability must contain, at the claimed spacetime point, a register that is entangled with the verifier's register. The first major contribution of our work is to define this new primitive and give a construction:

\begin{theorem*}[Entanglement localization (informal)]
There exists a non-destructive entanglement localization scheme in the ideal obfuscation model with perfect completeness. For any possibly nonlocal prover strategy that makes polynomially many oracle queries and is accepted with probability $\eta \ge 1/\poly(\secp)$, an extractor acting on the prover's state at the claimed spacetime point can, with probability $\eta$, recover a register whose joint state with the verifier's register has fidelity $1-1/\poly(\secp)$ with the $n$ EPR pairs.
\end{theorem*}

Our scheme crucially satisfies an additional property: namely, that the localization test is \emph{non-destructive} -- an honest prover's quantum state remains undisturbed by the protocol.\footnote{It is worth pointing out that existing QPV schemes satisfy \emph{neither} the entanglement localization or non-destructiveness property achieved here.} We can therefore apply the test successively at multiple spacetime points (i.e. a spacetime \emph{trajectory}) while preserving the entanglement throughout.

\begin{theorem*}[Trajectory verification (informal)]
For every physically realizable trajectory $L(\cdot)$, there exists a non-destructive trajectory verification scheme in the ideal obfuscation model with perfect completeness. For any possibly nonlocal prover strategy that makes polynomially many oracle queries and is accepted with probability $\eta \ge 1/\poly(\secp)$, for every time $t$ along the trajectory there is an extractor acting on the prover's state at spacetime point $(L(t),t)$ that, with probability $\eta$, recovers a register whose joint state with the verifier's register has fidelity $1-1/\poly(\secp)$ with $n$ EPR pairs.
\end{theorem*}

According to this definition, there is no successful strategy for trajectory verification which merely involves by stationing different provers at various checkpoints: the same unique piece of quantum information must be localized according to the claimed trajectory. This gives an operational meaning to tracking quantum information through space and time, and provides a stronger foundation for position-based cryptography.

Besides entanglement localization and trajectory verification, we also define and construct two additional primitives, which we call \emph{state localization} and \emph{functionality localization}. Briefly, state localization says that one can extract from a successful prover at the claimed spacetime point a specific unclonable state that is handed over to the prover at the beginning. Entanglement localization is then morally just a special case of this where the state is an EPR half. Functionality localization, on the other hand, says that from any successful prover one can extract the ability to compute a certain functionality. This becomes especially meaningful when this functionality is copy-protectable -- that is, when we are guaranteed that the ability to compute this functionality can be found in at most one physical location at any time. As an example application, we describe a position-based authentication scheme which works by localizing a copy-protected signature functionality.

\paragraph{Techniques.} We'll now briefly present the intuition behind our construction of the entanglement localization scheme. The first challenge is to test entanglement without destroying it. The standard way of testing whether two registers are entangled is to measure them in either the standard or Hadamard basis and compare the outcomes. This test certifies entanglement, but necessarily collapses the state and therefore cannot be repeated. Our key observation is that we do not need to actually perform the measurement in order to verify that the prover \emph{could} perform it successfully. We instead construct a non-collapsing test whose success implies that the corresponding collapsing test would also succeed.

A simple version of this idea is to hide the entangled qubits among a number BB84-state decoys, on which the test's measurements would be non-collapsing. The prover is then asked to answer a test on a randomly chosen subset of the qubits. Since the prover does not know which qubits are entangled, passing the test requires preserving the entire register. Unfortunately, repeated application of this simple construction eventually reveals which qubits are never tested. To obtain a reusable construction, we replace the random permutation by a random linear change of basis -- the resulting family of states corresponds to a variant of subspace coset states. From the prover's perspective, the relevant subspaces and cosets are computationally hidden, so the locations of the entangled qubits remain hidden even across repeated tests. This gives us a reusable non-collapsing test of entanglement.

The remaining challenge is to ensure \emph{localization}: a distributed team of provers might otherwise manipulate the quantum state at one location and defer the relevant operation to other locations. This issue is particularly apparent in the standard $f$-measure protocol, which admits a perfectly successful strategy where the entanglement is distributed among provers arbitrarily far from the claimed location, and the middle prover is completely unentangled with the verifiers' half of the EPR pair! Our construction prevents this type of attack by requiring the prover to answer two incompatible non-collapsing tests in succession -- one in the standard basis and one in the Hadamard basis. If the prover strategy attempts to answer these challenges without operating on the entanglement precisely at the claimed location, then we show that it can be adapted into a monogamy-of-entanglement game in which two separated parties must answer incompatible challenges consistently. We show that no efficient adversary can succeed in this game with probability significantly better than one half. Consequently, any successful prover must effectively perform the required measurement operation at the claimed spacetime point.

 We briefly elaborate on the techniques used to prove the above statement. The aforementioned monogamy-of-entanglement game to we which we reduce cheating prover strategies is a new kind of game which we call a \emph{multi-stage monogamy of entanglement} (multi-stage MoE) game. Much of the technical content of this work is devoted to studying this game through a series of oracle reductions, finally resulting in a close variant of the simultaneous Goldreich-Levin MoE game for subspace coset states, which was first studied by \cite{cryptoeprint:2025/1219}. The reduction to multi-stage MoE from a cheating prover in our entanglement localization scheme also requires extending previously studied techniques. Namely, it is a \emph{relativistic} oracle reduction similar to the one in \cite{unruh2014pvqrom}, but requires a much more involved decomposition of the prover into causally separated components, some of which are then subject to oracle reprogramming indistinguishability by these causal constraints. Putting this together, we obtain the statement that a successful prover strategy in our entanglement localization scheme must be capable of answering both standard-basis and Hadamard-basis challenges \emph{precisely at the claimed location}.

Finally, we turn this into an extraction statement. Using quantum recording techniques together with the compressed-oracle representation, we can extract the standard- and Hadamard-basis measurement outcomes from the prover's state at the claimed location. Combining the two extraction procedures with known EPR self-testing results then gives an extractor that recovers a $\breg$ register maximally entangled with the verifier's $\areg$ register, so long as the prover succeeds to answer challenges with high probability. Sequential repetition of our protocol then amplifies this guarantee from the high-success-probability regime to the inverse-polynomial regime needed for our full localization theorem.

The trajectory verification construction is then a conceptually simple extension. The verifier hands out the $\breg$ register to the prover and runs the non-destructive localization test successively along the claimed trajectory. Since the same $\areg$ register remains with the verifier and the entanglement is not destroyed by an honest execution, completeness as well as the localization guarantee applies at every point along the trajectory. By monogamy of entanglement, the anchored quantum information cannot simultaneously appear at another location, so a successful prover must carry the same quantum information along the trajectory rather than replacing it with independent information at different checkpoints.

The same framework also extends beyond entanglement. For \emph{state localization}, we replace the purified entangled state by an unclonable family of quantum states. The extractor must recover the particular state supplied to the prover at the claimed spacetime point, while uniqueness guarantees that the same state cannot be reproduced elsewhere. Our construction shows that any unclonable family that is a union of orthonormal bases can be localized non-destructively in the ideal obfuscation model. For \emph{functionality localization}, we replace the localized quantum state by a copy-protected implementation of a secret functionality. A successful prover must be able to evaluate the functionality at the claimed spacetime point, while the uniqueness requirement ensures that there cannot be two separated locations that can independently evaluate the same functionality. We obtain this by combining our localization techniques with best-possible quantum copy-protection, showing that any copy-protectable functionality can be localized in the ideal obfuscation model.


%% file: clarification.tex
We thank the reviewers and PC members for their thorough evaluation of our paper and would like to offer clarification for points 1 and 2 below. We emphasize that although both issues are important technical gaps which are necessary to fix before publication, we have already identified simple solutions to both. Furthermore, neither issue poses a fundamental roadblock to the validity or conceptual meaning of our results.
\begin{enumerate}
    \item As the reviewers point out, it is indeed the case that in order to invoke the extraction machinery of Theorem 7.9 as stated, Theorems 9.3 and 9.6 should provide the extractor with $\sk\gets{\sf GenAnchorState}(1^\lambda)$, i.e. $\sk=(U_{\sf shift},u,v)$ -- which contains the requisite basis information for extracting entanglement from the vessel state. In fact, in constructions 9.2, 9.5, and 9.8, each define an ``extraction key'' $\sf ek$ as a component of setup. However, Theorems 9.3 and 9.6 (by the reasoning in Remark 8.2) erroneously claim that extraction is possible even without providing this input $\sf ek$, by the incorrect assumption that $\sf ek$ could be recovered given unbounded queries to the public oracles.

    To correct this error, we will restate all of our definitions in terms of certifying the existence of an (efficient) extractor which is given the extraction key $(U_{\sf shift},u,v)$. The delta from our submitted version is that the extractor will now be given the entire matrix $U_{\sf shift}$ which contains strictly more information than the descriptions of $S$ and $T$. The fact that the extractor is given this extra information still does not in any way trivialize the operational meaning of localizing the prover's entanglement with the anchor state.

    \item For the second point, the submitted version does contain an error in the statement of Definition 7.8 and subsequent constructions: Construction 8.3, Construction 9.2, and Construction 9.8. The two challenges $\theta_0,\theta_1$ should not be given to $\cal{A}_M^0$ at the beginning. The first round challenge $\theta_0$ is given to $\cal{A}_M^1$ and the second round challenge $\theta_1$ is given to $\cal{A}_M^2$. To align this with later constructions, we must slightly edit the way that the challenge is released in Construction 8.3, Construction 9.2, and Construction 9.8. Formally, we will additionally include a challenge oracle $\cal{O}_\chal$ that determines the challenge $\theta_i=\cal{O}_\chal(x_{L,i}\oplus x_{R,i})$. By Unruh's lemma, the adversary only learns the challenge once $x_{L,i}$ and $x_{R,i}$ meet at the designated point. Therefore, these modified constructions do not have any affect on the results of our paper. We post the new Construction 9.2 here as clarification:\\
    
    Let $\secp\in\mathbb{N}$ be the security parameter and $n(\secp),m(\secp)\geq\secp$ be polynomials. Take a localization parameter $\Delta(\secp)\leq 1$ and a repetition factor $\gamma(\secp)$. We consider the following non-destructive entanglement localization scheme:
    \begin{itemize}
        \item There are two verifiers: $\cal{V}_L$ at $-1$ and $\cal{V}_R$ at $1$.
        \item $\Setup(1^\secp)$: 
        \begin{enumerate}
            \item Generate $\brackets{\sk,\ket{\chain^\sk}_{\areg\breg}}\gets\genchain(1^\secp)$, with $(U_{\rm shift},u,v):=\sk$, and let $S,T$ be the corresponding subspaces determined by $U_{\rm shift}$. The extraction key is $\ek=(U_{\rm shift},u,v)$, which is not given to the prover.
            \item Sample the public challenge oracle $\cal{O}_{\chal}:\mathbb{F}_2^\secp\rightarrow\set{0,1}$.
            \item Sample the public oracle $\cal{O}:\{0,1\}\times\mathbb{F}_2^\secp\times\mathbb{F}_2^{3n}\rightarrow\mathbb{F}_2^m\cup\set{\bot}$, defined with respect to a random oracle $\cal{O}_{\sf random}:\{0,1\}\times\mathbb{F}_2^\secp\rightarrow \mathbb{F}_2^m$ that is \textbf{not} publicly available:
            \begin{equation*}
                \cal{O}(\theta,x,z)=\left\{
                \begin{aligned}
                    &\cal{O}_{\sf random}(\theta,x) \quad &\text{if }\theta=0\text{ and }z\in T+v,\\
                    &\cal{O}_{\sf random}(\theta,x) \quad &\text{if }\theta=1\text{ and }z\in S^\perp+u,\\
                    &\bot \quad &\text{otherwise.}
                \end{aligned}
                \right.
            \end{equation*}
        \end{enumerate}
        Then $\Setup$ outputs:
        \begin{itemize}
            \item the public parameter $\param=(\cal{O}_\chal,\cal{O})$.
            \item The secret parameter $\sparam=(u,v)$
            \item The anchor state $\ket{\chain^\sk}_{\areg\breg}$, with $\anchor$ renamed $\areg$ and $\vessel$ renamed $\breg$.
        \end{itemize}
        \item $\Localize\brackets{\cal{P}(\breg)\rightleftharpoons \cal{V}(\sparam)}(\param,L,t)$: Let $\delta=\Delta/(\gamma+2)$ and define checkpoints $t_i=t+i\delta$ for $i=0,1,\ldots,\gamma$.
        \begin{itemize}
            \item For $i=0,1,\cdots,\gamma$, $\cal{V}_L$ samples $x_{L,i}\rand\mathbb{F}_2^\secp$ and broadcasts it at time $t_i-(L+1)$. It expects responses $y_{L,i}$ at time $t_i+(L+1)$.
            \item For $i=0,1,\cdots,\gamma$, $\cal{V}_R$ samples $x_{R,i}\rand\mathbb{F}_2^\secp$ and broadcasts it at time $t_i-(1-L)$. It expects responses $y_{R,i}$ at time $t_i+(1-L)$.
            \item If any of these responses is missing, the verifier outputs $\bot$. Otherwise, the verifier outputs $\top$ if and only if for each $i=0,1,\cdots,\gamma$, one of the following holds:
            \begin{itemize}
                \item $\theta_i:=\cal{O}_{\chal}(x_{L,i}\oplus x_{R,i})=0$ and $y_{L,i}=y_{R,i}=\cal{O}(0,x_{L,i}\oplus x_{R,i},v)$.
                \item $\theta_i:=\cal{O}_{\chal}(x_{L,i}\oplus x_{R,i})=1$ and $y_{L,i}=y_{R,i}=\cal{O}(1,x_{L,i}\oplus x_{R,i},u)$.
            \end{itemize}
            \item The honest prover consists of a single party always holding the $\breg$ register at location $L$. For each $i=0,1,\ldots,\gamma$, it receives $x_{L,i},x_{R,i}$ at time $t_i$. If $\theta_i=0$, it coherently evaluates $\cal{O}(0,x_{L,i}\oplus x_{R,i},\cdot)$ on the $\breg$ register in the standard basis and broadcasts the result. If $\theta_i=1$, it coherently evaluates $\cal{O}(1,x_{L,i}\oplus x_{R,i},\cdot)$ on the $\breg$ register in the Hadamard basis and broadcasts the result.
        \end{itemize}
    \end{itemize}
\end{enumerate}

We will shortly upload an arXiv version containing these modifications.

%% file: introduction.tex
\section{Introduction}
\label{sec:intro}

\paragraph{Position-based cryptography.} Is it possible to base cryptographic identity on a party's \emph{physical location}? This is in contrast with the typical cryptographic paradigm in which a party's identity, i.e. their \emph{credential}, is tantamount to their knowledge of some secret information such as a password or private key. Instead, in \emph{position-based cryptography}, it is a party's physical presence at an authorized location which serves as their credential~\cite{chandran2009position}.
The aim is to support tasks such as \emph{position-based key exchange}, which seeks to establish secret communication with a party at an authorized location, and \emph{position-based authentication}, wherein messages are signed with the guarantee of originating at an authorized location.
A promising approach to establishing this type of positional credential is known as \emph{quantum position verification} (QPV)~\cite{buhrman2014position,unruh2014pvqrom}.

In a QPV protocol, verifiers attempt to verify the physical location of a party (called the \emph{prover}) by sending challenges (which are, in general, quantum states) and receiving responses at precise times.
Intuitively, the goal is to ensure that a successful prover, responding within some timing constraint, must be at a purported location~$L$.
As argued by~\cite{KMS11qubit-routing,chandran2009position,buhrman2014position}, secure position verification is impossible in the classical setting because a team of spoofers (all of whom are \emph{not} in location~$L$) can copy and forward messages to each other to simulate a fictitious prover in location~$L$.
On the other hand, the No-Cloning Theorem of quantum mechanics stymies such copy-and-forwarding attacks, and in fact secure QPV protocols have been established in the bounded-qubits model~\cite{Tomamichel_2013,bluhm2021positionbased,asadi2025ranklowerboundsnonlocal,Asadi_2025} and the random oracle model~\cite{unruh2014pvqrom}.
Recently, QPV has been generalized to \emph{privacy-preserving proofs of location} that only reveal part of a prover's location and nothing else~\cite{girish2026private}. For example, a user can prove to the police that they were somewhere far from the scene of a crime, without revealing any other information about their specific whereabouts. There have also recently been successful QPV experiments by multiple groups \cite{kanneworff2025experimental,kavuri2026quantumpositionverificationremote,fanyuan2026relativistic}, marking a new type of NISQ-era-implementable quantum advantage with many real-world applications.


\paragraph{Trajectory verification, and a conceptual gap.} 
Despite the promise of QPV, closer examination reveals a gap between the intuitive goals of position-based cryptography and the formal notion of QPV security that has been considered thus far. In prior work, security of QPV was defined as follows: a team of spoofers $\{\cal{P}_1,\ldots,\cal{P}_k\}$ can succeed in a QPV protocol with high probability only if \emph{at least one} of the spoofers $\cal{P}_i$ is in the correct location $L$. This is a rather weak guarantee, however. For one thing, it should intuitively be possible to \emph{attribute} a position-based credential to some party, who can then use it in downstream applications --- but in QPV, a successful prover has only established that \emph{somebody} exists at the given location, which lacks an attribution mechanism. To illustrate some of the issues stemming from this conceptual gap, we will introduce a natural generalization of QPV which we call \emph{trajectory verification}.

In trajectory verification, we want to verify that a prover has traveled along a trajectory $L(\cdot)$ described as a function of time $t$. A natural approach is to run separate QPV protocols $\set{\Pi_i}_i$ for many spacetime points $\set{(L(t_i),t_i)}_i$, respectively, such that these points divide the trajectory into sufficiently small segments. An immediate problem is that the existing security guarantees of $\Pi_1,\Pi_2,\ldots$ do not distinguish between whether a \emph{single} prover has traveled along the trajectory $L(\cdot)$ or whether an entire team $\{\cal{P}_1,\cal{P}_2,\ldots\}$ of provers participated in the protocol, where each prover $\cal{P}_i$ is stationed at location $L(t_i)$, only participating in protocol $\Pi_i$ at time $t_i$. To rule out the latter scenario, one needs to explain why a line of stationary provers is fundamentally different from a single prover that is moving! QPV on its own is clearly ill-equipped to do this, since in both scenarios it is trivially true that \emph{somebody} is at each location.


\paragraph{This work.} The issues described above call for a notion of position verification in which a successful proof is attributable to some \emph{persistent information} about the participating prover, rather than a simple binary stating that the location was occupied. Our work achieves this by formally introducing \emph{localization}: pinpointing the behavior of an adversary who might \emph{a priori} be distributed throughout space and time, to a persistent quantum state whose location is precisely known.

The main contributions of this paper are to formalize several notions of \emph{quantum localization} -- the aforementioned ``pinpointing'' -- and establish their feasibility via constructions in the (classical) ideal obfuscation model,\footnote{We provide more details about this model, which is sometimes referred to as the classical oracle model, in \Cref{subsec:ideal-obf}.} which can be heuristically instantiated using post-quantum indistinguishability obfuscation of classical circuits. We then show that these notions enable meaningful notions of trajectory verification, yielding the first feasibility result for verifying the trajectory of an entity through spacetime.

We believe that our notions of localization come closer to capturing the concept of secure position verification. They give operational meaning to tracking the physical location of quantum states and computations, and lay the proper foundation for more complicated tasks such as trajectory verification. 


\subsection{Entanglement Localization and Trajectory Verification}
We present stronger security notions for QPV that we call \emph{quantum localization}, or \emph{localization} for short. Informally, these security notions will guarantee that any successful prover strategy (which may involve a team of provers $\{\cal{P}_1,\ldots,\cal{P}_k\}$ moving around) must contain a specific object at the correct point $(L,t)$ in spacetime. By ``contain,'' we mean that there exists a procedure called the \emph{extractor} acting only on the prover's strategy at the point $(L,t)$ that recovers the desired object. In this paper we identify three types of objects that can be localized: entanglement, (unclonable families of) quantum states, and (copy-protectable) functionalities. 

We first describe \emph{entanglement} localization and then describe how to build on entanglement localization in order to achieve a meaningful notion of \emph{trajectory} verification. 

\paragraph{Entanglement Localization.} Let $\psi$ be a bipartite entangled state on registers $\areg$ and $\breg$. Consider a QPV protocol where the verifiers generate $\psi$, keep register $\areg$, and send register $\breg$ to the prover who is purportedly at spacetime point $(L,t)$. 

We say that the protocol \emph{localizes the entanglement in $\psi$ at $(L,t)$} if for any (possibly nonlocal) prover strategy that is accepted with probability $\eta$, there exists an extractor $\cal{E}$ that acts on the quantum state at spacetime point $(L,t)$ and achieves the following. With probability $\eta$, it outputs a register $\breg$ that, together with the verifier's register $\areg$, is close to being in the state $\psi$. In other words, it has recovered the entangled state $\psi$ between register $\areg$ and the spacetime point $(L,t)$ with the same probability as the probability that the prover succeeds in the protocol. We formally define this in \Cref{def:NDEL}.

Entanglement localization captures the idea that, in order to be successful, a team of provers \emph{must} ensure that the quantum entanglement in the state $\psi$ arrives at spacetime point $(L,t)$. Suppose that $\psi$ is a pure state that is maximally entangled across the $V : P$ cut. By monogamy of entanglement, the verifier's register $\areg$ must then be unentangled with any other location $L' \neq L$ at time $t$.

\paragraph{Trajectory Verification.} The notion of entanglement localization gives rise to a natural definition of trajectory verification. Let $L(\cdot)$ denote a trajectory in spacetime. Intuitively, a trajectory verification protocol for $L(\cdot)$ is secure if there exists a maximally entangled state $\psi$ such that entanglement localization can be performed with respect to $\psi$ at spacetime points $(L(t),t)$ for all time $t$ along the trajectory $L(\cdot)$. Importantly, the verifier's part of $\psi$ remains at the same register $V$ throughout the protocol, while the prover's part of $\psi$ may travel around. This definition captures the idea of a \emph{single} entity moving along the trajectory $L(\cdot)$ by tracking the movement of the entanglement $\psi$. If the prover is accepted by the protocol, monogamy of entanglement implies that $\psi$ couldn't have strayed from the trajectory $L(\cdot)$. 

We note that in order to repeatedly run an entanglement localization protocol on the same state, it must satisfy an additional completeness guarantee which we call \emph{non-destructive}. That is, the prover's state must remain intact after interacting with the verifier. Therefore, our focus in this work is on building \emph{non-destructive} versions of quantum localization.


\paragraph{Construction.} 
We show how to construct secure trajectory verification in the ideal obfuscation model. In this model, the \Setup samples the description of an efficient classical circuit, which all parties are given \emph{black-box} access to throughout the protocol. Any protocol in this model can be heuristically instantiated in the plain model by applying a candidate post-quantum indistinguishability obfuscator to the classical circuit sampled by \Setup. 

Trajectory verification as (informally) defined above immediately yields a protocol for entanglement localization as a special case (where the trajectory is stationary at $L$). For the sake of exposition we first describe (a simplified version of) the entanglement localization protocol, and then describe how to extend it to obtain trajectory verification. Moreover, we focus here on the ``high success probability'' regime, where we only guarantee extraction success if the prover passes the protocol with probability close to 1. 

Assume that space is one-dimensional and the location to be verified is the origin $L = 0$. Place verifiers $\cal{V}_L,\cal{V}_R$ at locations $-1$ and $+1$, respectively. We assume messages travel one unit of space per unit time. 

The entangled state used in our protocol is what we call a \emph{quantum anchor state}. Let $S \leq T \leq \F_2^{3n}$ be subspaces of dimension $n$ and $2n$, respectively, and let $u, v \in \F_2^{3n}$ be some vectors. Let $\CS = \{x_1,x_2,\ldots\}$ denote some canonical set of $2^n$ coset representatives of $S$ within $T$, indexed by $i \in [2^n]$. We define the state
\[
	\ket{\chain}_{\anchor \vessel} = \frac{1}{\sqrt{2^n}} \sum_{i \in [2^n]} \ket{i}_{\anchor} \otimes \ket{\psi_i}_{\vessel}
\]
where for each $i \in [2^n]$,
\[
	\ket{\psi_i}_{\vessel} = \frac{1}{\sqrt{|S|}} \sum_{s \in S} (-1)^{\langle s, u \rangle} \ket{s + x_i + v}_{\vessel}~.
\]
We use the icons $\anchor$ and $\vessel$ to denote the ``anchor'' and ``vessel'' registers, respectively. The reason for this naming will become apparent shortly.

Readers familiar with unclonable cryptography may recognize $\ket{\psi_i}$ as a \emph{coset state}. That is, the state $\ket{\chain}$ can be seen as a uniform superposition over a set of $2^n$ possible coset states, each defined with respect to subspace $S$, dual shift $u$, and primal shift $v + x_i$ for some choice of $x_i \in \CS$. Note further that $\ket{\chain}$ is maximally entangled across the $\anchor : \vessel$ cut with Schmidt rank $2^n$.


%
%
%

In the \Setup phase of the protocol (run any time before time $t = 0$), the verifiers generate the quantum anchor state $\ket{\chain}$ corresponding to random subspaces $S \leq T \leq \F_2^{3n}$ and shifts $u,v \in \F_2^{3n}$. Furthermore, the verifiers prepare a classical oracle $\cal{O}$ (i.e., ideal obfuscation of some classical functionality) which depends on $S,T,u,v$ and whose behavior we describe shortly. The verifiers publish the oracle $\cal{O}$, which all parties (honest or adversarial) can access as a black box. They also release the vessel register $\vessel$ to the prover, while keeping hold of the entangled anchor register $\anchor$.

In the ``online'' phase of the protocol (at time $t = 0$), the verifiers sample random strings $a,b \in \{0,1\}^n$, the left verifier sends string $a$, and the right verifier sends string $b$. The honest prover with register $\vessel$ at location $L = 0$ gets $a,b$ at time $t = 1$. Then,
\begin{enumerate}
	\item The honest prover coherently queries the oracle $\cal{O}$ on input $\ket{0,a \oplus b, z}$ where $\ket{z}$ is the state of the vessel register $\vessel$ in the standard basis. The oracle, when the first bit of the input is $0$, checks whether $z \in T + v$ (For the honest prover, this will be true.) and if so, outputs a string $r_0 = H(0,a \oplus b)$, where $H$ is a random oracle.\footnote{Technically, $H$ will be implemented by a pseudorandom function so that the oracle $\cal{O}$ is efficiently computable.} Otherwise, the oracle outputs $\bot$.
	\item The honest prover then queries the oracle $\cal{O}$ on input $\ket{1,a \oplus b,w}$ where $\ket{w}$ is the state of the vessel register $\vessel$ in the Hadamard basis. The oracle, when the first bit of the input is $1$, checks whether $w \in S^\perp + u$ (which will be the case) and if so, outputs a string $r_1 = H(1,a \oplus b)$. Otherwise, the oracle outputs $\bot$.
	\item The honest prover returns $(r_0,r_1)$ back to the verifiers.
\end{enumerate}
We note that after each query, the vessel register and the anchor state $\ket{\chain}$ have not been disturbed. 

The verifiers check that the responses came back by time $t = 2$, and that $r_0,r_1$ are indeed equal to $H(0,a\oplus b), H(1,a \oplus b)$.\footnote{One way to check this is to have the verifiers remember a secret verification key $\sk = (u,v)$ consisting of the shifts and have them query $\cO(a \oplus b,v)$ and $\cO(a \oplus b,u)$.} The idea is that the only way for the prover to have ``unlocked'' $r_0,r_1$ is if it performed the honest strategy above in the correct location $L = 0$. 

As mentioned above, this is the basic idea behind a protocol that allows us to extract from any prover that passes with probability close to 1.  We formalize this in \Cref{thm:NDEL_overall}. Then, in \Cref{subsec:entanglement-localization}, we apply an appropriate notion of sequential repetition in order to establish a more general result, informally stated as follows.


\begin{theorem}[Entanglement localization (informal)]
\label{thm:intro-qpv-entanglement-localization}
Let $(L,t)$ denote the spacetime point being verified. Sequentially repeating the protocol\footnote{For technical reasons, our actual protocol is slightly more elaborate than this, where the standard- and Hadamard-basis queries are sequentially staggered (by a tiny time difference). We elaborate more on this in \Cref{sec:overview}.} above satisfies the following properties:
\begin{itemize}
    \item \textbf{Completeness}: The honest prover strategy at spacetime point $(L,t)$ is accepted by the verifiers with probability $1$.

    \item \textbf{Extraction soundness}: Let $\cal{A}$ be any (possibly nonlocal) prover strategy that makes at most a polynomial number of queries to $\cal{O}$ and is accepted with probability at least $\eta = 1/\poly(\secp)$. There exists an extractor $\cal{E}$ that takes as input the quantum state generated by the strategy at spacetime point $(L,t)$ and, with probability roughly $\eta$, outputs a vessel register $\vessel$ such that the joint state of the anchor register $\anchor$ (held by the verifiers) and vessel register $\vessel$ has fidelity $1-1/\poly(\secp)$ with the quantum anchor state $\ket{\chain}$. 
    
\end{itemize}
\end{theorem}

We note that our extractor requires knowledge of $S,T,u,v$ and $\CS$. Indeed, we do not require any computational (or query) bound on the extractor for the above notion to be meaningful, since monogamy-of-entanglement is an information-theoretic property. However, we do show that the extractor is efficient if \emph{given} a secret ``extraction key'' $\ek = (S,T,u,v,\CS)$, which may be a useful feature, depending on the application.

Given our entanglement localization protocol, the extension to trajectory verification is simple. Let $L(\cdot)$ be a physically realizable trajectory in spacetime that starts at time $0$ and ends at $\tau$, and the spatial points are all contained within the convex hull of the verifiers (i.e., strictly between $L = -1$ and $L = 1$). Discretize time $0 = t_1 < t_2 < \cdots < t_m = \tau$, and let $L(t_i)$ denote the spatial location of the trajectory at time $t_i$. 


The trajectory verification protocol has the same \Setup as the entanglement localization protocol. In the online phase of the protocol, the verifiers run $m$ invocations of the entanglement localization protocol in sequence for spacetime points $(L(t_i),t_i)$, sending freshly sampled $a_i,b_i$ each time. We can compose the entanglement localization guarantees of the QPV protocol above to obtain \Cref{thm:TV_overall}, informally stated as follows.

\begin{theorem}[Trajectory verification (informal)]
\label{thm:intro-trajectory-verification}
The trajectory verification protocol above\footnote{Again, the actual trajectory verification protocol we construct is slightly different, for technical reasons.} satisfies the following properties:
\begin{itemize}
    \item \textbf{Completeness}: An honest prover traversing the trajectory $L(\cdot)$ is accepted with probability $1$.

    \item \textbf{Extraction soundness}: Let $\cal{A}$ be any (possibly nonlocal) prover strategy that makes at most a polynomial number of queries to $\cal{O}$ and is accepted with probability at least $\eta = 1/\poly(\secp)$. There exists an extractor $\cal{E}$ that, for any $i \in [m]$, takes as input the quantum state generated by the strategy at spacetime point $(L(t_i),t_i)$ and, with probability roughly $\eta$, outputs a vessel register $\vessel$ such that the joint state of the anchor register $\anchor$ (held by the verifiers) and vessel register $\vessel$ has fidelity $1-1/\poly(\secp)$ with the quantum anchor state $\ket{\chain}$. 
\end{itemize}
\end{theorem}


At this point, we hope the ``anchor'' and ``vessel'' terminology should appear natural. The $\anchor$ part of the state $\ket{\chain}$ is ``anchored'' at the verifier's location, while the $\vessel$ part is launched into the spacetime region under the provers' control. The provers can \emph{a priori} perform arbitrarily complex operations on the vessel, but in the end, if they are to pass the trajectory verification protocol, they must faithfully sail the quantum vessel $\vessel$ along the trajectory $L(\cdot)$. The fact that the entanglement persists between the anchor and the points along the trajectory captures the fact that we are tracking the same quantum information through space and time.

\begin{figure}[h]
    \centering
    \includegraphics[scale = 0.35]{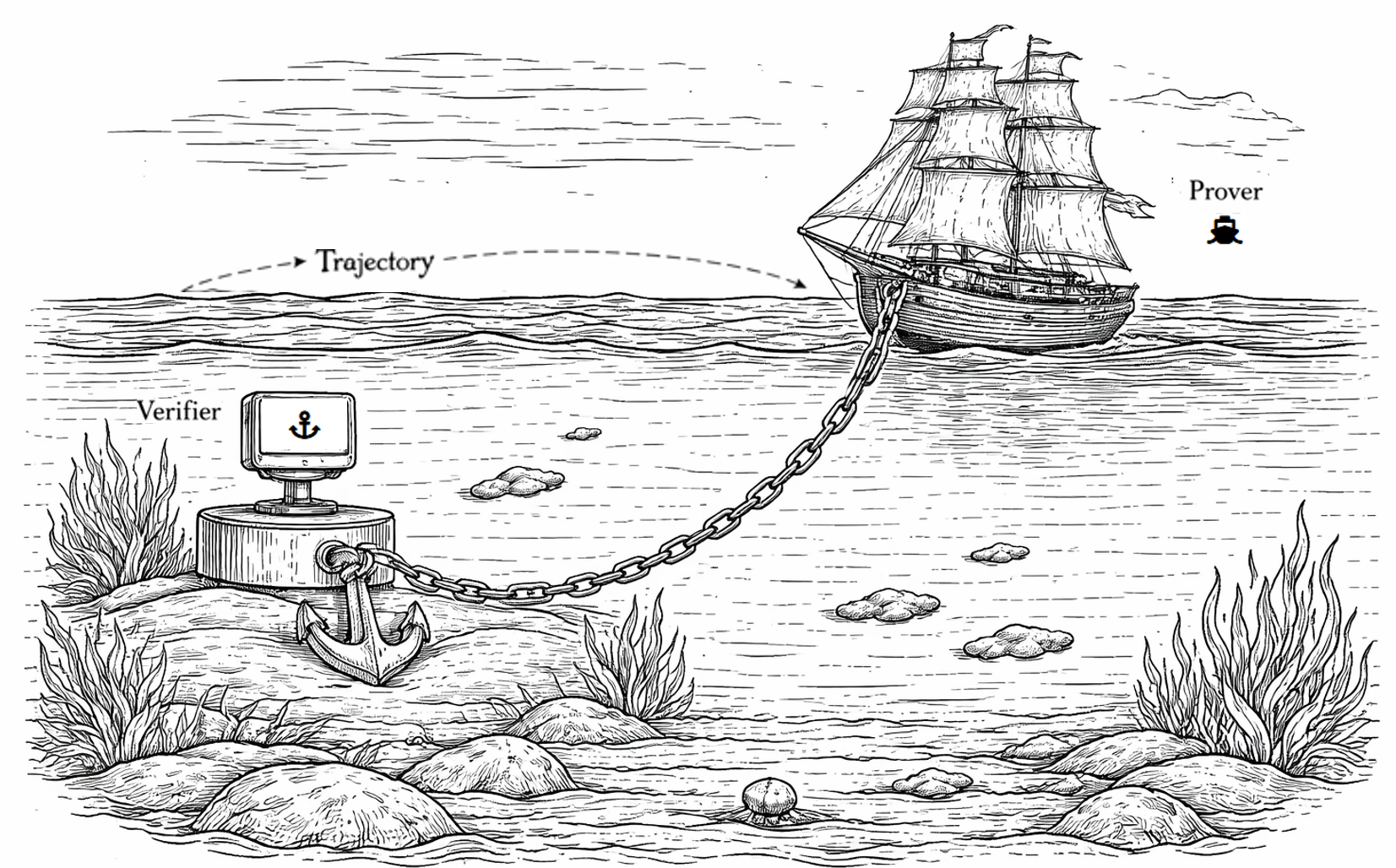}
    \caption{A quantum anchor. Generated by GPT.}
\end{figure}

\subsection{State Localization}

A curious aspect of the entanglement localization and trajectory verification protocols described above is that the verifier's register (the quantum anchor) is never acted upon after the entangled state is prepared. In fact, the verifiers don't even need this register to check the prover's responses. Thus, a slightly different way to describe the generation of the anchor state is as follows.

\begin{itemize}
    \item Sample $S,T,u,v,\CS$ as above, and define an isometry $\Enc_{S,T,u,v,\CS}$ that maps 
    \[\Enc_{S,T,u,v,\CS}: \ket{i} \to \ket{\psi_i} = \frac{1}{\sqrt{|S|}} \sum_{s \in S} (-1)^{\langle s, u \rangle} \ket{s + x_i + v}.\]
    \item Prepare the maximally entangled state \[\frac{1}{\sqrt{2^n}}\sum_{i \in [2^n]}\ket{i}_\anchor\ket{i}_\vessel\] and apply $\Enc_{S,T,u,v,\CS}$ to the $\vessel$ register.
    \item Discard $\anchor$.
\end{itemize}

Imagine, then, that instead of applying $\Enc_{S,T,u,v,\CS}$ to half of a maximally entangled state, the verifiers sample a pure state $\psi_k$ from some unclonable family of states $\{\psi_k\}_{k \in [K]}$ (say, the BB84 states) and output $\Enc_{S,T,u,v,\CS}\ket{\psi_k}$. If the mixed state produced by this sampling procedure is the maximally mixed state, then this is completely identical from both the prover's \emph{and} extractor's perspective, since neither touches the $\anchor$ register. It follows then that at each intermediate time $t_i$, the extractor acting on spacetime point $(L(t_i),t_i)$ will output a register in a state close to $\Enc_{S,T,u,v,\CS}\ket{\psi_k}$. In fact, since the extractor knows $S,T,u,v,\CS$, it can undo the isometry to recover $\ket{\psi_k}$ itself.


Thus, even though there is no entanglement to keep track of anymore in this scenario, there is still a meaningful sense in which we have localized a quantum state $\psi_k$ along the trajectory: At any given time $t_i$, the state $\psi_k$ cannot be found at any other location in space since otherwise the adversary / extractor would have been able to clone it. 


This motivates the following general definition of \emph{state localization}. Let $\cal{S} = \{\psi_k \}_k$ denote a family of quantum states indexed by $k$ (for example, BB84 states or coset states). Suppose the verifiers sample $k$, generate (and potentially encode) $\ket{\psi_k}$, and send the resulting state to the prover. We say that a protocol \emph{localizes the state family $\cal{S}$} if for any (possibly nonlocal) prover strategy that is accepted with high probability, there exists an extractor $\cal{E}$ that acts on the quantum state at spacetime point $(L,t)$ and outputs $\ket{\psi_k}$ with high fidelity. Furthermore, the state extracted from $(L,t)$ must be \emph{unique}, in that, in expectation over $\psi_k \gets \cal{S}$, it is impossible for any adversary to output \emph{two} copies of $\psi_k$ at any time $t$. We present a formal definition in \Cref{def:state-localization}. 

By applying the techniques used to prove \Cref{thm:NDEL_overall,thm:TV_overall}, we prove \Cref{thm:state-localization}, informally stated as follows.


\begin{theorem}[State localization (informal)]
\label{thm:intro-state-localization}
    Let $\cal{S}$ be any set of $n$-qubit states that consist of a union of orthonormal bases. Then if $\cal{S}$ is unclonable, there exists a protocol that (non-destructively) localizes $\cal{S}$ in the ideal obfuscation model.
\end{theorem}


While state localization is closely related to entanglement localization, there are some meaningful differences. The extractor produces a quantum state (close to) $\psi_k$ that is classically correlated with the verifiers' choice of $k$. That is, unlike in the case of entanglement localization, the verifier has an actual classical description $k$ of the state being tracked, as opposed to some entangled register. Thus, there may exist some other side information in the protocol that also depends on $k$, which may be a useful feature depending on the application. Furthermore, in the setting of purely classical verifiers (such as that of~\cite{liu2021beating}) we cannot obtain entanglement localization, but state localization is still potentially possible (e.g., the verifiers can delegate the preparation of $\psi_k$ to the prover). 

\paragraph{Trajectory verification with localized states.} Similarly to how trajectory verification follows naturally from our (non-destructive) entanglement localization protocol, it is straightforward to see that one could build trajectory verification from our (non-destructive) state localization protocol as well. We do not make this explicit, and only present trajectory verification based on entanglement localization in the body.

\subsection{Functionality Localization}

Entanglement and state localization capture the idea that the prover must carry a specific quantum state to some spacetime point $(L,t)$. However, can we say anything about the prover's \emph{computational capabilities} at $(L,t)$? This could be desirable in a scenario like export control of secret, proprietary programs (like a foundation language model), where a company or a government would like to guarantee that the ability to run the secret program is localized to a specific location (say within an authorized datacenter). 

Let $\cal{F} = \{ f_k \}_k$ denote a family of functions indexed by some key $k$ (for example, a pseudorandom function family). Suppose that at the beginning of the protocol, the verifiers sample $k$ privately. We say that a protocol \emph{localizes functionality $\cal{F}$} if for any (possibly nonlocal) prover strategy that is accepted with high probability, there exists an extractor $\cal{E}$ that acts on the quantum state at verified spacetime point $(L,t)$ and is able to compute $f_k(x)$ for a random input $x$. On the other hand, there must be no (efficient) adversary that can, for random inputs $x,y$, compute $f_k(x)$ and $f_k(y)$ at two separate locations without communicating. In other words, the ability to compute $f_k(x)$ has been localized to $(L,t)$, and no other location (at time $t$) has the ability to compute $f_k$. This is formally defined in \Cref{def:functionality_localization_revised}. 

The notion of functionality localization is closely related to (and in fact strengthens) the notion of \emph{quantum copy-protection}, which is the concept of encoding a function $f$ in a quantum state $\psi_f$ that can be used to evaluate $f$, but cannot be used to compute $f$ in two locations that cannot communicate \cite{10.1109/CCC.2009.42}. There are broad classes of function families that are known to be copy-protectable under various cryptographic assumptions and in various models (see e.g., \cite{10.1007/978-3-030-84242-0_19, 10.1007/978-3-030-84242-0_20, Coladangelo2024quantumcopy, 10.1007/978-3-031-22318-1_11, 10.1007/978-3-031-68394-7_1, 10.1007/978-3-032-25291-3_19,10.1007/978-3-032-25291-3_17}). Here, we show that any functionality $\cal{F}$ that can \emph{in principle} be copy-protected can also be localized. In fact, we don't even need to know an explicit copy-protection scheme for $\cal{F}$. Our construction relies on a combination of techniques described above as well as the ``best-possible'' copy-protection guarantees given by quantum state obfuscation \cite{CG24,BBV24}. We defer further details to \Cref{sec:overview}, and establish the following theorem, stated informally.

\begin{theorem}[Functionality localization (informal)]
    \label{thm:intro-functionality-localization}
    Let $\cal{F} = \{ f_k \}_k$ denote any copy-protectable functionality. Then there exists a protocol that (non-destructively) localizes the functionality $\cal{F}$ in the ideal obfuscation model. 
\end{theorem}








\subsection{The Ideal Obfuscation Model}\label{subsec:ideal-obf}

We now provide some more context on the ideal obfuscation model (also referred to as the classical oracle model) that we use to establish the security of our protocols. In this model, one can prepare and then ``obfuscate'' any polynomial-time computable classical functionality $f$, and thereafter all entities are  granted black-box (superposition) access to $f$. That is, anyone can apply the unitary $U_f : \ket{x}\ket{z} \to \ket{x}\ket{z \oplus f(x)}$ without learning anything else about $f$.

The ideal obfuscation model has a long history of study in quantum cryptography. In particular, it has been used to establish the feasibility of primitives for which no prior construction existed, e.g. publicly-verifiable quantum money \cite{10.1145/2213977.2213983}, signature tokens \cite{BenDavid2023quantumtokens}, copy-protection for unlearnable programs \cite{10.1007/978-3-030-84242-0_19}, witness encryption for QMA \cite{bartusek_et_al:LIPIcs.ITCS.2022.15}, and obfuscation for quantum circuits \cite{10.1145/3564246.3585179,BBV24,11369076,huang2026obfuscationarbitraryquantumcircuits}.

While ideal obfuscation of classical circuits has long been known to be impossible to achieve for certain contrived classes of functionalities \cite{JACM:BGIRSVY12}, the ideal obfuscation model remains a useful model in which to obtain feasibility results, for the following reasons. 

\begin{itemize}
    \item First, any construction in the ideal obfuscation model (that does not require obfuscating the contrived functions from \cite{JACM:BGIRSVY12}) yields a plausibly secure construction in the plain model by using \emph{indistinguishability obfuscation} (iO) to obfuscate the function $f$. While we may not have a formal reduction to the security of iO, it is a reasonable heuristic to expect that the plain model construction is secure. This is analogous to how the community proves security in the \emph{random oracle} model and then replaces the random oracle with a concrete hash function in practice, without necessarily having a formal reduction to any security property of the hash function.
    \item  Second, a recent work \cite{bartusek2025newapproachargumentsquantum}, building on \cite{10.1007/978-3-031-38551-3_8}, has shown that the (quantum-accessible) ideal obfuscation model can be instantiated from indistinguishability obfuscation plus the heuristic use of a \emph{hash function} (as opposed to an obfuscator) modeled as a ``pseudorandom oracle.'' This brings the ideal obfuscation model even closer to the widely-used random oracle model in that the only cryptographic object we need to treat heuristically is a hash function, which can be instantiated with a cryptographic hash such as SHA3.
\end{itemize}

We also mention that our use of oracles is to be expected, as our protocols are strengthenings of plain position-verification, which (in the unbounded entanglement setting) is only known in the random oracle model \cite{unruh2014pvqrom}. It would be considered a major breakthrough to prove the security of any of our protocols, or indeed position-verification itself, in the plain model.

Finally, we remark on a slight technical gap in the informal theorem statements above. While the ideal obfuscation model requires the circuit $f$ to be obfuscated to be polynomial-time computable, many of our constructions are most simply described as using a random oracle (plus other manipulations) to define $f$. Thus, in order to make $f$ efficient, we technically have to replace the random oracle with a (post-quantum) pseudorandom function, which is known from any post-quantum one-way function \cite{10.1145/3450745}. Thus, all of our results additionally assume one-way functions.

\subsection{\texorpdfstring{Non-Localizability of $f$-BB84}{Non-Localizability of f-BB84}}\label{sec:non-local-fbb84}

Now that we have stronger notions of position verification and have shown how to achieve them, it is natural to wonder whether these notions can be achieved with simpler protocols. We show that the well-studied $f$-BB84 protocol does \emph{not} satisfy entanglement localization. Let $f(x,y)$ be a random boolean function. In this protocol, the honest prover at the proper location $L$ receives classical challenges $x,y$ as well as half of an EPR pair. 
It is supposed to compute a bit $\theta = f(x,y)$ and measure its qubit in the standard or Hadamard basis according to $\theta$ to obtain a bit $b$, which it sends back in response. The verifier performs the same measurement on its half of the EPR pair and checks if it got the same outcome. 

Unruh showed that if $f$ is modeled as a random oracle, then unless the spoofers $\{ \cal{P}_1,\ldots,\cal{P}_k\}$ make $\exp(\Omega(n))$ queries to $f$ (where we think of $x,y$ as $n$-bit strings), then any successful spoofing strategy requires at least one of the provers $\cal{P}_i$ to be in the correct location $L$~\cite{unruh2014pvqrom}. Now, is it possible to strengthen Unruh's result to show that the EPR entanglement can be localized to where the honest prover was supposed to be? We show that, perhaps surprisingly, this is not possible. We describe one attack here and will actually describe a slightly different one in \Cref{sec:overview}.
\begin{figure}[H]
    \hspace{-2em}\resizebox{1.1\linewidth}{!}{%
    \begin{tikzpicture}[
        >=Latex,
        line width=0.9pt,
        line cap=round,
        line join=round,
        font=\small,
        sbox/.style={
            draw,
            rounded corners=3pt,
            minimum width=1.60cm,
            minimum height=0.88cm,
            inner sep=2.5pt,
            align=center
        },
        rbox/.style={
            draw,
            rounded corners=3pt,
            text width=2.05cm,
            minimum height=1.12cm,
            inner sep=3pt,
            align=center
        },
        mbox/.style={
            draw,
            dashed,
            rounded corners=3pt,
            text width=2.40cm,
            minimum height=1.50cm,
            inner sep=4pt,
            align=center
        },
        axis/.style={->, line width=0.9pt},
        msg/.style={->, line width=0.9pt},
        share/.style={line width=0.8pt},
        lab/.style={fill=white, inner sep=1.3pt, outer sep=0pt}
    ]

    \draw[axis] (-10.0,-0.90) -- (10.05,-0.90) node[right] {space};
    \draw[axis] (-10.0,-0.90) -- (-10.0,7.90) node[right] {time};

    \node (VL0) at (-9.10,0.35) {$V_L$};
    \node (VR0) at ( 9.10,0.35) {$V_R$};
    \node (VL1) at (-9.10,7.30) {$V_L$};
    \node (VR1) at ( 9.10,7.30) {$V_R$};

    \node[sbox] (PL0) at (-5.75,1.95) {$P_L$\\[-1mm]$r,s$};
    \node[sbox] (PR0) at ( 5.75,1.95) {$P_R$\\[-1mm]$r,s$};

    \node[mbox] (PM) at (0,4.12)
    {$P_M$\\[0.3mm]$\theta=f(x,y)$\\[0.3mm]$\mathsf{M}_{\theta}(X^rZ^sQ)=c$};

    \node[rbox] (PL1) at (-5.75,6.30)
    {$P_L$\\[-1mm]$\theta=0:\ c\oplus r$\\[-1mm]$\theta=1:\ c\oplus s$};

    \node[rbox] (PR1) at ( 5.75,6.30)
    {$P_R$\\[-1mm]$\theta=0:\ c\oplus r$\\[-1mm]$\theta=1:\ c\oplus s$};

    \coordinate (PL0in)  at ($(PL0.west)!0.30!(PL0.south west)$);
    \coordinate (PL0out) at ($(PL0.east)!0.22!(PL0.north east)$);
    \coordinate (PR0in)  at ($(PR0.east)!0.30!(PR0.south east)$);
    \coordinate (PR0out) at ($(PR0.west)!0.22!(PR0.north west)$);

    \coordinate (PMlin)  at ($(PM.west)!0.66!(PM.south west)$);
    \coordinate (PMrin)  at ($(PM.east)!0.66!(PM.south east)$);
    \coordinate (PMlout) at ($(PM.west)!0.66!(PM.north west)$);
    \coordinate (PMrout) at ($(PM.east)!0.66!(PM.north east)$);

    \coordinate (PL1in)  at ($(PL1.east)!0.30!(PL1.south east)$);
    \coordinate (PL1out) at ($(PL1.west)!0.22!(PL1.north west)$);
    \coordinate (PR1in)  at ($(PR1.west)!0.30!(PR1.south west)$);
    \coordinate (PR1out) at ($(PR1.east)!0.22!(PR1.north east)$);

    \draw[msg] (VL0) -- (PL0in);
    \draw[msg] (VR0) -- (PR0in);
    \draw[share]
        ($(PL0.south east)+(-0.08,-0.02)$)
        .. controls (-3.20,0.50) and (3.20,0.50) ..
        ($(PR0.south west)+(0.08,-0.02)$);
    \draw[msg] (PL0out) -- (PMlin);
    \draw[msg] (PR0out) -- (PMrin);
    \draw[msg] (PMlout) -- (PL1in);
    \draw[msg] (PMrout) -- (PR1in);
    \draw[msg] (PL1out) -- (VL1);
    \draw[msg] (PR1out) -- (VR1);

    \node[lab] at (-8.03,1.45) {$Q,x$};
    \node[lab] at ( 8.03,1.45) {$y$};
    \node[lab] at (-3.8,3.08) {$X^rZ^sQ,x$};
    \node[lab] at ( 3.8,3.08) {$y$};
    \node[lab] at (-2.15,5.58) {$(\theta,c)$};
    \node[lab] at ( 2.15,5.58) {$(\theta,c)$};
    \node[lab] at (-7.49,6.96) {$b$};
    \node[lab] at ( 7.49,6.96) {$b$};
    \node[lab] at (0,0.78) {pre-share $r,s$};

    \end{tikzpicture}%
    }
    \caption{Non-localizability attack on $f$-BB84}
    \label{fig:attack}
\end{figure}
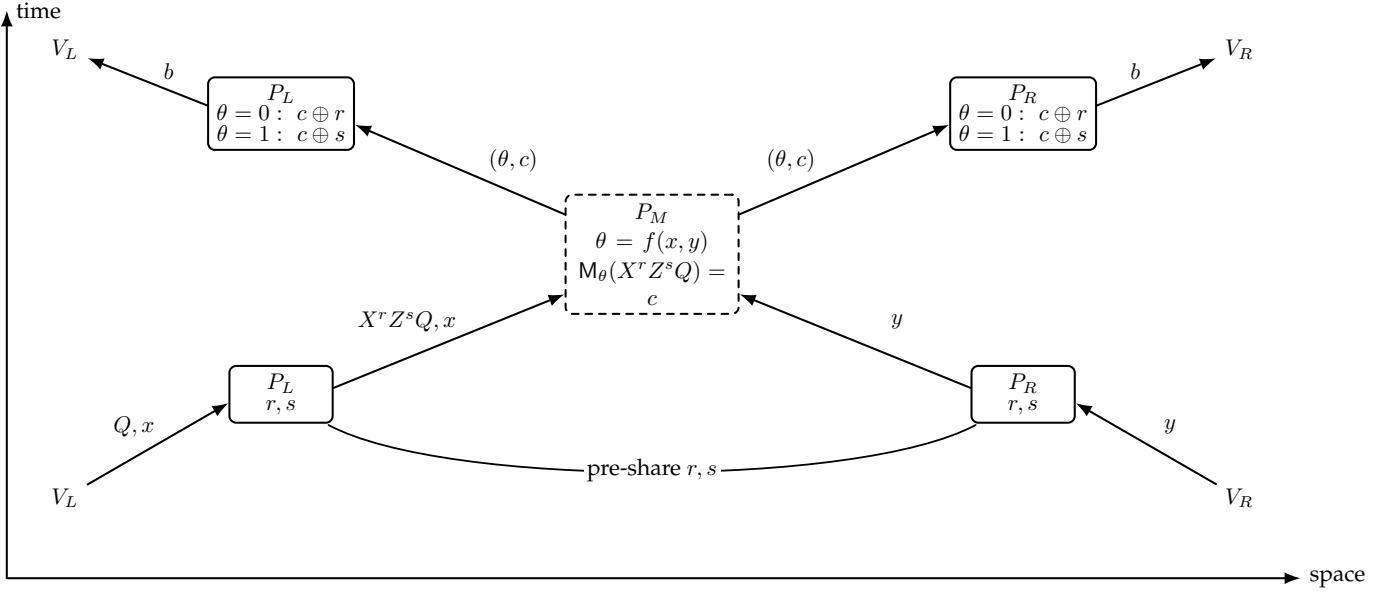

Consider the following $3$ prover strategy, pictured in \Cref{fig:attack}. There are provers $\cal{P}_L,\cal{P}_M,\cal{P}_R$ (for ``left'', ``middle'', and ``right''). Prover $\cal{P}_M$ is in the correct location but $\cal{P}_L,\cal{P}_R$ are on either side of $\cal{P}_M$. Suppose that $\cal{P}_L$ and $\cal{P}_R$ share a random quantum one-time pad key $r,s \in \{0,1\}$ in superposition. Suppose that the leftmost verifier sends half of the EPR pair; call this qubit $Q$. It gets intercepted by $\cal{P}_L$, who applies a quantum one-time pad $X^r Z^s$ to $Q$, and forwards it to $\cal{P}_M$. The middle prover $\cal{P}_M$ gets the one-time padded qubit, along with $x,y$. It computes $\theta = f(x,y)$, and measures the qubit in the corresponding basis to obtain a bit $c$. The middle prover $\cal{P}_M$ then sends $\theta,c$ to both $\cal{P}_L$ and $\cal{P}_R$. Both $\cal{P}_L$ and $\cal{P}_R$ compute $b$ as follows: if $\theta = 0$, then $b = c \oplus r$, and otherwise $b = c \oplus s$. All operations are controlled by the superposition of $r,s$. It is easy to check that this strategy satisfies the timing constraints and will be accepted with probability $1$.

Although the middle prover is performing the correct measurement, the quantum entanglement in the EPR pair cannot be localized to the middle prover. At all times in this strategy, the marginal state of the middle prover is completely uncorrelated with $\theta$ and the verifier's EPR qubit. Thus there is no way for the middle prover to locally extract the other end of the EPR pair; it has been nonlocally distributed between the three provers $\cal{P}_L,\cal{P}_M,\cal{P}_R$. 

Thus, (entanglement) localization is a rather strong security condition that doesn't hold for all QPV protocols studied in the literature.





\subsection{Outlook and Future Directions}

\paragraph{Position-based cryptography.} In light of our new notions of localization and trajectory verification, it will be interesting to revisit implications for position-based cryptography. As mentioned, there has been a subtle but important gap between the goals of position-based cryptography and the previous formalizations of position security. In \Cref{subsec:additional}, we illustrate this gap by presenting an ``attack'' on the position-based signature protocol of \cite{buhrman2014position} that does not violate their security definition but intuitively should not be allowed in any reasonable real-world application of position-based authentication. We then further discuss why we believe our concepts and techniques should provide a stronger foundation for position-based cryptography. We leave further exploration of new definitions and constructions of position-based cryptography to future work.



\paragraph{Future work.} Our work raises several other directions for future exploration, which we list here.

\begin{itemize}
    \item In this work, we restrict our attention to one spatial dimension. While some subtleties typically arise when generalizing to more dimensions (see e.g. \cite{unruh2014pvqrom}), we expect that our techniques should apply in higher dimensions, and we leave a formalization of this to future work.
	\item A recent work \cite{girish2026private} has shown how to build position \emph{commitments} and \emph{zero-knowledge} position-verification. Can we build on their techniques to obtain \emph{trajectory commitments} and \emph{zero-knowledge localization} protocols?
	\item \cite{liu2021beating,kaleoglu2025equivalenceclassicalpositionverification} have shown how to achieve position verification with purely classical communication. Can we obtain quantum localization and trajectory verification with only classical communication? 
	\item Can we further generalize the notion of quantum localization? For example, can we localize any (unclonable) QMA witness, perhaps by integrating our techniques with recent work \cite{bartusek2025newapproachargumentsquantum, kalai2026classicallyverifyquantumcat} that shows how to non-destructively verify QMA witnesses?
	\item Our localization and trajectory verification protocols all rely on highly entangled and difficult-to-implement states. Might there exist simpler protocols that are more amenable to experimental realization in the short term? We remark that simpler localization protocols are likely easier to achieve if we drop the non-destructive requirement, but such protocols would not serve as building blocks to trajectory verification.
	\item Our protocols assume noiseless transmission of quantum states, and that the prover is in exactly the right position at the right time. Can we achieve more \emph{robust} versions of localization and trajectory verification, which resist environmental noise and allow the prover some leeway in their declared position / trajectory?
    \item In this work, we discuss the way to achieve publicly verifiable trajectory verification and publicly verifiable localization schemes except publicly verifiable functionality localization (See \Cref{footnote:unknown_public_verifiable_functionality_localization} for more details). Can we construct a publicly verifiable functionality localization scheme?
\end{itemize}

\ifblind
\else
\paragraph{Acknowledgments.} We thank Tal Malkin, Alex May, and Saachi Mutreja for helpful discussions. HY is supported by the Sloan Foundation. LO is supported by an NSF Graduate Fellowship. 
\fi

%% file: technical_overview.tex
\section{Technical Overview}\label{sec:overview}

In this overview, we first introduce the main techniques used to construct trajectory verification. Along the way, we'll see how to perform entanglement and state localization. Then, we'll cover some additional ideas required to obtain functionality localization.

\subsection{Entanglement Localization and Trajectory Verification}

To develop our trajectory verification scheme, we first develop a \emph{non-destructive test of entanglement}, next upgrade it to \emph{entanglement localization}, and finally derive trajectory verification. A non-destructive test of entanglement distributes an entangled state between two parties, $\cal{V}$ (the verifier) and $\cal{P}$ (the prover), and then specifies an interactive protocol that can be repeatedly used to test whether the prover still holds the state entangled with the verifier. An entanglement localization scheme is a proof system where the prover can convince the verifier that the entanglement exists across the verifier's state and a state within a specific defined region, such as an interval $[L-\Delta, L+\Delta]$. We first consider how to verify entanglement in a non-destructive manner.

\paragraph{Is verifying entanglement always destructive?}
We begin with the standard approach for verifying one bit of entanglement. Suppose two parties $\cal{V}$ (an honest verifier) and $\cal{P}$ (a potentially adversarial prover) initially share an $\epr$ pair. The prover $\cal{P}$ may perform some operations on its subsystem, resulting in a joint state $\rho_{\areg\breg}$, where $\areg$ is held by $\cal{V}$ and $\breg$ is held by $\cal{P}$. As discussed for example in  \cite{vidick2021classicalproofsquantumknowledge}, the entanglement can be tested via the following protocol:
\begin{enumerate}
    \item $\cal{V}$ samples a random bit $b \rand \set{0,1}$. If $b=0$, it measures $\areg$ in the standard basis; otherwise, it measures $\areg$ in the Hadamard basis. Let the measurement outcome be $x$. The verifier then sends $b$ to $\cal{P}$.
    \item Upon receiving $b$, the prover $\cal{P}$ performs the corresponding measurement operation and outputs a value $x'$.
\end{enumerate}
If $\cal{P}$ can perfectly predict the standard or Hadamard basis measurement outcome on $\cal{V}$'s side, i.e., $x = x'$ with probability $1$ over the choice of $b$, then $\cal{V}$ is convinced that its qubit was maximally entangled with $\cal{P}$'s system. However, this procedure necessarily collapses the entanglement, and at first glance this appears unavoidable.

Our main observation here is the following:
\begin{center}
    {\it To verify entanglement, it suffices to argue the \textbf{existence} of a procedure that would pass the measurement test, without actually executing this procedure. In other words, there is a \textbf{non-collapsing test} such that, if the prover passes the non-collapsing test, it implies that the prover could be used to pass the collapsing test, which is only used in the analysis.}
\end{center}

More concretely, we design two tests $\Pi_{\sf ncol}$ and $\Pi_{\sf col}$ for $\cal{P}$. The test $\Pi_{\sf ncol}$ is the non-collapsing test that we actually execute. We then argue, via an indistinguishability argument, that any prover passing $\Pi_{\sf ncol}$ must also perform well in a corresponding collapsing test $\Pi_{\sf col}$, which certifies the presence of entanglement.

\paragraph{Deploying decoy qubits.}
We now present a first toy example. In addition to the original $\epr$ pair shared between $\cal{V}$ and $\cal{P}$, we sample $n$ random standard-basis qubits from $\set{\ket{0},\ket{1}}$ and $n$ random Hadamard-basis qubits from $\set{\ket{+},\ket{-}}$. These $2n$ qubits, together with the prover's half of the $\epr$ pair, are randomly shuffled into the register $\breg$, which now contains $2n+1$ qubits.

To test for entanglement, instead of querying the unique (hidden) entangled qubit directly, $\cal{V}$ proceeds as follows:
\begin{enumerate}
    \item Sample a random basis (standard or Hadamard).
    \item Send $\cal{P}$ the indices of the $n$ qubits prepared in that basis and request their measurement outcomes in the corresponding basis.
    \item Check whether the answers are correct.
\end{enumerate}

From the prover's perspective, all $2n+1$ qubits appear as random BB84 states and are thus maximally mixed. The $2n$ additional qubits serve as decoys that hide the location of the entangled qubit. To succeed, $\cal{P}$ must effectively retain all qubits, and hence preserve the entanglement. Discarding any qubit risks failing the test, since that qubit may be queried.

To formalize this intuition, consider the single-shot setting with a prover $\cal{P}$ that passes the protocol with probability 1. Let $S$ denote the indices of the standard-basis qubits, and consider some arbitrary position $x \in S$. We claim that $\cal{P}$ would also answer correctly if, instead of querying $S$, we queried $(S \setminus \set{x}) \cup \set{e}$, where $e$ is the index of the entangled qubit. Here, correctness for $e$ means consistency with the standard-basis measurement outcome on $\areg$. This follows because $S$ and $(S \setminus \set{x}) \cup \set{e}$ are indistinguishable from the prover's perspective.

This illustrates our earlier principle: the test on $S$ is non-collapsing, while the test on $(S \setminus \set{x}) \cup \set{e}$ is collapsing, as it measures the entangled qubit.

However, this approach has a clear limitation. If the test is repeated multiple times, $\cal{P}$ may learn which qubit is never queried, thereby identifying the entangled qubit. Once identified, the prover can discard it and still answer future queries correctly.

\paragraph{Reusable security via coset states.} To address this issue, we introduce a stronger method for hiding the decoys, based on \emph{coset states}. A coset state $\ket{A_{s,s'}}$, for a subspace $A \leq \mathbb{F}_2^n$ of dimension $n/2$ and $s, s' \in \mathbb{F}_2^n$, is defined as
\[
\ket{A_{s,s'}} = \frac{1}{\sqrt{|A|}} \sum_{a \in A} (-1)^{\ipd{a,s'}} \ket{a+s}.
\]
This state admits the following alternative interpretation:
\begin{enumerate}
    \item Prepare an $n$-qubit state in which the first $n/2$ qubits are random Hadamard-basis states, and the remaining $n/2$ qubits are random standard-basis states.
    \item Sample a random change-of-basis (invertible) matrix $U_{\sf shift} \in \mathbb{F}_2^{n \times n}$ and apply the unitary $\cal{U}_{\sf shift}$ defined as
    \[
    \cal{U}_{\sf shift}\ket{x} := \ket{U_{\sf shift} x}.
    \]
\end{enumerate}
The second step plays a crucial role: without knowledge of $U_{\sf shift}$, the subspace $A$ and the cosets $A+s$ and $A^\perp + s'$ are computationally hidden. In particular, given $\ket{A_{s,s'}}$, one cannot distinguish between:
\begin{itemize}
    \item Access to membership oracles $\cal{O}_{A+s}$ and $\cal{O}_{A^\perp+s'}$, and
    \item Access to membership oracles $\cal{O}_{B+s}$ and $\cal{O}_{C+s'}$, where $B$ and $C$ are random superspaces of dimensions $n - \omega(\log n)$ containing $A$ and $A^\perp$, respectively.
\end{itemize}
Here, $\cal{O}_S$ denotes the membership oracle for a set $S$, outputting $\top$ on inputs in $S$ and $\bot$ otherwise.

Intuitively, the random change of basis hides the ``positions'' of the standard- and Hadamard-basis qubits, in a strictly stronger manner than the random permutation from before. We now apply this idea to also hide the positions of \emph{entangled} qubits. Consider the following construction:
\begin{enumerate}
    \item Let $\anchor$ be an $n$-qubit register and $\vessel$ a $3n$-qubit register. Initialize the first $n$ qubits of $\vessel$ as random Hadamard-basis states, the next $n$ qubits as halves of $\epr$ pairs with $\anchor$, and the final $n$ qubits as random standard-basis states.
    \item Sample a random change-of-basis matrix $U_{\sf shift} \in \mathbb{F}_2^{3n \times 3n}$ and apply the unitary $\cal{U}_{\sf shift}$ defined as
    \[
    \cal{U}_{\sf shift}\ket{x} := \ket{U_{\sf shift} x}
    \]
    to the register $\vessel$.
\end{enumerate}
We refer to the resulting state as the \emph{anchor state} $\ket{\chain^\sk}_{\anchor\vessel}$, where $\sk = (U_{\sf shift},u,v)$.
\begin{enumerate}
    \item $S$ is the subspace spanned by the first $n$ columns of $U_{\sf shift}$.
    \item $T$ is the subspace spanned by the first $2n$ columns of $U_{\sf shift}$.
    \item $u := U_{\sf shift}^{-t} u'$ where $u'$ is the vector obtained by appending $2n$ zeros to the vector of Hadamard-basis values (where $\ket{+}$ is mapped to $0$ and $\ket{-}$ is mapped to $1$), and $U_{\sf shift}^{-t}$ is the inverse transpose of $U_{\sf shift}$.
    \item $v := U_{\sf shift} v'$ where $v'$ is the vector obtained by appending $2n$ zeros before the vector of the standard-basis values.
\end{enumerate}
We clarify some facts here. First, if we measure $\vessel$ in the standard basis, we will get a vector in $T + v$. If we measure $\vessel$ in the Hadamard basis, we will get a vector in $S^\perp + u$. Combining our previous non-destructive test with the hiding property achieved by oracles, we define the following test:
\begin{enumerate}
    \item In the setup phase, $\anchor$ is given to $\cal{A}$ and $\vessel$ is given to $\cal{B}$. Sample $\theta \rand \{0,1\}$ and $s \rand \{0,1\}^\secp$. $\theta$ is given to $\cal{B}$.
    \item If $\theta=0$, give the oracle $\cal{O}_{T+v}^{\sf ncol}$ to $\cal{B}$, where the oracle is defined as follows:
    \begin{itemize}
        \item $\cal{O}_{T+v}^{\sf ncol}$: On input $z$, check whether $z \in T + v$. Output $s$ if it is, and output $\bot$ otherwise.
    \end{itemize}
    If $\theta=1$, give the oracle $\cal{O}_{S^\perp+u}^{\sf ncol}$ to $\cal{B}$, where the oracle is defined as follows:
    \begin{itemize}
        \item $\cal{O}_{S^\perp+u}^{\sf ncol}$: On input $z$, check whether $z \in S^\perp + u$. Output $s$ if it is, and output $\bot$ otherwise.
    \end{itemize}
    $\cal{B}$ is supposed to return $s$.
\end{enumerate}
An honest $\cal{B}$ can simply perform the corresponding basis evaluation coherently. That is, for $\theta=0$, it computes $\cal{O}_{T+v}^{\sf ncol}$ coherently on register $\vessel$ and outputs the measurement result. For example, \Cref{fig:non-collapsing} shows the non-collapsing test when $U_{\sf shift}=I$ and $\theta=0$. Only the correctness of the standard-basis qubits is tested. One may notice that this is a non-collapsing test, and the state is unchanged after the test is completed.
\begin{figure}[h]
    
    \begin{tikzpicture}[scale=1]
    
    \node at (0, 5) {$\areg$:};
    \node at (0, 3) {$\breg$:};
    
    \node[draw, circle, minimum size=0.8cm] at (1, 3) (h0) {$+$};
    \node[draw, circle, minimum size=0.8cm] at (2, 3) {$-$};
    \node at (3, 3) {...};
    \node[draw, circle, minimum size=0.8cm] at (4, 3) {$+$};
    \node[draw, circle, minimum size=0.8cm] at (5, 3) (h3) {$+$};
    
    \draw[decorate,decoration={brace,amplitude=10pt,mirror}] 
        (h0.south) -- (h3.south) node[midway,below=6pt] {Hadamard-basis Qubits};
    
    \node[draw, circle, minimum size=0.8cm] (q0) at (6, 3) {};
    \node[draw, circle, minimum size=0.8cm] (q1) at (7, 3) {};
    \node at (8, 3) {...};
    \node[draw, circle, minimum size=0.8cm] (q2) at (9, 3) {};
    \node[draw, circle, minimum size=0.8cm] (q3) at (10, 3) {};
    
    \node[draw, circle, minimum size=0.8cm] (p0) at (6, 5) {};
    \node[draw, circle, minimum size=0.8cm] (p1) at (7, 5) {};
    \node at (8, 5) {...};
    \node[draw, circle, minimum size=0.8cm] (p2) at (9, 5) {};
    \node[draw, circle, minimum size=0.8cm] (p3) at (10, 5) {};

    \draw[decorate, decoration={snake,amplitude=1mm,segment length=5mm}] (p0) -- (q0);
    \draw[decorate, decoration={snake,amplitude=1mm,segment length=5mm}] (p1) -- (q1);
    \draw[decorate, decoration={snake,amplitude=1mm,segment length=5mm}] (p2) -- (q2);
    \draw[decorate, decoration={snake,amplitude=1mm,segment length=5mm}] (p3) -- (q3);
    
    \node[draw, circle, minimum size=0.8cm] at (11, 3) (s0) {$0$};
    \node[draw, circle, minimum size=0.8cm] at (12, 3) {$1$};
    \node at (13, 3) {...};
    \node[draw, circle, minimum size=0.8cm] at (14, 3) {$1$};
    \node[draw, circle, minimum size=0.8cm] at (15, 3) (s3) {$0$};
    
    \draw[decorate,decoration={brace,amplitude=10pt,mirror}] 
        (s0.south) -- (s3.south) node[midway,below=6pt] {Standard-basis Qubits};

    \node at (11.5,3.5) {$\checkmark$};
    \node at (12.5,3.5) {$\checkmark$};
    \node at (14.5,3.5) {$\checkmark$};
    \node at (15.5,3.5) {$\checkmark$};
    \end{tikzpicture}
    \caption{The Non-collapsing Test. In an actual test, all qubits in $\breg$ are shuffled.}
    \label{fig:non-collapsing}
\end{figure}
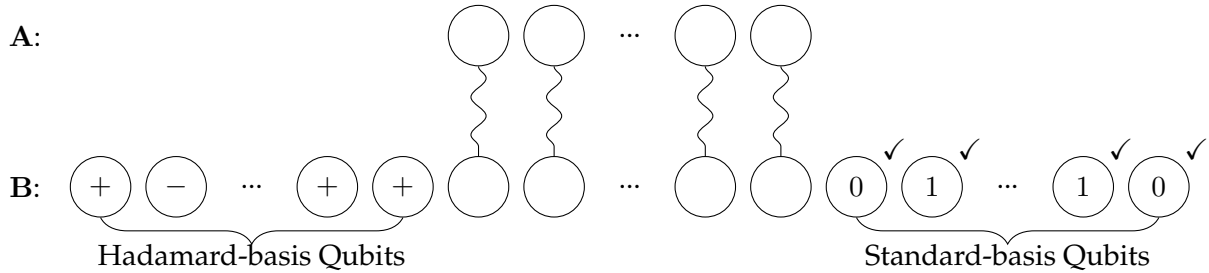

Our next step will be to define a \emph{collapsing} test that is designed to be indistinguishable from the non-collapsing test to $\cal{B}$. To show indistinguishability, we invoke an argument based on \emph{subspace bloating}, which says that a query-efficient algorithm cannot distinguish between membership oracles for a random subspace $S$, and a randomly sampled superspace $T\geq S$ -- more intuition is provided in \Cref{sec:collapsing-indistinguishability}.
\begin{enumerate}
    \item In the setup phase, $\anchor$ is given to $\cal{A}$ and $\vessel$ is given to $\cal{B}$. Sample $\theta \rand \{0,1\}$ and a random oracle $\cal{O}_{\sf random}$. $\theta$ is given to $\cal{B}$.
    \item If $\theta=0$, give the oracle $\cal{O}_{T+v}^{\sf col}$ to $\cal{B}$, where the oracle is defined as follows:
    \begin{itemize}
        \item $\cal{O}_{T+v}^{\sf col}$: On input $z$, check whether $z \in T + v$. Output \textcolor{red}{$\cal{O}_{\sf random}(\Can_S(z))$ where $\Can_S(z)$ is the canonical representation of coset $S+z$} if it is, and output $\bot$ otherwise.
    \end{itemize}
    If $\theta=1$, give the oracle $\cal{O}_{S^\perp+u}^{\sf col}$ to $\cal{B}$, where the oracle is defined as follows:
    \begin{itemize}
        \item $\cal{O}_{S^\perp+u}^{\sf col}$: On input $z$, check whether $z \in S^\perp + u$. Output \textcolor{red}{$\cal{O}_{\sf random}(\Can_{T^\perp}(z))$ where $\Can_{T^\perp}(z)$ is the canonical representation of coset $T^\perp + z$} if it is, and output $\bot$ otherwise.
    \end{itemize}
    \item If $\theta = 0$, the verifier measures $\anchor$ in the standard basis to obtain $x$. Otherwise, if $\theta = 1$, the verifier measures $\anchor$ in the Hadamard basis to obtain $x$.
\end{enumerate}
One can think of the canonical representation of a coset as a unique identifier for the coset. Now suppose that $\theta=0$ and the prover coherently computes $\cal{O}_{T+v}^{\sf col}$ on register $\vessel$. A measurement on the oracle output will collapse the state according to which coset of $S$ the vector $z$ belongs to, where $z$ is the standard-basis value on $\vessel$. This is equivalent to measuring all entangled qubits in the standard basis. For example, \Cref{fig:collapsing} shows the collapsing test when $U_{\sf shift}=I$ and $\theta=0$.
\begin{figure}[h]
    \begin{tikzpicture}[scale=1]
    
    \node at (0, 5) {$\areg$:};
    \node at (0, 3) {$\breg$:};
    
    \node[draw, circle, minimum size=0.8cm] at (1, 3) (h0) {$+$};
    \node[draw, circle, minimum size=0.8cm] at (2, 3) {$-$};
    \node at (3, 3) {...};
    \node[draw, circle, minimum size=0.8cm] at (4, 3) {$+$};
    \node[draw, circle, minimum size=0.8cm] at (5, 3) (h3) {$+$};
    
    \draw[decorate,decoration={brace,amplitude=10pt,mirror}] 
        (h0.south) -- (h3.south) node[midway,below=6pt] {Hadamard-basis Qubits};
    
    \node[draw, circle, minimum size=0.8cm] (q0) at (6, 3) {$1$};
    \node[draw, circle, minimum size=0.8cm] (q1) at (7, 3) {$1$};
    \node at (8, 3) {...};
    \node[draw, circle, minimum size=0.8cm] (q2) at (9, 3) {$0$};
    \node[draw, circle, minimum size=0.8cm] (q3) at (10, 3) {$0$};
    
    \node (m0) at (6, 4) {};
    \node (m1) at (7, 4) {};
    \node (m2) at (9, 4) {};
    \node (m3) at (10, 4) {};
    
    \node[draw, circle, minimum size=0.8cm] (p0) at (6, 5) {$1$};
    \node[draw, circle, minimum size=0.8cm] (p1) at (7, 5) {$1$};
    \node at (8, 5) {...};
    \node[draw, circle, minimum size=0.8cm] (p2) at (9, 5) {$0$};
    \node[draw, circle, minimum size=0.8cm] (p3) at (10, 5) {$0$};

    \draw[decorate, decoration={snake,amplitude=1mm,segment length=5mm}] (p0) -- (m0);
    \draw[decorate, decoration={snake,amplitude=1mm,segment length=5mm}] (p1) -- (m1);
    \draw[decorate, decoration={snake,amplitude=1mm,segment length=5mm}] (p2) -- (m2);
    \draw[decorate, decoration={snake,amplitude=1mm,segment length=5mm}] (p3) -- (m3);

    \draw[dashed] (5,4)--(11,4);
    \node at (5,4) {\Rightscissors};

    \draw[decorate, decoration={snake,amplitude=1mm,segment length=5mm}] (m0) -- (q0);
    \draw[decorate, decoration={snake,amplitude=1mm,segment length=5mm}] (m1) -- (q1);
    \draw[decorate, decoration={snake,amplitude=1mm,segment length=5mm}] (m2) -- (q2);
    \draw[decorate, decoration={snake,amplitude=1mm,segment length=5mm}] (m3) -- (q3);
    
    \node[draw, circle, minimum size=0.8cm] at (11, 3) (s0) {$0$};
    \node[draw, circle, minimum size=0.8cm] at (12, 3) {$1$};
    \node at (13, 3) {...};
    \node[draw, circle, minimum size=0.8cm] at (14, 3) {$1$};
    \node[draw, circle, minimum size=0.8cm] at (15, 3) (s3) {$0$};
    
    \draw[decorate,decoration={brace,amplitude=10pt,mirror}] 
        (s0.south) -- (s3.south) node[midway,below=6pt] {Standard-basis Qubits};

    \node at (11.5,3.5) {$\checkmark$};
    \node at (12.5,3.5) {$\checkmark$};
    \node at (14.5,3.5) {$\checkmark$};
    \node at (15.5,3.5) {$\checkmark$};
    \end{tikzpicture}
    \caption{The Collapsing Test. In an actual test, all qubits in $\breg$ are shuffled.}
    \label{fig:collapsing}
\end{figure}
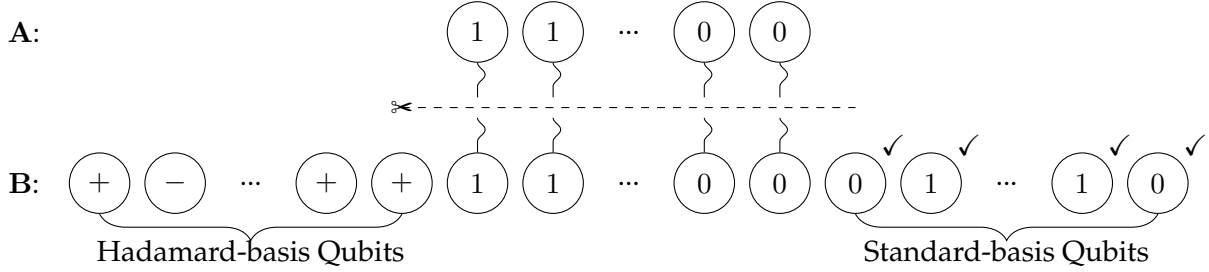

In the collapsing test, each potential measurement result $x$ obtained by the verifier corresponds to a different input on which the honest prover queried $\cal{O}_{\sf random}$. In the main proof, we show that \emph{any} (potentially adversarial) prover that passes with probability close to $1$ in the non-collapsing test will yield the following behavior in the collapsing test with overwhelming probability:
\begin{itemize}
    \item If $\theta=0$, the only input on which $\cal{O}_{\sf random}$ is accessed is $\Can_S(U_{\sf shift}(0^n \times x \times 0^n) + v)$, where $x$ is the standard-basis measurement result obtained by the verifier.
    \item If $\theta=1$, the only input on which $\cal{O}_{\sf random}$ is accessed is $\Can_{T^\perp}(U_{\sf shift}^{-t}(0^n \times x \times 0^n) + u)$, where $x$ is the Hadamard-basis measurement result obtained by the verifier.
\end{itemize}
Hence, we can test entanglement by repeatedly asking the prover to answer non-collapsing tests. If the prover is able to answer all of them, they would also succeed in the collapsing test, and we will be convinced that the prover is holding the entanglement. Indeed, a prover that passes the collapsing test will predict the verifier's standard  / Hadamard basis measurement results, allowing the extraction of EPR pairs following the strategy of \cite{vidick2021classicalproofsquantumknowledge}.
\begin{remark}
    In \cite{vidick2021classicalproofsquantumknowledge}, it was shown that classically-verifiable proofs of knowledge of an unclonable state cannot be non-destructive, otherwise an unbounded cloner would exist. Our discussion above can be used to obtain a non-destructive classically-verifiable \emph{argument} of knowledge for quantum money states. This does not contradict the result in \cite{vidick2021classicalproofsquantumknowledge} since soundness holds only against efficient malicious provers. 
    
\end{remark}

\paragraph{Non-localizability of $f$-BB84.} We now move to the second component mentioned above: Localization of entanglement. As discussed in the introduction, the famous $f$-BB84 protocol does not satisfy entanglement localization. Here we give an alternative attack on this protocol that better conveys the motivation for our own construction. 

Consider the following strategy involving three provers. There are provers $\cal{P}_L$, $\cal{P}_M$, and $\cal{P}_R$ (for ``left", ``middle", and ``right"). Prover $\cal{P}_M$ is in the correct location, while $\cal{P}_L$ and $\cal{P}_R$ are positioned on either side of $\cal{P}_M$. Upon receiving the $\epr$ pair $\ket{\Phi}_{\areg\breg}$, the prover applies the following unitary:
\[
\ket{x}_{\breg} \rightarrow \frac{1}{\sqrt{2}} \sum_{b \in \set{0,1}} \ket{b}_{\lreg} \ket{b}_{\mreg_0} \ket{b \oplus x}_{\mreg_1} \ket{b \oplus x}_{\rreg}
\]
and sends $\lreg$, $\mreg = \mreg_0 \otimes \mreg_1$, and $\rreg$ to $\cal{P}_L$, $\cal{P}_M$, and $\cal{P}_R$, respectively. It can be seen that there is no entanglement between $\areg$ and $\mreg$. To pass the challenge, the prover does the following.\\

\noindent If $\theta = 0$:
\begin{itemize}
    \item The middle prover $\cal{P}_M$ sends $\mreg_1$ to $\cal{P}_L$ and $\mreg_0$ to $\cal{P}_R$.
    \item The left prover $\cal{P}_L$ measures $\lreg$ and $\mreg_1$. If the bits are the same, it outputs $0$; otherwise, it outputs $1$.
    \item The right prover $\cal{P}_R$ measures $\rreg$ and $\mreg_0$. If the bits are the same, it outputs $0$; otherwise, it outputs $1$.
\end{itemize}
If $\theta = 1$:
\begin{itemize}
    \item The middle prover $\cal{P}_M$ sends $\mreg_0$ to $\cal{P}_L$ and $\mreg_1$ to $\cal{P}_R$.
    \item The left prover $\cal{P}_L$ outputs $0$ if its state is $\frac{1}{\sqrt{2}} \ket{00}_{\lreg \mreg_0} + \frac{1}{\sqrt{2}} \ket{11}_{\lreg \mreg_0}$; otherwise, it outputs $1$.
    \item The right prover $\cal{P}_R$ outputs $0$ if its state is $\frac{1}{\sqrt{2}} \ket{00}_{\rreg \mreg_1} + \frac{1}{\sqrt{2}} \ket{11}_{\rreg \mreg_1}$; otherwise, it outputs $1$.
\end{itemize}
Intuitively, this attack works because the left prover and the right prover ``encrypt'' the answer, such that the answer is stored in the interference between $\mreg$ and $\lreg \rreg$, while the middle prover has no clue about the answer when receiving the challenge. The middle prover does not measure anything, and the measurement is delayed until the global interference becomes local at the left and right provers.

\paragraph{Forcing a measurement via monogamy-of-entanglement.}

To bypass this attack, the key idea is to \emph{force} the prover to perform a particular ``measurement'' (in quotes because in reality they will be performing a non-destructive oracle query) at the position of the middle prover. More precisely, we want to ensure that the prover demonstrates its ability, in principle, to measure the state in a specific basis. We begin by replacing the random basis test on the $\epr$ pair with a non-destructive test for entanglement:

\begin{enumerate}
    \item In the setup phase, $\anchor$ is given to $\cal{A}$ and $\vessel$ is given to $\cal{B}$. Sample $s \rand \{0,1\}^\secp$.
    \item The left verifier $\cal{V}_L$ samples $x_L$, and the right verifier $\cal{V}_R$ samples $x_R$. Both $x_L$ and $x_R$ are broadcast at the appropriate time to ensure they meet at the correct position.
    \item The middle prover computes $\theta = f(x_L, x_R)$ and performs the following:
    \begin{itemize}
        \item If $\theta = 0$, it is granted access to the oracle $\cal{O}_{T+v}^{\sf ncol}$ and performs the coherent standard-basis measurement described above.
        \item If $\theta = 1$, it is granted access to the oracle $\cal{O}_{S^\perp + u}^{\sf ncol}$ and performs the coherent Hadamard-basis measurement described above.
        \item It broadcasts the result.
    \end{itemize}
    \item Both the left and right verifiers receive the result broadcast by the prover. The verifier returns $\top$ if both messages are on time and correct (equal to $s$); otherwise, it returns $\bot$.
\end{enumerate}

One may ask how the prover can be given access to an oracle during the middle of the computation. It can be assumed that there is a class of oracles indexed by $x_L, x_R$, and a particular oracle is selected for the parties once they know both $x_L$ and $x_R$. 

Now, actually proving the extractability of this simple protocol appears challenging. Our solution is to repeat this protocol twice: one standard round followed by one Hadamard round, or vice versa. For example, let's consider the case where the first round is a standard-basis round and the second round is a Hadamard-basis round.

\begin{enumerate}
    \item In the setup phase, $\anchor$ is given to $\cal{A}$ and $\vessel$ is given to $\cal{B}$. Sample $s_0, s_1 \rand \{0,1\}^\secp$.
    \item The first round has $\theta_0 = 0$. The prover is granted access to the oracle $\cal{O}_{T+v}^{\sf ncol}$ with $s_0$ encoded as the secret and performs the coherent standard-basis measurement described above.
    \item After a brief interval, the second round arrives with $\theta_1 = 1$. The prover is granted access to the oracle $\cal{O}_{S^\perp + u}^{\sf ncol}$ with $s_1$ encoded as the secret and performs the coherent Hadamard-basis measurement described above.
    \item Both the left and right verifiers receive the results broadcast by the prover. The verifier returns $\top$ if all messages are on time and correct; otherwise, it returns $\bot$.
\end{enumerate}



Now, let's examine how this two-round protocol prevents the attack described above. Suppose there are three provers $\cal{P}_L, \cal{P}_M, \cal{P}_R$ that succeed at passing the protocol, but where only $\cal{P}_M$ is at the correct spacetime location. Furthermore, suppose that $\cal{P}_M$ fails to ``measure'' and determine the answer for the first challenge at the correct location and time. In this case, the strategy can be recast as a successful adversary $\cal{A} = (\cal{A}_L, \cal{A}_M, \cal{A}_R)$ in the following monogamy-of-entanglement type game.\footnote{We note that the description of this game is slightly inaccurate, and we refer the reader to \Cref{sec:monogomy}, and in particular \Cref{lem:asymmetric_oracle_monogamy}, for formal details.} 

\begin{enumerate}
    \item The challenger generates $\ket{\chain^{\sk}}$ and gives the register $\vessel$ to $\cal{A}_M$.
    \item $\cal{A}_M$ generates a bipartite state on $\lreg\rreg$. It sends $\lreg$ to $\cal{A}_L$ and sends $\rreg$ to $\cal{A}_R$.
    \item The challenger samples $s_0, s_1 \rand \{0,1\}^\secp$. It gives $\cal{O}_{T+v}^{\sf ncol}$ to $\cal{A}_L$ that encodes $s_0$ as its underlying secret, and it gives $\cal{O}_{S^\perp + u}^{\sf ncol}$ to $\cal{A}_R$ that encodes $s_1$ as the underlying secret. In addition, it gives $s_0$ to $\cal{A}_R$.
    \item $\cal{A}_L$ returns $s_0'$ and $\cal{A}_R$ returns $s_1'$.
    \item The challenge is passed if $s_0' = s_0$ and $s_1' = s_1$.
\end{enumerate}

Intuitively, the reduction works as follows: $\cal{P}_M$ corresponds to $\cal{A}_M$, the left prover $\cal{P}_L$'s process for solving the $\cal{O}_{T+v}^{\sf ncol}$ challenge is used to construct $\cal{A}_L$, and the right prover $\cal{P}_R$'s process for solving the $\cal{O}_{S^\perp + u}^{\sf ncol}$ challenge is used to construct $\cal{A}_R$. Due to the structure of the game, we also have to include the correct ``left'' answer $s_0$ in the view of the right prover, which is a non-standard variant of monogamy-of-entanglement. 

In \Cref{sec:monogomy}, we show, via a sequence of reductions, that this game is hard for any adversary to win with probability better than $\frac{1}{2} + \negl(\secp)$. Thus, we can eliminate this type of prover, meaning that any successful prover must indeed ``measure'' (meaning query the oracle on its state) at the correct location and time. Finally, we argue that this implies the existence of an \emph{extractor} that, using quantum recording techniques \cite{zhandry2019record}, can obtain either the result of measuring $\vessel$ in the standard or the Hadamard basis at the correct location and time. By the previous entanglement verification arguments, we can thus localize the entanglement to the desired spacetime location.

\paragraph{Trajectory verification.} As mentioned earlier, trajectory verification follows fairly naturally from our entanglement localization protocol. At a high level, the verifier first sends the vessel register $\vessel$ to the prover, and then runs the entanglement localization protocol at a series of steps along the prover's claimed trajectory. In slightly more detail, at each step, the verifier sends a random standard- or Hadamard- basis challenge.\footnote{We note that one could alternatively specify that the verifier sent strictly alternating challenges, but this requires it to keep some state in between each challenge, which may be undesirable.} Therefore, each \emph{pair} of steps gives rise to a slightly modified version of the two-round protocol described above, where the cadence is not fixed standard followed by Hadamard, but rather both are randomly chosen. We show that in this modified game, one can still extract entanglement from any successful prover at the time of the first challenge. This means that one can extract entanglement from the prover at any step on their trajectory, assuming they have a high probability of passing all challenges. This is formalized in \Cref{sec:trajectory}.

\paragraph{Optimizations and applications.} We conclude by highlighting two additional features of our constructions. First, while the informal guarantees stated above only yield meaningful extraction when the adversary succeeds with probability $1-o(1)$, this limitation can be overcome via sequential repetition. By executing sufficiently many copies of the localization protocol in rapid succession, we can amplify soundness so that any prover accepted with non-negligible probability must still contain the localized object at the verified spacetime point. In this way, our localization guarantees can be strengthened to hold against adversaries whose success probability is merely inverse polynomial. 

Second, we discuss how to upgrade our protocols to achieve \emph{public verifiability}. Intuitively, because the verification procedure is implemented using public classical oracles, one can arrange the protocol so that anyone with oracle access can verify the prover’s responses, without requiring access to any secret verification key. We defer the precise formulation and construction of publicly-verifiable localization protocols to \Cref{subsec:public_verifiability}. After that, in \Cref{subsec:additional}, we discuss broader implications for position-based cryptography and related applications.

\subsection{Functionality Localization}

Functionality localization can be seen as a strengthening of quantum \emph{copy-protection}, which enables encoding a function $f$ (sampled from some family $\cal{F}$) into a quantum state $\rho_f$ such that (i) $\rho_f$ can be used to compute $f(x)$ for any input $x$, and (ii) no adversary can create two disjoint registers that can simultaneously be used to compute $f$ (over some distribution $\cal{D}$ on inputs $x$). That is, copy-protection establishes that certain functionalities can be rendered unclonable, meaning they cannot exist in two places at once. This then raises the possibility that functionalities can be \emph{localized}! In this work, we show that any functionality $\cal{F}$ that can in principle be copy-protected, can also be localized, in the classical ideal obfuscation model.\footnote{Technically, we say that a functionality is ``copy-protectable'' if there exists some distribution $\cal{D}$ over inputs such that $\cal{F}$ is copy-protectable with respect to independent challenges $x_0,x_1$ sampled from $\cal{D}$. See \Cref{subsec:functionality} for more details.}

Our definition of functionality localization extends copy-protection by adding a new security property, \emph{extraction soundness}, in addition to the usual unclonability guarantee. Extraction soundness, in a similar fashion to the notion of extractability for our other primitives, ensures that a successful prover can be converted into an evaluator for the copy-protected functionality by means of a physical procedure called the extractor. Here, the extractor is given access to the registers of a prover in the functionality localization protocol at a particular time, and should be able to evaluate the functionality on random inputs with high probability. Combined with the typical copy-protection security (which we term ``uniqueness''), this guarantees that a successful prover must possess the \emph{only} working implementation of the functionality at a specific point in spacetime. As discussed later on, this can have interesting applications to quantum cryptographic protocols, such as position-based authentication.

Our construction of functionality localization combines the localization techniques introduced in the previous sections with ``best-possible'' copy-protection. A best-possible copy-protector is a scheme that takes as input a function $f \gets \cal{F}$ and outputs a state $\rho_f$ such that, if there \emph{exists} any method of copy-protecting $f$, then $\rho_f$ is itself a copy-protected version of $f$. The idea of best-possible copy-protection was introduced by \cite{CG24}, who showed that any ``quantum state obfuscator'' serves as a best-possible copy protection scheme. Quantum state obfuscation was then shown to exist in the classical oracle model \cite{BBV24}.

Let $\QSO$ be a quantum state obfuscator, let $(\Authenticate,\Ver)$ be a message authentication code, and consider the following (simplified) proposal for functionality localization.

\begin{itemize}
    \item In the setup phase, the verifiers sample a functionality $f \gets \cal{F}$, an authentication key $\sk$, and a random coset state $\ket{S_{v,u}}$. Let $\widehat{f}_\sk$ be the functionality that, on input $x$, outputs $f(x)$ along with a signature $\sigma$ on the pair $(x,f(x))$. Output $\rho_f \gets \QSO(\widehat{f}_\sk)$ and $\ket{S_{v,u}}$.
    \item In the online phase, the verifiers sample $\theta \gets \{0,1\}, x \gets \cal{D}$, and a random oracle $\cal{O}_{\sf random}$, and do the following.
    \begin{itemize}
        \item If $\theta = 0$, they send a secret-sharing of $\theta,x$ to the prover along with an oracle $\cal{O}_{\sk,S,v}$ that takes as input $(z,x,y,\sigma)$ and if (i) $z \in S+v$ and (ii) $\Ver(\sk,(x,y),\sigma) = \top$, outputs $\cO_{\sf random}(x,y)$.
        \item If $\theta = 1$, they send a secret-sharing of $\theta,x$ to the prover along with an oracle $\cal{O}_{\sk,S^\perp,u}$ that takes as input $(z,x,y,\sigma)$ and if (i) $z \in S^\perp+u$ and (ii) $\Ver(\sk,(x,y),\sigma) = \top$, outputs $\cO_{\sf random}(x,y)$.
    \end{itemize}
    \item The honest prover, sitting at the correct spacetime point, uses $\rho_f$ to compute $f(x), \sigma_{x,f(x)}$, and then queries its oracle using either $\ket{S_{v,u}}$ in the standard or Hadamard basis to obtain the value $\cal{O}_{\sf random}(x,f(x))$, which it sends back to the verifier.
\end{itemize}

We note that there are several ways in which our final protocol differs from this simplified description, and we refer the reader to \Cref{subsec:functionality} for more details.

The intuition here is that, in order to succeed in the protocol, the prover \emph{must} query their oracle at the correct location on $z,x,f(x),\sigma_{x,f(x)}$, where $z$ is either in $S+v$ or $S^\perp + u$. On the other hand, due to the security of the message authentication code and the hiding of $\QSO$, they will not be able to learn a signature $\sigma_{x,y}$ on \emph{any} $y \neq f(x)$. Hence, the oracle $\cal{O}_{\sf random}$ will be queried on $(x,f(x))$ but not on $(x,y)$ for any $y \neq f(x)$. This yields the existence of an \emph{extractor} that, using quantum recording techniques, can extract the value of $f(x)$ from any successful adversarial prover at the location to be verified. We note that we crucially rely on the techniques developed above in the context of entanglement localization in order to establish that the adversary must indeed query their oracle at the correct spacetime location. In particular, even though there is no entanglement in this setting, we can view a random coset state as a ``degenerate'' version of the anchor state described above with the number of EPR pairs set to 0, and our techniques carry over naturally to this setting.

%% file: preliminary.tex
\section{Preliminaries}
\label{sec:prelims}
\subsection{General Notation}
In many contexts, objects are parameterized by a security parameter $\secp\in\N$, which is often omitted to reduce clutter. We write PPT and QPT to denote probabilistic polynomial-time and quantum polynomial-time, respectively. Algorithms are denoted with calligraphic letters, such as $\cA,\cB,\cal{P},\cal{V}$. Quantum registers are denoted with bold letters, such as $\areg,\breg,\xreg$. We write $\negl(\secparam)$ to denote a negligible function in $\secparam$. For a protocol $\sf Prot$ and two parties $\cal{P}$ and $\cal{V}$ (prover and verifier), we write $Y_{\cal{P}},Y_{\cal{V}}\gets{\sf Prot}(\cal{P}(\areg)\rightleftharpoons \cal{V}(\breg))(x)$ to denote the experiment where $\cal P$ and $\cal V$ interact by running $\sf Prot$ with public input $x$, private prover input $\areg$, and private verifier input $\breg$. Here, $Y_{\cal P}$ is some set of prover outputs and $Y_{\cal{V}}$ is some set of verifier outputs. $Y_{\cal V}$ typically includes a bit $b\in\{\top,\bot\}$ known as the acceptance decision, with $\top$ meaning the verifiers accept and $\bot$ meaning that they reject.
\subsection{Modeling Computations and Interactions in Spacetime}
In this section, we describe our model for how algorithms and interactive protocols take place in physical spacetime, which will be important for the interpretation of our main results. This model is somewhat new to our work, but is inspired by similar sections in \cite{CGMO09PBC},\cite{unruh2014pvqrom}, and \cite{girish2026private}. Note that later sections of this paper, as well as our main results, do not rely closely on the wording of this section --- instead, they tend to use terminology of a broader and more standard flavor. We strongly believe that our results are general enough that there is a wide variety of acceptable modeling choices under which they can retain correctness and meaningful interpretability.

\paragraph{Spacetime.} All of our protocols are set in a bounded one-dimensional space\footnote{It is natural to ask whether our protocols can be extended to work in standard three-dimensional space. We expect that, similarly to \cite{unruh2014pvqrom}, there is a natural generalization to three dimensions which can be then compiled to a more abstracted causal circuit, and analyzed without dealing with three-dimensional geometry. However, we leave this generalization as an open problem.}, represented as $[-1,1]$. This is the line segment between two statically placed verifiers, $\cal V_L$ at position $-1$ and $\cal V_R$ at position $1$. We model time as $\R^{\geq0}$, and we assume that all parties have access to a synchronized clock that reports the current time $t \in \R^{\geq0}$. Units of space and time are taken such that light travels 1 unit of space in 1 unit of time. A \emph{point in spacetime} is a pair $(L,t) \in [-1,1]\times\R^{\geq0}$, combining a spatial point with a time.

\paragraph{Messages.} 
We always consider messages to be directional --- i.e. they travel either only to the left, or only to the right --- and to travel with speed exactly 1 (the speed of light). We treat this process as the movement of a physical quantum register through space at the speed of light. In general, a message is a quantum state, although we often distinguish between the classical and quantum parts of a given message (for example, a party may send a string $x$ and a qubit $\ket{\psi}$ in the same message).
\paragraph{Algorithms.}
We define the following model of computation in spacetime. An algorithm consists of a finite number of movable ``cells'', which are quantum registers labeled $\rreg_1,\rreg_2,\dots,\rreg_M$. Corresponding to each cell is a predetermined, continuous trajectory through space, which we'll label $T_1,\dots,T_M:\R^{\geq0}\to[-1,1]$. $T_i(t)$ denotes the spatial location of cell $\rreg_i$ at time $t$. Cells can move at any speed up to the speed of light. Inputs are given to an algorithm by setting the initial value of some cells designated as input cells. Additionally, for each message that could possibly be sent or received by the algorithm, we add a cell which travels along the same path as this message, in perpetuity or until the message goes ``out of bounds'', in which case the cell lingers at $-1$ or $1$. Prior to the departure time of a message, its corresponding cell will simply be dormant in the message's starting location. The reason for this addition will become clear in the next paragraph.

The computational units of an algorithm are contained at the intersection points between its cells\footnote{This restriction is without loss of generality. Since we assume -- as is standard in QPV literature -- that all computation is instantaneous, any computation which is being done locally at one cell, whether the cell is stationary or moving, can always be freely pushed forward or backward to the nearest intersection points with other cells.} --- let $X$ be the set containing all of these points. At each intersection point $(L,t)\in X$, the algorithm contains a quantum circuit $C_{L,t}$ acting on all of the cells which intersect there. $C_{L,t}$ can implement any unitary on these cells, and the output registers can be arbitrarily subdivided among the participating cells. The size of a circuit, denoted $|C_{L,t}|$, refers to the number of gates used to implement it, taken from some fixed universal quantum gate set. No matter the size of a circuit, it is always assumed to be an instantaneous operation. An algorithm is said to eventually \emph{terminate} if $X$ is finite, and it is said to have terminated whenever the last intersection point is reached. At termination, some of its cells are designated as outputs, and the rest discarded. As mentioned above, all messages are also treated as cells --- to send a message, the algorithm transfers some information into the cell having that message's trajectory, which occurs at the intersection point corresponding to the message departing. To receive a message, an algorithm might place a SWAP circuit between a message cell and another ancilla cell, for example.

We sometimes specify that a family of algorithms $\{\cA_\secp\}_\secp$ is QPT, which is defined with respect to a security parameter, and abbreviated as ``$\cal A$ is QPT'' if dependence on $\secp$ is clear from context. Let $X_\secp$ and $\{C_{\secp}^{L,t}\}_{(L,t)\in X_\secp}$ be the sets of intersection points and circuits, respectively, for $\cA_\secp$. QPT is then shorthand for the following requirement: there exists some polynomial $p(\secp)$ such that for all $\secp\in\N$,
\begin{enumerate}
    \item $\cA_\secp$ terminates, and
    \item $\sum_{(L,t)\in X_\secp}|C_{\secp}^{L,t}|\leq p(\secp)~.$
\end{enumerate}

\paragraph{Positional Protocols.}
Our protocols take place between two kinds of parties: provers and verifiers. As stated above, the verifiers always consist of two parties $\cal V_L$ and $\cal V_R$, whose locations are publicly known to be $-1$ and $1$, respectively. Provers are, in general, allowed to have circuits anywhere in $[-1,1]$. The protocol typically includes a $\Setup$ operation, whose outputs are distributed between these parties unequally. We clarify this distinction, as well as other distinctions between provers and verifiers, below:
\begin{itemize}
    \item All verifiers can:
    \begin{itemize}
        \item Take as input some secret information $\sparam$ from $\Setup$, in addition to the public information $\param$. In our protocols, these are both classical.
        \item Send classical or quantum messages to any position $(L,t)$.
        \item Send classical or quantum messages to other verifiers through private, trusted channels.
    \end{itemize}
    \item A prover can:
    \begin{itemize}
        \item Take as input some prover-specific information from $\Setup$. In our protocols, this will always be a quantum state $\rho$.
        \item Send classical or quantum messages to any position $(L,t)$.
        \item Send classical or quantum messages between its various distributed components through private, trusted channels.
    \end{itemize}
    \item We assume that the prover's trajectory is a continuous function $L(\cdot)$ with $L(t)\in[-1,1]$ and $\big|\frac{L(t_1)-L(t_0)}{t_1-t_0}\big|\leq 1$ for all $t,t_0,t_1\in\R^{\geq0}$. The second condition is just the speed of light constraint, as measured by the secant line of the trajectory.
\end{itemize}

The following lemma is imported from \cite{unruh2014pvqrom}. It states that no adversary can distinguish an oracle reprogramming on the position decided by the xor of random domain element sampled by $\cal{V}_L,\cal{V}_R$ if one of the domain elements hasn't reach the adversary.
\begin{lemma}[\cite{unruh2014pvqrom}]\label{lem:oracle_reprogram}
    Let $\cal{O}_{\sf random}$ be a random oracle with exponential domain and arbitrary image size and the size is a power of $2$. Any QPT prover cannot distinguish between the following two cases with noticeable probability:
    \begin{itemize}
        \item $\cal{V}_L$ samples a uniform random vector $x_l$ and broadcasts it at time $t_l$. $\cal{V}_R$ samples a uniform random vector $x_r$ and broadcasts it at time $t_r$. 
        \item $\cal{V}_L$ samples a uniform random vector $x_l$ and broadcasts it at time $t_l$. $\cal{V}_R$ samples a uniform random vector $x_r$ and broadcasts it at time $t_r$. Then the oracle $\cal{O}_{\sf random}$ is reprogrammed to a uniform random output $y$ on input $x_l\oplus x_r$ for all space-time coordinates $(p,t)$ within the region:
        \begin{itemize}
            \item $t-t_l\geq p+1$.
            \item $t-t_r\geq 1-p$.
        \end{itemize}
    \end{itemize}
\end{lemma}
\subsection{Specifying Quantum Information in Spacetime}
We introduce notation, $\register[\cdot](\cdot)$, for convenience in formalizing many of our extraction-based security definitions in this paper. A definition is given below.
\begin{definition}
    Let ${\sf Exp}_{\cal P,V}$ refer to some experiment, taking place between a prover $\cal P$ and verifier $\cal V$, and let $R_1,\dots,R_k$ be disjoint regions of spacetime. Then, we use the following notation,
    \[\areg_1,\dots,\areg_k\gets\register[R_1;\dots;R_k]({\sf Exp}_{\cal P,V})~,\]
    to mean that for $i\in[k]$, $\areg_i$ is the tensor product of the following:
    \begin{enumerate}
        \item All of $\cal P$'s registers within region $R_i$.
        \item All message registers currently inflight within region $R_i$.
        \item The transcript of all classical messages sent by $\cal V$ which causally precede any points in $R_i$. In other words, any verifier messages which would be ``heard'' within $R_i$.
    \end{enumerate}
\end{definition}
Note that this may not correspond to a physical operation -- in particular, if $R$ contains points which are at the same spatial location at different times, these might have states which cannot physically coexist -- but we still allow this to be defined for notational convenience. Also, when convenient, we will use the notation ``$S~@~t$'', where $S$ is a spatial region and $t$ a time, to refer to the spacetime region $S\times\{t\}$.
\subsection{Compressed Oracle}
In this subsection, we recall the technique introduced by Zhandry \cite{zhandry2019record}. For more details please refer to \cite{zhandry2019record}. The following part is adapted from~\cite{hao2026needquantummemoryshort}. We will show two equivalent oracle forms: the standard oracle and the compressed oracle. Note that Fourier basis is considered in other papers that use this technique but here we use Hadamard basis instead for simplicity.

We first model an oracle quantum algorithm. It consists of the following registers:
\begin{itemize}
    \item $\xreg$ is the register that stores either an oracle query or an answer waiting to be written.
    \item $\ureg$ is the register that stores the oracle's response or is used to store phase for the output process (will explain later).
    \item $\wreg$ is the register that stores ancilla qubits in the computation.
\end{itemize}
\paragraph{Standard oracle.} Let $\cal{O}:\mathbb{F}_2^n \rightarrow \mathbb{F}_2^m$ be a random oracle. We can view an algorithm that runs in the random oracle model with respect to $\cal{O}$ as the algorithm itself concatenated with a random oracle register $\dreg$ that is initialized to $\sum_{\cal{O}}\ket{\cal{O}}\bra{\cal{O}}_{\dreg}$ (ignoring the normalizing factor). The register $\dreg$ stores the random function $\ket{\cal{O}}_{\dreg} = \ket {\cal{O}(0^n)} \ket {\cal{O}(0^{n-1}1)} \cdots \ket {\cal{O}(1^n)}$. The oracle unitary $\sto$ can be written as follows:
\begin{align*}
    \sto \ket{x}_{\xreg} \ket {u}_{\ureg} \ket {w}_{\wreg} \otimes \ket {\cal{O}}_{\dreg} = \ket{x}_{\xreg} \ket {u + \cal{O}(x)}_{\ureg} \ket {w}_{\wreg} \otimes \ket {\cal{O}}_{\dreg},  
\end{align*}
The following lemma shows that the output distribution using a standard oracle is exactly the same as using a random oracle. 

\begin{lemma}[{\cite[Lemma 2]{zhandry2019record}}] \label{lem:purifyoracle}
    Let $\As$ be an (unbounded) quantum algorithm making oracle queries. The output of $\As$ given a random function $\cal{O}$ is exactly identical to the output of $\As$ given access to a standard oracle. 
    Therefore, a random oracle with quantum query access can be perfectly simulated as a standard oracle. 
\end{lemma}

\paragraph*{Compressed oracle.} 
The compressed oracle can be viewed a type of lazy sampling technique. Instead of initializing $\cal{O}$ at the very beginning, the compressed oracle creates a database $\ket{D}_{\dreg}=\ket{D(0^n)}_{\dreg_{0^n}}\ket{D(0^{n-1}1)}_{\dreg_{0^{n-1}1}}\cdots\ket{D(1^n)}_{\dreg_{1^n}}$ where $D(x)\in\mathbb{F}_2^m\cup\set{\bot}$ and $\ket{D}$ is initialized to $\ket{\emptyset}_{\dreg}=\ket{\bot, \bot, \cdots, \bot}$ where $\bot$ is a symbol that indicates the lack of information of the algorithm on certain function value. Let $|D|$ denote the number of entries in $D$ that are not $\bot$. 
The database is initialized as an empty list $D_0$ of length $N$, in other words, it is initialized as the pure state $\ket{\emptyset} := \ket{\bot, \bot, \cdots, \bot}$. 

For any $D$ and $x$ such that $D(x) = \bot$, we define $D'=D \cup (x, u)$ to be the database $D'$, such that for every $x' \ne x$, $D'(x') = D(x)$ and at the input $x$, $D'(x) = u$. 

The compressed oracle is the unitary $\csto := \stddecomp \cdot \csto' \cdot \stddecomp$, where
\begin{itemize}
    
    \item $\csto'$ writes $D(x)$ to the answer register $\ureg$ by writing $u+ D(x)$ into it when $D(x)\neq \bot$ as usual but does nothing when $D(x)=\bot$. Or to say that we can define addition for $\bot$: $u+\bot=u$, $\forall u\in\mathbb{F}_2^m$. 
    Formally,
    \begin{equation*}
        \csto'\ket{x}_{\xreg}\ket{u}_{\ureg}\ket{w}_{\wreg}\otimes\ket{D}_{\dreg} = \ket{x}_{\xreg}\ket{u+ D(x)}_{\ureg}\ket{w}_{\wreg}\otimes\ket{D}_{\dreg}.
    \end{equation*}
    
    \item When the algorithm queries, the database calls $\stddecomp$ which unfolds the database and samples a value $y$ for positions that the algorithm does not know what the value is. More specifically, $\stddecomp \ket{x}_{\xreg}\ket{u}_{\ureg}\ket{w}_{\wreg}\otimes\ket{D}_{\dreg} := \ket{x}_{\xreg}\ket{u}_{\ureg}\ket{w}_{\wreg}\otimes\stddecomp_x\ket{D}_{\dreg}$, where $\stddecomp_x$ works on $\dreg_x$. 
    \begin{itemize}
        \item If $D(x) = \bot$, $\stddecomp_x$ maps $\ket \bot$ to 
        \[
        \frac{1}{\sqrt{2^m}} \sum_{y\in\mathbb{F}_2^m} \ket y.
        \]
        
        \item If $D(x) \ne \bot$, $\stddecomp_x$ works on the $x$-th register, and it is an identity on 
        \[
        \frac{1}{\sqrt{2^m}} \sum_{y\in\mathbb{F}_2^m} (-1)^{\langle u,y\rangle} \ket y
        \]
        for all $u \ne 0$; it maps the uniform superposition $\frac{1}{\sqrt{2^m}} \sum_{y\in\mathbb{F}_2^m} \ket y$ to $\ket \bot$. 
        
        More formally, for a $D'$ such that $D'(x) = \bot$, 
        \begin{align*}
            \stddecomp_x  \frac{1}{\sqrt{2^m}} \sum_{y\in\mathbb{F}_2^m} (-1)^{\langle u,y\rangle} \ket{D' \cup (x, y)}_{\dreg}
                        = \frac{1}{\sqrt{2^m}} \sum_{y\in\mathbb{F}_2^m} (-1)^{\langle u,y\rangle}  \ket{D' \cup (x, y)}_{\dreg}
        \end{align*}
        for any $u\ne 0$ and, 
        \begin{align*}   
            \stddecomp_x  \frac{1}{\sqrt{2^m}} \sum_{y\in\mathbb{F}_2^m} \ket{D' \cup (x, y)}_{\dreg}  =\ket{D'}_{\dreg}.
        \end{align*}
    \end{itemize}
    Intuitively, it swaps a uniform superposition $\frac 1{\sqrt{2^m}} \sum_{y\in\mathbb{F}_2^m} \ket y$ with $\ket \bot$ on $\dreg_x$ and does nothing on other orthogonal basis. So it is a well defined unitary.
\end{itemize}
Note that the compressed oracle can be implemented in an efficient way where we don't create $\dreg_x$ for all $x$ and we only store all non-$\bot$ entries. For more details, please refer to \cite{zhandry2019record}. Zhandry proves that, ${\sf StO}$ and \csto{} are perfectly indistinguishable to any \textit{unbounded} quantum algorithm. 
\begin{lemma}[{\cite[Lemma 4]{zhandry2019record}}] 
    Let $\As$ be an (unbounded) quantum algorithm making oracle queries. The output of $\As$ given access to the standard oracle is exactly identical to the output of $\As$ given access to a compressed oracle. 
\end{lemma}

%% file: quantum_anchor.tex
\section{The Quantum Anchor State}
In this section we introduce the quantum anchor state, a structured bipartite quantum state that will serve as a basic primitive throughout the paper. Intuitively, the generation process starts from several $\epr$ pairs shared between two registers $\anchor$ and $\vessel$. Some `decoy' qubits are then tensored to the $\vessel$ register. A random unitary $U_{\sf shift}$ is applied to $\vessel$ to `scramble' and hide the entanglement among all qubits in $\vessel$. This ensures that once $\vessel$ is given to another party who does not know $U_{\sf shift}$, it cannot separate the part of $\vessel$ that is entangled with $\anchor$ from the other `decoy' qubits. 

If not specified, in remaining paragraphs we consider vectors and matrices with elements in $\mathbb{F}_2$. Normally, we use $\cal{O}_S^b$, where $S$ is a set of elements (subsets, cosets, subspaces) and $b$ is any object (such as symbols, bit strings, and field elements), to denote the oracle that outputs $b$ on input $x\in S$ and outputs $\bot$ otherwise. When $b=\top$, it is called a membership oracle and we write $\cal{O}_S$ for simplicity. Also, if not specified, the domain of $\cal{O}_S^b$ is the natural domain depending on $S$. For example, when $S$ is a subset/coset/subspace of $\mathbb{F}_2^n$, then the domain is $\mathbb{F}_2^n$. We abuse notations like $\set{0,1}^n$ here so they can interact with field element such as $x\in\mathbb{F}_2^n$. For example, we use $0^n\times x$ to denote a vector in $\mathbb{F}_2^{2n}$ with its first $n$ bits all zeros and its last $n$ bits being $x$. Another example is that we use $U_{\sf shift}(0^n\times S)$, where $S$ is a subset/coset/subspace of $\mathbb{F}_2^n$ and $U_{\sf shift}$ is a change of basis matrix, to denote the subset/coset/subspace of $\mathbb{F}_2^{2n}$ consisting of all vectors of the form $U_{\sf shift}(0^n\times x)$ for $x\in S$.
\begin{remark}
    When we define subspaces in this paper, we use the columns of a matrix to describe the subspace so that we can define the canonical representation of the subspace and its dual. For example, let $A\leq\mathbb{F}_2^n$ be a subspace of dimension $m$ defined by $m$ columns of a $n\times n$ change of basis (invertible) matrix $U$. Let the set of indices of these columns be $C$. This means that $A$ is the span of the set of $m$ columns of $U$ with indices in $C$. Now we define $\Can_A(s)$ and $\Can_{A^\perp}(s')$ to be the canonical representations of the cosets $A+s$ and $A^\perp+s'$:
    \begin{itemize}
        \item Let $w$ be the vector obtained by replacing $m$ bits of $U^{-1}s$ with indices in $C$ by zeros. Define $\Can_A(s):=Uw$.
        \item Let $w'$ be the vector obtained by replacing $n-m$ bits of $U^ts'$ with indices not in $C$ by zeros. Define $\Can_{A^\perp}(s'):=U^{-t}w'$. Here $U^{-t}$ is defined as $\brackets{U^{t}}^{-1}$.
    \end{itemize}
     To understand what this definition means, we use the following example. Let $A$ be a $n/2$ dimensional subspace of $\mathbb{F}_2^n$ spanned by the first $n/2$ columns of $U_A$. Let $x=U_A\brackets{x_A\times 0^{n/2}+0^{n/2}\times x_{A^\perp}}$ be a vector in $\mathbb{F}_2^n$ where $x_A$ and $x_{A^\perp}$ are vectors in $\mathbb{F}_2^{n/2}$. Then we have that $\Can_A(x)=U_A\brackets{0^{n/2}\times x_{A^\perp}}$ and $\Can_{A^\perp}(x)=U_A^{-t}\brackets{x_A\times 0^{n/2}}$. Define $\CS(A):=\set{\Can_A(s):s\in\mathbb{F}_2^n}$ and $\CS(A^\perp):=\set{\Can_{A^\perp}(s'):s'\in\mathbb{F}_2^n}$. 
\end{remark}
\begin{definition}[Quantum Anchor State]\label{def:quantum_anchor}
    Let $\secp$ be the security parameter and let $n_e(\secp), n_s(\secp), n_h(\secp)$ be polynomials. Take a uniformly random $(n_e+n_h+n_s)\times (n_e+n_h+n_s)$ change of basis (invertible) matrix $U_{\sf shift}$. The subspaces $S,T$ are described by the first $n_h$ columns and the first $n_h+n_e$ columns of $U_{\sf shift}$, respectively. The anchor state $\ket{\chain_{n_e,n_h,n_s}^\sk}_{\anchor\vessel}$ is a bipartite state on a $n_e$-qubit anchor register $\anchor$ and a $(n_e+n_h+n_s)$-qubit vessel register $\vessel$. This state is also parameterized by a secret key $\sk$, which we specify later. The state $\ket{\chain_{n_e,n_h,n_s}^\sk}_{\anchor\vessel}$ is generated in the following way:
    \begin{enumerate}
        \item First generate a $2(n_e+n_h+n_s)$-qubit $\epr$ state between $\ureg\otimes\anchor\otimes\vreg$ and $\vessel$. In other words, after this step, the registers $\ureg$ of size $n_h$, $\anchor$ of size $n_e$, and $\vreg$ of size $n_s$ are fully entangled (in the form of $\epr$ pairs) with the first $n_h$ qubits of $\vessel$, the middle $n_e$ qubits of $\vessel$, and the last $n_s$ qubits of $\vessel$, respectively.
        \item Measure $\ureg$ in the Hadamard basis to obtain a vector $u'\in\mathbb{F}_2^{n_h}$ ($\ket{+}$ is mapped to $0$ and $\ket{-}$ is mapped to $1$ on each coordinate). Then measure $\vreg$ in the standard basis to obtain a vector $v'\in\mathbb{F}_2^{n_s}$. Compute $u=U_{\sf shift}^{-t}\brackets{u'\times 0^{n_e+n_s}}$ and $v=U_{\sf shift}\brackets{0^{n_h+n_e}\times v'}$.
        \item Apply the following change of basis unitary $\cal{U}_{\sf shift}$ on $\vessel$:
        \begin{equation*}
            \cal{U}_{\sf shift}\ket{x}_{\vessel}\rightarrow\ket{U_{\sf shift}x}_{\vessel}.
        \end{equation*}
    \end{enumerate}
    This procedure is shown in \Cref{fig:genchain_domainextension}. Define the functionality $\genchain_{n_e,n_h,n_s}(1^\secp)$ as the function that samples random $S,T$, generates $\ket{\chain^\sk_{n_e,n_h,n_s}}_{\anchor\vessel}$, and outputs $\brackets{\sk:=(U_{\sf shift},u,v),\ket{\chain^\sk_{n_e,n_h,n_s}}_{\anchor\vessel}}$. One can see that by storing $U_{\sf shift}$ in $\sk$, we also store $S,T$ in it. 
\end{definition}
The generation process creates the following state on $\anchor\vessel$:
\begin{equation*}
    \propto\sum_{x\in\mathbb{F}_2^{n_e}}\ket{x}_{\anchor}\sum_{s\in S}(-1)^{\ipd{s,u}}\ket{s+U_{\sf shift}\brackets{0^{n_h}\times x\times 0^{n_s}}+v}_\vessel.
\end{equation*}
\begin{remark}
    We may omit the subscript when $n_e=n_h=n_s=n$ and $n$ is clear from the context. For example, $\brackets{\sk,\ket{\chain^\sk}_{\anchor\vessel}}\gets\genchain(1^\secp)$ means $\brackets{\sk,\ket{\chain^\sk_{n,n,n}}_{\anchor\vessel}}\gets\genchain_{n,n,n}(1^\secp)$ when $n$ is clear from the context.
\end{remark}

\newcommand{\domainextension}{{\sf DomainExtension}}

\begin{definition}[Domain Extension and Oracle Simulation]
    Define $\domainextension_{n'_h, n'_s}$ as follows:
    \begin{enumerate}
        \item It takes $1^\secp$ and works as an isometry on the $\vessel$ register of an anchor state $\ket{\chain^\sk_{n_e,n_h,n_s}}$.
        \item Sample a uniformly random $(n_e+n_h+n_s+n'_h+n'_s)\times (n_e+n_h+n_s+n'_h+n'_s)$ change of basis matrix $U_{\sf shift}^*$. Define $S^*$ to be the subspace described by the first $n'_h$ columns of $U_{\sf shift}^*$ and define $T^*$ to be the subspace described by the first $n'_h+n_h+n_e+n_s$ columns of $U_{\sf shift}^*$. 
        \item Create a $n'_h$-qubit register $\ureg^*$ and a $n'_s$-qubit register $\vreg^*$. Then extend $\vessel$ by tensoring it with $n'_h$/$n'_s$ qubits at the front/end that are fully entangled with $\ureg^*$/$\vreg^*$ in the form of $\epr$ pairs, respectively. The extended register is called $\vessel^*$.
        \item Measure $\ureg^*$ in the Hadamard basis to obtain a vector ${u'}^*\in\mathbb{F}_2^{n'_h}$ ($\ket{+}$ is mapped to $0$ and $\ket{-}$ is mapped to $1$ on each coordinate). Then measure $\vreg^*$ in the standard basis to obtain a vector ${v'}^*\in\mathbb{F}_2^{n'_s}$. Compute $u^*={U_{\sf shift}^*}^{-t}({u'}^*\times 0^{n_h+n_e+n_s+n'_s})$ and $v^*=U_{\sf shift}^*(0^{n'_h+n_h+n_e+n_s}\times{v'}^*)$. 
        \item Apply the following change of basis unitary $\cal{U}_{\sf shift}^*$ on $\vessel^*$:
        \begin{equation*}
            \cal{U}_{\sf shift}^*\ket{x}_{\vessel^*}\rightarrow\ket{U_{\sf shift}^*x}_{\vessel^*}.
        \end{equation*}
        \item Let $\sk^*=(U_{\sf shift}^*,u^*,v^*)$. Return $\sk^*$ and $\vessel^*$.
    \end{enumerate}
    Also we define
    \begin{itemize}
        \item $S_{\sf ext}:=U_{\sf shift}^*\brackets{\set{0,1}^{n'_h}\times S\times 0^{n'_s}}$, described by the first $n_h+n'_h$  columns of $U_{\sf extshift}$. Where $U_{\sf extshift}:=U_{\sf shift}^*
        \begin{pmatrix}
            I_{n'_h} & & \\
             & U_{\sf shift} & \\
             & & I_{n'_s}
        \end{pmatrix}$.
        \item $T_{\sf ext}:=U_{\sf shift}^*\brackets{\set{0,1}^{n'_h}\times T\times 0^{n'_s}}$, described by the first $n'_h + n_h + n_e$ columns of $U_{\sf extshift}$. 
        \item $u_{\sf ext}:={U_{\sf shift}^*}^{-t}\brackets{0^{n'_h}\times u\times 0^{n'_s}}+u^*$.
        \item $v_{\sf ext}:=U_{\sf shift}^*\brackets{0^{n'_h}\times v\times 0^{n'_s}}+v^*$.
    \end{itemize}
\end{definition}

$\genchain$ and $\domainextension$ are shown in \Cref{fig:genchain_domainextension}.
\begin{figure}[h]
    \centering
    \begin{tikzpicture}[
        font=\small,
        register/.style={draw, rounded corners, minimum width=2.1cm, minimum height=0.8cm, align=center},
        vesselblock/.style={draw, rounded corners, minimum width=2.6cm, minimum height=0.8cm, align=center},
        arrow/.style={->, thick},
        epr/.style={
            thick,
            decorate,
            decoration={snake, amplitude=0.7mm, segment length=3mm}
        },
        node distance=1.1cm and 1.6cm
    ]
    
    \node[vesselblock] (v0)
    {\footnotesize first $n'_h$ of $\vessel^*$};
    \node[vesselblock, right=0pt of v0] (v1)
    {\footnotesize first $n_h$ of $\vessel$};
    \node[vesselblock, right=0pt of v1] (v2)
    {\footnotesize middle $n_e$ of $\vessel$};
    \node[vesselblock, right=0pt of v2] (v3)
    {\footnotesize last $n_s$ of $\vessel$};
    \node[vesselblock, right=0pt of v3] (v4)
    {\footnotesize last $n'_s$ of $\vessel^*$};
    
    \node[register, above=0.8cm of v0] (uregs) {$\ureg^*$};
    \node[register, above=0.8cm of v1] (ureg) {$\ureg$};
    \node[register, above=0.8cm of v2] (anchor) {$\anchor$};
    \node[register, above=0.8cm of v3] (vreg) {$\vreg$};
    \node[register, above=0.8cm of v4] (vregs) {$\vreg^*$};
    
    
    \draw[epr] (uregs.south) -- (v0.north);
    \draw[epr] (ureg.south) -- (v1.north);
    \draw[epr] (anchor.south) -- (v2.north);
    \draw[epr] (vreg.south) -- (v3.north);
    \draw[epr] (vregs.south) -- (v4.north);
    
    \node[register, above=0.6cm of uregs] (measus) {H-basis};
    \node[register, above=0.6cm of ureg] (measu) {H-basis};
    \node[register, above=0.6cm of vreg] (measv) {Std-basis};
    \node[register, above=0.6cm of vregs] (measvs) {Std-basis};
    
    \draw[arrow] (uregs.north) -- (measus.south);
    \draw[arrow] (ureg.north) -- (measu.south);
    \draw[arrow] (vreg.north) -- (measv.south);
    \draw[arrow] (vregs.north) -- (measvs.south);
    
    \node[above=0.5cm of measus] (uprimes) {${u'}^*\rightarrow u^*$};
    \node[above=0.5cm of measu] (uprime) {$u'\rightarrow u$};
    \node[above=0.5cm of measv] (vprime) {$v'\rightarrow v$};
    \node[above=0.5cm of measvs] (vprimes) {${v'}^*\rightarrow v^*$};
    
    \draw[arrow] (measus.north) -- (uprimes.south);
    \draw[arrow] (measu.north) -- (uprime.south);
    \draw[arrow] (measv.north) -- (vprime.south);
    \draw[arrow] (measvs.north) -- (vprimes.south);

    
    \draw[decorate,decoration={brace,amplitude=10pt,mirror}] 
        (v1.south west) -- (v3.south east) node[midway,below=6pt] (Ushift) {$U_{\sf shift}$};
        
    \draw[decorate,decoration={brace,amplitude=20pt,mirror}] 
        ([yshift=-1.3cm]v0.south west) -- ([yshift=-1.3cm]v4.south east) node[midway,below=20pt] (Ushifts) {$U_{\sf shift}^*$};
    
    \node[
        draw=red,
        dashed,
        thick,
        rounded corners,
        fit=(uprime)(v1)(v3)(vprime)(Ushift),
        inner sep=6pt,
        label={[red]below:{\small $\genchain$}}
    ] (GenChain) {};
    
    \node[
        draw=blue,
        dashed,
        thick,
        rounded corners,
        fit=(uprimes)(v0)(v4)(vprimes)(Ushifts)(GenChain),
        inner sep=6pt,
        label={[blue]below:{\small $\domainextension$}}
    ] (DomainExtension) {};
    
    \end{tikzpicture}
    \caption{Compact state-generation diagram. $\genchain$ and $\domainextension$.}
    \label{fig:genchain_domainextension}
\end{figure}
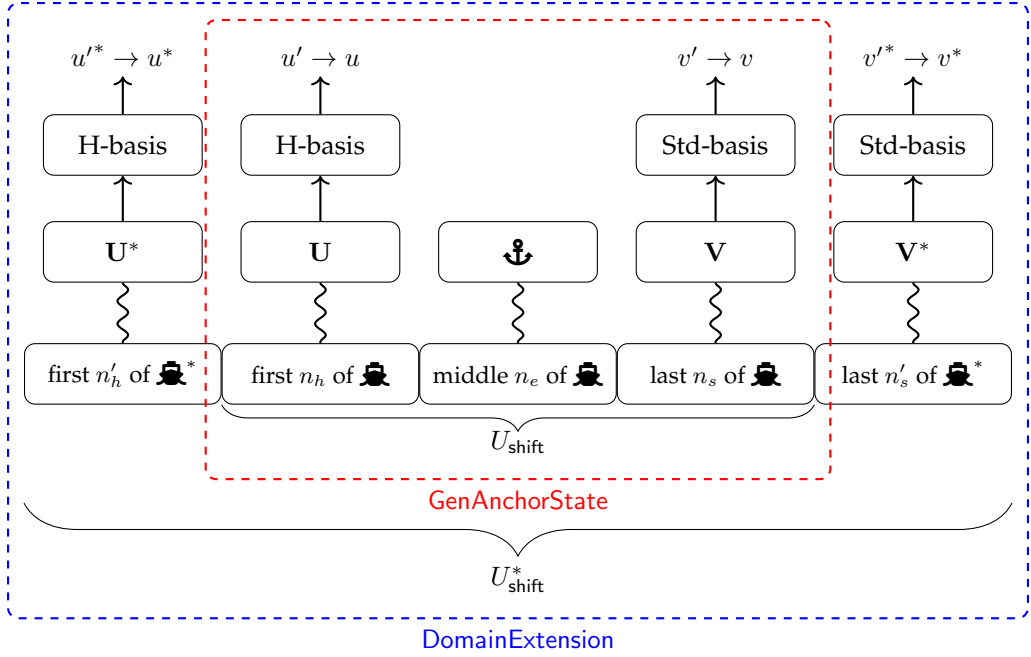
Here are some useful facts that we will use later in the proof.
\begin{fact}\label{fact:standard_hadamard_duality}
    Let $n_e,n_h,n_s$ be polynomials of $\secp$ and let $\brackets{\sk,\ket{\chain^\sk_{n_e,n_h,n_s}}_{\anchor\vessel}}\gets\genchain_{n_e,n_h,n_s}(1^\secp)$. The standard-basis measurement on $\vessel$ always outputs a vector in $T+v$, and the Hadamard-basis measurement on $\vessel$ always outputs a vector in $S^\perp+u$. Furthermore, the anchor state has the following symmetric property: the following two distributions are equal for fixed $n_e,n_h,n_s$.
    \begin{align*}
        &\set{S,T,u,v,\ket{\chain^\sk_{n_e,n_h,n_s}}_{\anchor\vessel}\middle\vert\brackets{\sk,\ket{\chain^\sk_{n_e,n_h,n_s}}_{\anchor\vessel}}\gets\genchain_{n_e,n_h,n_s}(1^\secp)}\\
        =&\set{T^\perp,S^\perp,v,u,\brackets{H^{\otimes n_e}_{\anchor}\otimes H^{\otimes (n_e+n_h+n_s)}_{\vessel}}\ket{\chain^\sk_{n_e,n_s,n_h}}_{\anchor\vessel}\middle\vert\brackets{\sk,\ket{\chain^\sk_{n_e,n_s,n_h}}_{\anchor\vessel}}\gets\genchain_{n_e,n_s,n_h}(1^\secp)}.
    \end{align*}
\end{fact}

\begin{fact}\label{fact:extension_indistinguishability}
    Let $n_e,n_h,n_s,n'_h,n'_s$ be polynomials of $\secp$. The following two distributions
    \begin{equation*}
        \set{\substack{U_{\sf extshift},\\S_{\sf ext},T_{\sf ext},\\u_{\sf ext},v_{\sf ext},\\\text{state on }\anchor\vessel^*}
        \middle\vert\substack{\brackets{\sk,\ket{\chain^\sk_{n_e,n_h,n_s}}_{\anchor\vessel}}\gets\genchain_{n_e,n_h,n_s}(1^\secp)\\ \sk^*\gets\domainextension_{n'_h,n'_s}(1^\secp)\text{ on }\vessel\text{ to obtain }\vessel^*}}
    \end{equation*}
    and
    \begin{equation*}
        \set{\substack{U_{\sf shift},\\S,T,\\u,v,\\\text{state on }\anchor\vessel}
        \middle\vert\brackets{\sk,\ket{\chain^\sk_{n_e,n_h+n'_h,n_s+n'_s}}_{\anchor\vessel}}\gets\genchain_{n_e,n_h+n'_h,n_s+n'_s}(1^\secp)}
    \end{equation*}
    are equal.
\end{fact}

\begin{fact}\label{fact:simulating_extension_oracles}
    Let $\brackets{\sk,\ket{\chain^{\sk}_{n_e,n_h,n_s}}_{\anchor\vessel}}\gets\genchain_{n_e,n_h,n_s}(1^\secp)$ and let $\sk^*$ be generated by $\domainextension_{n'_h,n'_s}(1^\secp)$ working on $\vessel$. Any party with $\sk^*$ can do the following:
    \begin{itemize}
        \item It can simulate oracle access to $\cal{O}_{T^*+v^*}^{b}$ with $b$: on input $z$, check whether ${U_{\sf shift}^*}^{-1}(z-v^*)$ has its last $n'_s$ bits equal to zero. Output $b$ if so and output $\bot$ otherwise.
        \item It can simulate oracle access to $\cal{O}_{{S^*}^\perp+u^*}^{b}$ with $b$: on input $z$, check whether ${U_{\sf shift}^*}^t(z-u^*)$ has its first $n'_h$ bits equal to zero. Output $b$ if so and output $\bot$ otherwise.
        \item It can simulate oracle access to $\cal{O}_{T_{\sf ext}+v_{\sf ext}}^{b}$ given oracle access to $\cal{O}_{T+v}^b$: on input $z$, let $z'$ be the result of discarding the first $n'_h$ bits and the last $n'_s$ bits of ${U_{\sf shift}^*}^{-1}(z-v^*)$. If any of the last $n'_s$ discarded bits is non-zero, output $\bot$; otherwise forward the output of $\cal{O}_{T+v}^b(z')$.
        \item It can simulate oracle access to $\cal{O}_{S_{\sf ext}^\perp+u_{\sf ext}}^{b}$ given oracle access to $\cal{O}_{S^\perp+u}^b$: on input $z$, let $z'$ be the result of discarding the first $n'_h$ bits and the last $n'_s$ bits of ${U_{\sf shift}^*}^{t}(z-u^*)$. If any of the first $n'_h$ discarded bits is non-zero, output $\bot$; otherwise forward the output of $\cal{O}_{S^\perp+u}^b(z')$.
    \end{itemize}
\end{fact}

%% file: collapsing_oracles_versus_non-collapsing_oracles.tex
\section{Collapsing vs Non-collapsing Oracles}
\label{sec:collapsing-indistinguishability}
Let $a(\secp)\leq b(\secp)< c(\secp)\leq d(\secp)$ be polynomials of $\secp$. Let $W\leq\mathbb{F}_2^d$ be a subspace of dimension $c$. Let $U^{(\cdot)}$ be any efficient isometry with superposition access to a classical oracle. We introduce a subspace oracle indistinguishability lemma below. In the remaining part of this section $S,T$ are sampled so that $T\leq W$ is a uniform random subspace of dimension $b$ and $S\leq T$ be a uniform random subspace of dimension $a$. 
\begin{lemma}\label{lem:subspace_oracle_indistinguishability} Let $\set{\cal{O}_S}_S$ be any family of oracles such that $\cal{O}_S$ outputs $\bot$ on $z\notin S$. Let $\set{\cal{O}_{S,T}}_{S,T}$ be any family of oracles such that $\cal{O}_{S,T}$ outputs $\cal{O}_S(z)$ on $z\in S$ and outputs $\bot$ on $z\notin T$. There exists a polynomial $q(\secp)$ such that the following holds. Define $\ket{\psi_S}:=U^{\cal{O}_S}\ket{S}$, $\ket{\psi_T}:=U^{\cal{O}_{S,T}}\ket{S}$, $\ket{\phi_S}=\frac{1}{\sqrt{\binom{c}{b}_2\binom{b}{a}_2}}\sum_{S,T}\ket{S,T}\ket{\psi_{S}}$ and $\ket{\phi_T}=\frac{1}{\sqrt{\binom{c}{b}_2\binom{b}{a}_2}}\sum_{S,T}\ket{S,T}\ket{\psi_{T}}$ we have
\begin{equation*}
    \TraceDist{\ket{\phi_S}\bra{\phi_S}}{\ket{\phi_T}\bra{\phi_T}}\leq\frac{q}{2^{(c-b)/2}}.
\end{equation*}
\end{lemma}
\begin{proof}
    We prove it by hybrid argument. Let the number of oracle queries of $U^\cal{O}$ be $q'(\secp)$. Define $\ket{\psi^i}$ to be the state after applying $U^\cal{O}$ on $\ket{S}$, where the first $i$ queries use oracle $\cal{O}_{S}$ and other queries use oracle $\cal{O}_{S,T}$. We have $\ket{\psi^0}=\ket{\psi_T}$ and $\ket{\psi^{q'}}=\ket{\psi_S}$. Similarly we define $\ket{\phi^i}=\frac{1}{\sqrt{\binom{c}{b}_2\binom{b}{a}_2}}\sum_{S,T}\ket{S,T}\ket{\psi^i}$. To use hybrid argument, we need the following claim.
    \begin{claim}
        For all $0\leq i<q'$,
        \begin{equation*}
            \TraceDist{\ket{\phi^i}\bra{\phi^i}}{\ket{\phi^{i+1}}\bra{\phi^{i+1}}}\leq\frac{2}{2^{(c-b)/2}}.
        \end{equation*}
    \end{claim}
    \begin{proof}
        Let $\ket{\psi^{\leq i}}$ be the state of applying $U^{\cal{O}_S}$ on $\ket{S}$ but stop before the $i+1$-th oracle query. Define $\Pi_{T\setminus S}$ to be the projector onto states where the value $z$ on the oracle query register is in $T$ but is not in $S$. We have
        \begin{align*}
            &\TraceDist{\ket{\phi^i}\bra{\phi^i}}{\ket{\phi^{i+1}}\bra{\phi^{i+1}}}\\
            =&\frac{1}{\binom{c}{b}_2\binom{b}{a}_2}\TraceDist{\brackets{\sum_{S,T}\ket{S,T}\cal{O}_{S,T}\ket{\psi^{\leq i}}}\brackets{\sum_{S,T}\bra{S,T}\bra{\psi^{\leq i}}\cal{O}_{S,T}^{\dagger}}}{\brackets{\sum_{S,T}\ket{S,T}\cal{O}_S\ket{\psi^{\leq i}}}\brackets{\sum_{S,T}\bra{S,T}\bra{\psi^{\leq i}}\cal{O}_S^{\dagger}}}\\
            =&\sqrt{1-\abs{\expectc{S,T}{\bra{\psi^{\leq i}}\cal{O}_{S,T}^{\dagger}\cal{O}_S\ket{\psi^{\leq i}}}}^2}\\
            =&\sqrt{1-\abs{\expectc{S,T}{\bra{\psi^{\leq i}}\Pi_{T\setminus S}\cal{O}_{S,T}^{\dagger}\cal{O}_S\Pi_{T\setminus S}\ket{\psi^{\leq i}}+\bra{\psi^{\leq i}}(I-\Pi_{T\setminus S})\cal{O}_{S,T}^{\dagger}\cal{O}_S(I-\Pi_{T\setminus S})\ket{\psi^{\leq i}}}}^2}\\
            \leq&\sqrt{1-\brackets{\abs{\expectc{S,T}{\bra{\psi^{\leq i}}(I-\Pi_{T\setminus S})\cal{O}_{S,T}^{\dagger}\cal{O}_S(I-\Pi_{T\setminus S})\ket{\psi^{\leq i}}}}-\abs{\expectc{S,T}{\bra{\psi^{\leq i}}\Pi_{T\setminus S}\cal{O}_{S,T}^{\dagger}\cal{O}_S\Pi_{T\setminus S}\ket{\psi^{\leq i}}}}}^2}\\
            \leq&\sqrt{1-\brackets{1-\frac{2}{2^{c-b}}}^2}\\
            \leq&\frac{2}{2^{(c-b)/2}}.
        \end{align*}
        Now we explain these (in)equations:
        \begin{itemize}
            \item The equation between line 2 and line 3 is by the definition of trace distance between pure state 
            \[
            \TraceDist{\ket{\phi}\bra{\phi}}{\ket{\psi}\bra{\psi}}=\sqrt{1-\abs{\braket{\phi}{\psi}}^2}.
            \]
            And also notice that all cross terms for different $S,T$ in the inner product is zero because $\braket{S,T}{S',T'}=0$ for $(S,T)\neq (S',T')$.
            \item The equation between line 3 and line 4 is because that the cross terms, for example,
            \[
                \bra{\psi^{\leq i}}(I-\Pi_{T\setminus S})\cal{O}_{S,T}^{\dagger}\cal{O}_S\Pi_{T\setminus S}\ket{\psi^{\leq i}}
            \]
            is zero because $\Pi_{T\setminus S}$ and $I-\Pi_{T\setminus S}$ maps the state to disjoint basis on the query register.
            \item The inequality between line 4 and line 5 is because $\abs{a+b}^2\geq(\abs{b}-\abs{a})^2$.
            \item The inequality between line 5 and line 6 is because that $\cal{O}_{S,T}$ and $\cal{O}_S$ is the same oracle if the input state is in $I-\Pi_{T\setminus S}$. We have
            \begin{align*}
                &\abs{\expectc{S,T}{\bra{\psi^{\leq i}}(I-\Pi_{T\setminus S})\cal{O}_{S,T}^{\dagger}\cal{O}_S(I-\Pi_{T\setminus S})\ket{\psi^{\leq i}}}}\\
                =&\expectc{S,T}{\abs{(I-\Pi_{T\setminus S})\ket{\psi^{\leq i}}}^2}\geq 1-\frac{1}{2^{c-b}}
            \end{align*}
            and
            \begin{align*}
                &\abs{\expectc{S,T}{\bra{\psi^{\leq i}}\Pi_{T\setminus S}\cal{O}_{S,T}^{\dagger}\cal{O}_S\Pi_{T\setminus S}\ket{\psi^{\leq i}}}}\\
                \leq&\expectc{S,T}{\abs{\Pi_{T\setminus S}\ket{\psi^{\leq i}}}^2}\leq\frac{1}{2^{c-b}}.
            \end{align*}
            The last steps of both calculation are because that $\ket{\psi^{\leq i}}$ is independent of $T$ thus the average weight on $\Pi_{T\setminus S}$ is at most $\frac{1}{2^{c-b}}$.
            \item The inequality between line 6 and line 7 is because $\sqrt{1-(1-x)^2}\leq\sqrt{2x}$.
        \end{itemize}
    \end{proof}
    Setting $q(\secp)=2q'(\secp)$, we see that the lemma follows.
\end{proof}
Let $m(\secp)$ be a polynomial. Now we consider the case that $c=d$ and let $w$ be a uniform random vector in $\mathbb{F}_2^c$. We further define some oracles:
\begin{itemize}
    \item $\cal{O}_{\sf any}:\mathbb{F}_2^{c}\rightarrow\mathbb{F}_2^m$ be an oracle that is sampled from a distribution depending on $S,T,w$ such that every output alone is independent of $T$ given $S,w$. (But maybe with two outputs of inputs, one can recover $S,T,w$)
    \item $\cal{O}_{T+w}^{\sf ncol}:\mathbb{F}_2^{c}\rightarrow\mathbb{F}_2^{m}\cup\set{\bot}$ on input $z$ checks if it is in $T+w$. It returns $\bot$ if it is not. Otherwise, it returns the fixed vector $\cal{O}_{\sf any}(\Can_S(w))$, disregarding the input $z$.
    \item $\cal{O}_{T+w}^{\sf col}:\mathbb{F}_2^{c}\rightarrow\mathbb{F}_2^{m}\cup\set{\bot}$ on input $z$ checks if it is in $T+w$. It returns $\bot$ if it is not. Otherwise, it returns $\cal{O}_{\sf any}(\Can_S(z))$.
\end{itemize}
Intuitively, $\cal{O}_{\sf any}$ can just be thought of a random oracle mapping cosets of $S$ to length-$m$ labels, while the other oracles are two types of membership oracles for $T+w$. The non-collapsing oracle returns the single fixed label associated with $\Can_S(w)$ for every accepted query, whereas the collapsing oracle returns the particular label associated with $z$ through $\cal{O}_{\sf any}(\Can_S(z))$.
\begin{lemma}\label{lem:collapsing_versus_non-collapsing}
    For any efficient isometry
    \begin{equation*}
        \cal{U}^{\cal{O}}:\cal{H}_{\sreg}\otimes\cal{H}_{\wreg}\rightarrow\cal{H}_\zreg
    \end{equation*}
    with oracle access to $\cal{O}$ where $\cal{O}$ is either $\cal{O}_{T+w}^{\sf ncol}$ or $\cal{O}_{T+w}^{\sf col}$. There exists a polynomial $q(\secp)$ such that let $\ket{\psi_{T+w}^{\sf ncol}}=\cal{U}^{\cal{O}_{T+w}^{\sf ncol}}\ket{S}_{\sreg}\ket{w}_{\wreg}$ and $\ket{\psi_{T+w}^{\sf col}}=\cal{U}^{\cal{O}_{T+w}^{\sf col}}\ket{S}_{\sreg}\ket{w}_{\wreg}$. We have
    \begin{equation*}
        \expectc{S,T,w,\cal{O}_{\sf any}}{\TraceDist{\ket{\psi_{T+w}^{\sf ncol}}\bra{\psi_{T+w}^{\sf ncol}}}{\ket{\psi_{T+w}^{\sf col}}\bra{\psi_{T+w}^{\sf col}}}}\leq\frac{q}{2^{(c-b)/2}}.
    \end{equation*}
\end{lemma} 


\begin{proof}
    We prove it by hybrid argument.\\
    \textbf{Hybrid 0}: This is the hybrid corresponding to $\ket{\psi_{T+w}^{\sf ncol}}=\cal{U}^{\cal{O}_{T+w}^{\sf ncol}}\ket{S}_{\sreg}\ket{w}_{\wreg}$.\\
    \textbf{Hybrid 1}: In this hybrid, the final state is $\ket{\psi_{S+w}^{\sf ncol}}=\cal{U}^{\cal{O}_{S+w}^{\sf ncol}}\ket{S}_{\sreg}\ket{w}_{\wreg}$, where the new oracle is defined as:
    \begin{itemize}
        \item $\cal{O}_{S+w}^{\sf ncol}:\mathbb{F}_2^{c}\rightarrow\mathbb{F}_2^m\cup\set{\bot}$ on input $z$ checks if it is in $S+w$. It returns $\bot$ if it is not. Otherwise, it returns $\cal{O}_{\sf any}(\Can_S(w))$.
    \end{itemize}
    \begin{claim}\label{claim:templabel-example1}
        There exists polynomial $q_0(\secp)$ such that
        \begin{equation*}
            \expectc{S,T,w,\cal{O}_{\sf any}}{\TraceDist{\ket{\psi_{T+w}^{\sf ncol}}\bra{\psi_{T+w}^{\sf ncol}}}{\ket{\psi_{S+w}^{\sf ncol}}\bra{\psi_{S+w}^{\sf ncol}}}}\leq\frac{q_0}{2^{(c-b)/2}}.
        \end{equation*}
    \end{claim}
    \begin{proof}
        $T$ is a uniform random subspace containing $S$ given $S,w$. Also the output of $\cal{O}_{S+w}^{\sf ncol}$ is independent of $T$. Thus by \Cref{lem:subspace_oracle_indistinguishability} there exists such polynomial $q_0(\secp)$. The reason that we can call this lemma is because $w$ is given at the beginning so $\cal{O}_{T+w}^{\sf ncol}$ is equivalent to $\cal{O}_{T}^{\sf ncol}(\cdot -w)$ and $\cal{O}_{S+w}^{\sf ncol}$ is equivalent to $\cal{O}_{S}^{\sf ncol}(\cdot -w)$ where for $A=S,T$ we define
        \begin{itemize}
            \item $\cal{O}_{A}^{\sf ncol}:\mathbb{F}_2^{c}\rightarrow\mathbb{F}_2^{m}\cup\set{\bot}$ on input $z$ checks if it is in $A$. It returns $\bot$ if it is not. Otherwise, it returns $\cal{O}_{\sf any}(\Can_S(w))$.
        \end{itemize}
        While $\set{\cal{O}_{S}^{\sf ncol}}_S$ and $\set{\cal{O}_{T}^{\sf ncol}}_T$ satisfy the requirement in \Cref{lem:subspace_oracle_indistinguishability}.
    \end{proof}
    \textbf{Hybrid 2}: In this hybrid, the final state is $\ket{\psi_{S+w}^{\sf col}}=\cal{U}^{\cal{O}_{S+w}^{\sf col}}\ket{S}_{\sreg}\ket{w}_{\wreg}$, where the new oracle is defined as:
    \begin{itemize}
        \item $\cal{O}_{S+w}^{\sf col}:\mathbb{F}_2^{c}\rightarrow\mathbb{F}_2^m\cup\set{\bot}$ on input $z$ checks if it is in $S+w$. It returns $\bot$ if it is not. Otherwise, it returns $\cal{O}_{\sf any}(\Can_S(z))$.
    \end{itemize}
    \begin{claim}
        For all $S,T,w,\cal{O}_{\sf any}$,
        \begin{equation*}
            \TraceDist{\ket{\psi_{S+w}^{\sf ncol}}\bra{\psi_{S+w}^{\sf ncol}}}{\ket{\psi_{S+w}^{\sf col}}\bra{\psi_{S+w}^{\sf col}}}=0.
        \end{equation*}
    \end{claim}
    \begin{proof}
        For any vector $z\in S+w$, $S+w$ and $S+z$ are the same coset. These two hybrids are identical because $\cal{O}_{\sf any}(\Can_S(S+w))=\cal{O}_{\sf any}(\Can_S(S+z))$.
    \end{proof}
    \textbf{Hybrid 3}: In this hybrid, the final state is $\ket{\psi_{T+w}^{\sf col}}=\cal{U}^{\cal{O}_{T+w}^{\sf col}}\ket{S}_{\sreg}\ket{w}_{\wreg}$.
    \begin{claim}
        There exists polynomial $q_1(\secp)$ such that
        \begin{equation*}
            \expectc{S,T,w,\cal{O}_{\sf any}}{\TraceDist{\ket{\psi_{S+w}^{\sf col}}\bra{\psi_{S+w}^{\sf col}}}{\ket{\psi_{T+w}^{\sf col}}\bra{\psi_{T+w}^{\sf col}}}}\leq\frac{q_1}{2^{(c-b)/2}}.
        \end{equation*}
    \end{claim}
    \begin{proof}
        $T$ is a uniform random subspace containing $S$ given $S,w$. Also the output of $\cal{O}_{S+w}^{\sf col}$ is independent of $T$. Thus by \Cref{lem:subspace_oracle_indistinguishability} there exists such polynomial $q_1(\secp)$. The reason is similar to the one mentioned in the proof of \Cref{claim:templabel-example1}.
    \end{proof}
    Combine all above, let $q=q_0+q_1$, we proved the lemma.
\end{proof}

%% file: monogamy-of-entanglement_games_with_abort.tex
\section{Monogamy-of-Entanglement Games}\label{sec:monogomy}
\newcommand{\CosetMonogamy}{{\sf CosetMonogamy}\xspace}
\newcommand{\CosetMonogamyWithAbort}{{\sf CosetMonogamy}\xspace}
\newcommand{\AsymmetricCosetMonogamy}{{\sf AsymmetricCosetMonogamy}\xspace}
\newcommand{\AsymmetricAnchorMonogamy}{{\sf AsymmetricAnchorMonogamy}\xspace}


\newcommand{\MultiStageDecisionMonogamy}{{\sf MultiStageDecisionMonogamy}\xspace}
\newcommand{\MultiStageSearchMonogamy}{{\sf MultiStageSearchMonogamy}\xspace}
\newcommand{\DualMultiStageSearchMonogamy}{{\sf DualMultiStageSearchMonogamy}\xspace}
\newcommand{\BeforeSplit}{\textbf{BeforeSplit}\xspace}
The purpose of this section is to prove properties of the following game.
\begin{definition}[Multi-Stage Decision/Search Monogamy-of-Entanglement Game]\label{def:multi-Stage_monogamy-of-entanglement}
    Let $\secp\in\mathbb{N}$ be the security parameter and let $n(\secp)\geq\secp$, and $n_e(\secp),m_0(\secp),m_1(\secp)$ be polynomials. Consider the following game between the challenger and an adversary $\cal{A}=\allowdisplaybreaks(\cal{A}_M^0,\cal{A}_M^1,\cal{A}_M^2,\cal{A}_L^0,\cal{A}_L^1,\cal{A}_R^0,\cal{A}_R^1)$:
    \begin{enumerate}
        \item The challenger generates $\brackets{\sk,\ket{\chain^{\sk}_{n_e,n,n}}_{\anchor\vessel}}\gets\genchain_{n_e,n,n}(1^\secp)$ and gives the register $\vessel$ to $\cal{A}_M^0$. All parties of the adversary are given oracle access to $\cal{O}_{T+v}$ and $\cal{O}_{S^\perp+u}$.
        \item $\cal{A}_M^0$ can choose to abort in this step. If it aborts, the output of the game is $\bot$. Otherwise, $\cal{A}_M^0$ generates a tripartite state on $\lreg\mreg\rreg$. It sends $\lreg$ to $\cal{A}_L^0$, sends $\mreg$ to $\cal{A}_M^1$ and sends $\rreg$ to $\cal{A}_R^0$.
        \item $\cal{A}_M^1$ is given the access to $\cal{O}_{T+v}^{b_0}$ for random $b_0\rand\mathbb{F}_2^{m_0}$. We can imagine implementing this oracle in a purified manner, where we create a register that is initialized to $\ket{+_\breg}_\breg$, the uniform superposition over all $b_0$, and $\cal{A}_M^1$ is able to perform the following operation in order to access $\breg$:
        \begin{equation*}
            \sum_{b_0\in\mathbb{F}_2^{m_0}}\ket{b_0}\bra{b_0}_{\breg}\otimes\cal{O}_{T+v}^{b_0}.
        \end{equation*}
        $\cal{A}_M^1$ with access to this oracle and $\mreg$ generates a tripartite state on $\mreg_L\mreg'\mreg_R$. It sends $\mreg_L$ to $\cal{A}_L^0$, sends $\mreg'$ to $\cal{A}_M^2$ and sends $\mreg_R$ to $\cal{A}_R^0$. We will name the time just after all operations of $\cal{A}_M^1$ are complete and before renaming the registers and splitting them as \BeforeSplit.
        \item $\cal{A}_L^0$ and $\cal{A}_R^0$ are given access to $\cal{O}_{T+v}^{b_0}$. $\cal{A}_L^0$ on $\lreg\mreg_L$ produces answer $b_0^l$ and a state on register $\lreg'$ that is sent to $\cal{A}_L^1$. $\cal{A}_R^0$ on $\rreg\mreg_R$ produces answer $b_0^r$ and a state on register $\rreg'$ that is sent to $\cal{A}_R^1$.
        \item $\cal{A}_M^2$ is given access to $\cal{O}_{T+v}^{b_0}$ and $\cal{O}_{S^\perp+u}^{b_1}$ for random $b_1\rand\mathbb{F}_2^{m_1}$. It generates a bipartite state on $\mreg'_L\mreg'_R$. $\mreg'_L$ is given to $\cal{A}_L^1$ and $\mreg'_R$ is given to $\cal{A}_R^1$. 
        \item $\cal{A}_L^1$ and $\cal{A}_R^1$ are given access to $\cal{O}_{T+v}^{b_0}$ and $\cal{O}_{S^\perp+u}^{b_1}$. $\cal{A}_L^1$ on $\lreg'\mreg'_L$ generates the answer $b_1^l$. $\cal{A}_R^1$ on $\rreg'\mreg'_R$ generates the answer $b_1^r$. The adversary wins iff $b_0^l=b_0^r=b_0$ and $b_1^l=b_1^r=b_1$, where $b_0$ is obtained by measuring the $\breg$ register.
    \end{enumerate}
    This game is displayed in \Cref{fig:multi-stage_monogamy-of-entanglement}. Let $\MultiStageSearchMonogamy(\cal{A},1^\secp)$ be the random variable that takes value $1/0/\bot$ if the adversary $\cal{A}$ wins/loses/aborts in the above game, respectively. Specifically, we use  $\MultiStageDecisionMonogamy(\cal{A},1^\secp)$ to denote the same random variable \textbf{when $m_0(\secp)=1$} (but $m_1$ can still be arbitrary).
\end{definition}
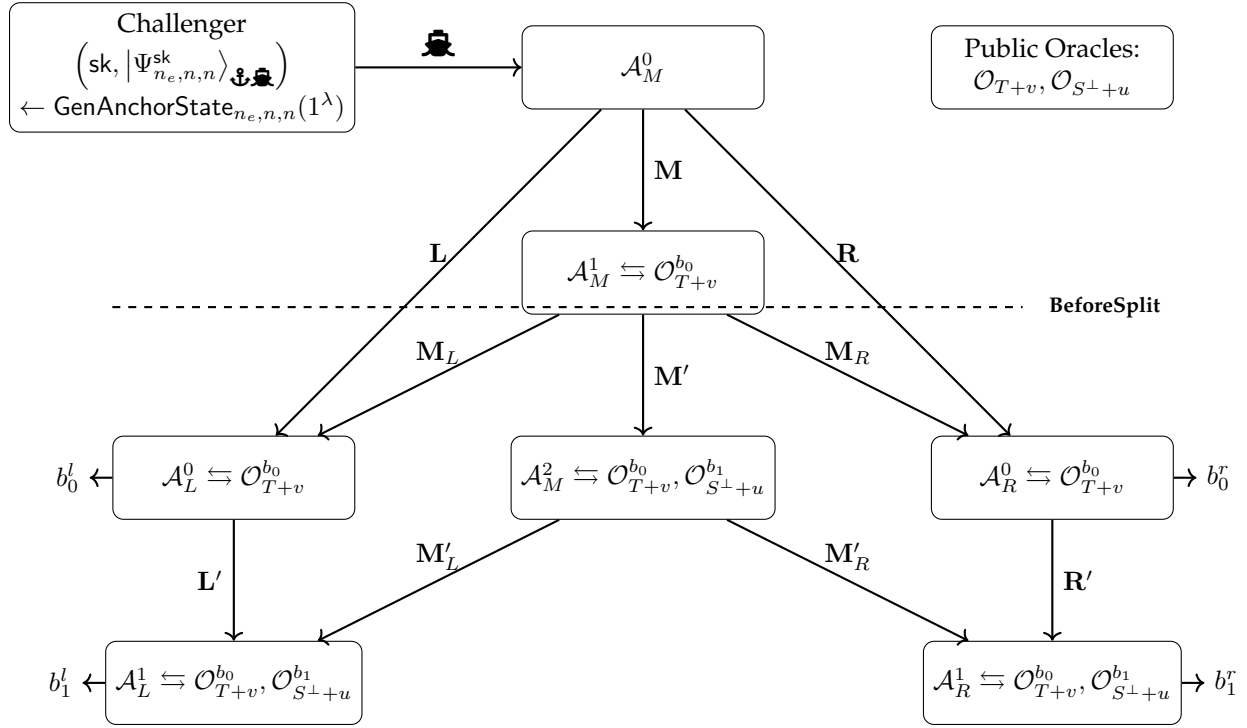
\begin{figure}[h]
    \centering
    \begin{tikzpicture}[
        node distance=1.6cm and 2.2cm,
        every node/.style={font=\small},
        box/.style={draw, rounded corners, align=center, minimum width=3.2cm, minimum height=1.1cm},
        oracle/.style={draw, dashed, rounded corners, align=center, minimum width=3.2cm, minimum height=1.1cm},
        arr/.style={->, thick}
    ]
    
    \node[box] (chal) {Challenger\\
    $\brackets{\sk,\ket{\chain^{\sk}_{n_e,n,n}}_{\anchor\vessel}}$\\
    $\gets\genchain_{n_e,n,n}(1^\secp)$};
    
    \node[box, right=of chal] (AM0) {$\mathcal A_M^0$};
    \node[box, right=of AM0] (PublicOracles) {Public Oracles:\\
    $\cal{O}_{T+v},\cal{O}_{S^\perp+u}$};
    
    \node[box, below=of AM0] (AM1) {$\mathcal A_M^1\leftrightarrows\cal{O}_{T+v}^{b_0}$};
    \node[box, below left=of AM1] (AL0) {$\mathcal A_L^0\leftrightarrows\cal{O}_{T+v}^{b_0}$};
    \node[box, below right=of AM1] (AR0) {$\mathcal A_R^0\leftrightarrows\cal{O}_{T+v}^{b_0}$};
    
    \node[box, below=of AM1] (AM2) {$\mathcal A_M^2\leftrightarrows\cal{O}_{T+v}^{b_0},\cal{O}_{S^\perp+u}^{b_1}$};
    
    \node[box, below=of AL0] (AL1) {$\mathcal A_L^1\leftrightarrows\cal{O}_{T+v}^{b_0},\cal{O}_{S^\perp+u}^{b_1}$};
    \node[box, below=of AR0] (AR1) {$\mathcal A_R^1\leftrightarrows\cal{O}_{T+v}^{b_0},\cal{O}_{S^\perp+u}^{b_1}$};
    
    \draw[arr] (chal) -- node[above] {$\vessel$} (AM0);
    
    \draw[arr] (AM0) -- node[above] {$\lreg$} (AL0);
    \draw[arr] (AM0) -- node[above] {$\rreg$} (AR0);
    \draw[arr] (AM0) -- node[right] {$\mreg$} (AM1);
    
    \draw[arr] (AM1) -- node[above] {$\mreg_L$} (AL0);
    \draw[arr] (AM1) -- node[above] {$\mreg_R$} (AR0);
    \draw[arr] (AM1) -- node[right] {$\mreg'$} (AM2);
    
    \draw[arr] (AL0) -- node[left] {$\lreg'$} (AL1);
    \draw[arr] (AR0) -- node[right] {$\rreg'$} (AR1);
    
    \draw[arr] (AM2) -- node[above] {$\mreg'_L$} (AL1);
    \draw[arr] (AM2) -- node[above] {$\mreg'_R$} (AR1);
    \path
    ($(AL0.north west)+(-1.2cm,0)$)
    ($(AR0.north east)+(1.2cm,0)$);
    
    \draw[dashed, thick]
    ($(AL0.north west)+(0cm,1.7cm)$) --
    ($(AR0.north east)+(-2.0cm,1.7cm)$)
    node[right=6pt, font=\scriptsize] {\BeforeSplit};

    \node[left=0.3cm of AL0] (b0l){$b_0^l$};
    \node[right=0.3cm of AR0] (b0r){$b_0^r$};
    \node[left=0.3cm of AL1] (b1l){$b_1^l$};
    \node[right=0.3cm of AR1] (b1r){$b_1^r$};
    \draw[arr] (AL0) -- node[left] {} (b0l);
    \draw[arr] (AR0) -- node[right] {} (b0r);
    \draw[arr] (AL1) -- node[left] {} (b1l);
    \draw[arr] (AR1) -- node[right] {} (b1r);
    
    \end{tikzpicture}

    \caption{The multi-stage monogamy-of-entanglement game \MultiStageSearchMonogamy/\MultiStageDecisionMonogamy.}
    \label{fig:multi-stage_monogamy-of-entanglement}
\end{figure}
The main result of this section is \Cref{lem:multi-stage_search_monogamy-of-entanglement}. Intuitively, it states that any adversary passing the above game with good probability has to recover the value of $b_0$ by querying its oracle already at the time \textbf{BeforeSplit}. Formally, we show that for each value of $b_0$, the states corresponding to $b_0$ at \textbf{BeforeSplit} are almost orthogonal to each other. Our proof establishing this fact will involve several reductions between various monogamy-of-entanglement games.
\subsection{Asymmetric Monogamy-of-Entanglement Games}
\begin{theorem}[\cite{cryptoeprint:2025/1219} Theorem 5.4]\label{thm:coset_monogamy}
    Let $A\leq\mathbb{F}_2^{n}$ be a uniformly random subspace of dimension $n/2$ described by the first $n/2$ columns of a uniformly random change of basis matrix $U$. Let $s$ be a vector in $\CS(A)$ and $t$ be a vector in $\CS(A^\perp)$. Define the coset state as
    \begin{equation*}
        \ket{A_{s,t}}=\frac{1}{\sqrt{\abs{A}}}\sum_{a\in A}(-1)^{\ipd{a,t}}\ket{a+s},
    \end{equation*}
    Consider the following coset monogamy game for adversary $\cal{A}=(\cal{A}_M,\cal{A}_L,\cal{A}_R)$.
    \begin{enumerate}
        \item The challenger picks a random change of basis matrix $U$ that describes a uniformly random subspace $A\subseteq\mathbb{F}_2^n$ of dimension $n/2$, samples $s\rand \CS(A)$ and $t\rand\CS(A^\perp)$. The challenger sends $\ket{A_{s,t}}_{\mreg}$ to the adversary $\cal{A}_M$.
        \item $\cal{A}_M$ generates a bipartite state on $\lreg\rreg$, sends $\lreg$ to $\cal{A}_L$ and sends $\rreg$ to $\cal{A}_R$.
        \item The challenger samples $r_L,r_R\rand\mathbb{F}_2^n$. $\cal{A}_L$ is given $(U,r_L)$ and $\cal{A}_R$ is given $(U,r_R)$. $\cal{A}_L$ returns $b^l_0$ and $\cal{A}_R$ returns $b^r_1$.
    \end{enumerate}
    The adversary wins if and only if $b^l_0\oplus b^r_1=\ipd{r_L,s}\oplus\ipd{r_R,t}$. Let $\CosetMonogamy(\cal{A},1^\secp)$ be the random variable that takes value $1/0/\bot$ if the adversary $\cal{A}$ wins/loses/aborts in the above game, respectively. For any adversary $\cal{A}$,
    \begin{equation*}
        \prob{\CosetMonogamy(\cal{A},1^\secp)=1}\leq\frac{1}{2}+\negl(\secp).
    \end{equation*}
\end{theorem}
\begin{remark}
    The canonical representation set $\CS$ defined in \cite{cryptoeprint:2025/1219} is a bit different from ours. In \cite{cryptoeprint:2025/1219} they don't have $U_A$, instead they define $\Can_A(z)$ as the lexicographically smallest element in $A+z$. However, it is easy to see that the proof in \cite{cryptoeprint:2025/1219} holds even if we define the canonical representation our way.
\end{remark}
Now we present a modified version of the above game where the prover can choose to abort before the test.
\begin{lemma}[Coset Monogamy-of-Entanglement Game With Abort]\label{lem:coset_monogamy_with_abort}
    Let $A\leq\mathbb{F}_2^{n}$ be a uniformly random subspace of dimension $n/2$ described by the first $n/2$ columns of a uniformly random change of basis matrix $U$. Let $s$ be a vector in $\CS(A)$ and $t$ be a vector in $\CS(A^\perp)$. Define the coset state as
    \begin{equation*}
        \ket{A_{s,t}}=\frac{1}{\sqrt{\abs{A}}}\sum_{a\in A}(-1)^{\ipd{a,t}}\ket{a+s},
    \end{equation*}
    Consider the following coset monogamy game with abort for adversary $\cal{A}=(\cal{A}_M,\cal{A}_L,\cal{A}_R)$. Modifications are highlighted.
    \begin{enumerate}
        \item The challenger picks a random change of basis matrix $U$ that describes a uniformly random subspace $A\subseteq\mathbb{F}_2^n$ of dimension $n/2$, samples $s\rand \CS(A)$ and $t\rand\CS(A^\perp)$. The challenger sends $\ket{A_{s,t}}_{\mreg}$ to the adversary $\cal{A}_M$.
        \item \textcolor{red}{$\cal{A}_M$ can choose to abort in this step. If it aborts, the output of the game is $\bot$. Otherwise,} $\cal{A}_M$ generates a bipartite state on $\lreg\rreg$, sends $\lreg$ to $\cal{A}_L$ and sends $\rreg$ to $\cal{A}_R$.
        \item The challenger samples $r_L,r_R\rand\mathbb{F}_2^n$. $\cal{A}_L$ is given $(U,r_L)$ and $\cal{A}_R$ is given $(U,r_R)$. $\cal{A}_L$ returns $b^l_0$ and $\cal{A}_R$ returns $b^r_1$.
    \end{enumerate}
    The adversary wins if and only if $b^l_0\oplus b^r_1=\ipd{r_L,s}\oplus\ipd{r_R,t}$. Let $\CosetMonogamyWithAbort(\cal{A},1^\secp)$ be the random variable that takes value $1/0/\bot$ if the adversary $\cal{A}$ wins/loses/aborts in the above game, respectively. For any adversary $\cal{A}$ such that 
    \begin{equation*}
        \prob{\CosetMonogamyWithAbort(\cal{A},1^\secp)\neq\bot}\geq \frac{1}{\poly(\secp)},
    \end{equation*}
    we have
    \begin{equation*}
        \frac{\prob{\CosetMonogamyWithAbort(\cal{A},1^\secp)=1}}{\prob{\CosetMonogamyWithAbort(\cal{A},1^\secp)\neq\bot}}\leq\frac{1}{2}+\negl(\secp).
    \end{equation*}
\end{lemma}
\begin{proof}
    Suppose that there exists an adversary $\cal{A}$ and an inverse polynomial $\delta(\secp)$ such that 
    \begin{equation*}
        \prob{\CosetMonogamyWithAbort(\cal{A},1^\secp)\neq\bot}\geq\delta(\secp),
    \end{equation*}
    and
    \begin{equation*}
        \frac{\prob{\CosetMonogamyWithAbort(\cal{A},1^\secp)=1}}{\prob{\CosetMonogamyWithAbort(\cal{A},1^\secp)\neq\bot}}\geq\frac{1}{2}+\delta(\secp).
    \end{equation*}
    Then we construct an adversary $\cal{B}$ that breaks \Cref{thm:coset_monogamy}. $\cal{B}$ works as follows:
    \begin{enumerate}
        \item $\cal{B}_M$ runs $\cal{A}_M$. If $\cal{A}_M$ aborts then $\cal{B}_M$ sends an abort symbol to $\cal{B}_L$ and $\cal{B}_R$. Otherwise, it forwards $\lreg\rreg$, $\cal{A}_M$'s output to $\cal{B}_L$ and $\cal{B}_R$, respectively.
        \item $\cal{B}_L$ outputs a random bit if an abort symbol is received. Otherwise, it runs $\cal{A}_L$ and forwards the answer.
        \item $\cal{B}_R$ outputs a random bit if an abort symbol is received. Otherwise, it runs $\cal{A}_R$ and forwards the answer.
    \end{enumerate}
    Condition on not aborting, the winning probability of $\cal{B}$ is at least $\frac{1}{2}+\delta(\secp)$. On the other hand, if $\cal{A}_M$ aborts, the winning probability of $\cal{B}$ is exactly $\frac{1}{2}$ because it outputs two random bits. Thus, the overall winning probability of $\cal{B}$ is
    \begin{equation*}
        \delta(\secp)\cdot\brackets{\frac{1}{2}+\delta(\secp)}+\brackets{1-\delta(\secp)}\cdot\frac{1}{2}=\frac{1}{2}+(\delta(\secp))^2.
    \end{equation*}
    By contradiction, we proved the lemma.
\end{proof}
\begin{lemma}\label{lem:asymmetric_coset_monogamy}
    Consider the following asymmetric coset monogamy game for adversary $\cal{A}=(\cal{A}_M,\cal{A}_L,\cal{A}_R)$.
    \begin{enumerate}
        \item The challenger picks a random change of basis matrix $U$ that describes a uniformly random subspace $A\subseteq\mathbb{F}_2^n$ of dimension $n/2$ using its first $n/2$ columns, samples $s\rand \CS(A)$ and $t\rand\CS(A^\perp)$. The challenger sends $\ket{A_{s,t}}_{\mreg}$ to the adversary $\cal{A}_M$.
        \item $\cal{A}_M$ can choose to abort in this step. If it aborts, the output of the game is $\bot$. Otherwise, $\cal{A}_M$ generates a bipartite state on $\lreg\rreg$, sends $\lreg$ to $\cal{A}_L$ and sends $\rreg$ to $\cal{A}_R$.
        \item The challenger samples $r_L,r_R\rand\mathbb{F}_2^n$. $\cal{A}_L$ is given $(U,r_L)$ and $\cal{A}_R$ is given $(\ipd{r_L,s},U,r_R)$. $\cal{A}_L$ returns $b^l_0$ and $\cal{A}_R$ returns $b^r_1$.
    \end{enumerate}
    The adversary wins if and only if $b^l_0=\ipd{r_L,s}$ and $b^r_1=\ipd{r_R,t}$. Let $\AsymmetricCosetMonogamy(\cal{A},1^\secp)$ be the random variable that takes value $1/0/\bot$ if the adversary $\cal{A}$ wins/loses/aborts in the above game, respectively. For any efficient adversary $\cal{A}$ such that
    \begin{equation*}
        \prob{\AsymmetricCosetMonogamy(\cal{A},1^\secp)\neq\bot}\geq\frac{1}{\poly(\secp)},
    \end{equation*}
    we have
    \begin{equation*}
        \frac{\prob{\AsymmetricCosetMonogamy(\cal{A},1^\secp)=1}}{\prob{\AsymmetricCosetMonogamy(\cal{A},1^\secp)\neq\bot}}\leq\frac{1}{2}+\negl(\secp).
    \end{equation*}
\end{lemma}
\begin{proof}
    The difference between this lemma and \Cref{lem:coset_monogamy_with_abort} is that $\cal{A}_R$ is given the correct answer of $\cal{A}_L$. For any adversary $\cal{A}$ that wins with probability at least $\frac{1}{2}+\delta(\secp)$ in this game condition on not aborting and its non-aborting probability is at least $\delta(\secp)$ for some inverse polynomial $\delta(\secp)$, we construct adversary $\cal{B}$ that violates \Cref{lem:coset_monogamy_with_abort}. $\cal{B}$ is constructed as follows:
    \begin{enumerate}
        \item $\cal{B}_M$ runs $\cal{A}_M$ on $\mreg$ to obtain $\lreg\rreg$. It aborts if $\cal{A}_M$ aborts.
        \item $\cal{B}_L$ receives $(U,r_L)$ and runs $\cal{A}_L$ on $\lreg$. It outputs $b^l_0$, the output of $\cal{A}_L$.
        \item $\cal{B}_R$ receives $(U,r_R)$ and samples $b_0\rand\mathbb{F}_2$ as a guess to the answer. It runs $\cal{A}_R$ with input $(b_0,U,r_R)$ on $\rreg$. It outputs $b^r_1$, the output of $\cal{A}_R$.
    \end{enumerate}
    The aborting probability of $\cal{B}$ is exactly the same as the aborting probability of $\cal{A}$, now let us investigate the winning probability of $\cal{B}$ condition on not aborting. Define 
    \begin{equation*}
        p_{l,r}^{b}=\probs{U,s,t,r_L,r_R,b_0}{\substack{\mathds{1}[b^l_0=\ipd{r_L,s}]=l\\\mathds{1}[b^r_1=\ipd{r_R,t}]=r}\middle\vert \substack{\mathds{1}[b_0=\ipd{r_L,s}]=b\\ \cal{A}\text{ does not abort}}}.
    \end{equation*}
    where $\mathds{1}[a=b]$ is $1$ when $a=b$ and $0$ otherwise.
    From our assumption on $\cal{A}$ we have
    \begin{equation*}
        p_{1,1}^1=\probs{U,s,t,r_L,r_R,b_0}{\substack{b^l_0=\ipd{r_L,s}\\b^r_1=\ipd{r_R,t}}\middle\vert \substack{b_0=\ipd{r_L,s}\\ \cal{A}\text{ does not abort}}}\geq\frac{1}{2}+\delta(\secp).
    \end{equation*}
    Notice that whether the event $b_0=\ipd{r_L,s}$ happens or not cannot be detected by both parties with noticeable probability. For $\cal{B}_L$, it does not know whether $b_0$ is sampled correctly, so the correctness of $b_0$ will not affect its winning probability on outputting the correct $b^l_0=\ipd{r_L,s}$. For $\cal{B}_R$, as long as $s\neq 0$ which happens with $1-\negl$ probability, it has no information about whether $r_L$ satisfies $\ipd{r_L,s}=b_0$. So whether $\ipd{r_L,s}=b_0$ or not will not affect the winning probability of $\cal{B}_R$. Thus we have for $l=0,1$,
    \begin{equation*}
        p_{l,0}^1+p_{l,1}^1=p_{l,0}^0+p_{l,1}^0.
    \end{equation*}
    And we have for $r=0,1$,
    \begin{align*}
        \abs{(p_{0,r}^1+p_{1,r}^1)-(p_{0,r}^0+p_{1,r}^0)}=\negl(\secp)
    \end{align*}
    Thus there exists a constant $c$ such that 
    \begin{equation*}
        \norm{
        \begin{pmatrix}
            p_{0,0}^0 & p_{0,1}^0 \\
            p_{1,0}^0 & p_{1,1}^0 \\
        \end{pmatrix}
        -
        \begin{pmatrix}
            p_{0,0}^1 & p_{0,1}^1 \\
            p_{1,0}^1 & p_{1,1}^1 \\
        \end{pmatrix}
        +c
        \begin{pmatrix}
            1 & -1 \\
            -1 & 1 \\
        \end{pmatrix}
        }_{\infty}=\negl(\secp).
    \end{equation*}
    From this we know that the total winning probability of $\cal{B}$ is
    \begin{align*}
        &\frac{1}{2}(p_{0,0}^0+p_{1,1}^0+p_{0,0}^1+p_{1,1}^1)-\negl(\secp)\\
        \geq&\frac{1}{2}(-p_{0,0}^0+p_{1,1}^0+p_{0,0}^1+p_{1,1}^1)-\negl(\secp)\\
        \geq&\frac{1}{2}(-p_{0,0}^1+p_{1,1}^1+p_{0,0}^1+p_{1,1}^1)-\negl(\secp)\\
        \geq&p_{1,1}^1-\negl(\secp)\\
        \geq&\frac{1}{2}+\delta(\secp)-\negl(\secp).
    \end{align*}
    This yields a contradiction, proving the lemma.
\end{proof}
Intuitively, $\ket{A_{s,t}}$ and $\ket{\chain^\sk_{0,n/2,n/2}}$ have the same distribution. The coset state is exactly the anchor state with $n_e=0,n_h=n_s=n/2$. We can use this theorem to prove a similar statement in our case.
\begin{lemma}\label{lem:asymmetric_anchor_monogamy}
    Let $n(\secp)\geq\secp,n_e(\secp)$ be polynomials. Define the following asymmetric anchor monogamy game \AsymmetricAnchorMonogamy for adversary $\cal{A}=(\cal{A}_M,\cal{A}_L,\cal{A}_R)$.
    \begin{enumerate}
        \item The challenger samples $\brackets{\sk,\ket{\chain^{\sk}_{n_e,n/2,n/2}}_{\anchor\vessel}}\gets\genchain_{n_e,n/2,n/2}(1^\secp)$ and sends $\vessel$ to the adversary $\cal{A}_M$.
        \item $\cal{A}_M$ can choose to abort in this step. If it aborts, the output of the game is $\bot$. Otherwise, $\cal{A}_M$ generates a bipartite state on $\lreg\rreg$, sends $\lreg$ to $\cal{A}_L$ and sends $\rreg$ to $\cal{A}_R$.
        \item The challenger samples $r_L,r_R\rand\mathbb{F}_2^{n_e+n}$. $\cal{A}_L$ is given $(U_{\sf shift},r_L)$ and $\cal{A}_R$ is given $(\ipd{r_L,v},U_{\sf shift},r_R)$. $\cal{A}_L$ returns $b^l_0$ and $\cal{A}_R$ returns $b^r_1$.
    \end{enumerate}
    The adversary wins if and only if $b^l_0=\ipd{r_L,v}$ and $b^r_1=\ipd{r_R,u}$. Let $\AsymmetricAnchorMonogamy(\cal{A},1^\secp)$ be the random variable that takes value $1/0/\bot$ if the adversary $\cal{A}$ wins/loses/aborts in the above game, respectively. For any efficient adversary $\cal{A}$ such that
    \begin{equation*}
        \prob{\AsymmetricAnchorMonogamy(\cal{A},1^\secp)\neq\bot}\geq\frac{1}{\poly(\secp)},
    \end{equation*}
    we have
    \begin{equation*}
        \frac{\prob{\AsymmetricAnchorMonogamy(\cal{A},1^\secp)=1}}{\prob{\AsymmetricAnchorMonogamy(\cal{A},1^\secp)\neq\bot}}\leq\frac{1}{2}+\negl(\secp).
    \end{equation*}
\end{lemma}
\begin{proof}
    The proof is similar to the domain extension technique, the only difference is that we are extending the $\epr$ part. For any adversary $\cal{A}$ that wins with probability at least $\frac{1}{2}+\delta(\secp)$ in this game condition on not aborting and its non-aborting probability is at least $\delta(\secp)$ for some inverse polynomial $\delta(\secp)$, we construct adversary $\cal{B}$ that violates \Cref{lem:asymmetric_coset_monogamy}:
    \begin{enumerate}
        \item $\cal{B}_M$ receives $\ket{A_{s,t}}$ on $n$-qubit register $\areg$, and tensors it with a $2n_e$-qubit $\epr$ state $\ket{\epr}_{\anchor\areg'}^{\otimes n_e}$ which is fully entangled between $n_e$-qubit register $\anchor$ and $n_e$-qubit register $\areg'$. Let $\vessel=\areg\otimes\areg'$, it samples a uniform random change of basis matrix $U_{\sf shift}$ in $\mathbb{F}_2^{n_e+n}$ and applies
        \begin{equation*}
            \cal{U}_{\sf shift}\ket{z}_{\vessel}:=\ket{U_{\sf shift}z}_{\vessel}
        \end{equation*}
        on $\areg\otimes\areg'$. Finally, it runs $\cal{A}_M$ on $\vessel$ to obtain $\lreg\rreg$. It aborts if $\cal{A}_M$ aborts. It sends $\lreg,U_{\sf shift}$ to $\cal{B}_L$ and sends $\rreg,U_{\sf shift}$ to $\cal{B}_R$.
        \item $\cal{B}_L$ receives $\lreg$ and $(U,r_L)$. It samples $r'_L\rand\mathbb{F}_2^{n_e}$, computes 
        \[U'_{\sf shift}=U_{\sf shift}\begin{pmatrix}
                U &  \\
                 & I_{n_e} \\
            \end{pmatrix}\begin{pmatrix}
                I_{n/2} & & \\
                 & & I_{n/2}\\
                 & I_{n_e}\\
            \end{pmatrix}.\]
        It runs $\cal{A}_L$ on $\lreg$ with input
        \begin{equation*}
            \brackets{U'_{\sf shift},U_{\sf shift}^{-t}\brackets{r_L\times r'_L}}.
        \end{equation*}
        It outputs $b^l_0$, the output of $\cal{A}_L$.
        \item $\cal{B}_R$ receives $\rreg$ and $(b_0,U,r_R)$. It samples $r'_R\rand\mathbb{F}_2^{n_e}$, computes $U'_{\sf shift}$ in the same way and runs $\cal{A}_R$ on $\rreg$ with input 
        \begin{equation*}
            \brackets{b_0,U'_{\sf shift},U_{\sf shift}\brackets{r_R\times r'_R}}.
        \end{equation*}
        It outputs $b^r_1$, the output of $\cal{A}_R$.
    \end{enumerate}
    First note that the distribution of the state on $\anchor\vessel$ after step 1 together with $U'_{\sf shift}$ has the same distribution as $\brackets{U_{\sf shift},\ket{\chain^{\sk}_{n_e,n/2,n/2}}}$ sampled by $\genchain_{n_e,n/2,n/2}(1^\secp)$. Let $v'$ be the last $n/2$ bits of $U^{-1}v$. If $\cal{A}_L$ returns correctly, its output is
    \begin{align*}
        &\ipd{U_{\sf shift}^{-t}(r_L\times r'_L),U_{\sf shift}(s\times 0^{n_e})}\\
        =&\ipd{r_L\times r'_L,s\times 0^{n_e}}\\
        =&\ipd{r_L,s}.
    \end{align*}
    Similarly, if $\cal{A}_R$ returns correctly, its output is
    \begin{align*}
        &\ipd{U_{\sf shift}(r_R\times r'_R),U_{\sf shift}^{-t}(t\times 0^{n_e})}\\
        =&\ipd{r_R\times r'_R,t\times 0^{n_e}}\\
        =&\ipd{r_R,t}.
    \end{align*}
\end{proof}
Now we switch from the Goldreich-Levin style query to the oracle style query.
\newcommand{\AsymmetricOracleMonogamy}{{\sf AsymmetricOracleMonogamy}\xspace}
\begin{lemma}\label{lem:asymmetric_oracle_monogamy}
    Let $n(\secp)\geq\secp$, $n_e(\secp)$, and $m_1(\secp)$ be polynomials. Define the following asymmetric oracle monogamy game \AsymmetricOracleMonogamy for adversary $\cal{A}=(\cal{A}_M,\cal{A}_L,\cal{A}_R)$.
    \begin{enumerate}
        \item The challenger generates $\brackets{\sk,\ket{\chain^{\sk}_{n_e,3n/4,3n/4}}_{\anchor\vessel}}\gets\genchain_{n_e,3n/4,3n/4}(1^\secp)$ and gives the register $\vessel$ to $\cal{A}_M$.
        \item $\cal{A}_M$ can choose to abort in this step. If it aborts, the output of the game is $\bot$. Otherwise, $\cal{A}_M$ generates a bipartite state on $\lreg\rreg$. It sends $\lreg$ to $\cal{A}_L$ and sends $\rreg$ to $\cal{A}_R$.
        \item The challenger samples $b_0\rand\mathbb{F}_2$ and $b_1\rand\mathbb{F}_2^{m_1}$. It gives $\cal{O}_{T+v}^{b_0}$ to $\cal{A}_L$. Then it gives $b_0$ and $\cal{O}_{S^\perp+u}^{b_1}$ to $\cal{A}_R$.
        \item $\cal{A}_L$ returns $b_0^l$ and $\cal{A}_R$ returns $b_1^r$. 
    \end{enumerate}
    The adversary wins if $b_0^l=b_0$ and $b_1^r=b_1$. Note that the oracle given to $\cal{A}_L$ has one-bit output, but the oracle given to $\cal{A}_R$ has $m_1$ bits output. Let $\AsymmetricOracleMonogamy(\cal{A},1^\secp)$ be the random variable that takes value $1/0/\bot$ if the adversary $\cal{A}$ wins/loses/aborts in the above game, respectively. For any efficient adversary $\cal{A}$ such that
    \begin{equation*}
        \prob{\AsymmetricOracleMonogamy(\cal{A},1^\secp)\neq\bot}\geq\frac{1}{\poly(\secp)},
    \end{equation*}
    we have
    \begin{equation*}
        \frac{\prob{\AsymmetricOracleMonogamy(\cal{A},1^\secp)=1}}{\prob{\AsymmetricOracleMonogamy(\cal{A},1^\secp)\neq\bot}}\leq\frac{1}{2}+\negl(\secp).
    \end{equation*}
\end{lemma}
\begin{proof}
    For any adversary $\cal{A}$ that wins with probability at least $\frac{1}{2}+\delta(\secp)$ in this game condition on not aborting and its non-aborting probability is at least $\delta(\secp)$ for some inverse polynomial $\delta(\secp)$, we consider a series of hybrids. \\
    \textbf{Hybrid 0:} This hybrid corresponds to the adversary $\cal{A}$ in the $\AsymmetricOracleMonogamy$ game. The non-aborting probability for $\cal{A}$ is at least $\delta(\secp)$ and conditional on that, the winning probability is at least $\frac{1}{2}+\delta(\secp)$.\\
    \textbf{Hybrid 1:} This hybrid corresponds to the adversary $\cal{A}$ in the following game:
    \begin{enumerate}
        \item The challenger generates \textcolor{red}{$\brackets{\sk,\ket{\chain^{\sk}_{n_e,n/2,n/2}}_{\anchor\vessel}}\gets\genchain_{n_e,n/2,n/2}(1^\secp)$. It then runs $\sk^*\gets\domainextension_{n/4,n/4}(1^\secp)$ on $\vessel$ and gives the register $\vessel^*$ to $\cal{A}_M$.}
        \item $\cal{A}_M$ can choose to abort in this step. If it aborts, the output of the game is $\bot$. Otherwise, $\cal{A}_M$ generates a bipartite state on $\lreg\rreg$. It sends $\lreg$ to $\cal{A}_L$ and sends $\rreg$ to $\cal{A}_R$.
        \item The challenger samples $b_0\rand\mathbb{F}_2$ and $b_1\rand\mathbb{F}_2^{m_1}$. It gives \textcolor{red}{$\cal{O}_{T_{\sf ext}+v_{\sf ext}}^{b_0}$} to $\cal{A}_L$. Then it gives $b_0$ and \textcolor{red}{$\cal{O}_{S_{\sf ext}^\perp+u_{\sf ext}}^{b_1}$} to $\cal{A}_R$.
        \item $\cal{A}_L$ returns $b_0^l$ and $\cal{A}_R$ returns $b_1^r$. The adversary wins if $b_0^l=b_0$ and $b_1^r=b_1$.
    \end{enumerate}
    \begin{claim}
        The non-aborting probability for $\cal{A}$ in \textbf{Hybrid 1} is at least $\delta(\secp)$ and conditional on that, the winning probability is at least $\frac{1}{2}+\delta(\secp)$.
    \end{claim}
    \begin{proof}
        By \Cref{fact:extension_indistinguishability} and \Cref{fact:simulating_extension_oracles}, the input distribution of $\cal{A}$ doesn't change.
    \end{proof}
    \textbf{Hybrid 2:} This hybrid corresponds to the adversary $\cal{A}$ in the following game:
    \begin{enumerate}
        \item The challenger generates $\brackets{\sk,\ket{\chain^{\sk}_{n_e,n/2,n/2}}_{\anchor\vessel}}\gets\genchain_{n_e,n/2,n/2}(1^\secp)$. It then runs $\sk^*\gets\domainextension_{n/4,n/4}(1^\secp)$ on $\vessel$ and gives the register $\vessel^*$ to $\cal{A}_M$.
        \item $\cal{A}_M$ can choose to abort in this step. If it aborts, the output of the game is $\bot$. Otherwise, $\cal{A}_M$ generates a bipartite state on $\lreg\rreg$. It sends $\lreg$ to $\cal{A}_L$ and sends $\rreg$ to $\cal{A}_R$.
        \item The challenger samples $b_0\rand\mathbb{F}_2$ and $b_1\rand\mathbb{F}_2^{m_1}$. It gives \textcolor{red}{$\cal{O}_{T^*+v^*}^{b_0}$} to $\cal{A}_L$. Then it gives $b_0$ and \textcolor{red}{$\cal{O}_{{S^*}^\perp+u^*}^{b_1}$} to $\cal{A}_R$.
        \item $\cal{A}_L$ returns $b_0^l$ and $\cal{A}_R$ returns $b_1^r$. The adversary wins if $b_0^l=b_0$ and $b_1^r=b_1$.
    \end{enumerate}
    \begin{claim}
        The non-aborting probability for $\cal{A}$ in \textbf{Hybrid 2} is at least $\delta(\secp)$ and conditional on that, the winning probability is at least $\frac{1}{2}+\delta(\secp)-\negl(\secp)$.
    \end{claim}
    \begin{proof}
        $T^*$ is a uniformly random subspace of dimension $n_e+5n/4$ satisfying $T_{\sf ext}\leq T^*\leq\mathbb{F}_2^{n_e+3n/2}$ and it is also independent of $v_{\sf ext}$. $S^*$ is a uniformly random subspace of $S_{\sf ext}$ of dimension $n/4$ and it is also independent of $u_{\sf ext}$. The claim follows from \Cref{lem:subspace_oracle_indistinguishability}.
    \end{proof}
    \textbf{Hybrid 3:} This hybrid corresponds to the adversary $\cal{A}$ in the following game:

    \begin{enumerate}
        \item The challenger generates $\brackets{\sk,\ket{\chain^{\sk}_{n_e,n/2,n/2}}_{\anchor\vessel}}\gets\genchain_{n_e,n/2,n/2}(1^\secp)$. It then runs $\sk^*\gets\domainextension_{n/4,n/4}(1^\secp)$ on $\vessel$ and gives the register $\vessel^*$ to $\cal{A}_M$. 
        \item $\cal{A}_M$ can choose to abort in this step. If it aborts, the output of the game is $\bot$. Otherwise, $\cal{A}_M$ generates a bipartite state on $\lreg\rreg$. It sends $\lreg$ to $\cal{A}_L$ and sends $\rreg$ to $\cal{A}_R$.
        \item The challenger samples \textcolor{red}{$r_L\rand\mathbb{F}_2^{n_e+3n/2}$} and $b_1\rand\mathbb{F}_2^{m_1}$. \textcolor{red}{It also samples $w_v\rand T_{\sf ext}+v_{\sf ext}$ and gives the following oracle to $\cal{A}_L$:}
        \begin{itemize}
            \item \textcolor{red}{$\cal{O}_{T^*+w_v}^{{\sf ncol},r_L}:\mathbb{F}_2^{n_e+3n/2}\rightarrow\mathbb{F}_2\cup\set{\bot}$ on input $z$ checks if it is in $T^*+w_v$. It returns $\bot$ if it is not. Otherwise, it returns $\ipd{r_L,\Can_{T_{\sf ext}}(w_v)}$}.
        \end{itemize}
        Then it gives \textcolor{red}{$\ipd{r_L,\Can_{T_{\sf ext}}(w_v)}$} and $\cal{O}_{{S^*}^\perp+u^*}^{b_1}$ to $\cal{A}_R$.
        \item $\cal{A}_L$ returns $b_0^l$ and $\cal{A}_R$ returns $b_1^r$. The adversary wins if \textcolor{red}{$b_0^l=\ipd{r_L,\Can_{T_{\sf ext}}(w_v)}$} and $b_1^r=b_1$.
    \end{enumerate}
    \begin{claim}
        The non-aborting probability for $\cal{A}$ in \textbf{Hybrid 3} is at least $\delta(\secp)$ and conditional on that, the winning probability is at least $\frac{1}{2}+\delta(\secp)-\negl(\secp)$.
    \end{claim}
    \begin{proof}
        In the adversary's view, $\ipd{r_L,v_{\sf ext}}$ is a uniform random bit and is equivalent to $b_0$. Also $T_{\sf ext}+v_{\sf ext}$ and $T_{\sf ext}+w_v$ are the same coset.
    \end{proof}
    \textbf{Hybrid 4:} This hybrid corresponds to the adversary $\cal{A}$ in the following game:
    \begin{enumerate}
        \item The challenger generates $\brackets{\sk,\ket{\chain^{\sk}_{n_e,n/2,n/2}}_{\anchor\vessel}}\gets\genchain_{n_e,n/2,n/2}(1^\secp)$. It then runs $\sk^*\gets\domainextension_{n/4,n/4}(1^\secp)$ on $\vessel$ and gives the register $\vessel^*$ to $\cal{A}_M$. 
        \item $\cal{A}_M$ can choose to abort in this step. If it aborts, the output of the game is $\bot$. Otherwise, $\cal{A}_M$ generates a bipartite state on $\lreg\rreg$. It sends $\lreg$ to $\cal{A}_L$ and sends $\rreg$ to $\cal{A}_R$.
        \item The challenger samples $r_L\rand\mathbb{F}_2^{n_e+3n/2}$ and $b_1\rand\mathbb{F}_2^{m_1}$. It also samples $w_v\rand T_{\sf ext}+v_{\sf ext}$ and gives the following oracle to $\cal{A}_L$:
        \begin{itemize}
            \item \textcolor{red}{$\cal{O}_{T^*+w_v}^{{\sf col},r_L}$}$:\mathbb{F}_2^{n_e+3n/2}\rightarrow\mathbb{F}_2\cup\set{\bot}$ on input $z$ checks if it is in $T^*+w_v$. It returns $\bot$ if it is not. Otherwise, \textcolor{red}{it returns $\ipd{r_L,\Can_{T_{\sf ext}}(z)}$}.
        \end{itemize}
        Then it gives \textcolor{red}{$\ipd{r_L,v_{\sf ext}}$} and $\cal{O}_{{S^*}^\perp+u^*}^{b_1}$ to $\cal{A}_R$.
        \item $\cal{A}_L$ returns $b_0^l$ and $\cal{A}_R$ returns $b_1^r$. The adversary wins if \textcolor{red}{$b_0^l=\ipd{r_L,v_{\sf ext}}$} and $b_1^r=b_1$.
    \end{enumerate}
    \begin{claim}\label{claim:use_collapsing}
        The non-aborting probability for $\cal{A}$ in \textbf{Hybrid 4} is at least $\delta(\secp)$ and conditional on that, the winning probability is at least $\frac{1}{2}+\delta(\secp)-\negl(\secp)$.
    \end{claim}
    \begin{proof}
        Note that $w_v$ is a uniform random vector independent of $T_{\sf ext}$ and $T^*$, with $\Can_{T_{\sf ext}}(w_v)=v_{\sf ext}$. For any coset of $T_{\sf ext}$ contained in $T^*+v^*$, $\ipd{r_L,\Can_{T_{\sf ext}}(z)}$ alone is a bit independent of $T^*$ given $T_{\sf ext}$. To use \Cref{lem:collapsing_versus_non-collapsing} we only need to show that we can generate the anchor state and simulate the whole game using $T_{\sf ext}$ and $w_v$. To be more specific, we will rewrite \textbf{Hybrid 3} and \textbf{Hybrid 4} as two games with input $(T_{\sf ext},T^*,w_v)$ where one of them is the collapsing version and the other one is the non-collapsing version in order to use \Cref{lem:collapsing_versus_non-collapsing}. Now we rewrite \textbf{Hybrid 3} and \textbf{Hybrid 4} for fixed $T_{\sf ext}$, $T^*$ and $w_v$. 
        \begin{enumerate}
            \item Sample random subspaces $S^*\leq S_{\sf ext}\leq T_{\sf ext}$ where $S^*$ is a $n/4$ dimensional subspace and $S_{\sf ext}$ is a $3n/4$ dimensional subspace. Sample a $(n_e+3n/2)\times(n_e+3n/2)$ uniform random change of basis matrix $U_{\sf shift}$ condition on the span of its first $n/4,3n/4,n_e+3n/4,n_e+5n/4$ columns are $S^*,S_{\sf ext},T_{\sf ext},T^*$, respectively. Sample $u_{\sf ext}\rand\CS(S_{\sf ext}^\perp)$ and compute $u^*$ accordingly.
            \item Decompose $w_v=w_s+w_x+v_{\sf ext}$ where $w_s\in S_{\sf ext}$ and $w_x=U_{\sf shift}\brackets{0^{3n/4}\times x'\times 0^{3n/4}}$ for a unique $x'$. Generate the following state:
            \begin{align*}
                \propto&\sum_{x}\ket{x}_{\anchor}\sum_{s\in S_{\sf ext}}(-1)^{\ipd{u_{\sf ext},s}}\ket{s+U_{\sf shift}\brackets{0^{3n/4}\times x\times 0^{3n/4}}+w_v}_{\vessel^*}\\
                \propto&\sum_{x}\ket{x}_{\anchor}\sum_{s\in S_{\sf ext}}(-1)^{\ipd{u_{\sf ext},s}}\ket{(s+w_s)+U_{\sf shift}\brackets{0^{3n/4}\times (x+x')\times 0^{3n/4}}+v_{\sf ext}}_{\vessel^*}\\
                \propto&\sum_{x}(-1)^{\ipd{u_{\sf ext},w_s}}\ket{x-x'}_{\anchor}\sum_{s\in S_{\sf ext}}(-1)^{\ipd{u_{\sf ext},s}}\ket{s+U_{\sf shift}\brackets{0^{3n/4}\times x\times 0^{3n/4}}+v_{\sf ext}}_{\vessel^*}.
            \end{align*}
            Note that the adversary $\cal{A}$ can only access $\vessel$, thus it cannot detect the phase shift $(-1)^{\ipd{u_{\sf ext},w_s}}$ and the control shift $\ket{x-x'}_{\anchor}$. The result of the game is the same if we give the following state that is exactly the anchor state, instead of the above state.
            \[\propto\sum_{x}\ket{x}_{\anchor}\sum_{s\in S_{\sf ext}}(-1)^{\ipd{u_{\sf ext},s}}\ket{s+U_{\sf shift}\brackets{0^{3n/4}\times x\times 0^{3n/4}}+v_{\sf ext}}_{\vessel^*}.\]
            \item Run $\cal{A}_M$, $\cal{A}_L$, $\cal{A}_R$, provide the corresponding oracles and check the outcome.
        \end{enumerate}
        From the above description, we can use \Cref{lem:collapsing_versus_non-collapsing} to show the indistinguishability.
    \end{proof}
    \textbf{Hybrid 5:} This hybrid corresponds to the adversary $\cal{A}$ in the following game:
    \begin{enumerate}
        \item The challenger generates $\brackets{\sk,\ket{\chain^{\sk}_{n_e,n/2,n/2}}_{\anchor\vessel}}\gets\genchain_{n_e,n/2,n/2}(1^\secp)$. It then runs $\sk^*\gets\domainextension_{n/4,n/4}(1^\secp)$ on $\vessel$ and gives the register $\vessel^*$ to $\cal{A}_M$. 
        \item $\cal{A}_M$ can choose to abort in this step. If it aborts, the output of the game is $\bot$. Otherwise, $\cal{A}_M$ generates a bipartite state on $\lreg\rreg$. It sends $\lreg$ to $\cal{A}_L$ and sends $\rreg$ to $\cal{A}_R$.
        \item The challenger samples $r_L\rand\mathbb{F}_2^{n_e+3n/2}$, \textcolor{red}{$r_R\rand\mathbb{F}_2^{n_e+3n/2}$ and $b'_1\rand\mathbb{F}_2^{m_1-1}$. It gives $\cal{O}_{T^*+v^*}^{{\sf col},r_L}$ to $\cal{A}_L$. Then it samples $w_u\rand S_{\sf ext}^\perp+u_{\sf ext}$, gives $\ipd{r_L,v_{\sf ext}}$ and this oracle to $\cal{A}_R$:}
        \begin{itemize}
            \item \textcolor{red}{$\cal{O}_{{S^*}^\perp+w_u}^{{\sf ncol},r_R,b'_1}:\mathbb{F}_2^{n_e+3n/2}\rightarrow\mathbb{F}_2^{m_1}\cup\set{\bot}$ on input $z$ checks if it is in ${S^*}^\perp+w_u$. It returns $\bot$ if it is not. Otherwise, it returns $\ipd{r_R,\Can_{S_{\sf ext}^\perp}(w_u)}\times b'_1$}.
        \end{itemize}
        \item $\cal{A}_L$ returns $b_0^l$ and $\cal{A}_R$ returns $b_1^r$. The adversary wins if $b_0^l=\ipd{r_L,v_{\sf ext}}$ and \textcolor{red}{$b_1^r=\ipd{r_R,\Can_{S_{\sf ext}^\perp}(w_u)}\times b'_1$}.
    \end{enumerate}
    \begin{claim}
        The non-aborting probability for $\cal{A}$ in \textbf{Hybrid 5} is at least $\delta(\secp)$ and conditional on that, the winning probability is at least $\frac{1}{2}+\delta(\secp)-\negl(\secp)$.
    \end{claim}
    \begin{proof}
        In the adversary's view, $\ipd{r_R,u_{\sf ext}}\times b'_1$ is a uniform random string and is equivalent to $b_1$. Also $T^*+v^*$ and $T^*+w_v$ are the same coset, and similarly ${S^*}^\perp+u^*$ and ${S^*}^\perp+w_u$ are the same coset.
    \end{proof}
    \textbf{Hybrid 6:} This hybrid corresponds to the adversary $\cal{A}$ in the following game:
    \begin{enumerate}
        \item The challenger generates $\brackets{\sk,\ket{\chain^{\sk}_{n_e,n/2,n/2}}_{\anchor\vessel}}\gets\genchain_{n_e,n/2,n/2}(1^\secp)$. It then runs $\sk^*\gets\domainextension_{n/4,n/4}(1^\secp)$ on $\vessel$ and gives the register $\vessel^*$ to $\cal{A}_M$. 
        \item $\cal{A}_M$ can choose to abort in this step. If it aborts, the output of the game is $\bot$. Otherwise, $\cal{A}_M$ generates a bipartite state on $\lreg\rreg$. It sends $\lreg$ to $\cal{A}_L$ and sends $\rreg$ to $\cal{A}_R$.
        \item The challenger samples $r_L\rand\mathbb{F}_2^{n_e+3n/2}$, $r_R\rand\mathbb{F}_2^{n_e+3n/2}$ and $b'_1\rand\mathbb{F}_2^{m_1-1}$. It gives $\cal{O}_{T^*+v^*}^{{\sf col},r_L}$ to $\cal{A}_L$. Then it samples $w_u\rand S_{\sf ext}^\perp+u_{\sf ext}$, gives $\ipd{r_L,v_{\sf ext}}$ and this oracle to $\cal{A}_R$:
        \begin{itemize}
            \item \textcolor{red}{$\cal{O}_{{S^*}^\perp+w_u}^{{\sf col},r_R,b'_1}$}$:\mathbb{F}_2^{n_e+3n/2}\rightarrow\mathbb{F}_2^{m_1}\cup\set{\bot}$ on input $z$ checks if it is in ${S^*}^\perp+w_u$. It returns $\bot$ if it is not. Otherwise, \textcolor{red}{it returns $\ipd{r_R,\Can_{S_{\sf ext}^\perp}(z)}\times b'_1$}.
        \end{itemize}
        \item $\cal{A}_L$ returns $b_0^l$ and $\cal{A}_R$ returns $b_1^r$. The adversary wins if $b_0^l=\ipd{r_L,v_{\sf ext}}$ and \textcolor{red}{$b_1^r=\ipd{r_R,u_{\sf ext}}\times b'_1$}.
    \end{enumerate}
    \begin{claim}
        The non-aborting probability for $\cal{A}$ in \textbf{Hybrid 6} is at least $\delta(\secp)$ and conditional on that, the winning probability is at least $\frac{1}{2}+\delta(\secp)-\negl(\secp)$.
    \end{claim}
    \begin{proof}
        Similar to the proof of \Cref{claim:use_collapsing} except this time in the Hadamard basis.
    \end{proof}
    \textbf{Hybrid 7:} This hybrid corresponds to the adversary $\cal{A}$ in the following game:
    \begin{enumerate}
        \item The challenger generates $\brackets{\sk,\ket{\chain^{\sk}_{n_e,n/2,n/2}}_{\anchor\vessel}}\gets\genchain_{n_e,n/2,n/2}(1^\secp)$. It then runs $\sk^*\gets\domainextension_{n/4,n/4}(1^\secp)$ on $\vessel$ and gives the register $\vessel^*$ to $\cal{A}_M$. 
        \item $\cal{A}_M$ can choose to abort in this step. If it aborts, the output of the game is $\bot$. Otherwise, $\cal{A}_M$ generates a bipartite state on $\lreg\rreg$. It sends $\lreg$ to $\cal{A}_L$ and sends $\rreg$ to $\cal{A}_R$.
        \item The challenger samples $r_L\rand\mathbb{F}_2^{n_e+3n/2}$, $r_R\rand\mathbb{F}_2^{n_e+3n/2}$ and $b'_1\rand\mathbb{F}_2^{m_1-1}$. It gives $\cal{O}_{T^*+v^*}^{{\sf col},r_L}$ to $\cal{A}_L$. \textcolor{red}{Then it gives $\ipd{r_L,v_{\sf ext}}$ and $\cal{O}_{{S^*}^\perp+u^*}^{{\sf col},r_R,b'_1}$ to $\cal{A}_R$.}
        \item $\cal{A}_L$ returns $b_0^l$ and $\cal{A}_R$ returns $b_1^r$. The adversary wins if $b_0^l=\ipd{r_L,v_{\sf ext}}$ and $b_1^r=\ipd{r_R,u_{\sf ext}}\times b'_1$.
    \end{enumerate}
    \begin{claim}
        The non-aborting probability for $\cal{A}$ in \textbf{Hybrid 7} is at least $\delta(\secp)$ and conditional on that, the winning probability is at least $\frac{1}{2}+\delta(\secp)-\negl(\secp)$.
    \end{claim}
    \begin{proof}
        ${S^*}^\perp+u^*$ and ${S^*}^\perp+w_u$ are the same coset.
    \end{proof}
    Using \textbf{Hybrid 7} as a tool, we construct an adversary $\cal{B}$ that violates \Cref{lem:asymmetric_anchor_monogamy}. $\cal{B}$ simulates part of the challenger in \textbf{Hybrid 7}.
    \begin{enumerate}
        \item $\cal{B}_M$ receives the $\vessel$ register. It then runs $\sk^*\gets\domainextension_{n/4,n/4}(1^\secp)$ on $\vessel$ and gives the register $\vessel^*$ to $\cal{A}_M$ to obtain $\lreg\rreg$. $\cal{B}_M$ aborts if $\cal{A}_M$ aborts. It samples $r_L^s\rand\mathbb{F}_2^{n/4}$. It sends $\sk^*,\lreg,r_L^s$ to $\cal{B}_L$ and sends $\sk^*,\rreg,r_L^s$ to $\cal{B}_R$.
        \item $\cal{B}_L$ receives $\sk^*,\lreg,r_L^s$ from $\cal{B}_M$ and $(U_{\sf shift},r_L)$ from the challenger. It samples $r_L^h\rand\mathbb{F}_2^{n/4}$ and computes $r'_L={U_{\sf shift}^*}^{-t}\brackets{r_L^h\times r_L\times r_L^s}$. It runs $\cal{A}_L$ on $\lreg$ with oracle access to $\cal{O}_{T^*+v^*}^{{\sf col},r'_L}$. Let the output of $\cal{A}_L$ be $b^l_0$. It outputs $b^l_0\oplus\ipd{{U_{\sf shift}^*}^{-t}\brackets{0^{n_e+5n/4}\times r_L^s},v^*}$.
        \item $\cal{B}_R$ receives $\sk^*,\rreg,r_L^s$ from $\cal{B}_M$ and $(U_{\sf shift},r_R,b_0)$ from the challenger. It computes $b'_0=b_0\oplus\ipd{{U_{\sf shift}^*}^{-t}\brackets{0^{n_e+5n/4}\times r_L^s},v^*}$. It samples $b'_1\rand\mathbb{F}_2^{m_1-1}$, $r_R^h,r_R^s\rand\mathbb{F}_2^{n/4}$ and computes $r'_R=U_{\sf shift}^*\brackets{r_R^h\times r_R\times r_R^s}$. It runs $\cal{A}_R$ on $\rreg$ with $b'_0$ and oracle access to $\cal{O}_{{S^*}^\perp+u^*}^{{\sf col},r'_R,b'_1}$. Let the first bit of the output of $\cal{A}_R$ be $b^r_1$. It outputs $b^r_1\oplus\ipd{U_{\sf shift}^*\brackets{r_R^h\times 0^{n_e+5n/4}},u^*}$.
    \end{enumerate}
    \begin{claim}
        We have
        \begin{equation*}
            \prob{\AsymmetricAnchorMonogamy(\cal{B},1^\secp)\neq\bot}\geq\delta(\secp)
        \end{equation*}
        and 
        \begin{equation*}
            \frac{\prob{\AsymmetricAnchorMonogamy(\cal{B},1^\secp)=1}}{\prob{\AsymmetricAnchorMonogamy(\cal{B},1^\secp)\neq\bot}}\geq\frac{1}{2}+\delta(\secp)-\negl(\secp).
        \end{equation*}
    \end{claim}
    \begin{proof}
        Clearly the aborting probability of $\cal{B}$ is the same as the aborting probability of $\cal{A}$ in \textbf{Hybrid 7}. For the winning probability of $\cal{B}$, first we show that if $\cal{A}$ outputs correctly (as in \textbf{Hybrid 7}) then the output of $\cal{B}$ is correct. Suppose the output of $\cal{A}_L$ is correct: $b_0^l=\ipd{r'_L,v_{\sf ext}}$, then the output of $\cal{B}_L$ is correct:
        \begin{align*}
            &b_0^l\oplus\ipd{{U_{\sf shift}^*}^{-t}\brackets{0^{n_e+5n/4}\times r_L^s},v^*}\\
            =&\ipd{r'_L,v_{\sf ext}}\oplus\ipd{{U_{\sf shift}^*}^{-t}\brackets{0^{n_e+5n/4}\times r_L^s},v^*}\\
            =&\ipd{{U_{\sf shift}^*}^{-t}\brackets{r_L^h\times r_L\times r_L^s},U_{\sf shift}^*\brackets{0^{n/4}\times v\times 0^{n/4}}+v^*}\oplus\ipd{{U_{\sf shift}^*}^{-t}\brackets{0^{n_e+5n/4}\times r_L^s},v^*}\\
            =&\ipd{{U_{\sf shift}^*}^{-t}\brackets{r_L^h\times r_L\times r_L^s},U_{\sf shift}^*\brackets{0^{n/4}\times v\times 0^{n/4}}}\oplus\ipd{{U_{\sf shift}^*}^{-t}\brackets{r_L^h\times r_L\times 0^{n/4}},v^*}\\
            =&\ipd{r_L,v}\oplus\ipd{{U_{\sf shift}^*}^{-t}\brackets{r_L^h\times r_L\times 0^{n/4}},U_{\sf shift}^*\brackets{0^{n_e+5n/4}\times {v'}^*}}\\
            =&\ipd{r_L,v}.
        \end{align*}
        Suppose the output of $\cal{A}_R$ is correct for its first bit: $b^r_1=\ipd{r'_R,u_{\sf ext}}$, then the output of $\cal{B}_R$ is correct:
        \begin{align*}
            &b^r_1\oplus\ipd{U_{\sf shift}^*\brackets{r_R^h\times 0^{n_e+5n/4}},u^*}\\
            =&\ipd{r'_R,u_{\sf ext}}\oplus\ipd{U_{\sf shift}^*\brackets{r_R^h\times 0^{n_e+5n/4}},u^*}\\
            =&\ipd{U_{\sf shift}^*\brackets{r_R^h\times r_R\times r_R^s},{U_{\sf shift}^*}^{-t}\brackets{0^{n/4}\times u\times 0^{n/4}}+u^*}\oplus\ipd{U_{\sf shift}^*\brackets{r_R^h\times 0^{n_e+5n/4}},u^*}\\
            =&\ipd{U_{\sf shift}^*\brackets{r_R^h\times r_R\times r_R^s},{U_{\sf shift}^*}^{-t}\brackets{0^{n/4}\times u\times 0^{n/4}}}\oplus\ipd{U_{\sf shift}^*\brackets{0^{n/4}\times r_R\times r_R^s},u^*}\\
            =&\ipd{r_R,u}\oplus\ipd{U_{\sf shift}^*\brackets{0^{n/4}\times r_R\times r_R^s},{U_{\sf shift}^*}^{-t}\brackets{{u'}^*\times 0^{n_e+5n/4}}}\\
            =&\ipd{r_R,u}.
        \end{align*}
        Now we show that the input distribution of $\cal{A}$ in $\cal{B}$ is the same as the input distribution it receives from the challenger in \textbf{Hybrid 7}. By definition, $r'_L$ and $r'_R$ are two uniform random vectors in $\mathbb{F}_2^{n_e+3n/2}$ and $b'_1$ is a uniform random vector in $\mathbb{F}_2^{m_1-1}$. And $b'_0$, given to $\cal{A}_R$, is the correct answer for $\cal{A}_L$:
        \begin{align*}
            b'_0&=b_0\oplus\ipd{{U_{\sf shift}^*}^{-t}\brackets{0^{n_e+5n/4}\times r_L^s},v^*}\\
            &=\ipd{r_L,v}\oplus\ipd{{U_{\sf shift}^*}^{-t}\brackets{0^{n_e+5n/4}\times r_L^s},v^*}\\
            &=\ipd{r'_L,v_{\sf ext}}.
        \end{align*}
        The last equation comes from previous calculations on the correctness of $\cal{B}_L$'s output.
    \end{proof}
\end{proof}
\subsection{Multi-Stage Monogamy-of-Entanglement Games}
\begin{lemma}\label{lem:multi-stage_decision_monogamy-of-entanglement}
    Take any efficient adversary $\cal{A}$ in the multi-stage monogamy-of-entanglement game described in \Cref{def:multi-Stage_monogamy-of-entanglement} such that
    \begin{equation*}
        \prob{\MultiStageDecisionMonogamy(\cal{A},1^\secp)\neq\bot}=\eps'(\secp)
    \end{equation*}
    and
    \begin{equation*}
        \frac{\prob{\MultiStageDecisionMonogamy(\cal{A},1^\secp)=1}}{\prob{\MultiStageDecisionMonogamy(\cal{A},1^\secp)\neq\bot}}=1-\eps(\secp)
    \end{equation*}
    for any function $\eps(\secp)$ and any inverse polynomial $\eps'(\secp)$.. Consider the state $\key\anchor\lreg\mreg\rreg\breg$ at \BeforeSplit conditioned on $\cal{A}_M^0$ not aborting, which we can write as (we also purify $\sk$ in register $\key$ here):
    
    \begin{equation*}
        \sum_{\sk}\sum_{x}\alpha_{\sk,x}\ket{\sk}_{\key}\ket{x}_{\anchor}\brackets{\frac{1}{\sqrt{2}}\ket{\psi^\sk_{x,0}}_{\lreg\mreg\rreg}\ket{0}_\breg+\frac{1}{\sqrt{2}}\ket{\psi^\sk_{x,1}}_{\lreg\mreg\rreg}\ket{1}_\breg}.
    \end{equation*}
    where the state $\ket{\psi^\sk_{x,0}}$ and $\ket{\psi^\sk_{x,1}}$ are normalized states. We have,
    \begin{equation*}
        \expectc{\sk,x}{\abs{\braket{\psi^\sk_{x,0}}{\psi^\sk_{x,1}}}^2}\leq4\eps(\secp)+\negl(\secp)
    \end{equation*}
    where the expectation is over $x$, the standard basis value on $\anchor$ and over $\sk$, generated by the process of $\brackets{\sk,\ket{\chain^\sk_{n_e,n,n}}_{\anchor\vessel}}\gets\genchain_{n_e,n,n}(1^\secp)$. The expectation is weighted where the weight on $\sk,x$ is $\abs{\alpha_{\sk,x}}^2$. \textbf{When we say expectation over this quantity throughout this section, we mean the above weighted definition.}
\end{lemma}
\begin{remark}
    Here we assume WLOG that $\cal{A}$'s state is always purified. The purification can be taken to live within registers $\lreg$ or $\rreg$, and ignored by later computations. The reason we do not allow it to go into $\mreg$ is because we later extract from the middle prover, as will be made clear in later theorems. In later of this section, expectations are weighted in the same manner where these uneven weights are due to the fact that we are conditioning on $\cal{A}_M^0$ not aborting that may depend on $\sk,x$.
\end{remark}
\begin{proof}
    Consider this modified version of the multi-stage monogamy-of-entanglement game, where the challenge oracles given to $\cal{A}_M^1$ and $(\cal{A}_L^0,\cal{A}_R^0)$ return independently sampled random bits ($b_0'$ and $b_0$, respectively).
    \newcommand{\ModifiedMultiStageDecisionMonogamy}{{\sf ModifiedMultiStageDecisionMonogamy}\xspace}
    \begin{enumerate}
        \item The challenger generates $\brackets{\sk,\ket{\chain^{\sk}_{n_e,n,n}}_{\anchor\vessel}}\gets\genchain_{n_e,n,n}(1^\secp)$ and gives the register $\vessel$ to $\cal{A}_M^0$. All parties of the adversary are given oracle access to $\cal{O}_{T+v}$ and $\cal{O}_{S^\perp+u}$.
        \item $\cal{A}_M^0$ can choose to abort in this step. If it aborts, the output of the game is $\bot$. Otherwise, it generates a tripartite state on $\lreg\mreg\rreg$. It sends $\lreg$ to $\cal{A}_L^0$, sends $\mreg$ to $\cal{A}_M^1$ and sends $\rreg$ to $\cal{A}_R^0$.
        \item $\cal{A}_M^1$ is given the access to \textcolor{red}{$\cal{O}_{T+v}^{b'_0}$} for random \textcolor{red}{$b'_0\rand\mathbb{F}_2$}. We can consider purified version where we create a register that stores $b'_0$, initialized to $\ket{+}_{\breg'}$, a uniform superposition over all \textcolor{red}{$b'_0$}. And $\cal{A}_M^1$ is able to perform the following operation to access $\breg'$:
        \begin{equation*}
            \sum_{b'_0\in\mathbb{F}_2}\ket{b'_0}\bra{b'_0}_{\breg'}\otimes\cal{O}_{T+v}^{b'_0}.
        \end{equation*}
        $\cal{A}_M^1$ with access to this oracle and $\mreg$ generates a tripartite state on $\mreg_L\mreg'\mreg_R$. It sends $\mreg_L$ to $\cal{A}_L^0$, sends $\mreg'$ to $\cal{A}_M^2$ and sends $\mreg_R$ to $\cal{A}_R^0$.
        \item $\cal{A}_L^0$ and $\cal{A}_R^0$ are given the access to $\cal{O}_{T+v}^{b_0}$ \textcolor{red}{for fresh generated $b_0\rand\mathbb{F}_2$. Similarly, we do this in the purified way and introduce a $\ket{+_\breg}_\breg$ as the control register.} $\cal{A}_L^0$ on $\lreg\mreg_L$ produces answer $b_0^l$ and a state on register $\lreg'$ that is sent to $\cal{A}_L^1$. $\cal{A}_R^0$ on $\rreg\mreg_R$ produces answer $b_0^r$ and a state on register $\rreg'$ that is sent to $\cal{A}_R^1$.
        \item $\cal{A}_M^2$ is given the access to $\cal{O}_{T+v}^{b_0}$ and $\cal{O}_{S^\perp+u}^{b_1}$ for random $b_1\rand\mathbb{F}_2^{m_1}$. It generates a bipartite state on $\mreg'_L\mreg'_R$. $\mreg'_L$ is given to $\cal{A}_L^1$ and $\mreg'_R$ is given to $\cal{A}_R^1$. 
        \item $\cal{A}_L^1$ and $\cal{A}_R^1$ are given the access to $\cal{O}_{T+v}^{b_0}$ and $\cal{O}_{S^\perp+u}^{b_1}$. $\cal{A}_L^1$ on $\lreg'\mreg'_L$ generates the answer $b_1^l$. $\cal{A}_R^1$ on $\rreg'\mreg'_R$ generates the answer $b_1^r$. The adversary wins iff $b_0^l=b_0^r=b_0$ and $b_1^l=b_1^r=b_1$. 
    \end{enumerate}
    Let $\ModifiedMultiStageDecisionMonogamy(\cal{A},1^\secp)$ be the random variable that takes value $1/0/\bot$ if the adversary $\cal{A}$ wins/loses/aborts in the above game, respectively. This game is displayed in \Cref{fig:modified_multi-stage_monogamy-of-entanglement}.
    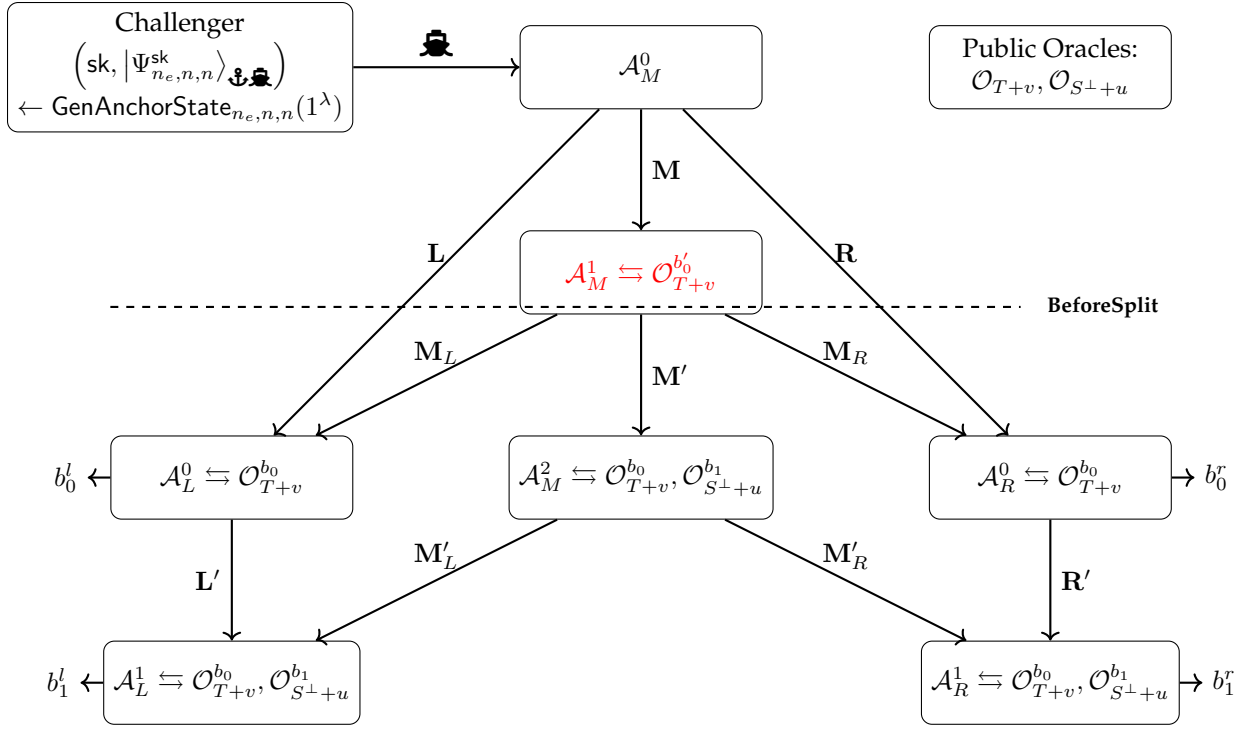
\begin{figure}[h]
        \centering
        \begin{tikzpicture}[
            node distance=1.6cm and 2.2cm,
            every node/.style={font=\small},
            box/.style={draw, rounded corners, align=center, minimum width=3.2cm, minimum height=1.1cm},
            oracle/.style={draw, dashed, rounded corners, align=center, minimum width=3.2cm, minimum height=1.1cm},
            arr/.style={->, thick}
        ]
        
        \node[box] (chal) {Challenger\\
        $\brackets{\sk,\ket{\chain^{\sk}_{n_e,n,n}}_{\anchor\vessel}}$\\
        $\gets\genchain_{n_e,n,n}(1^\secp)$};
        
        \node[box, right=of chal] (AM0) {$\mathcal A_M^0$};
        \node[box, right=of AM0] (PublicOracles) {Public Oracles:\\
        $\cal{O}_{T+v},\cal{O}_{S^\perp+u}$};
        
        \node[box, below=of AM0] (AM1) {\textcolor{red}{$\mathcal A_M^1\leftrightarrows\cal{O}_{T+v}^{b'_0}$}};
        \node[box, below left=of AM1] (AL0) {$\mathcal A_L^0\leftrightarrows\cal{O}_{T+v}^{b_0}$};
        \node[box, below right=of AM1] (AR0) {$\mathcal A_R^0\leftrightarrows\cal{O}_{T+v}^{b_0}$};
        
        \node[box, below=of AM1] (AM2) {$\mathcal A_M^2\leftrightarrows\cal{O}_{T+v}^{b_0},\cal{O}_{S^\perp+u}^{b_1}$};
        
        \node[box, below=of AL0] (AL1) {$\mathcal A_L^1\leftrightarrows\cal{O}_{T+v}^{b_0},\cal{O}_{S^\perp+u}^{b_1}$};
        \node[box, below=of AR0] (AR1) {$\mathcal A_R^1\leftrightarrows\cal{O}_{T+v}^{b_0},\cal{O}_{S^\perp+u}^{b_1}$};
        
        \draw[arr] (chal) -- node[above] {$\vessel$} (AM0);
        
        \draw[arr] (AM0) -- node[above] {$\lreg$} (AL0);
        \draw[arr] (AM0) -- node[above] {$\rreg$} (AR0);
        \draw[arr] (AM0) -- node[right] {$\mreg$} (AM1);
        
        \draw[arr] (AM1) -- node[above] {$\mreg_L$} (AL0);
        \draw[arr] (AM1) -- node[above] {$\mreg_R$} (AR0);
        \draw[arr] (AM1) -- node[right] {$\mreg'$} (AM2);
        
        \draw[arr] (AL0) -- node[left] {$\lreg'$} (AL1);
        \draw[arr] (AR0) -- node[right] {$\rreg'$} (AR1);
        
        \draw[arr] (AM2) -- node[above] {$\mreg'_L$} (AL1);
        \draw[arr] (AM2) -- node[above] {$\mreg'_R$} (AR1);
        \path
        ($(AL0.north west)+(-1.2cm,0)$)
        ($(AR0.north east)+(1.2cm,0)$);
        
        \draw[dashed, thick]
        ($(AL0.north west)+(0cm,1.7cm)$) --
        ($(AR0.north east)+(-2.0cm,1.7cm)$)
        node[right=6pt, font=\scriptsize] {\BeforeSplit};
        
        \node[left=0.3cm of AL0] (b0l){$b_0^l$};
        \node[right=0.3cm of AR0] (b0r){$b_0^r$};
        \node[left=0.3cm of AL1] (b1l){$b_1^l$};
        \node[right=0.3cm of AR1] (b1r){$b_1^r$};
        \draw[arr] (AL0) -- node[left] {} (b0l);
        \draw[arr] (AR0) -- node[right] {} (b0r);
        \draw[arr] (AL1) -- node[left] {} (b1l);
        \draw[arr] (AR1) -- node[right] {} (b1r);
        
        \end{tikzpicture}
    
        \caption{The modified multi-stage decision monogamy-of-entanglement game \ModifiedMultiStageDecisionMonogamy.}
        \label{fig:modified_multi-stage_monogamy-of-entanglement}
    \end{figure}
    \begin{claim}\label{claim:half_minus_epsilon_gap}
        For all $\secp\in\mathbb{N}$ and any adversary $\cal{A}$, the aborting probabilities are the same in both games,
        \begin{align*}
            &\prob{\ModifiedMultiStageDecisionMonogamy(\cal{A},1^\secp)=\bot}\\
            =&\prob{\MultiStageDecisionMonogamy(\cal{A},1^\secp)=\bot}
        \end{align*}
        and the winning probabilities conditioned on not aborting are related
        \begin{align*}
            &\frac{\prob{\ModifiedMultiStageDecisionMonogamy(\cal{A},1^\secp)=1}}{\prob{\ModifiedMultiStageDecisionMonogamy(\cal{A},1^\secp)\neq\bot}}\\
            \geq&\frac{\prob{\MultiStageDecisionMonogamy(\cal{A},1^\secp)=1}}{\prob{\MultiStageDecisionMonogamy(\cal{A},1^\secp)\neq\bot}}-\frac{1}{2}+\frac{1}{4}\expectc{\sk,x}{\abs{\braket{\psi^\sk_{x,0}}{\psi^\sk_{x,1}}}^2}.
        \end{align*}
    \end{claim}
    \begin{proof}
        Given an adversary $\cal{A}$ we are able to construct two corresponding projector $\Pi_0^{\cal{A}},\Pi_1^{\cal{A}}$ which simulate the game from \BeforeSplit and project on the adversary winning in the case that $b_0=0$, $b_0=1$, respectively. It is crucial to notice that this projector does not act on $\breg'$, meaning that $\key$, $\anchor$, $\breg'$ and $\breg$ all serve as control registers when calculating the maximum distinguishing probability between \ModifiedMultiStageDecisionMonogamy and \MultiStageDecisionMonogamy. In the game \ModifiedMultiStageDecisionMonogamy, the state at \BeforeSplit is
        \begin{align*}
            \propto\sum_{\sk}\sum_x\sum_{b'_0,b_0\in\mathbb{F}_2}\alpha_{\sk,x}\ket{\sk}_{\key}\ket{x}_{\anchor}\ket{b'_0}_{\breg'}\ket{b_0}_\breg\ket{\psi^\sk_{x,b'_0}}_{\lreg\mreg\rreg}.
        \end{align*}
        In the game \MultiStageDecisionMonogamy, the state at \BeforeSplit is 
        \begin{align*}
            \propto\sum_{\sk}\sum_x\sum_{b'_0,b_0\in\mathbb{F}_2}\alpha_{\sk,x}\ket{\sk}_{\key}\ket{x}_{\anchor}\ket{b'_0}_{\breg'}\ket{b_0}_\breg\ket{\psi^\sk_{x,b_0}}_{\lreg\mreg\rreg}.
        \end{align*}
        if we include an untouched $\breg'$ register. Thus the maximum distinguishing probability for any $\Pi_0^{\cal{A}},\Pi_1^{\cal{A}}$ is
        \begin{align*}
            &\frac{\prob{\MultiStageDecisionMonogamy(\cal{A},1^\secp)=1}}{\prob{\MultiStageDecisionMonogamy(\cal{A},1^\secp)\neq\bot}}\\
            &-\frac{\prob{\ModifiedMultiStageDecisionMonogamy(\cal{A},1^\secp)=1}}{\prob{\ModifiedMultiStageDecisionMonogamy(\cal{A},1^\secp)\neq\bot}}\\
            =&\frac{1}{2}\expectc{\sk,x}{\TraceDist{\ket{\psi^\sk_{x,0}}\bra{\psi^\sk_{x,0}}}{\ket{\psi^\sk_{x,1}}\bra{\psi^\sk_{x,1}}}}\\
            =&\frac{1}{2}\expectc{\sk,x}{\sqrt{1-\abs{\braket{\psi^\sk_{x,0}}{\psi^\sk_{x,1}}}^2}}\\
            \leq&\frac{1}{2}\expectc{\sk,x}{1-\frac{1}{2}\abs{\braket{\psi^\sk_{x,0}}{\psi^\sk_{x,1}}}^2}\\
            =&\frac{1}{2}-\frac{1}{4}\expectc{\sk,x}{\abs{\braket{\psi^\sk_{x,0}}{\psi^\sk_{x,1}}}^2}.
        \end{align*}
    \end{proof}
    Now we prove another claim showing that there is not noticeable advantage for any adversary to win \ModifiedMultiStageDecisionMonogamy.
    \begin{claim}\label{claim:relate_to_extended_MoE}
        The winning probability of any adversary $\cal{A}$ in \ModifiedMultiStageDecisionMonogamy is at most $\frac{1}{2}+\negl$. For any adversary $\cal{A}$ such that
        \begin{equation*}
            \prob{\ModifiedMultiStageDecisionMonogamy(\cal{A},1^\secp)\neq\bot}\geq\frac{1}{\poly(\secp)},
        \end{equation*}
        then we have
        \begin{equation*}
            \frac{\prob{\ModifiedMultiStageDecisionMonogamy(\cal{A},1^\secp)=1}}{\prob{\ModifiedMultiStageDecisionMonogamy(\cal{A},1^\secp)\neq\bot}}\leq\frac{1}{2}+\negl(\secp).
        \end{equation*}
    \end{claim}
    \begin{proof}
        For any adversary $\cal{A}$ that does not abort with at least $\delta(\secp)$ probability in the game \ModifiedMultiStageDecisionMonogamy, and conditioned on that it wins with probability at least $\frac{1}{2}+\delta(\secp)$ for some inverse polynomial $\delta(\secp)$, we construct an adversary $\cal{B}=(\cal{B}_M,\cal{B}_L,\cal{B}_R)$ that wins the following game with probability at least $\frac{1}{2}+\delta(\secp)$ conditioned on not aborting, and the non-aborting probability is at least $\delta(\secp)$.
        \begin{enumerate}
            \item The challenger generates $\brackets{\sk,\ket{\chain^{\sk}_{n_e,n,n}}_{\anchor\vessel}}\gets\genchain_{n_e,n,n}(1^\secp)$ and gives the register $\vessel$ to $\cal{B}_M$. Membership oracles $\cal{O}_{T+v}$ and $\cal{O}_{S^\perp+u}$ are made to public.
            \item $\cal{B}_M$ can choose to abort in this step. If it aborts, the output of the game is $\bot$. Otherwise, it generates a bipartite state on $\lreg\rreg$. It sends $\lreg$ to $\cal{B}_L$ and sends $\rreg$ to $\cal{B}_R$.
            \item The challenger samples $b_0\rand\mathbb{F}_2$ and $b_1\rand\mathbb{F}_2^{m_1}$. It gives $\cal{O}_{T+v}^{b_0}$ to $\cal{B}_L$. Then it gives both $\cal{O}_{T+v}^{b_0}$ and $\cal{O}_{S^\perp+u}^{b_1}$ to $\cal{B}_R$.
            \item $\cal{B}_L$ returns $b_0^l$, $\cal{B}_R$ returns $b_1^r$. The adversary wins if $b_0^l=b_0$ and $b_1^r=b_1$.
        \end{enumerate}
        $\cal{B}$ is constructed as follows (shown in \Cref{fig:adversary_B}):
        \begin{enumerate}
            \item $\cal{B}_M$ receives $\vessel$ from the challenger, it calls $\cal{A}_M^0$ to obtain $\lreg\mreg\rreg$. It aborts if $\cal{A}_M^0$ aborts. It then samples $b'_0$ by itself, simulates $\cal{O}_{T+v}^{b'_0}$ on its own using $\cal{O}_{T+v}$, calls $\cal{A}_M^1$ to obtain $\mreg_L\mreg'\mreg_R$. It gives $\lreg\mreg_L$ (as the $\lreg$ register in $\AsymmetricOracleMonogamy$) to $\cal{B}_L$ and gives $\mreg'\mreg_R\rreg$ (as the $\rreg$ register in $\AsymmetricOracleMonogamy$) to $\cal{B}_R$.
            \item $\cal{B}_L$ with access to $\cal{O}_{T+v}^{b_0}$ calls $\cal{A}_L^0$. It forwards the output of $\cal{A}_L^0$ as $b_0^l$.
            \item $\cal{B}_R$ calls $\cal{A}_M^2$ with access to $\cal{O}_{T+v}^{b_0}$ and $\cal{O}_{S^\perp+u}^{b_1}$ on $\mreg'$ to obtain $\mreg'_L\mreg'_R$. Then it calls $\cal{A}_R^0$ on $\rreg\mreg_R$ with access to $\cal{O}_{T+v}^{b_0}$ to obtain $\rreg'$. Finally it calls $\cal{A}_R^1$ on $\mreg'_R\rreg'$ with access to $\cal{O}_{T+v}^{b_0},\cal{O}_{S^\perp+u}^{b_1}$ and forwards the output of $\cal{A}_R^1$ as $b_1^r$.
        \end{enumerate}
        \begin{figure}[h]
            \centering
            \begin{tikzpicture}[
                node distance=1.6cm and 2.2cm,
                every node/.style={font=\small},
                box/.style={draw, rounded corners, align=center, minimum width=3.2cm, minimum height=1.1cm},
                oracle/.style={draw, dashed, rounded corners, align=center, minimum width=3.2cm, minimum height=1.1cm},
                arr/.style={->, thick}
            ]
            
            \node[box] (chal) {Challenger\\
            $\brackets{\sk,\ket{\chain^{\sk}_{n_e,n,n}}_{\anchor\vessel}}$\\
            $\gets\genchain_{n_e,n,n}(1^\secp)$};
            
            \node[box, right=of chal] (AM0) {$\mathcal A_M^0$};
            \node[box, right=of AM0] (PublicOracles) {Public Oracles:\\
            $\cal{O}_{T+v},\cal{O}_{S^\perp+u}$};
            
            \node[box, below=of AM0] (AM1) {$\mathcal A_M^1\leftrightarrows\cal{O}_{T+v}^{b'_0}$};
            \node[box, below left=of AM1] (AL0) {$\mathcal A_L^0\leftrightarrows\cal{O}_{T+v}^{b_0}$};
            \node[box, below right=of AM1] (AR0) {$\mathcal A_R^0\leftrightarrows\cal{O}_{T+v}^{b_0}$};
            
            \node[box, below=of AM1] (AM2) {$\mathcal A_M^2\leftrightarrows\cal{O}_{T+v}^{b_0},\cal{O}_{S^\perp+u}^{b_1}$};
            
            \node[box, below=of AL0] (AL1) {$\mathcal A_L^1\leftrightarrows\cal{O}_{T+v}^{b_0},\cal{O}_{S^\perp+u}^{b_1}$};
            \node[box, below=of AR0] (AR1) {$\mathcal A_R^1\leftrightarrows\cal{O}_{T+v}^{b_0},\cal{O}_{S^\perp+u}^{b_1}$};
            
            \draw[arr] (chal) -- node[above] {$\vessel$} (AM0);
            
            \draw[arr] (AM0) -- node[above] {$\lreg$} (AL0);
            \draw[arr] (AM0) -- node[above] {$\rreg$} (AR0);
            \draw[arr] (AM0) -- node[right] {$\mreg$} (AM1);
            
            \draw[arr] (AM1) -- node[above] {$\mreg_L$} (AL0);
            \draw[arr] (AM1) -- node[above] {$\mreg_R$} (AR0);
            \draw[arr] (AM1) -- node[right] {$\mreg'$} (AM2);
            
            \draw[arr] (AL0) -- node[left] {$\lreg'$} (AL1);
            \draw[arr] (AR0) -- node[right] {$\rreg'$} (AR1);
            
            \draw[arr] (AM2) -- node[above] {$\mreg'_L$} (AL1);
            \draw[arr] (AM2) -- node[above] {$\mreg'_R$} (AR1);
            \path
            ($(AL0.north west)+(-1.2cm,0)$)
            ($(AR0.north east)+(1.2cm,0)$);
            
            \draw[dashed, thick]
            ($(AL0.north west)+(0cm,1.7cm)$) --
            ($(AR0.north east)+(-2.0cm,1.7cm)$)
            node[right=6pt, font=\scriptsize] {\BeforeSplit};
            
            \node[left=0.3cm of AL0] (b0l){$b_0^l$};
            \node[right=0.3cm of AR0] (b0r){$b_0^r$};
            \node[left=0.3cm of AL1] (b1l){$b_1^l$};
            \node[right=0.3cm of AR1] (b1r){$b_1^r$};
            \draw[arr] (AL0) -- node[left] {} (b0l);
            \draw[arr] (AR0) -- node[right] {} (b0r);
            \draw[arr] (AL1) -- node[left] {} (b1l);
            \draw[arr] (AR1) -- node[right] {} (b1r);
            \node[
                draw=red,
                dashed,
                thick,
                rounded corners,
                fit=(AM0)(AM1),
                inner sep=6pt,
                label={[red]above:{\small $\cal{B}_M$}}
            ] (BMbox) {};
            \node[
                draw=red,
                dashed,
                thick,
                rounded corners,
                fit=(AL0),
                inner sep=6pt,
                label={[red]above:{\small $\cal{B}_L$}}
            ] (BLbox) {};
            
            \node[
                draw=red,
                dashed,
                thick,
                rounded corners,
                fit=(AM2)(AR0)(AR1),
                inner sep=6pt,
                label={[red]right:{\small $\cal{B}_R$}}
            ] (BRbox) {};
            
            \end{tikzpicture}

            \caption{The adversary $\cal{B}=(\cal{B}_M,\cal{B}_L,\cal{B}_R)$.}
            \label{fig:adversary_B}
        \end{figure}
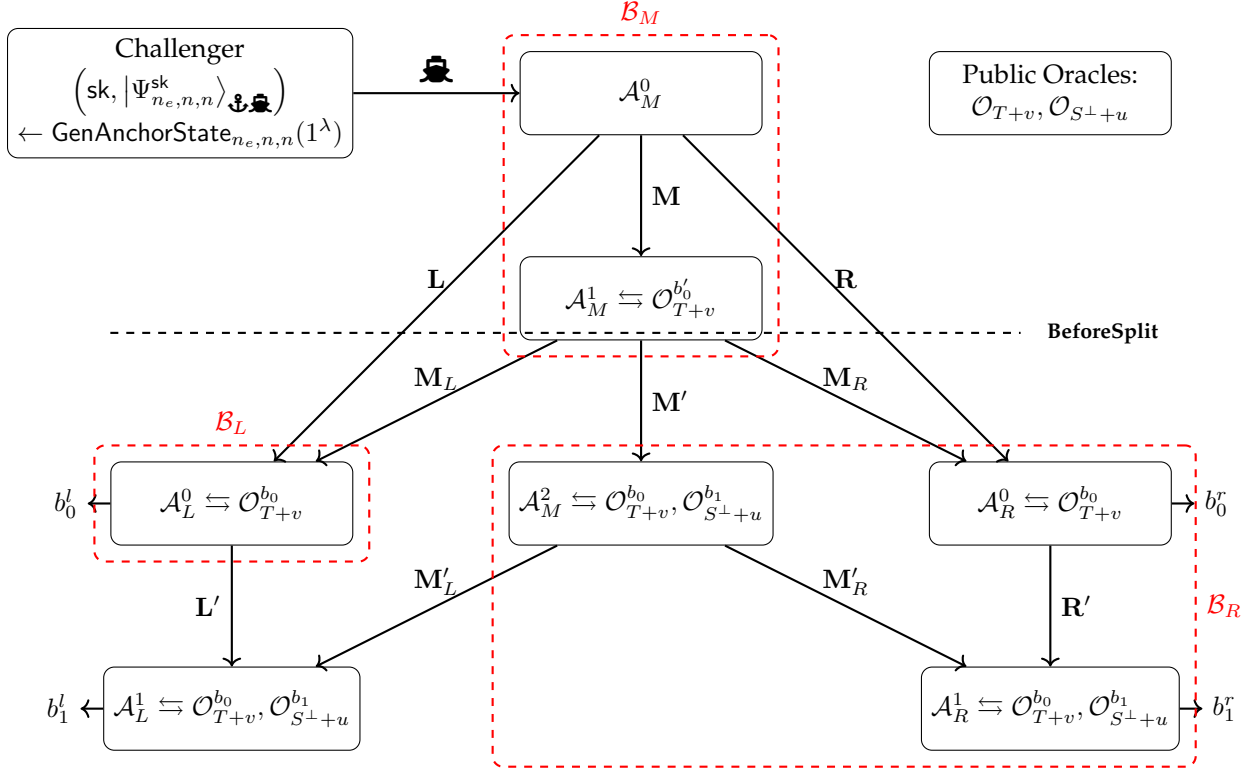
        The same adversary $\cal{B}=(\cal{B}_M,\cal{B}_L,\cal{B}_R)$ also wins the following game with probability at least $\frac{1}{2}+\delta(\secp)-\negl$ conditioned on not aborting, and the non-aborting probability is at least $\delta(\secp)-\negl(\secp)$.
        \begin{enumerate}
            \item The challenger generates $\brackets{\sk,\ket{\chain^{\sk}_{n_e,n,n}}_{\anchor\vessel}}\gets\genchain_{n_e,n,n}(1^\secp)$ and gives the register $\vessel$ to $\cal{B}_M$. It samples uniformly random subspace $T^*$ of dimension $n_e+7n/4$ such that $T\leq T^*\leq\mathbb{F}_2^{n_e+2n}$. It also samples uniformly random subspace $S^*\leq S$ of dimension $n/4$. Membership oracles $\cal{O}_{T^*+v}$ and $\cal{O}_{{S^*}^\perp+u}$ are made public.
            \item $\cal{B}_M$ can choose to abort in this step. If it aborts, the output of the game is $\bot$. Otherwise, it generates a bipartite state on $\lreg\rreg$. It sends $\lreg$ to $\cal{B}_L$ and sends $\rreg$ to $\cal{B}_R$.
            \item The challenger samples $b_0\rand\mathbb{F}_2$ and $b_1\rand\mathbb{F}_2^{m_1}$. It gives $\cal{O}_{T+v}^{b_0}$ to $\cal{B}_L$. Then it gives both $\cal{O}_{T^*+v}^{b_0}$ and $\cal{O}_{S^\perp+u}^{b_1}$ to $\cal{B}_R$.
            \item $\cal{B}_L$ returns $b_0^l$, $\cal{B}_R$ returns $b_1^r$. The adversary wins if $b_0^l=b_0$ and $b_1^r=b_1$.
        \end{enumerate}
        This is because that we can apply hybrid argument and use \Cref{lem:subspace_oracle_indistinguishability} multiple times on $\cal{O}_{T+v}$, $\cal{O}_{S^\perp+u}$, and $\cal{O}_{T+v}^{b_0}$ given to $\cal{B}_R$. Now using $\cal{B}$ we construct an adversary $\cal{C}=(\cal{C}_M,\cal{C}_L,\cal{C}_R)$ that wins \AsymmetricOracleMonogamy (defined in \Cref{lem:asymmetric_oracle_monogamy}) with probability $\frac{1}{2}+\delta(\secp)-\negl$ conditioned on not aborting, and the non-aborting probability is at least $\delta(\secp)-\negl(\secp)$.
        \begin{enumerate}
            \item $\cal{C}_M$ receives $\vessel$ and runs $\sk^*\gets\domainextension_{n/4,n/4}(1^\secp)$ on $\vessel$. It runs $\cal{B}_M$ on $\vessel^*$ to obtain $\lreg\rreg$. It aborts if $\cal{B}_M$ aborts. It gives $\sk^*,\lreg$ to $\cal{B}_L$ and $\sk^*,\rreg$ to $\cal{B}_R$. All queries to $\cal{O}_{T^*+v}$ is equivalent to $\cal{O}_{T^*+v^*}$ and can be simulated by $\sk^*$. So do queries to $\cal{O}_{{S^*}^\perp+u}$ and other oracles below.
            \item $\cal{C}_L$ runs $\cal{B}_L$ on $\lreg$ with oracle access to $\cal{O}_{T_{\sf ext}+v_{\sf ext}}^{b_0}$. It outputs $b^l_0$, the output of $\cal{B}_L$.
            \item $\cal{C}_R$ runs $\cal{B}_R$ on $\rreg$ with oracle access to $\cal{O}_{T^*+v}^{b_0}$ (simulated by $\sk^*$ and $b_0$ given by the challenger) and $\cal{O}_{S_{\sf ext}^\perp+u_{\sf ext}}^{b_1}$. It outputs $b^r_1$, the output of $\cal{B}_R$.
        \end{enumerate}
        By contradiction, such adversary $\cal{B}$ cannot exist and hence $\cal{A}$ also doesn't exist.
    \end{proof}
    Finally, combining \Cref{claim:half_minus_epsilon_gap} and \Cref{claim:relate_to_extended_MoE} we get
    \begin{align*}
        &\expectc{\sk,x}{\abs{\braket{\psi^\sk_{x,0}}{\psi^\sk_{x,1}}}^2}\\
        \leq&4\brackets{\frac{\prob{\ModifiedMultiStageDecisionMonogamy(\cal{A},1^\secp)=1}}{\prob{\ModifiedMultiStageDecisionMonogamy(\cal{A},1^\secp)\neq\bot}}+\frac{1}{2}-\frac{\prob{\MultiStageDecisionMonogamy(\cal{A},1^\secp)=1}}{\prob{\MultiStageDecisionMonogamy(\cal{A},1^\secp)\neq\bot}}}\\
        =&4\brackets{\frac{\prob{\ModifiedMultiStageDecisionMonogamy(\cal{A},1^\secp)=1}}{\prob{\ModifiedMultiStageDecisionMonogamy(\cal{A},1^\secp)\neq\bot}}+\frac{1}{2}-(1-\eps(\secp))}\\
        \leq&4\brackets{\frac{1}{2}+\negl(\secp)+\frac{1}{2}-(1-\eps(\secp))}\\
        \leq&4\eps(\secp)+\negl(\secp).
    \end{align*}
\end{proof}
\begin{lemma}\label{lem:multi-stage_search_monogamy-of-entanglement}
    Take any efficient adversary $\cal{A}$ in the multi-stage monogamy-of-entanglement game described in \Cref{def:multi-Stage_monogamy-of-entanglement} such that
    \begin{equation*}
        \prob{\MultiStageSearchMonogamy(\cal{A},1^\secp)\neq\bot}=\eps'(\secp)
    \end{equation*}
    and
    \begin{equation*}
        \frac{\prob{\MultiStageSearchMonogamy(\cal{A},1^\secp)=1}}{\prob{\MultiStageSearchMonogamy(\cal{A},1^\secp)\neq\bot}}=1-\eps(\secp)
    \end{equation*}
    for any function $\eps(\secp)$ and any inverse polynomial $\eps'(\secp)$.. Consider the state $\key\anchor\lreg\mreg\rreg\breg$ at \BeforeSplit conditioned on $\cal{A}_M^0$ not aborting, which we can write as (we also purify $\sk$ in register $\key$ here):
    
    \begin{equation*}
        \ket{\psi}_{\key\anchor\lreg\mreg\rreg\breg}=\sum_{\sk}\sum_{x}\sum_{b_0\in\mathbb{F}_2^{m_0}}\frac{\alpha_{\sk,x}}{\sqrt{2^{m_0}}}\ket{\sk}_{\key}\ket{x}_{\anchor}\ket{\psi^\sk_{x,b_0}}_{\lreg\mreg\rreg}\ket{b_0}_\breg.
    \end{equation*}
    where $\ket{\psi^\sk_{x,b_0}}$ for all $b_0$ are normalized states. We have,
    \begin{equation*}
        \trace{\ket{+_\breg}\bra{+_\breg}_\breg\cdot\ket{\psi}\bra{\psi}_{\key\anchor\lreg\mreg\rreg\breg}}\leq\frac{1}{2^{m_0}}+2\sqrt{2}\eps^{1/2}+\negl(\secp)
    \end{equation*}
    where $\ket{+_\breg}_\breg=\frac{1}{\sqrt{2^{m_0}}}\sum_{b_0\in\mathbb{F}_2^{m_0}}\ket{b_0}_\breg$ is the uniform superposition over $b_0$.
\end{lemma}
\begin{proof}
    Define the probability that the adversary $\cal{A}$ wins (both parties output the correct $b_0$) when $b_0=b_0^*$ conditioned on not aborting as $1-\eps_{b_0^*}(\secp)$. We have,
    \begin{equation*}
        \eps=\expectc{b_0}{\eps_{b_0}}.
    \end{equation*}
    For any different $b'_0,b''_0\in\mathbb{F}_2^{m_0}$, we must have 
    \begin{equation}\label{eq:expected_ipd}
        \expectc{\sk,x}{\abs{\braket{\psi^\sk_{x,b'_0}}{\psi^\sk_{x,b''_0}}}^2}\leq 2\brackets{\eps_{b'_0}(\secp)+\eps_{b''_0}(\secp)}+\negl(\secp).
    \end{equation}
    To prove this we construct an adversary $\cal{A}'$ in \Cref{lem:multi-stage_decision_monogamy-of-entanglement}. $\cal{A}'$ runs $\cal{A}$ and reprograms all the oracle outputs: reprograms $0$ to $b'_0$, $1$ to $b''_0$, and $\bot$ unchanged. Finally, $\cal{A}'$ also reprograms the output of $\cal{A}$: reprograms $b'_0$ to $0$ and $b''_0$ to $1$. For this $\cal{A}'$, 
    \begin{equation*}
        \frac{\prob{\MultiStageDecisionMonogamy(\cal{A}',1^\secp)=1}}{\prob{\MultiStageDecisionMonogamy(\cal{A}',1^\secp)\neq\bot}}=1-\frac{1}{2}\brackets{\eps_{b'_0}(\secp)+\eps_{b''_0}(\secp)}.
    \end{equation*}
    By \Cref{lem:multi-stage_decision_monogamy-of-entanglement}, the state of $\cal{A}'$ at \BeforeSplit satisfies
    \begin{equation*}
        \expectc{\sk,x}{\abs{\braket{\psi^\sk_{x,0}}{\psi^\sk_{x,1}}}^2}\leq 2\brackets{\eps_{b'_0}(\secp)+\eps_{b''_0}(\secp)}+\negl(\secp).
    \end{equation*}
    Hence we proved \Cref{eq:expected_ipd}. Combine all above, we have
    \begin{align*}
        &\trace{\ket{+_\breg}\bra{+_\breg}_\breg\cdot\ket{\psi}\bra{\psi}_{\key\anchor\lreg\mreg\rreg\breg}}\\
        \leq&\frac{1}{2^{2m_0}}\expectc{\sk,x}{\sum_{b'_0,b''_0}\abs{\braket{\psi^\sk_{x,b'_0}}{\psi^\sk_{x,b''_0}}}}\\
        =&\frac{1}{2^{2m_0}}\expectc{\sk,x}{2^{m_0}+\sum_{b'_0\neq b''_0}\abs{\braket{\psi^\sk_{x,b'_0}}{\psi^\sk_{x,b''_0}}}}\\
        =&\frac{1}{2^{m_0}}+\frac{1}{2^{2m_0}}\sum_{b'_0\neq b''_0}\expectc{\sk,x}{\abs{\braket{\psi^\sk_{x,b'_0}}{\psi^\sk_{x,b''_0}}}}\\
        \leq&\frac{1}{2^{m_0}}+\frac{1}{2^{2m_0}}\sum_{b'_0\neq b''_0}\sqrt{\expectc{\sk,x}{\abs{\braket{\psi^\sk_{x,b'_0}}{\psi^\sk_{x,b''_0}}}^2}}\\
        \leq&\frac{1}{2^{m_0}}+\frac{1}{2^{2m_0}}\sum_{b'_0\neq b''_0}\sqrt{2\brackets{\eps_{b'_0}(\secp)+\eps_{b''_0}(\secp)}+\negl(\secp)}\\
        \leq&\frac{1}{2^{m_0}}+\frac{1}{2^{2m_0}}\sum_{b'_0\neq b''_0}\brackets{\sqrt{2\eps_{b'_0}(\secp)+\negl(\secp)}+\sqrt{2\eps_{b''_0}(\secp)+\negl(\secp)}}\\
        \leq&\frac{1}{2^{m_0}}+2\sqrt{2}\expectc{b_0}{\sqrt{\eps_{b_0}(\secp)+\negl(\secp)}}\\
        \leq&\frac{1}{2^{m_0}}+2\sqrt{2}\sqrt{\expectc{b_0}{\eps_{b_0}(\secp)+\negl(\secp)}}\\
        \leq&\frac{1}{2^{m_0}}+2\sqrt{2}\eps^{1/2}+\negl(\secp)
    \end{align*}
    The inequality between line 4 and line 5 uses $\textbf{Var}(X)=\expect{X^2}-\expect{X}^2\geq 0$, as well as the inequality between line 8 and line 9.
\end{proof}

%% file: entanglement_extraction_with_abort.tex
\section{Entanglement Extraction}\label{sec:extraction}
\subsection{Extracting Incompatible Measurement Results via Collapsing Oracles}
In later sections, we consider $n_e(\lambda)=n(\lambda)$ if not specified. 
\begin{lemma}\label{lem:extracting_standard_basis_measurement_results}
    Let $n(\secp)\geq\secp$ and $m_0(\secp)=m_1(\secp)=m(\secp)\geq\secp$ be polynomials of $\secp$. There exists an efficient extractor $\cal{E}_S$ with access to $\sk$ and the register $\mreg$ such that the following holds. Take any efficient adversary $\cal{A}$ in the multi-stage monogamy-of-entanglement game described in \Cref{def:multi-Stage_monogamy-of-entanglement} such that
    \begin{equation*}
        \prob{\MultiStageSearchMonogamy(\cal{A},1^\secp)\neq\bot}=\eps'(\secp)
    \end{equation*}
    and
    \begin{equation*}
        \frac{\prob{\MultiStageSearchMonogamy(\cal{A},1^\secp)=1}}{\prob{\MultiStageSearchMonogamy(\cal{A},1^\secp)\neq\bot}}=1-\eps(\secp)
    \end{equation*}
    for any function $\eps(\secp)$ and any inverse polynomial $\eps'(\secp)$.. Then running $\cal{E}_S^{\cal{A}_M^1}$ on the state $\mreg$ conditioned on $\cal{A}_M^0$ not aborting gives $x$, the standard basis measurement result on $\anchor$, with probability at least $1-2\sqrt{2}\eps^{1/2}-\negl(\secp)$. That is, with $\xreg$ denoting the output register of $\cal{E}_S^{\cal{A}_M^1}$,
    \begin{equation*}
        \expect{\tr\brackets{\Pi_{\sf eq}\cal{E}_S^{\cal{A}_M^1}\rho_{\anchor\mreg}\brackets{\cal{E}_S^{\cal{A}_M^1}}^\dagger\Pi_{\sf eq}}\middle\vert\substack{\brackets{\sk,\ket{\chain^{\sk}}_{\anchor\vessel}}\gets\genchain(1^\secp)\\ \cal{A}_M^0\text{ does not abort and produces }\rho_{\anchor\mreg}\text{ on }\xreg\mreg}}\geq 1-2\sqrt{2}\eps^{1/2}-\negl(\secp)
    \end{equation*}
    where $\Pi_{\sf eq}=\sum_{x}\ket{x}_{\anchor}\bra{x}\otimes\ket{x}_{\xreg}\bra{x}$ is the projector that tests whether the output is correct. Furthermore, $\cal{E}_S^{\cal{A}_M^1}$ only makes black-box access to $\cal{A}_M^1$ and acts on $\mreg$. When we say $\cal{A}_M^1$, we refer to the process only up to \BeforeSplit, excluding the renaming and splitting procedure.
\end{lemma}
\begin{proof}
    In this proof we consider all notations aligned with \Cref{lem:collapsing_versus_non-collapsing} with $a=n,b=2n$ and $c=3n$. In the following proof, all procedures are conditioned on $\cal{A}_M^0$ not aborting. Here is the construction of $\cal{E}_S^{\cal{A}_M^1}$:
    \begin{enumerate}
        \item Set up a random oracle $\cal{O}_{\sf random}$ using a compress oracle with the database on $\ket{D}_{\dreg}$. Run $\cal{A}_M^1$ in \MultiStageSearchMonogamy on $\mreg$. Simulate $\cal{O}_{S^\perp+u}$ using $\sk$. Replace $\cal{O}_{T+v}^{b_0}$ with 
        \begin{itemize}
            \item $\cal{O}_{T+v}^{\sf col,\cal{O}_{\sf random}}:\mathbb{F}_2^{3n}\rightarrow\mathbb{F}_2^{m_0}\cup\set{\bot}$ on input $z$ checks if it is in $T+v$. It returns $\bot$ if it is not. Otherwise, it returns $\cal{O}_{\sf random}(\Can_S(z))$.
        \end{itemize}
        And also simulate $\cal{O}_{T+v}$ using the above oracle.
        \item Measure the database register, if there is no non-$\bot$ entry or more than one non-$\bot$ entry, then output $\bot$. Otherwise, let the unique entry be $(w',y)$. Output the middle $n$ bits of $U_{\sf shift}^{-1}w'$.
    \end{enumerate}
    We calculate the success probability of $\cal{E}_S^{\cal{A}_M^1}$ using a hybrid argument. In the following hybrids, when we say compute something based on $x$, the standard basis value on $\anchor$, we mean coherently do that controlled by $x$ on $\anchor$.\\
    \textbf{Hybrid 0:} This hybrid corresponds to running $\cal{E}_S^{\cal{A}_M^1}$ on $\mreg$.\\
    \textbf{Hybrid 1:} This hybrid corresponds to running $\cal{E}_{S,{\sf ncol}}^{\cal{A}_M^1}$ on $\mreg$ where $\cal{E}_{S,{\sf ncol}}^{\cal{A}_M^1}$ is controlled by $\anchor$ and is constructed as follows:
    \begin{enumerate}
        \item Set up a random oracle $\cal{O}_{\sf random}$ using a compress oracle with the database on $\ket{D}_{\dreg}$. Run $\cal{A}_M^1$ in \MultiStageSearchMonogamy on $\mreg$. Simulate $\cal{O}_{S^\perp+u}$ using $\sk$. Sample $w_x\rand U_{\sf shift}\brackets{\set{0,1}^n\times x\times 0^n}+v$ where $x$ is the standard basis value on $\anchor$. Replace $\cal{O}_{T+v}^{b_0}$ with 
        \begin{itemize}
            \item $\cal{O}_{T+w_x}^{\sf ncol,\cal{O}_{\sf random}}:\mathbb{F}_2^{3n}\rightarrow\mathbb{F}_2^{m_0}\cup\set{\bot}$ on input $z$ checks if it is in $T+w_x$. It returns $\bot$ if it is not. Otherwise, it returns $\cal{O}_{\sf random}\brackets{\Can_S(w_x)}$.
        \end{itemize}
        And also simulate $\cal{O}_{T+v}$ using the above oracle.
        \item Measure the database register, if there is no non-$\bot$ entry or more than one non-$\bot$ entry, then output $\bot$. Otherwise, let the unique entry be $(w',y)$. Output the middle $n$ bits of $U_{\sf shift}^{-1}w'$.
    \end{enumerate}
    \begin{claim}
        We view $x$ on $\anchor$ as a classical control string. Let $P_{D,x}$ be the distribution of database register $\ket{D}_{\dreg}$ in \textbf{Hybrid 0} conditioned on the value on $\anchor$ is $x$. Let $P'_{D,x}$ be the corresponding distribution for \textbf{Hybrid 1}. For a uniform random $x$, the expected statistical distance between $P_{D,x}$ and $P'_{D,x}$ is negligible.
    \end{claim}
    \begin{proof}
        First note that $w_x$ is an independent random vector given $S$ and $T$ is an independent subspace containing $S$ given $S,w_x$. To use \Cref{lem:collapsing_versus_non-collapsing} we need to show that we can construct the correct state on $\vessel$ conditioned on the standard basis value $x$ on $\anchor$ given $S,w_x$ so that we describe the whole hybrid as an isometry on $\ket{S,w_x}$ for uniform random $x$ (if we trace out $\anchor$). Decompose $w_x$ as $u_x+U_{\sf shift}\brackets{0^n\times x\times 0^n}+v$. We need to construct 
        \begin{equation*}
            \propto\sum_{x\in\mathbb{F}_2^n}\ket{x}_{\anchor}\sum_{s\in S}(-1)^{\ipd{s,u}}\ket{s+U_{\sf shift}\brackets{0^{n_h}\times x\times 0^{n_s}}+v}_\vessel
        \end{equation*}
        from
        \begin{equation*}
            \propto\sum_{x\in\mathbb{F}_2^n}\ket{x}_{\anchor}\ket{S,w_x}
        \end{equation*}
        by operating on $\vessel$. We can generate the state by first generating the state on $\vessel$ assuming $x=0$ on $\anchor$ with a random $u$ and with $v=0$. 
        \begin{equation*}
            \propto\sum_{x\in\mathbb{F}_2^n}\ket{x}_{\anchor}\sum_{s\in S}(-1)^{\ipd{s,u}}\ket{s}.
        \end{equation*}
        Then we XOR $w_x$ on the standard basis. 
        \begin{align*}
            \propto&\sum_{x\in\mathbb{F}_2^n}\ket{x}_{\anchor}\sum_{s\in S}(-1)^{\ipd{s,u}}\ket{s+w_x}\\
            \propto&\sum_{x\in\mathbb{F}_2^n}\ket{x}_{\anchor}\sum_{s\in S}(-1)^{\ipd{s,u}}\ket{s+u_x+U_{\sf shift}\brackets{0^n\times x\times 0^n}+v}\\
            \propto&\sum_{x\in\mathbb{F}_2^n}(-1)^{\ipd{u_x,u}}\ket{x}_{\anchor}\sum_{s'\in S}(-1)^{\ipd{s',u}}\ket{s'+U_{\sf shift}\brackets{0^n\times x\times 0^n}+v}.
        \end{align*}
        In this way we get the correct state except there is a relative phase between $\anchor$ and $\vessel$ depending on the standard basis value on $\anchor$ if you consider the whole purified anchor state. But note that the standard basis value on $\anchor$ only acts as a control so this relative phase doesn't affect the output distribution of the hybrid.
    \end{proof}
    \textbf{Hybrid 2:} This hybrid corresponds to running $\cal{A}_M^1$ on $\mreg$ but we replace $\cal{O}_{T+v}^{b_0}$ with a purified version of it:
    \begin{enumerate}
        \item Set up a random oracle $\cal{O}_{\sf random}$ using a compress oracle with the database on $\ket{D}_{\dreg}$. Run $\cal{A}_M^1$ in \MultiStageSearchMonogamy on $\mreg$. Simulate $\cal{O}_{S^\perp+u}$ using $\sk$. Sample $w_x\rand U_{\sf shift}\brackets{\set{0,1}^n\times x\times 0^n}+v$ where $x$ is the standard basis value on $\anchor$. But rename $\dreg_{\Can_S(w_x)}$ as $\breg$ that stores the purified $b_0$. And also simulate $\cal{O}_{T+v}$ with $\cal{O}_{T+v}^{b_0}$.
        \item Measure the database register, if there is no non-$\bot$ entry or more than one non-$\bot$ entry, then output $\bot$. Otherwise, let the unique entry be $(w',y)$. Output the middle $n$ bits of $U_{\sf shift}^{-1}w'$.
    \end{enumerate}
    \begin{claim}
         Let $P''_{D,x}$ be the distribution of database register $\ket{D}_{\dreg}$ in \textbf{Hybrid 2} conditioned on the value on $\anchor$ is $x$. For all $x$, $P'_{D,x}$ and $P''_{D,x}$ are identical.
    \end{claim}
    \begin{proof}
        $\cal{O}_{\sf random}\brackets{\Can_S(w_x)}$ is just a uniform random string identically distributed as $b_0$.
    \end{proof}
    Note that the probability of outputting the correct $x$ in \textbf{Hybrid 2} is at least $1-\frac{1}{2^m}-2\sqrt{2}\eps^{1/2}-\negl(\secp)$ because:
    \begin{itemize}
        \item There is at most one non-$\bot$ entry in the database because the only place that can be accessed is $\Can_S(w_x)$.
        \item The probability that $\Can_S(w_x)$ is also a $\bot$-entry is at most $\frac{1}{2^m}+2\sqrt{2}\eps^{1/2}+\negl$ by \Cref{lem:multi-stage_search_monogamy-of-entanglement}.
    \end{itemize}
    Thus by hybrid argument, the probability that in \textbf{Hybrid 0} it outputs $x$ is at least $1-\frac{1}{2^m}-2\sqrt{2}\eps^{1/2}-\negl(\secp)$.
\end{proof}
Now we look at a similar game \DualMultiStageSearchMonogamy that reverses the order of cosets given in the challenge.
\begin{definition}[Dual Multi-Stage Search Monogamy-of-Entanglement Game]\label{def:dual_multi-Stage_monogamy-of-entanglement}
    Let $\secp\in\mathbb{N}$ be the security parameter. Let $n(\secp)\geq\secp$ and $n_e(\secp),m_0(\secp),m_1(\secp)$ be polynomials. Consider the following game between the challenger and an adversary $\cal{A}=\allowdisplaybreaks(\cal{A}_M^0,\cal{A}_M^1,\cal{A}_M^2,\cal{A}_L^0,\cal{A}_L^1,\cal{A}_R^0,\cal{A}_R^1)$:
    \begin{enumerate}
        \item The challenger generates $\brackets{\sk,\ket{\chain^{\sk}_{n_e,n,n}}_{\anchor\vessel}}\gets\genchain_{n_e,n,n}(1^\secp)$ and gives the register $\vessel$ to $\cal{A}_M^0$. All parties of the adversary are given oracle access to $\cal{O}_{T+v}$ and $\cal{O}_{S^\perp+u}$.
        \item $\cal{A}_M^0$ can choose to abort in this step. If it aborts, the output of the game is $\bot$. Otherwise, it generates a tripartite state on $\lreg\mreg\rreg$. It sends $\lreg$ to $\cal{A}_L^0$, sends $\mreg$ to $\cal{A}_M^1$ and sends $\rreg$ to $\cal{A}_R^0$.
        \item $\cal{A}_M^1$ is given access to $\cal{O}_{S^\perp+u}^{b_0}$ for random $b_0\rand\set{0,1}^{m_0(\secp)}$. $\cal{A}_M^1$ with access to this oracle and $\mreg$ generates a tripartite state on $\mreg_L\mreg'\mreg_R$. It sends $\mreg_L$ to $\cal{A}_L^0$, sends $\mreg'$ to $\cal{A}_M^2$ and sends $\mreg_R$ to $\cal{A}_R^0$. We will name the time just after all operations of $\cal{A}_M^1$ are done and before renaming the registers and splitting them as \BeforeSplit.
        \item $\cal{A}_L^0$ and $\cal{A}_R^0$ are given access to $\cal{O}_{S^\perp+u}^{b_0}$. $\cal{A}_L^0$ on $\lreg\mreg_L$ produces answer $b_0^l$ and a state on register $\lreg'$ that is sent to $\cal{A}_L^1$. $\cal{A}_R^0$ on $\rreg\mreg_R$ produces answer $b_0^r$ and a state on register $\rreg'$ that is sent to $\cal{A}_R^1$.
        \item $\cal{A}_M^2$ is given access to $\cal{O}_{S^\perp+u}^{b_0}$ and $\cal{O}_{T+v}^{b_1}$ for random $b_1\rand\set{0,1}^{m_1(\secp)}$. It generates a bipartite state on $\mreg'_L\mreg'_R$. $\mreg'_L$ is given to $\cal{A}_L^1$ and $\mreg'_R$ is given to $\cal{A}_R^1$. 
        \item $\cal{A}_L^1$ and $\cal{A}_R^1$ are given access to $\cal{O}_{S^\perp+u}^{b_0}$ and $\cal{O}_{T+v}^{b_1}$. $\cal{A}_L^1$ on $\lreg'\mreg'_L$ generates the answer $b_1^l$. $\cal{A}_R^1$ on $\rreg'\mreg'_R$ generates the answer $b_1^r$. The adversary wins iff $b_0^l=b_0^r=b_0$ and $b_1^l=b_1^r=b_1$. 
    \end{enumerate}
    This game is displayed in \Cref{fig:dual_multi-stage_monogamy-of-entanglement}. Let $\DualMultiStageSearchMonogamy(\cal{A},1^\secp)$ be the random variable that takes value $1/0/\bot$ if the adversary $\cal{A}$ wins/loses/aborts in the above game, respectively.
\end{definition}
\begin{figure}[h]
    \centering
    \begin{tikzpicture}[
        node distance=1.6cm and 2.2cm,
        every node/.style={font=\small},
        box/.style={draw, rounded corners, align=center, minimum width=3.2cm, minimum height=1.1cm},
        oracle/.style={draw, dashed, rounded corners, align=center, minimum width=3.2cm, minimum height=1.1cm},
        arr/.style={->, thick}
    ]
    
    \node[box] (chal) {Challenger\\
    $\brackets{\sk,\ket{\chain^{\sk}_{n_e,n,n}}_{\anchor\vessel}}$\\
    $\gets\genchain_{n_e,n,n}(1^\secp)$};
    
    \node[box, right=of chal] (AM0) {$\mathcal A_M^0$};
    \node[box, right=of AM0] (PublicOracles) {Public Oracles:\\
    $\cal{O}_{T+v},\cal{O}_{S^\perp+u}$};
    
    \node[box, below=of AM0] (AM1) {$\mathcal A_M^1\rightleftharpoons\cal{O}_{S^\perp+u}^{b_0}$};
    \node[box, below left=of AM1] (AL0) {$\mathcal A_L^0\rightleftharpoons\cal{O}_{S^\perp+u}^{b_0}$};
    \node[box, below right=of AM1] (AR0) {$\mathcal A_R^0\rightleftharpoons\cal{O}_{S^\perp+u}^{b_0}$};
    
    \node[box, below=of AM1] (AM2) {$\mathcal A_M^2\rightleftharpoons\cal{O}_{S^\perp+u}^{b_0},\cal{O}_{T+v}^{b_1}$};
    
    \node[box, below=of AL0] (AL1) {$\mathcal A_L^1\rightleftharpoons\cal{O}_{S^\perp+u}^{b_0},\cal{O}_{T+v}^{b_1}$};
    \node[box, below=of AR0] (AR1) {$\mathcal A_R^1\rightleftharpoons\cal{O}_{S^\perp+u}^{b_0},\cal{O}_{T+v}^{b_1}$};
    
    \draw[arr] (chal) -- node[above] {$\vessel$} (AM0);
    
    \draw[arr] (AM0) -- node[above] {$\lreg$} (AL0);
    \draw[arr] (AM0) -- node[above] {$\rreg$} (AR0);
    \draw[arr] (AM0) -- node[right] {$\mreg$} (AM1);
    
    \draw[arr] (AM1) -- node[above] {$\mreg_L$} (AL0);
    \draw[arr] (AM1) -- node[above] {$\mreg_R$} (AR0);
    \draw[arr] (AM1) -- node[right] {$\mreg'$} (AM2);
    
    \draw[arr] (AL0) -- node[left] {$\lreg'$} (AL1);
    \draw[arr] (AR0) -- node[right] {$\rreg'$} (AR1);
    
    \draw[arr] (AM2) -- node[above] {$\mreg'_L$} (AL1);
    \draw[arr] (AM2) -- node[above] {$\mreg'_R$} (AR1);
    
    \node[left=0.3cm of AL0] (b0l){$b_0^l$};
    \node[right=0.3cm of AR0] (b0r){$b_0^r$};
    \node[left=0.3cm of AL1] (b1l){$b_1^l$};
    \node[right=0.3cm of AR1] (b1r){$b_1^r$};
    \draw[arr] (AL0) -- node[left] {} (b0l);
    \draw[arr] (AR0) -- node[right] {} (b0r);
    \draw[arr] (AL1) -- node[left] {} (b1l);
    \draw[arr] (AR1) -- node[right] {} (b1r);
    
    \end{tikzpicture}

    \caption{The dual multi-stage monogamy-of-entanglement game \DualMultiStageSearchMonogamy.}
    \label{fig:dual_multi-stage_monogamy-of-entanglement}
\end{figure}
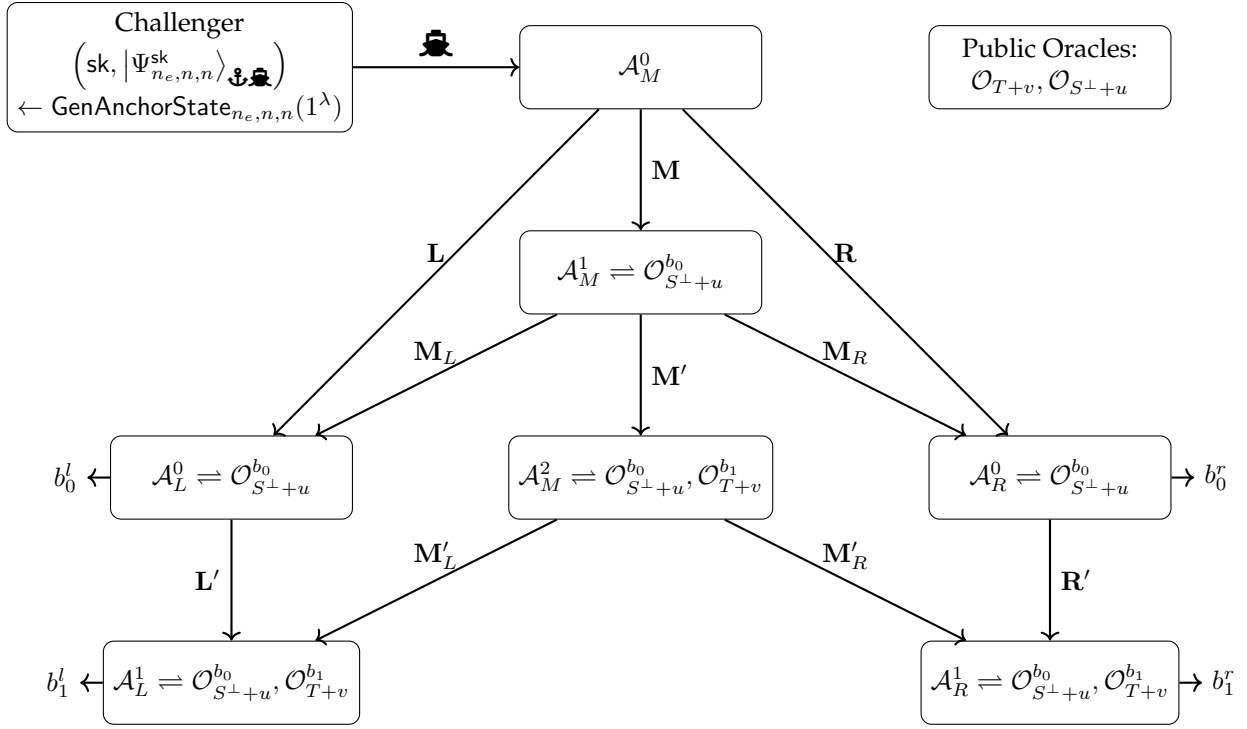
\begin{corollary}\label{cor:extracting_hadamard_basis_measurement_results}
    Let $n(\secp)\geq\secp$ and $m_0(\secp)=m_1(\secp)=m(\secp)\geq\secp$ be polynomials of $\secp$. There exists an efficient extractor $\cal{E}_H$ with access to $\sk$ and the register $\mreg$ such that the following holds. Take any efficient adversary $\cal{A}$ in the multi-stage monogamy-of-entanglement game described in \Cref{def:dual_multi-Stage_monogamy-of-entanglement} such that
    \begin{equation*}
        \prob{\DualMultiStageSearchMonogamy(\cal{A},1^\secp)\neq\bot}=\eps'(\secp)
    \end{equation*}
    and
    \begin{equation*}
        \frac{\prob{\DualMultiStageSearchMonogamy(\cal{A},1^\secp)=1}}{\prob{\DualMultiStageSearchMonogamy(\cal{A},1^\secp)\neq\bot}}=1-\eps(\secp)
    \end{equation*}
    for any function $\eps(\secp)$ and any inverse polynomial $\eps'(\secp)$.. Then running $\cal{E}_H^{\cal{A}_M^1}$ on the state $\mreg$ conditioned on $\cal{A}_M^0$ not aborting gives $x$, the Hadamard basis measurement result on $\anchor$, with probability at least $1-2\sqrt{2}\eps^{1/2}-\negl(\secp)$. That is, with $\xreg$ denoting the output register of $\cal{E}_H^{\cal{A}_M^1}$,
    \begin{equation*}
        \expect{\tr\brackets{\Pi_{\sf eq}\cal{E}_H^{\cal{A}_M^1}\rho_{\anchor\mreg}\brackets{\cal{E}_H^{\cal{A}_M^1}}^\dagger\Pi_{\sf eq}}\middle\vert\substack{\brackets{\sk,\ket{\chain^{\sk}}_{\anchor\vessel}}\gets\genchain(1^\secp)\\ \cal{A}_M^0\text{ does not abort and produces }\rho_{\anchor\mreg}\text{ on }\xreg\mreg}}\geq 1-2\sqrt{2}\eps^{1/2}-\negl(\secp)
    \end{equation*}
    where $\Pi_{\sf eq}=\sum_{x}H^{\otimes n}\ket{x}_{\anchor}\bra{x}H^{\otimes n}\otimes\ket{x}_{\xreg}\bra{x}$ is the projector that tests whether the output is correct. Furthermore, $\cal{E}_H^{\cal{A}_M^1}$ only makes black-box access to $\cal{A}_M^1$ and acts on $\mreg$. When we say $\cal{A}_M^1$, we refer to the process only up to \BeforeSplit, excluding the renaming and splitting procedure.
\end{corollary}
\begin{proof}
    Similar to \Cref{lem:extracting_standard_basis_measurement_results}. By the symmetric structure of $\ket{\chain^\sk}_{\anchor\vessel}$.
\end{proof}
\subsection{Entanglement Rigidity in Multi-Stage Monogamy-of-Entanglement Games}
Define the Pauli operators
\begin{equation*}
    \sigma_X=
    \begin{pmatrix}
        0 & 1 \\
        1 & 0 \\
    \end{pmatrix}
    \quad
    \sigma_Y=
    \begin{pmatrix}
        0 & -i\\
        i & 0\\
    \end{pmatrix}
    \quad
    \sigma_Z=
    \begin{pmatrix}
        1 & 0 \\
        0 & -1 \\
    \end{pmatrix}
\end{equation*}
\begin{lemma}[EPR Pair Rigidity, \cite{vidick2021classicalproofsquantumknowledge}, Lemma 5.4]\label{lem:EPR_pair_rigidity}
    Let $\ket{\psi}_{\areg\breg}\in\brackets{\mathbb{C}^2}^{\otimes n}_{\areg}\otimes\mathcal{H}_{\breg}$, where $\mathcal{H}_{\breg}$ is arbitrary. Suppose that for every $a\in\set{0,1}^n$ there exist observables $X^{\breg}(a)$ and $Z^{\breg}(a)$ on $\mathcal{H}_{\breg}$ such that 
    \begin{equation*}
        \forall W\in\set{X,Z},\quad \expectc{a}{\norm{\brackets{\sigma_W^{\areg}(a)-W^{\breg}(a)}\ket{\psi}}^2}\leq\eps,
    \end{equation*}
    for some $0\leq\eps\leq 1$ where $\sigma_W^{\areg}(a)=\sigma_W^{a_1}\otimes \sigma_W^{a_2}\otimes\cdots\otimes \sigma_W^{a_n}$ and the expectation is under the uniform distribution over $a\in\set{0,1}^n$. Then there exists an isometry
    \begin{equation*}
        \Phi_\breg:\mathcal{H}_{\breg}\rightarrow\brackets{\brackets{\mathbb{C}^2}^{\otimes n}}_{\breg'}\otimes\cal{H}_{\hat{\breg}}
    \end{equation*}
    such that
    \begin{equation*}
        \tr\brackets{\bra{\epr}^{\otimes n}_{\areg\breg'}\brackets{(I_{\areg}\otimes\Phi_{\breg})\brackets{\ket{\psi}_{\areg\breg}\bra{\psi}}}\ket{\epr}^{\otimes n}_{\areg\breg'}}=1-O\brackets{\eps^{1/2}}.
    \end{equation*}
    Moreover, the isometry $\Phi_\breg$ can be implemented as an $O(n)$-size quantum circuit acting on $\mathcal{H}_\breg$ as well as some ancilla qubits, and that uses controlled gates for $X^\breg(a)$ and $Z^\breg(b)$ (controlled on an ancilla register of dimension $2^n$ that contains $a$ or $b$) as black boxes. Precisely, $\Phi^\breg$ is defined by
    \begin{equation*}
        \Phi_\breg\ket{\phi}_\breg=\brackets{\frac{1}{2^n}\sum_{a,b}X^\breg(a)Z^\breg(b)\otimes\sigma_X(a)\sigma_Z(b)\otimes I}\ket{\phi}_\breg\ket{\epr}^{\otimes n},
    \end{equation*}
    and letting $\mathcal{H}_{\hat{\breg}}=\mathcal{H}_\breg\otimes\brackets{\mathbb{C}^2}^{\otimes n}$, with register $\breg'$ associated with the last $n$ copies of $\mathbb{C}^2$.
\end{lemma}
\begin{corollary}[EPR Pair Rigidity (corollary)]\label{cor:extract_epr_pair}
    Let $\ket{\psi}_{\areg\breg}\in\brackets{\mathbb{C}^2}^{\otimes n}_{\areg}\otimes\mathcal{H}_{\breg}$, where $\mathcal{H}_{\breg}$ is arbitrary. Let $\widetilde{\breg}$ be the sub-register of $\breg$ containing the first $n$ qubits and define:
    \begin{equation*}
        \Pi_{\sf eq}=\sum_{x\in\mathbb{F}_2^n}\ket{x}_\areg\bra{x}\otimes\ket{x}_{\widetilde{\breg}}\bra{x}\otimes I_{\widetilde{\breg}'}.
    \end{equation*}
    Suppose that there exists two isometries $\cal{U}_X^\breg,\cal{U}_Z^\breg:\cal{H}_\breg\rightarrow\cal{H}_{\widetilde{\breg}}\otimes\cal{H}_{\widetilde{\breg}'}$ such that 
    \begin{equation*}
        \norm{\Pi_{\sf eq}(I_\areg\otimes \cal{U}_Z^\breg)\ket{\psi}}^2\geq 1-\eps\quad\text{ and }\quad\norm{\Pi_{\sf eq}(H^{\otimes n}_\areg\otimes \cal{U}_X^\breg)\ket{\psi}}^2\geq 1-\eps.
    \end{equation*}Then there exists an isometry
    \begin{equation*}
        \Phi_\breg:\mathcal{H}_{\breg}\rightarrow\brackets{\brackets{\mathbb{C}^2}^{\otimes n}}_{\breg'}\otimes\cal{H}_{\hat{\breg}}
    \end{equation*}
    such that
    \begin{equation*}
        \tr\brackets{\bra{\epr}^{\otimes n}_{\areg\breg'}\brackets{(I_{\areg}\otimes\Phi_{\breg})\brackets{\ket{\psi}_{\areg\breg}\bra{\psi}}}\ket{\epr}^{\otimes n}_{\areg\breg'}}=1-O\brackets{\eps^{1/2}}.
    \end{equation*}
    Moreover, the isometry $\Phi_\breg$ can be implemented as an $O(n)$-size quantum circuit acting on $\mathcal{H}_\breg$ as well as some ancilla qubits, and that uses controlled gates for $\cal{U}_X^\breg$ and $\cal{U}_Z^\breg$ as black boxes.
\end{corollary}
\begin{proof}
    We construct observables $W^\breg(a)$ in \Cref{lem:EPR_pair_rigidity} using $\cal{U}_X^\breg$ and $\cal{U}_Z^\breg$. Let (we abuse the notation of $\cal{U}^\dagger$ to mean the reverse procedure of $\cal{U}$)
    \begin{equation*}
        W^\breg(a)=\brackets{\cal{U}_W^\breg}^{\dagger}\brackets{\sigma_Z^{\widetilde{\breg}}(a)\otimes I_{\widetilde{\breg}'}}\cal{U}_W^\breg\quad\forall W\in\set{X,Z}.
    \end{equation*}
    We have 
    \begin{align*}
        \norm{\Pi_{\sf eq}(I_\areg\otimes \cal{U}_Z)\ket{\psi}}^2\geq 1-\eps\Rightarrow& \expectc{a}{\norm{\brackets{\sigma_Z^{\areg}(a)-Z^{\breg}(a)}\ket{\psi}}^2}\leq2\eps\\
        \norm{\Pi_{\sf eq}(H_\areg^{\otimes n}\otimes \cal{U}_X)\ket{\psi}}^2\geq 1-\eps\Rightarrow& \expectc{a}{\norm{\brackets{\sigma_X^{\areg}(a)-X^{\breg}(a)}\ket{\psi}}^2}\leq2\eps.
    \end{align*}
    By \Cref{lem:EPR_pair_rigidity}, we obtain the result.
\end{proof}
\newcommand{\MultiStageIndependentMonogamy}{{\sf MultiStageIndependentMonogamy}\xspace}
\begin{definition}[Multi-Stage Independent Monogamy-of-Entanglement Game]\label{def:multi-Stage_independent_monogamy-of-entanglement}
    Let $\secp\in\mathbb{N}$ be the security parameter and let $n(\secp)$, $m(\secp)$ be polynomials. Consider the following game between the challenger and an adversary $\cal{A}=\allowdisplaybreaks(\cal{A}_M^0,\cal{A}_M^1,\cal{A}_M^2,\cal{A}_L^0,\cal{A}_L^1,\cal{A}_R^0,\cal{A}_R^1)$:
    \begin{enumerate}
        \item The challenger generates $\brackets{\sk,\ket{\chain^{\sk}}_{\anchor\vessel}}\gets\genchain(1^\secp)$ and gives the register $\vessel$ to $\cal{A}^0_M$.  All parties of the adversary are given oracle access to $\cal{O}_{T+v}$ and $\cal{O}_{S^\perp+u}$. Let $C_0=T+v$ and $C_1=S^\perp+u$.
        \item $\cal{A}_M^0$ can choose to abort in this step. If it aborts, the output of the game is $\bot$. Otherwise, it generates a tripartite state on $\lreg\mreg\rreg$. It sends $\lreg$ to $\cal{A}_L^0$, sends $\mreg$ to $\cal{A}_M^1$ and sends $\rreg$ to $\cal{A}_R^0$.
        \item The challenger samples $\theta_0\rand\set{0,1}$ and gives $\theta_0$ to $\cal{A}_M^1$. It also samples $b_0\rand\mathbb{F}_2^m$, and gives $\cal{A}_M^1$ oracle access to $\cal{O}_{C_{\theta_0}}^{b_0}$. Then, $\cal{A}_M^1$ on $\mreg$ generates a tripartite state on $\mreg_L\mreg'\mreg_R$. It sends $\mreg_L$ to $\cal{A}_L^0$, sends $\mreg'$ to $\cal{A}_M^2$ and sends $\mreg_R$ to $\cal{A}_R^0$. We will name the time just after all operations of $\cal{A}_M^1$ are done and before renaming the registers and splitting them as \BeforeSplit.
        \item $\cal{A}_L^0$ and $\cal{A}_R^0$ are given access to $\cal{O}_{C_{\theta_0}}^{b_0}$. $\cal{A}_L^0$ on $\lreg\mreg_L$ produces answer $b_0^l$ and a state on register $\lreg'$ that is sent to $\cal{A}_L^1$. $\cal{A}_R^0$ on $\rreg\mreg_R$ produces answer $b_0^r$ and a state on register $\rreg'$ that is sent to $\cal{A}_R^1$.
        \item The challenger samples $\theta_1\rand\set{0,1}$ and gives $\theta_1$ to $\cal{A}_M^2$. It also samples $b_1\rand\mathbb{F}_2^m$ gives $\cal{A}_M^2$ oracle access to $\cal{O}_{C_{\theta_0}}^{b_0}$ and $\cal{O}_{C_{\theta_1}}^{b_1}$. It generates a bipartite state on $\mreg'_L\mreg'_R$. $\mreg'_L$ is given to $\cal{A}_L^1$ and $\mreg'_R$ is given to $\cal{A}_R^1$. 
        \item $\cal{A}_L^1$ and $\cal{A}_R^1$ are given access to $\cal{O}_{C_{\theta_0}}^{b_0}$ and $\cal{O}_{C_{\theta_1}}^{b_1}$. $\cal{A}_L^1$ on $\lreg'\mreg'_L$ generates the answer $b_1^l$. $\cal{A}_R^1$ on $\rreg'\mreg'_R$ generates the answer $b_1^r$. The adversary wins iff $b_0^l=b_0^r=b_0$ and $b_1^l=b_1^r=b_1$. 
    \end{enumerate}
    Let $\MultiStageIndependentMonogamy(\cal{A},1^\secp)$ be the random variable that takes value $1/0/\bot$ if the adversary $\cal{A}$ wins/loses/aborts in the above game, respectively.
\end{definition}
\begin{theorem}\label{thm:extracting_EPR_pairs}
    Let $n(\secp)\geq\secp$ and $m_0(\secp)=m_1(\secp)=m(\secp)\geq\secp$ be polynomials of $\secp$. There exists an efficient extractor $\cal{E}$ with access to $\sk$ and the register $\mreg$ such that the following holds. Take any efficient adversary $\cal{A}$ in the multi-stage monogamy-of-entanglement game described in \Cref{def:multi-Stage_independent_monogamy-of-entanglement} such that
    \begin{equation*}
        \prob{\MultiStageIndependentMonogamy(\cal{A},1^\secp)\neq\bot}=\eps'(\secp)
    \end{equation*}
    and
    \begin{equation*}
        \frac{\prob{\MultiStageIndependentMonogamy(\cal{A},1^\secp)=1}}{\prob{\MultiStageIndependentMonogamy(\cal{A},1^\secp)\neq\bot}}=1-\eps(\secp)
    \end{equation*}
    for any function $\eps(\secp)$ and any inverse polynomial $\eps'(\secp)$.. Then running $\cal{E}^{\cal{A}_M^1}$ on the state $\mreg$ conditioned on $\cal{A}_M^0$ not aborting extracts a state on $\anchor\xreg$ that is very close to $n$ $\epr$-pairs. More specifically, with $\xreg$ denoting the output register of $\cal{E}^{\cal{A}_M^1}$,
    \begin{equation*}
        \expect{\tr\brackets{\bra{\epr}_{\anchor\xreg}^{\otimes n}\cal{E}^{\cal{A}_M^1}\rho_{\anchor\mreg}\brackets{\cal{E}^{\cal{A}_M^1}}^\dagger\ket{\epr}_{\anchor\xreg}^{\otimes n}}\middle\vert\substack{\brackets{\sk,\ket{\chain^{\sk}}_{\anchor\vessel}}\gets\genchain(1^\secp)\\ \cal{A}_M^0\text{ does not abort and produces }\rho_{\anchor\mreg}\text{ on }\xreg\mreg}}\geq 1-O\brackets{\eps^{1/4}}-\negl(\secp).
    \end{equation*}
    Furthermore, $\cal{E}^{\cal{A}_M^1}$ only makes black-box access to $\cal{A}_M^1$ and acts on $\mreg$. When we say $\cal{A}_M^1$, we refer to the process only up to \BeforeSplit, excluding the renaming and splitting procedure.
\end{theorem}
\begin{proof}
    Note that the game \MultiStageIndependentMonogamy is just a \MultiStageSearchMonogamy with $\frac{1}{4}$ probability, \DualMultiStageSearchMonogamy with $\frac{1}{4}$ and something else with $\frac{1}{2}$ probability. So such adversary $\cal{A}$ must succeed in \MultiStageSearchMonogamy and \DualMultiStageSearchMonogamy with probability at least $1-4\eps$. We can obtain two extractors $\cal{E}_S^{\cal{A}_M^1}$ and $\cal{E}_H^{\cal{A}_M^1}$ that, with probability $1-O\brackets{\eps^{1/2}}-\negl(\secp)$, extract the standard basis measurement and the Hadamard basis measurement by \Cref{lem:extracting_standard_basis_measurement_results} and \Cref{cor:extracting_hadamard_basis_measurement_results}. Then by \Cref{cor:extract_epr_pair}, we can construct $\cal{E}^{\cal{A}_M^1}$ from $\cal{E}_S^{\cal{A}_M^1}$ and $\cal{E}_H^{\cal{A}_M^1}$.
\end{proof}

%% file: warmup_entanglement_localization.tex
\providecommand{\info}{\ensuremath{\mathsf{info}}\xspace}

\section{Warmup: Localizing Entanglement in the High-Success-Probability Regime}

In this section, we present a warmup to the main results of the paper --- a protocol allowing us to localize EPR-pair halves to within some small region of spacetime, built and proven secure using the techniques developed in the previous sections. This is a simplified variant of the entanglement localization protocol in \Cref{subsec:entanglement-localization}, with the latter being a sequentially repeated version of the former. Sequential repetition will allow us to amplify extraction fidelity, even in the low-success-probability regime.

The definition and subsequent construction in this section will have a weaker form of extraction soundness than in the main result. In \Cref{subsec:entanglement-localization}, we give a stronger, more general definition accompanied by a construction which satisfies it, needing no additional assumptions.

\begin{definition}[Entanglement Localization, warmup version]
    An entanglement localization scheme consists of the following syntax:
    \begin{itemize}
        \item $\Setup(1^\secp)\to\param,\sparam,\ek,\areg,\breg$: The setup algorithm takes the security parameter $1^\secp$ and outputs the public parameter $\param$, the secret parameter $\sparam$, the extraction key $\ek$, and a state on the bipartite register $\areg\breg$.
        \item $\Localize\brackets{\cal{P}(\breg)\rightleftharpoons \cal{V}(\sparam)}(\param,L,t)\rightarrow \breg',b$: This is a positional protocol between a QPT prover holding a quantum register $\breg$ and a PPT\footnote{In general, we should allow any quantum verifier. Nevertheless, our construction achieves classical verification.} verifier with input $\sparam$, further specified by public inputs $\param$ and the claimed location $(L,t)$. The prover's output is a quantum register $\breg'$ and the verifier's output is an acceptance indicator $b\in\{\top,\bot\}$.
    \end{itemize}
    It should satisfy the following properties --- parameterized by completeness probability $\alpha(\secp)$, extraction fidelity $\beta(\secp,\eta)$, and localization parameter $\Delta(\secp)$ --- for all $\secp\in\N$.
    \begin{itemize}
        \item $\alpha$-\textbf{Completeness}: For any location $L$ and time $t$, there is a QPT prover $\cal{P}$ at position $(L,t)$ such that
        \[\Pr\left[b = \top : \begin{array}{r} \param,\sparam,\ek,\areg,\breg\gets \Setup(1^\secp) \\ \breg',b \gets \Localize\brackets{\cal{P}(\breg)\rightleftharpoons \cal{V}(\sparam)}(\param,L,t)\end{array}\right] \geq\alpha(\secp).\]
        \item $(\beta,\Delta)$-\textbf{Extraction Soundness}: 
        For any $L,t$ and QPT prover $\cal{P}^*$, let 
        \[1-\eta(\secp):=\Pr\left[b = \top : \begin{array}{r} \param,\sparam,\ek,\areg,\breg\gets \Setup(1^\secp) \\ \breg',b \gets \Localize\brackets{\cal{P}^*(\breg)\rightleftharpoons \cal{V}(\sparam)}(\param,L,t)\end{array}\right]\]
        be the probability that $\cal{P}^*$ passes the verification. Then there exists a local extractor $\cal{E}$ that extracts $\ket{\epr}^{\otimes n}$ from registers near location $L$ at time $t$ with fidelity at least $\beta(\secp,\eta)$:
        \[\hspace{-2.25em}\expect{\bra{\epr}^{\otimes n}_{\areg\breg^*}\rho_{\areg\breg^*}\ket{\epr}^{\otimes n}_{\areg\breg^*}:\begin{array}{r}\param,\sparam,\ek,\areg,\breg\gets \Setup(1^\secp)\\ \breg^*\gets\register[L_\Delta~@~t]\brackets{\Localize\brackets{\cal{P}^*(\breg)\rightleftharpoons \cal{V}(\sparam)}(\param,L,t)}\\\breg^*\gets\cal{E}(\param,\ek,\breg^*)\end{array}} \geq \beta(\secp,\eta)~,\]
        where  $\rho_{\areg\breg^*}$ is the joint state on $\areg\breg^*$ after running the extractor, and $L_\Delta=[L-\Delta,L+\Delta]$. Here we overload the notation $\breg^*$, meaning that we update the register $\breg^*$ by applying a channel to it.
    \end{itemize}

    We say that the entanglement localization scheme is \emph{non-destructive} if it additionally satisfies the following property.
    \begin{itemize}
        \item $\textbf{Non-destructive}$: For any spatial location $L$ and time $t$, 
        \[
        \expect{\TraceDist{\rho_{\areg\breg}}{\rho_{\areg\breg'}}:\begin{array}{r}\param,\sparam,\ek,\areg,\breg\gets \Setup(1^\secp)\\ \breg',b \gets \Localize\brackets{\cal{P}(\breg)\rightleftharpoons \cal{V}(\sparam)}(\param,L,t)\end{array}} \leq \negl(\secp)~,
        \]
        where $\rho_{\areg\breg}$ is the joint state on $\areg\breg$ after $\Setup$ and $\rho_{\areg\breg'}$ is the joint state on $\areg\breg'$ later after $\Localize$. 
    \end{itemize}
\end{definition}

\begin{remark}\label{remark:efficient-extractor}
    Note that this definition does not require the extractor $\cal{E}$ to be efficient. That is, $\cal{E}$ can run for unbounded time, and, in the oracle model, can make unbounded queries to its oracle. This is still meaningful, as entanglement is an information-theoretic notion that cannot be duplicated even given unbounded time. Moreover, the extractor never runs in the real world --- only in the analysis. Nevertheless, we note that our extractor \emph{is} efficient and requires an ``extraction key'' $\ek$ that is sampled by $\Setup$ (but not known to the prover). In our construction, $\ek=(U_{\sf shift},u,v)$. We mention this here because efficient extraction given $\ek$ may be a useful feature, depending on the application.
\end{remark}

\begin{construction}[Entanglement Localization, warmup version]\label{con:non-destructive_entanglement_localization}
        Let $\secp\in\mathbb{N}$ be the security parameter and $n(\secp),m(\secp)\geq\secp$ be polynomials. Take some localization parameter $\Delta(\secp)$. We consider the following non-destructive entanglement localization scheme:
        \begin{itemize}
            \item $\Setup(1^\secp)$:
            \begin{enumerate}
                \item Generate $\brackets{\sk,\ket{\chain^\sk}_{\areg\breg}}\gets\genchain(1^\secp)$, and let $U_{\sf shift},S,T,u,v$ be the basis-matrix, subspaces, and shifts determined by $\sk$.
                \item Sample the public challenge oracle $\cal{O}_{\chal}:\mathbb{F}_2^\secp\rightarrow\set{0,1}$.
                \item Sample the public oracle $\cal{O}:\{0,1\}\times\mathbb{F}_2^\secp\times\mathbb{F}_2^{3n}\rightarrow\mathbb{F}_2^m\cup\set{\bot}$ is based on a random oracle $\cal{O}_{\sf random}:\{0,1\}\times\mathbb{F}_2^\secp\rightarrow \mathbb{F}_2^m$ that is \textbf{not} publicly available.
                \begin{equation*}
                    \cal{O}(\theta,x,z)=\left\{
                    \begin{aligned}
                        &\cal{O}_{\sf random}(\theta,x) \quad &\text{if }\theta=0\text{ and }z\in T+v,\\
                        &\cal{O}_{\sf random}(\theta,x) \quad &\text{if }\theta=1\text{ and }z\in S^\perp+u,\\
                        &\bot \quad &\text{otherwise}
                    \end{aligned}
                    \right.
                \end{equation*}
            \end{enumerate}
            Then it generates $\ket{\chain^\sk}_{\areg\breg}$ and outputs:
            \begin{itemize}
                \item the public parameter $\param=(\cal{O}_{\chal},\cal{O})$.
                \item The secret parameters $\sparam=(u,v)$. This is given to the verifiers $\cal{V}_L$ and $\cal{V}_R$.
                \item The extraction key $\ek=(U_{\sf shift},u,v)$.
                \item The anchor state $\ket{\chain^\sk}_{\areg\breg}$, with $\anchor$ renamed $\areg$ and $\vessel$ renamed $\breg$.
            \end{itemize}
            \item $\Localize\brackets{\cal{P}(\breg)\rightleftharpoons \cal{V}(\sparam)}(\param,L,t)$: Let $\delta=\Delta/2$.
            \begin{itemize}
                \item For $i=0,1$, $\cal{V}_L$ samples $x_{L,i}\rand\mathbb{F}_2^\secp$ and broadcasts it at time $t+i\delta-(L+1)$. It expects responses $y_{L,i}$ at time $t+i\delta+(L+1)$.
                \item For $i=0,1$, $\cal{V}_R$ samples $x_{R,i}\rand\mathbb{F}_2^\secp$ and broadcasts it at time $t+i\delta-(1-L)$. It expects responses $y_{R,i}$ at time $t+i\delta+(1-L)$.
                \item If any of these responses is missing, the verifier outputs $\bot$. Otherwise, the verifier outputs $\top$ if and only if for each $i=0,1$, one of the following holds:
                \begin{itemize}
                    \item $\theta_i:=\cal{O}_{\chal}(x_{L,i}\oplus x_{R,i})=0$ and $y_{L,i}=y_{R,i}=\cal{O}(0,x_{L,i}\oplus x_{R,i},v)$.
                    \item $\theta_i:=\cal{O}_{\chal}(x_{L,i}\oplus x_{R,i})=1$ and $y_{L,i}=y_{R,i}=\cal{O}(1,x_{L,i}\oplus x_{R,i},u)$.
                \end{itemize}
                \item \item The honest prover consists of a single party always holding the $\breg$ register at location $L$. For each $i=0,1$, it receives $x_{L,i},x_{R,i}$ at time $t+i\delta$. If $\theta_i=0$, it coherently evaluates $\cal{O}(0,x_{L,i}\oplus x_{R,i},\cdot)$ on the $\breg$ register in the standard basis and broadcasts the result. If $\theta_i=1$, it coherently evaluates $\cal{O}(1,x_{L,i}\oplus x_{R,i},\cdot)$ on the $\breg$ register in the Hadamard basis and broadcasts the result.
            \end{itemize}
        \end{itemize}
    \end{construction}

\begin{theorem}[NDEL security, warmup version]
\label{thm:NDEL_overall}
    The construction in \Cref{con:non-destructive_entanglement_localization} is a non-destructive entanglement localization scheme with $1$-\textbf{Completeness} and $\brackets{\beta,\Delta}$-\textbf{Extraction Soundness}, where
    \[
        \beta(\secp,\eta)=1-O\!\left(\eta(\secp)^{1/4}\right)-\negl(\secp).
    \]
\end{theorem}
\begin{proof}
    For completeness, by \Cref{fact:standard_hadamard_duality}, the standard basis measurement on $\breg$ will always give a vector in $T+v$ and the Hadamard basis measurement on $\breg$ will always give a vector in $S^\perp+u$. Thus, the coherent query result will be $y_i=\cal{O}_{\sf random}(\theta_i,x_{L,i}\oplus x_{R,i})$ with probability $1$. From there, non-destructiveness also follows, by an application of gentle measurement.
    
    For extraction soundness, suppose that there exists a prover $\cal{P}^*$ that passes localize experiment at a point $(L,t_0)$ with probability $1-\eta$. We will turn $\cal{P}^*$ into a prover $\cal{A}=(\cal{A}_M^0,\cal{A}_M^1,\cal{A}_L^0,\cal{A}_R^0,\cal{A}_M^2,\cal{A}_L^1,\cal{A}_R^1)$ for the \MultiStageIndependentMonogamy game, by \emph{chunking} the execution of $\cal{P}^*$ according to its spacetime geometry. This is pictured in \Cref{fig:time-position-B}. Let $r(\secp)$ be the total number of queries to all oracles that $\cal{P}^*$ makes (from any location and at any time), then let $\set{F_k:\set{\cal{O}_{\chal}(\cdot),\cal{O}(0,\cdot,\cdot),\cal{O}(1,\cdot,\cdot)}\times\mathbb{F}_2^\secp\rightarrow\mathbb{F}_2^m}_{k}$ be a family of $2r$-wise independent hash functions. The first entry of the input is the `name' of the oracle that $F_k$ is used to simulate. Let $\delta:=\Delta/2$, $t_0=t$, and $t_1:=t+\delta$. For notational clarity, we overload $p,t$ to refer to position and time axes, respectively. We construct each component of $\cal{A}$ as follows:
    \begin{enumerate}
        \item $\cal{A}_M^0$ receives register $\breg:=\vessel$, membership oracles $\cal{O}_{C_0}$ and $\cal{O}_{C_1}$ (where $C_0=T+v$ and $C_1=S^\perp+u$, as in \Cref{def:multi-Stage_independent_monogamy-of-entanglement}) from the challenger. It then prepares to simulate the localization experiment by sampling a hash function $F_k$ and secret-shared challenges $x_{L,0},x_{R,0},x_{L,1},x_{R,1}$. Define $\info:=(k,x_{L,0},x_{R,0},x_{L,1},x_{R,1})$.
        
        $\cal{A}_M^0$ then begins simulating $\cal{P}^*$ on input $\breg$, answering all of $\cal{P}^*$'s queries to $\cal{O}$ as follows: a query $(\theta,x, z)$ is answered by checking membership in the corresponding coset using $\cal{O}_{C_\theta}$, and then outputting $F_k(\theta,x)$ in place of $\cal{O}_{\sf random}(\theta,x)$ if the membership check passed. Crucially, $\cal{P}^*$ is only simulated over a carefully chosen region of spacetime: $\cal{A}_M^0$ freezes each of the simulated prover's computations and messages \emph{just before} they would touch the spacetime line $t-p=t_0-L$ (on the right of $L$) or $t+p=t_0+L$ (on the left of $L$). This cutoff is illustrated by a pair of dashed lines in \Cref{fig:time-position-B}. Let $\lreg$ be the register that stores $\info$ and all simulated parties/messages with $p\in[-1,L-\delta]$ when the simulation ends. Similarly, let $\rreg$ be the register that stores $\info$ and all simulated parties/messages with $p\in[L+\delta,1]$ when the simulation ends. Finally, let $\mreg$ be the register that stores $\info$ and all parties/messages with $p\in[L-\delta,L+\delta]$ when the simulation ends.
        \item $\cal{A}_M^1$ receives $\mreg$ from $\cal{A}_M^0$, basis $\theta_0$, and oracle access to $\cal O^{b_0}_{C_{\theta_0}}$ from the challenger. It then resumes simulation of all of $\cal{P}^*$'s (real-time) computations and messages which would take place in the spacetime triangle -- colored green in \Cref{fig:time-position-B} -- described by the constraints
        \begin{itemize}
            \item $t-p\geq t_0-L$;
            \item $t+p\geq t_0+L$;
            \item $t<t_1$,
        \end{itemize}
        and halts simulation just before exiting this region. During this phase, the simulated oracle $\cal{O}$ is reprogrammed so that on a query $(\theta,x,z)$ with $(\theta,x)=(\theta_0,x_{L,0}\oplus x_{R,0})$, $\cal{A}_M^1$ answers with $\cal{O}_{C_{\theta_0}}^{b_0}(z)$. The other simulated oracle $\cal{O}_{\chal}$ is reprogrammed so that on the query $x=x_{L,0}\oplus x_{R,0}$ it outputs $\theta_0$. Other oracle outputs not mentioned here are still simulated using the membership oracles and $F_k$. Let $\mreg_L$ be the register that stores $\theta_0$ and all parties/messages with $p\in[L-\delta,L)$ when the simulation ends. Similarly, let $\mreg_R$ be the register that stores $\theta_0$ and all parties/messages with $p\in(L,L+\delta]$ when the simulation ends. Finally, append $\theta_0$ to $\info$ and let $\mreg'$ be the register that stores $\info$ and all parties/messages at the spacetime point $(L,t_1)$.
        \item $\cal{A}_L^0$ receives $\lreg$ from $\cal{A}_M^0$, $\mreg_L$ from $\cal{A}_M^1$, and oracle access to $\cal O^{b_0}_{C_{\theta_0}}$ from the challenger. It then resumes simulation of all of $\cal P^*$'s real-time computations and messages which would take place in the spacetime trapezoid -- colored red in \Cref{fig:time-position-B} -- described by the constraints
        \begin{itemize}
            \item $p\ge -1$;
            \item $t\ge t_1$;
            \item $t+p\ge t_0+L$;
            \item $t+p<t_1+L$,
        \end{itemize}
        and halts simulation just before exiting this region. During this simulation, $\cal O$ is reprogrammed in the same way as for $\cal A_M^1$. If the simulated prover did not send a classical message which would be received by $\cal{V}_L$ at time $t_0+L+1$, $\cal{A}_L^0$ outputs $\bot$. Otherwise, $\cal A_L^0$ uses $\cal P^*$'s response message sent to $\cal V_L$ as its output $b_0^L$. Let $\lreg'$ be the register that stores $\info$ and all parties/messages with $p\in[-1,L)$ when $\cal A_L^0$'s simulation ends.
        \item $\cal{A}_R^0$ receives $\rreg$ from $\cal{A}_M^0$, $\mreg_R$ from $\cal{A}_M^1$, and oracle access to $\cal O^{b_0}_{C_{\theta_0}}$ from the challenger. It then resumes simulation of all of $\cal P^*$'s real-time computations and messages which would take place in the spacetime trapezoid -- colored blue in \Cref{fig:time-position-B} -- described by the constraints
        \begin{itemize}
            \item $p\le 1$;
            \item $t\ge t_1$;
            \item $t-p\ge t_0-L$;
            \item $t-p<t_1-L$,
        \end{itemize}
        and halts simulation just before exiting this region. During this simulation, $\cal O$ is reprogrammed in the same way as for $\cal A_M^1$. If the simulated prover did not send a classical message which would be received by $\cal{V}_R$ at time $t_0+1-L$, $\cal{A}_R^0$ outputs $\bot$. Otherwise, $\cal A_R^0$ uses $\cal P^*$'s response message sent to $\cal V_R$ as its output $b_0^R$. Let $\rreg'$ be the register that stores $\info$ and all parties/messages with $p\in(L,1]$ when $\cal A_R^0$'s simulation ends.
        \item $\cal{A}_M^2$ receives $\mreg'$ from $\cal{A}_M^1$, basis $\theta_1$, and oracle access to $\cal O^{b_0}_{C_{\theta_0}}$ and $\cal O^{b_1}_{C_{\theta_1}}$ from the challenger. It then resumes simulation of all of $\cal P^*$'s real-time computations and messages which would take place exactly at the point $(L,t_1)$. On top of reprogramming the oracle $\cal O$ in the same way as for $\cal A_M^1$, we add two additional reprograms: 
        \begin{itemize}
            \item On a query $(\theta,x,z)$ to $\cal{O}$ with $(\theta,x)=(\theta_1,x_{L,1}\oplus x_{R,1})$, it answers with $\cal O^{b_1}_{C_{\theta_1}}(z)$.
            \item On the query $x=x_{L,1}\oplus x_{R,1}$ to $\cal{O}_{\chal}$ it outputs $\theta_1$.
        \end{itemize}
        Let $\mreg'_L$ be the register that stores $\theta_1$ and all messages sent to the left from $L$ at time $t_1$. Similarly, let $\mreg'_R$ be the register that stores $\theta_1$ and all messages sent to the right from $L$ at time $t_1$.
        \item $\cal{A}_L^1$ receives $\lreg'$ from $\cal A_L^0$, $\mreg'_L$ from $\cal A_M^2$, and oracle access to $\cal O^{b_0}_{C_{\theta_0}}$ and $\cal O^{b_1}_{C_{\theta_1}}$ from the challenger. It then resumes simulation of all of $\cal P^*$'s real-time computations and messages which would take place in the spacetime line segment described by
        \begin{itemize}
            \item $p\in[-1,L)$;
            \item $t+p=t_1+L$,
        \end{itemize}
        and halts simulation just before exiting this region. $\cal O$ is reprogrammed in the same way as for $\cal A_M^2$. If the simulated prover did not send a classical message which would be received by $\cal{V}_L$ at time $t_1+L+1$, $\cal{A}_L^1$ outputs $\bot$. Otherwise, $\cal A_L^1$ uses that response message as its output $b_1^L$.
        \item $\cal{A}_R^1$ receives $\rreg'$ from $\cal A_R^0$, $\mreg'_R$ from $\cal A_M^2$, and oracle access to $\cal O^{b_0}_{C_{\theta_0}}$ and $\cal O^{b_1}_{C_{\theta_1}}$ from the challenger. It then resumes simulation of all of $\cal P^*$'s real-time computations and messages which would take place in the spacetime line segment described by
        \begin{itemize}
            \item $p\in(L,1]$;
            \item $t-p=t_1-L$,
        \end{itemize}
        and halts simulation just before exiting this region. $\cal O$ is reprogrammed in the same way as for $\cal A_M^2$. If the simulated prover did not send a classical message which would be received by $\cal{V}_R$ at time $t_1+1-L$, $\cal{A}_R^1$ outputs $\bot$. Otherwise, $\cal A_R^1$ uses that response message as its output $b_1^R$.
    \end{enumerate}
    \begin{figure}[h]
        \centering
        \begin{tikzpicture}[>=Latex]
    
        \def\z{4}
        \fill[yellow!25]
          (0,0) --
          (-\z,\z) --
          (-\z,-0.5*\z) -- 
          (\z,-0.5*\z) --
          (\z,\z) -- cycle;
        \node at (0.56*\z,-0.36*\z) {$\mathcal A_M^0$};
        
        \draw[->] (-1.5*\z,0) -- (1.5*\z,0) node[right] {$t=t_0$};
        \draw[-] (0,-0.5*\z) node [below] {$p=L$} -- (0,0);
        \draw[-] (-\z,2*\z) -- (-\z,-0.5*\z) node[below] {$p=-1$};
        \draw[-] (\z,2*\z) -- (\z,-0.5*\z) node[below] {$p=1$};
        
        \def\d{0.3*\z}
        \def\p{0}
    
        \coordinate (EL) at (-2*\d,0);
        \coordinate (ER) at ( 2*\d,0);
        \fill[cyan!10]
          (EL) --
          (-\d,\d) --
          ( \d,\d) --
          (ER) -- cycle;
        \draw[dashed, green!55!black] (EL) -- (-\d,\d);
        \draw[dashed, green!55!black] (ER) -- ( \d,\d);
        \draw[densely dashed, gray!70] (-\d,0) -- (-\d,\d);
        \draw[densely dashed, gray!70] ( \d,0) -- ( \d,\d);
        \draw[line width=3.2pt, green!55!black] (EL) -- (ER);
        \draw[green!55!black]
          ($(EL)+(0,-0.09*\z)$) -- ($(EL)+(0,0.09*\z)$)
          ($(ER)+(0,-0.09*\z)$) -- ($(ER)+(0,0.09*\z)$)
          ($(-\d,0)+(0,-0.07*\z)$) -- ($(-\d,0)+(0,0.07*\z)$)
          ($( \d,0)+(0,-0.07*\z)$) -- ($( \d,0)+(0,0.07*\z)$);
        \node[below=12pt, font=\small, inner sep=0.5pt] at (EL) {$L-2\delta$};
        \node[below=5pt, font=\small, inner sep=0.5pt] at (-\d,0) {$L-\delta$};
        \node[below=5pt, font=\small, inner sep=0.5pt] at ( \d,0) {$L+\delta$};
        \node[below=12pt, font=\small, inner sep=0.5pt] at (ER) {$L+2\delta$};
        \node[align=center, inner sep=1.5pt] (ExtractionZoneLabel) at (-3.25*\d-0.53,0.62*\d) {extraction zone};
        \draw[->, green!55!black, shorten <=1pt, shorten >=1pt]
          (ExtractionZoneLabel.south east) -- (-1.95*\d,0);
        
        \fill[green!25]
          (0,0) --
          (-\d,\d) --
          (\d,\d) -- cycle;
        \node at (0,0.15*\z) {$\mathcal A_M^1$};
        
        \fill[red!25]
          (-\d,\d) --
          (-\z,\z) --
          (-\z,\z+\d+\p) --
          (\p,\d) -- cycle;
        \node at (-0.4*\z,0.58*\z) {$\mathcal A_L^0$};
        
        \fill[blue!25]
          (\d,\d) --
          (\z,\z) --
          (\z,\z+\d-\p) --
          (\p,\d) -- cycle;
        \node at (0.4*\z,0.53*\z) {$\mathcal A_R^0$};
        
        \draw[very thick, red!60] (\p,\d) -- (-\z,\z+\d+\p)
          node[midway,above,color=black] {$\mathcal A_L^1$};
        
        \draw[very thick, blue!60] (\p,\d) -- (\z,\z+\d-\p)
          node[midway,above=4pt,color=black] {$\mathcal A_R^1$};
        
        \draw[dashed] (-0.5*\z,-0.5*\z) -- (1.5*\z,1.5*\z) node[right] {$t-p=0$};
        \draw[dashed] (0.5*\z,-0.5*\z) -- (-1.5*\z,1.5*\z) node[left] {$t+p=0$};
        
        \draw[dotted] (-0.5*\z-\d+\p,-0.5*\z) -- (1.5*\z,1.5*\z+\d-\p) node[right] {$t-p=\delta$};
        \draw[dotted] (0.5*\z+\d+\p,-0.5*\z) -- (-1.5*\z,1.5*\z+\d+\p) node[left] {$t+p=\delta$};
        
        \draw[dotted, thick] (-1.5*\z,\d) -- (1.5*\z,\d) node[right] {$t=t_0+\delta$};
        \draw[line width=1.6pt, blue!65] (0,0) -- (-\d,\d) (0,0) -- (\d,\d);
        \node[align=center, inner sep=1.5pt] (MRegLabel) at (-0.01*\z,0.75*\z) {$\mreg$};
        \draw[->, blue!65, shorten <=1pt, shorten >=1pt]
          (MRegLabel.south) -- (-0.6*\d,0.6*\d);
        \fill[black] (0,\d) circle (2.2pt)
          node[above=3pt] {$\mathcal A_M^2$};
        
        \draw[dotted, thick] (-1.5*\z,\z) -- (1.5*\z,\z) node[right] {$t=t_0+1$};
        
    \end{tikzpicture}
        
        \caption{Spacetime decomposition of the prover $\cal A$ for our localization reduction, shown with $L=0$.}
    \label{fig:time-position-B}
    \end{figure}
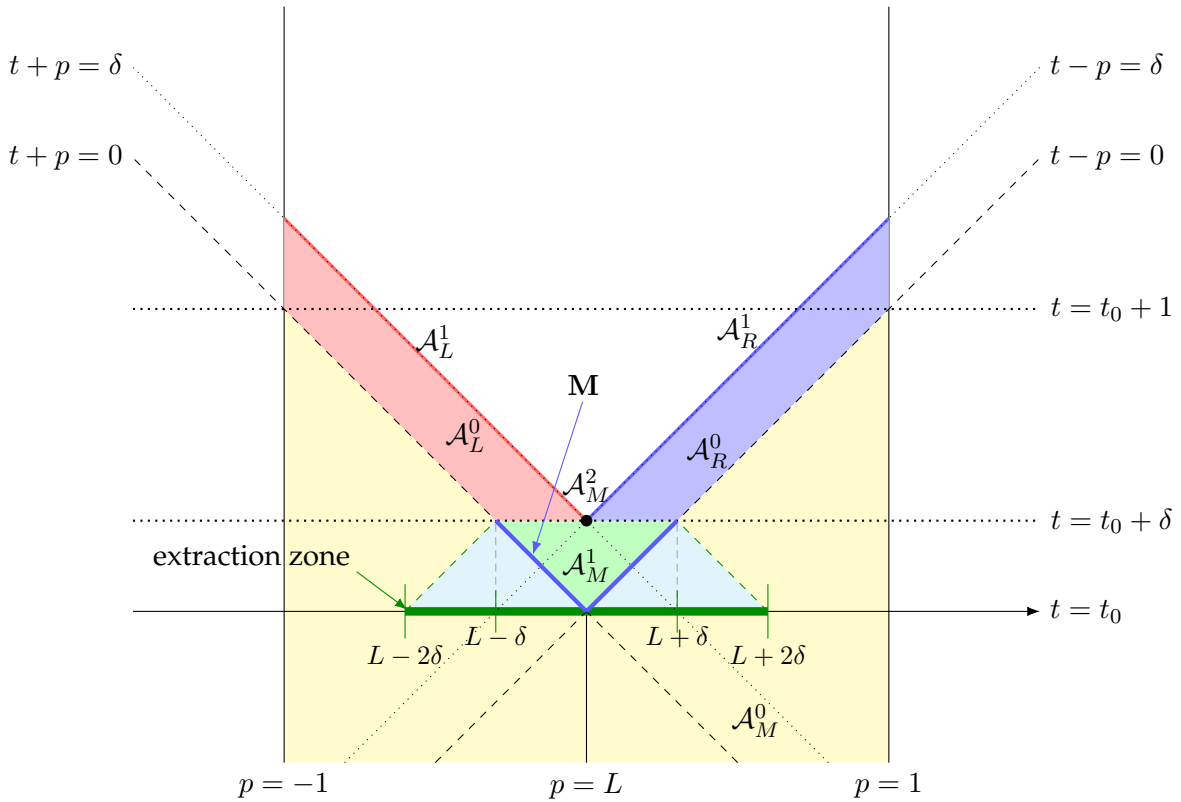
    \begin{claim}
        For any prover $\cal{P}^*$ that passes with probability $1-\eta$, the probability that $\cal{A}$ wins \Cref{def:multi-Stage_independent_monogamy-of-entanglement} is at least $1-\eta-\negl$.
    \end{claim}
    \begin{proof}
        By \cite[Theorem 3.1]{C:Zhandry12}, the $2r$-wise independent hash function perfectly simulates the random oracle for $\cal{P}^*$. Also we can see that $\cal{A}$ perfectly simulates $\cal{P}^*$, since for any event simulated in $\cal{A}$, all other events that can have a causal effect on that event have already been simulated. The only difference is that the oracles $\cal{O}_{\chal},\cal{O}$ are both reprogrammed twice, but by \Cref{lem:oracle_reprogram}, the resulting experiment is negligibly close to the actual $\Localize$ experiment. And there, $\cal{P}^*$ eventually outputs $y_{L,0}=y_{R,0}=\cal{O}_{\sf random}(\theta_0,x_{L,0}\oplus x_{R,0})$ and $y_{L,1}=y_{R,1}=\cal{O}_{\sf random}(\theta_1,x_{L,1}\oplus x_{R,1})$ with probability $1-\eta$. Therefore $\cal{A}$ will produce the correct outputs $b^L_0,b^R_0,b^L_1,b^R_1$ with probability $1-\eta-\negl(\secp)$.
    \end{proof}
    By the above claim and \Cref{thm:extracting_EPR_pairs}, there exists an extractor $\cal{E}^{\cal{A}_M^1}$ acting on $\mreg$ with access to $\cal{A}_M^1$ that extracts $n$ $\epr$ pairs with expected fidelity
    \[
        1-O\brackets{\eta^{1/4}}-\negl.
    \]
    Recall that $\mreg$ is the joint register of all parties/messages with $p\in[L-\delta,L+\delta]$ with $t-t_0=|p-L|$, corresponding to the two bottom boundaries of the green triangle in \Cref{fig:time-position-B}. These registers can be derived from access to all parties/messages of $\cal{P}^*$ in the spacetime region $[L-2\delta,L+2\delta]\times\{t_0\}$ by forward time-evolution, as pictured in the diagram. Note that this is exactly the extraction region $L_\Delta~@~t_0$ from the definition, since $2\delta=\Delta$. Now, let us define our final extractor $\cal{E}$ that takes as input $\param$, $\ek$, and $\register[L_\Delta~@~t_0]$ and outputs something close to $\ket{\epr}^{\otimes n}$. It does the following:
    \begin{itemize}
        \item Recover the joint register $\mreg$ via forward-simulation.
        \item Use the extraction key $\ek=(U_{\sf shift},u,v)$ and run $\cal{E}^{\cal{A}_M^1}$ on $(\ek,\mreg)$ with oracle access to oracles in $\param$.
    \end{itemize}
\end{proof}

%% file: localizing_quantum_information_sequential_repetition.tex
\section{Localizing Quantum Information}
\label{sec:localizing}
In this section, we present our formal definitions and constructions of the localization primitives mentioned in the introduction: entanglement localization, trajectory verification, state localization, and functionality localization. These definitions will have a stronger extraction soundness condition than was shown in the warmup, which is much more meaningful against low-success-probability adversaries. To achieve this stronger soundness, our constructions in this section employ a sequential repetition technique --- each protocol is the sequential chaining of $O(\gamma)$-many two-round protocols. 

\subsection{Entanglement Localization}\label{subsec:entanglement-localization}
\begin{definition}[Entanglement Localization]\label{def:NDEL}
    An entanglement localization scheme is parameterized by the security parameter $\secp$ and a bipartite state $\sigma$\footnote{\label{footnote:secp}In reality, we localize a family of bipartite states indexed by $\secp$, i.e. $\Sigma=\{\sigma_\secp\}_\secp$. As is standard in these contexts, we suppress the dependence on $\secp$ for notational convenience, and simply write $\sigma$.}. It consists of the following syntax:
    \begin{itemize}
        \item $\Setup(1^\secp)\to\param,\sparam,\ek,\areg,\breg$: The setup algorithm takes the security parameter $1^\secp$ and outputs the public parameter $\param$, the secret parameter $\sparam$, the extraction key $\ek$, and a state on the bipartite register $\areg\breg$.
        \item $\Localize\brackets{\cal{P}(\breg) \rightleftharpoons \cal{V}(\sparam)}(\param,L,t)\rightarrow \breg',b$: This is a positional protocol between a QPT prover holding a quantum register $\breg$ and a PPT verifier with input $\sparam$,  further specified by the following public inputs: public parameters $\param$, spatial location $L \in [-1,1]$, and time $t$. The prover's output is a quantum register $\breg'$ and the verifier's output is an acceptance indicator $b\in\{\top,\bot\}$.
    \end{itemize}
    It should satisfy the following properties --- parameterized by completeness probability $\alpha(\secp)$, extraction fidelity $\beta(\secp,\eta)$, extraction probability $\beta'(\secp,\eta)$, and localization parameter $\Delta(\secp)$ --- for all $\secp\in\N$.
    \begin{itemize}
        \item $\alpha$-\textbf{Completeness}: For any location $L$ and time $t$, there is a QPT prover $\cal P$ at position $(L,t)$ such that
        \[\Pr\left[b = \top : \begin{array}{r} \param,\sparam,\ek,\areg,\breg\gets \Setup(1^\secp) \\ \breg',b \gets \Localize\brackets{\cal{P}(\breg)\rightleftharpoons \cal{V}(\sparam)}(\param,L,t)\end{array}\right] \geq\alpha(\secp).\]
        \item $(\beta,\beta',\Delta)$-\textbf{Extraction Soundness}:
        For any $L,t$ and QPT prover $\cal{P}^*$, let
        \[\eta(\secp):=\Pr\left[b = \top : \begin{array}{r} \param,\sparam,\ek,\areg,\breg\gets \Setup(1^\secp) \\ \breg',b \gets \Localize\brackets{\cal{P}^*(\breg)\rightleftharpoons \cal{V}(\sparam)}(\param,L,t)\end{array}\right]\]
        be $\cal P^*$'s success probability. Assume $\eta(\secp)\ge 1/p(\secp)$ for some polynomial $p$. Then there exists a local extractor $\cal{E}$ acting on registers near location $L$ at time $t$ such that the probability of extracting $\sigma_\secp$ with fidelity at least $\beta(\secp,\eta)$ is at least $\beta'(\secp,\eta)$:
        \[
        \hspace{-4em}
        \prob{
            F\big(\sigma_\secp,\rho_{\areg\breg^*}\big)\ge \beta(\secp,\eta)~
        :
        \begin{array}{r}
            \param,\sparam,\ek,\areg,\breg\gets \Setup(1^\secp)\\ \breg^*\gets\register[L_\Delta~@~t]\brackets{\Localize\brackets{\cal{P}^*(\breg)\rightleftharpoons \cal{V}(\sparam)}(\param,L,t)}\\\breg^*\gets\cal{E}(\param,\ek,\breg^*)
        \end{array}
        }
        \geq \beta'(\secp,\eta)~,
        \]
        where $\rho_{\areg\breg^*}$ is the joint state on $\areg\breg^*$ after running the extractor, and $L_\Delta=[L-\Delta,L+\Delta]$.
    \end{itemize}

    We say that the entanglement localization scheme is \emph{non-destructive} if it additionally satisfies the following property.
    \begin{itemize}
        \item $\textbf{Non-destructive}$: For any spatial location $L$ and time $t$,
        \[
        \expect{\TraceDist{\rho_{\areg\breg}}{\rho_{\areg\breg'}}:\begin{array}{r}\param,\sparam,\ek,\areg,\breg\gets \Setup(1^\secp)\\ \breg',b \gets \Localize\brackets{\cal{P}(\breg)\rightleftharpoons \cal{V}(\sparam)}(\param,L,t)\end{array}} \leq \negl(\secp)~,
        \]
        where $\rho_{\areg\breg}$ is the joint state on $\areg\breg$ after $\Setup$ and $\rho_{\areg\breg'}$ is the joint state on $\areg\breg'$ later after $\Localize$.
    \end{itemize}
\end{definition}
\begin{construction}[Entanglement Localization Scheme]\label{con:entanglement_localization}
    Let $\secp\in\mathbb{N}$ be the security parameter and $n(\secp),m(\secp)\geq\secp$ be polynomials. Take a localization parameter $\Delta(\secp)\leq 1$ and a repetition factor $\gamma(\secp)$. We consider the following non-destructive entanglement localization scheme:
    \begin{itemize}
        \item There are two verifiers: $\cal{V}_L$ at $-1$ and $\cal{V}_R$ at $1$.
        \item $\Setup(1^\secp)$: 
        \begin{enumerate}
            \item Generate $\brackets{\sk,\ket{\chain^\sk}_{\areg\breg}}\gets\genchain(1^\secp)$, and let $U_{\sf shift},S,T,u,v$ be the basis-matrix, subspaces, and shifts determined by $\sk$.
            \item Sample the public challenge oracle $\cal{O}_{\chal}:\mathbb{F}_2^\secp\rightarrow\set{0,1}$.
            \item Sample the public oracle $\cal{O}:\{0,1\}\times\mathbb{F}_2^\secp\times\mathbb{F}_2^{3n}\rightarrow\mathbb{F}_2^m\cup\set{\bot}$ is defined with respect to a random oracle $\cal{O}_{\sf random}:\{0,1\}\times\mathbb{F}_2^\secp\rightarrow \mathbb{F}_2^m$ that is \textbf{not} publicly available:
            \begin{equation*}
                \cal{O}(\theta,x,z)=\left\{
                \begin{aligned}
                    &\cal{O}_{\sf random}(\theta,x) \quad &\text{if }\theta=0\text{ and }z\in T+v,\\
                    &\cal{O}_{\sf random}(\theta,x) \quad &\text{if }\theta=1\text{ and }z\in S^\perp+u,\\
                    &\bot \quad &\text{otherwise.}
                \end{aligned}
                \right.
            \end{equation*}
        \end{enumerate}
        Then $\Setup$ outputs:
        \begin{itemize}
            \item the public parameter $\param=(\cal{O}_\chal,\cal{O})$.
            \item The secret parameter $\sparam=(u,v)$
            \item The extraction key $\ek=(U_{\sf shift},u,v)$.
            \item The anchor state $\ket{\chain^\sk}_{\areg\breg}$, with $\anchor$ renamed $\areg$ and $\vessel$ renamed $\breg$.
        \end{itemize}
        \item $\Localize\brackets{\cal{P}(\breg)\rightleftharpoons \cal{V}(\sparam)}(\param,L,t)$: Let $\delta=\Delta/(\gamma+2)$ and define checkpoints $t_i=t+i\delta$ for $i=0,1,\ldots,\gamma$.
        \begin{itemize}
            \item For $i=0,1,\cdots,\gamma$, $\cal{V}_L$ samples $x_{L,i}\rand\mathbb{F}_2^\secp$ and broadcasts it at time $t_i-(L+1)$. It expects responses $y_{L,i}$ at time $t_i+(L+1)$.
            \item For $i=0,1,\cdots,\gamma$, $\cal{V}_R$ samples $x_{R,i}\rand\mathbb{F}_2^\secp$ and broadcasts it at time $t_i-(1-L)$. It expects responses $y_{R,i}$ at time $t_i+(1-L)$.
            \item If any of these responses is missing, the verifier outputs $\bot$. Otherwise, the verifier outputs $\top$ if and only if for each $i=0,1,\cdots,\gamma$, one of the following holds:
            \begin{itemize}
                \item $\theta_i:=\cal{O}_{\chal}(x_{L,i}\oplus x_{R,i})=0$ and $y_{L,i}=y_{R,i}=\cal{O}(0,x_{L,i}\oplus x_{R,i},v)$.
                \item $\theta_i:=\cal{O}_{\chal}(x_{L,i}\oplus x_{R,i})=1$ and $y_{L,i}=y_{R,i}=\cal{O}(1,x_{L,i}\oplus x_{R,i},u)$.
            \end{itemize}
            \item The honest prover consists of a single party always holding the $\breg$ register at location $L$. For each $i=0,1,\ldots,\gamma$, it receives $x_{L,i},x_{R,i}$ at time $t_i$. If $\theta_i=0$, it coherently evaluates $\cal{O}(0,x_{L,i}\oplus x_{R,i},\cdot)$ on the $\breg$ register in the standard basis and broadcasts the result. If $\theta_i=1$, it coherently evaluates $\cal{O}(1,x_{L,i}\oplus x_{R,i},\cdot)$ on the $\breg$ register in the Hadamard basis and broadcasts the result.
        \end{itemize}
    \end{itemize}
\end{construction}
\begin{theorem}\label{thm:entanglement_localization}
    The construction in \Cref{con:entanglement_localization} is a non-destructive entanglement localization scheme for the state $\sigma=\ketbra{\epr}{\epr}^{\otimes n}$ in the classical oracle model with $1$-\textbf{Completeness} and $\brackets{\beta,\beta',\Delta}$-\textbf{Extraction Soundness}, where 
    \[
        \beta(\secp,\eta)=1-O\!\left(\left(\frac{\log(1/\eta)}{\gamma}\right)^{1/4}\right)-\negl(\secp),
        \qquad
        \beta'(\secp,\eta)=0.99\eta(\secp).\footnote{This $0.99$ factor can be an arbitrary constant in $(0,1)$, it will affect the constant in the big-$O$ notation of $\beta(\secp,\eta)$.}
    \]
    In particular, since we only consider adversaries for which $1/\eta$ is at most polynomial in $\secp$, for any polynomial $t(\secp)$, we can set $\gamma=c\secp t^4$ for some constant $c$ to achieve
    \[\beta(\secp,\eta)\geq1-1/t(\secp)~.\]
\end{theorem}

\begin{proof}
    Completeness follows from \Cref{fact:standard_hadamard_duality}: the standard-basis measurement on $\breg$ always gives a vector in $T+v$, and the Hadamard-basis measurement on $\breg$ always gives a vector in $S^\perp+u$. Therefore the honest coherent query in every checkpoint returns $\cal{O}_{\sf random}(\theta_i,x_{L,i}\oplus x_{R,i})$, and the verifier accepts. The same fact implies non-destructiveness, since the oracle response is constant on the coset supporting the honest prover's state in the queried basis.

    We prove extraction soundness. Let $\cal P^*$ be accepted with probability $\eta$. Let $q_i$ denote the probability that $\cal P^*$ passes checkpoint $i$, conditioned on passing checkpoints $0,\ldots,i-1$. Since $\prod_{i=0}^{\gamma}q_i=\eta$, there exists an adjacent pair of checkpoints $i^\star,i^\star+1$ such that
    \[
        q_{i^\star}q_{i^\star+1}\ge \eta^{2/\gamma}.
    \]
    Write $t_0=t+i^\star\delta$ and $t_1=t_0+\delta$. Conditioned on passing checkpoints $0,\ldots,i^\star-1$, the checkpoints $i^\star$ and $i^\star+1$ form an identical experiment to the two-round localization test from \Cref{con:non-destructive_entanglement_localization}. We therefore apply the same spacetime chunking reduction from the proof of \Cref{thm:NDEL_overall}, with the two target challenges at these checkpoints. As usual, the reduction $\cal{A}$ simulates the hidden random oracle $\cal{O}_{\sf random}$ and the challenge oracle $\cal{O}_\chal$ by a $2r(\secp)$-wise independent function $F_k$, where $r$ is the number of oracle queries $\cal{P}^*$ makes.
    
    Importantly, due to the repetition, the first-stage reduction component $\cal{A}^0_M$ must simulate the $\Localize$ experiment from the beginning of time up through checkpoint $i^\star-1$, thus ensuring that $\cal{A}^0_M$ is non-aborting iff the simulated prover $\cal{P}^*$ makes it to checkpoint $i^\star$. To do this, $\cal{A}^0_M$ samples challenges $(x_{L,i},x_{R,i})$ for $i<i^*$ on its own and simulates the verifiers' interactions with the prover as in $\Localize$. If the prover responds with $y_{L,i},y_{R,i}$ such that $y_{L,i}\neq y_{R,i}$ or if $y_{L,i}\neq F_k(\theta_i,x_{L,i}\oplus x_{R,i})$, $\cal{A}_M^0$ aborts. The reduction reprograms $\cal{O}_\chal$ at both $x_{L,i^\star}\oplus x_{R,i^\star}$ and $x_{L,i^\star+1}\oplus x_{R,i^\star+1}$, reprograms $\cal O$ at both $(\theta_{i^\star},x_{L,i^\star}\oplus x_{R,i^\star},\cdot)$ and $(\theta_{i^\star+1},x_{L,i^\star+1}\oplus x_{R,i^\star+1},\cdot)$; by \Cref{lem:oracle_reprogram}, this changes the experiment by only a negligible amount. Thus the resulting \MultiStageIndependentMonogamy adversary does not abort with probability at least $\eta$ and wins, conditioned on not aborting, with probability at least $q_{i^\star}q_{i^\star+1}-\negl(\secp)$.

    Then by \Cref{thm:extracting_EPR_pairs}, there exists an extractor $\cal E^{\cal A_M^1}_{i^\star}$ acting on the middle register $\mreg$ for checkpoints $i^\star,i^\star+1$, which with probability $\geq\eta$ (correspond to the case that the prover does not abort before $i^*$ checkpoint) outputs a state having the following \textbf{expected} fidelity with $\ket{\epr}^{\otimes n}$:
    \[
        1-O\Big((1-q_{i^\star}q_{i^\star+1})^{1/4}\Big)-\negl.
    \]
    Since $q_{i^\star}q_{i^\star+1}\ge \eta^{2/\gamma}$, this expected fidelity is at least
    \[
        1-O\!\left(\left(1-\eta^{2/\gamma}\right)^{1/4}\right)-\negl(\secp)
        =
        1-O\!\left(\left(\frac{\log(1/\eta)}{\gamma}\right)^{1/4}\right)-\negl(\secp),
    \]
    where the last step uses 
    \[
        1-\eta^{2/\gamma}
        =
        1-\exp\!\left(-\frac{2\log(1/\eta)}{\gamma}\right)
        \leq \frac{2\log(1/\eta)}{\gamma}.
    \]
    By Markov inequality, with probability $\geq0.99\eta$ the actual fidelity is 
    \[
        1-O\!\left(\left(\frac{\log(1/\eta)}{\gamma}\right)^{1/4}\right)-\negl(\secp)
    \]
    with a different constant in the big-$O$ notation.
    
    It remains to translate the middle-register extractor into the extraction statement at the original time $t$. Let ${\sf Exp}=\Localize(\cal P^*(\breg)\rightleftharpoons\cal V(\sparam))(\param,L,t)$ be the real localization experiment. The middle register $\mreg$ is the same one used in the warmup proof: it lies on the two lower boundaries of the green triangle in \Cref{fig:time-position-B} for the time $t_0$. These registers can be derived from $\register[L_{2\delta}~@~t_0]({\sf Exp})$ by forward simulation. Since $t_0-t\le \gamma\delta$, recovering this $2\delta$ neighborhood at time $t_0$ from time $t$ requires the region
    \[
        \register[L_{(\gamma+2)\delta}~@~t]({\sf Exp})=\register[L_\Delta~@~t]({\sf Exp}).
    \]
    The final extractor performs this forward simulation and then runs $\cal E^{\cal A_M^1}_{i^\star}$.
\end{proof}

\subsection{Trajectory Verification}\label{sec:trajectory}
\newcommand{\Verify}{\ensuremath{\mathsf{Localize}}\xspace}
\begin{definition}[Trajectory Verification]\label{def:trajectory_verification}
    A trajectory verification scheme is parameterized by the security parameter $\secp$ and a bipartite state $\sigma$\footnote{Technically, a sequence of states indexed by $\secp$ -- see \Cref{footnote:secp}.}. It consists of the following syntax:
    \begin{itemize}
        \item $\Setup(1^\secp)\to\param,\sparam,\ek,\areg,\breg$: The setup algorithm takes the security parameter $1^\secp$ and outputs the public parameter $\param$, the secret parameter $\sparam$, the extraction key $\ek$, and a state on the bipartite register $\areg\breg$.
        \item $\Verify\left(\cal{P}(\breg)\rightleftharpoons \cal{V}(\sparam)\right)(\param,L)\rightarrow \breg',b$: The interactive verification protocol takes place between a QPT prover holding a quantum register $\breg$ and a PPT verifier with input $\sparam$. The public inputs are $\param$ and the description of a continuous trajectory $L:[0,\tau]\to[-1,1]$ with $\abs{\frac{\mathrm{d}L}{\mathrm{d}t}}<1$. $L(t)$ is the prover's purported location at any time $0\leq t\leq \tau$, where $\tau$ is some time limit included in the description of $L$. The prover's output is a quantum register $\breg'$ and the verifier's output is an acceptance indicator $b\in\{\top,\bot\}$.
    \end{itemize}
    It should satisfy the following properties --- parameterized by completeness probability $\alpha(\secp)$, extraction fidelity $\beta(\secp,\eta)$, extraction probability $\beta'(\secp,\eta)$, and localization parameter $\Delta(\secp)$ --- for all $\secp\in\N$.
    \begin{itemize}
        \item $\alpha$-\textbf{Completeness}: For any valid trajectory $L$, there is a QPT prover $\cal{P}$ moving along $L$ such that
        \[\Pr\left[b = \top : \begin{array}{r} \param,\sparam,\ek,\areg,\breg\gets \Setup(1^\secp) \\ \breg',b \gets \Verify\brackets{\cal{P}(\breg)\rightleftharpoons \cal{V}(\sparam)}(\param,L)\end{array}\right] \geq\alpha(\secp).\]
        \item $(\beta,\beta',\Delta)$-\textbf{Extraction Soundness}:
        For any trajectory $L$ and QPT prover $\cal{P}^*$, let
        \[\eta(\secp):=\Pr\left[b = \top : \begin{array}{r} \param,\sparam,\ek,\areg,\breg\gets \Setup(1^\secp) \\ \breg',b \gets \Verify\brackets{\cal{P}^*(\breg)\rightleftharpoons \cal{V}(\sparam)}(\param,L)\end{array}\right]\]
        be $\cal P^*$'s success probability in $\Verify$. Assume $\eta(\secp)\ge 1/p(\secp)$ for some polynomial $p$. Then for any time $t\in[0,\tau]$, there exists a (potentially unbounded, see \Cref{remark:efficient-extractor}) local extractor $\cal{E}$ acting on registers near the claimed trajectory at time $t$ such that the probability of extracting $\sigma$ with fidelity at least $\beta(\secp,\eta)$ is at least $\beta'(\secp,\eta)$:
        \[
        \hspace{-3em}
        \prob{
            F\big(\sigma,\rho_{\areg\breg^*}\big)\ge \beta(\secp,\eta)
        :
        \begin{array}{r}
            \param,\sparam,\ek,\areg,\breg\gets \Setup(1^\secp)\\
            \breg^*\gets\register[(L(t))_\Delta~@~t]\brackets{\Verify\brackets{\cal{P}^*(\breg)\rightleftharpoons \cal{V}(\sparam)}(\param,L)}\\
            \breg^*\gets\cal{E}(\param,\ek,\breg^*)
        \end{array}
        }
        \geq \beta'(\secp,\eta)~,
        \]
        where  $\rho_{\areg\breg^*}$ is the joint state on $\areg\breg^*$ after running the extractor $\cal{E}$, and $(L(t))_\Delta=[L(t)-\Delta,L(t)+\Delta]$.
    \end{itemize}
\end{definition}
Now we present a construction for trajectory verification.
\begin{construction}[Trajectory Verification Scheme]\label{con:trajectory_verification}
    Let $\secp\in\mathbb{N}$ be the security parameter and $n(\secp),m(\secp)\geq\secp$ be polynomials. Take some localization parameter $\Delta(\secp)\leq 1$ and a repetition factor $\gamma(\secp)$. We consider the following trajectory verification scheme:
    \begin{itemize}
        \item There are two verifiers: $\cal{V}_L$ at $-1$ and $\cal{V}_R$ at $1$.
        \item $\Setup(1^\secp)$:
        \begin{enumerate}
            \item Generate $\brackets{\sk,\ket{\chain^\sk}_{\areg\breg}}\gets\genchain(1^\secp)$, and let $U_{\sf shift},S,T,u,v$ be the basis-matrix, subspaces, and shifts determined by $\sk$.
            \item Sample the public challenge oracle $\cal{O}_{\chal}:\mathbb{F}_2^\secp\rightarrow\set{0,1}$.
            \item Sample the public oracle $\cal{O}:\set{0,1}\times\mathbb{F}_2^\secp\times\mathbb{F}_2^{3n}\rightarrow\mathbb{F}_2^m\cup\set{\bot}$ is defined with respect to a random oracle $\cal{O}_{\sf random}:\{0,1\}\times\mathbb{F}_2^\secp\rightarrow \mathbb{F}_2^m$ that is \textbf{not} publicly available.
            \begin{equation*}
                \cal{O}(\theta,x,z)=\left\{
                \begin{aligned}
                    &\cal{O}_{\sf random}(\theta,x) \quad &\text{if }\theta=0\text{ and }z\in T+v,\\
                    &\cal{O}_{\sf random}(\theta,x) \quad &\text{if }\theta=1\text{ and }z\in S^\perp+u,\\
                    &\bot \quad &\text{otherwise}
                \end{aligned}
                \right.
            \end{equation*}
        \end{enumerate}
        Then $\Setup$ outputs:
        \begin{itemize}
            \item the public parameter $\param=(\cal{O}_\chal,\cal{O})$.
            \item The secret parameter $\sparam=(u,v)$
            \item The extraction key $\ek=(U_{\sf shift},u,v)$.
            \item The anchor state $\ket{\chain^\sk}_{\areg\breg}$, with $\anchor$ renamed $\areg$ and $\vessel$ renamed $\breg$.
        \end{itemize}
        \item $\Verify\left(\cal{P}(\breg)\rightleftharpoons \cal{V}(\sparam)\right)(\param,L)$: Suppose the claimed trajectory is $L(t)$ with $t\in[0,\tau]$. Let $\delta=\Delta/(2\gamma+2)$, let $N:=\lceil \tau/\delta\rceil+\gamma$, and define an extended trajectory\footnote{To enable extraction from all points along the trajectory, we must send a round of challenges \emph{after} the prover has already reached the endpoint $(L(\tau),\tau)$. We define the extension $\widetilde{L}$ for notational convenience when sending these post-trajectory challenges.} $\widetilde L(s):=L(\min\{s,\tau\})$.
        \begin{itemize}
            \item For $i=0,1,\cdots,N$, $\cal{V}_L$ samples $x_{L,i}\rand\mathbb{F}_2^\secp$ and broadcasts it at time $i\delta-(\widetilde L(i\delta)+1)$. It expects responses $y_{L,i}$ at time $i\delta+(\widetilde L(i\delta)+1)$.
            \item For $i=0,1,\cdots,N$, $\cal{V}_R$ samples $x_{R,i}\rand\mathbb{F}_2^\secp$ and broadcasts it at time $i\delta-(1-\widetilde L(i\delta))$. It expects responses $y_{R,i}$ at time $i\delta+(1-\widetilde L(i\delta))$.
            \item If any of these responses is missing, the verifier outputs $\bot$. Otherwise, the verifier outputs $\top$ if and only if for each $i=0,1,\cdots,N$, one of the following holds:
            \begin{itemize}
                \item $\theta_i:=\cal{O}_{\chal}(x_{L,i}\oplus x_{R,i})=0$ and $y_{L,i}=y_{R,i}=\cal{O}(0,x_{L,i}\oplus x_{R,i},v)$.
                \item $\theta_i:=\cal{O}_{\chal}(x_{L,i}\oplus x_{R,i})=1$ and $y_{L,i}=y_{R,i}=\cal{O}(1,x_{L,i}\oplus x_{R,i},u)$.
            \end{itemize}
            \item  The honest prover consists of a single party always holding the $\breg$ register while moving along $L$ and then remaining at $L(\tau)$ for the final $\gamma$ checkpoints. For $i=0,1,\cdots,N$, it does the following:
            \begin{itemize}
                \item Receive $x_{L,i},x_{R,i}$ at time $i\delta$. If $\theta_i=0$, it coherently evaluates $\cal{O}(\theta_i,x_{L,i}\oplus x_{R,i},\cdot)$ on the $\breg$ register in the standard basis and sends the result to $\cal{V}_L,\cal{V}_R$. If $\theta_i=1$, it coherently evaluates $\cal{O}(\theta_i,x_{L,i}\oplus x_{R,i},\cdot)$ on the $\breg$ register in the Hadamard basis and sends the result to $\cal{V}_L,\cal{V}_R$.
            \end{itemize}
        \end{itemize}
    \end{itemize}
\end{construction}
\begin{theorem}\label{thm:TV_overall}
    The construction in \Cref{con:trajectory_verification} is a trajectory verification scheme for the state $\sigma=\ketbra{\epr}{\epr}^{\otimes n}$ in the classical oracle model with $1$-\textbf{Completeness} and $\brackets{\beta,\beta',\Delta}$-\textbf{Extraction Soundness} where
    \[
        \beta(\secp,\eta)=1-O\!\left(\left(\frac{\log(1/\eta)}{\gamma}\right)^{1/4}\right)-\negl(\secp),
        \qquad
        \beta'(\secp,\eta)=0.99\eta(\secp).
    \]
    In particular, since we only consider adversaries for which $1/\eta$ is at most polynomial in $\secp$, for any polynomial $t(\secp)$, we can set $\gamma=c\secp t^4$ for some constant $c$ to achieve
    \[\beta(\secp,\eta)\geq1-1/t(\secp)~.\]
\end{theorem}

\begin{proof}
    Completeness is identical to the completeness argument for \Cref{thm:NDEL_overall}: by \Cref{fact:standard_hadamard_duality}, the honest prover's coherent query returns $\cal O_{\sf random}(\theta_i,x_{L,i}\oplus x_{R,i})$ in every checkpoint and the vessel state is not disturbed.

    We prove extraction soundness by treating the protocol as a chained execution of the simple two-round localization scheme of \Cref{con:non-destructive_entanglement_localization}. Take a trajectory $L$, a QPT prover $\cal{P}^*$ succeeding in $\Verify$ with probability $\eta$, and fix a time $t\in[0,\tau]$. Let $N$ and $\widetilde L$ be as in the construction. Let $q_i$ denote the probability that $\cal P^*$ passes checkpoint $i$, conditioned on passing checkpoints $0,\ldots,i-1$, so $\prod_{i=0}^N q_i=\eta$. Set $k:=\lceil t/\delta\rceil$ and consider the future window of adjacent checkpoint indices
    \[
        I_t:=\{k,k+1,\ldots,k+\gamma-1\}.
    \]
    Since $N=\lceil\tau/\delta\rceil+\gamma$, this window is contained in $\{0,\ldots,N-1\}$, and corresponds to the checkpoint times in $[t,t+\gamma\delta]$. Since all $q_i\leq 1$, we have $\prod_{r=k}^{k+\gamma}q_r\geq\eta$. Therefore, there exists $i^\star\in I_t$ such that, conditioned on passing checkpoints $0,\ldots,i^\star-1$, the probability of passing checkpoints $i^\star$ and $i^\star+1$ is at least $\eta^{2/\gamma}$.

    The reduction to \MultiStageIndependentMonogamy is the same reduction as in the proof of \Cref{thm:NDEL_overall}, with the following trajectory-specific changes. First, the two localization challenge points are now $(\widetilde L(i^\star\delta),i^\star\delta)$ and $(\widetilde L((i^\star+1)\delta),(i^\star+1)\delta)$ rather than two points on a stationary worldline. Thus the middle/left/right components are obtained by cutting the prover's spacetime computation along the corresponding light-cone boundaries. This moving-prover setting for the spacetime decomposition is pictured in \Cref{fig:trajectory-time-position-B}. Concretely, in the seven-component adversary $\cal A$ from \Cref{thm:NDEL_overall}, replace $t_0,t_1$ by $i^\star\delta,(i^\star+1)\delta$, replace the first target location $L$ by $\widetilde L(i^\star\delta)$, and replace the second target location $L$ by $\widetilde L((i^\star+1)\delta)$ in the splitting of $\mreg_L,\mreg_R,\lreg',\rreg'$ and the final response lines.

    \begin{figure}[h]
    \centering
    \begin{tikzpicture}[>=Latex]

    \def\z{4}
    \def\d{0.3*\z}
    \def\p{0.16*\z}
    \def\yb{-0.4*\z}

    \fill[yellow!25]
      (0,0) --
      (-\z,\z) --
      (-\z,\yb) --
      (\z,\yb) --
      (\z,\z) -- cycle;

    \coordinate (EL) at (-2*\d,0);
    \coordinate (ER) at ( 2*\d,0);
    \coordinate (EBL) at (-2*\d+\yb,\yb);
    \coordinate (EBR) at ( 2*\d-\yb,\yb);
    \fill[cyan!10]
      (EBL) --
      (EBR) --
      ( \d,\d) --
      (-\d,\d) -- cycle;
    \draw[dashed, blue!55] (EBL) -- (-\d,\d);
    \draw[dashed, blue!55] (EBR) -- ( \d,\d);
    \draw[densely dashed, gray!70] (-\d,0) -- (-\d,\d);
    \draw[densely dashed, gray!70] ( \d,0) -- ( \d,\d);
    \draw[blue!55]
      ($(EL)+(0,-0.07*\z)$) -- ($(EL)+(0,0.07*\z)$)
      ($(ER)+(0,-0.07*\z)$) -- ($(ER)+(0,0.07*\z)$)
      ($(-\d,0)+(0,-0.07*\z)$) -- ($(-\d,0)+(0,0.07*\z)$)
      ($( \d,0)+(0,-0.07*\z)$) -- ($( \d,0)+(0,0.07*\z)$);
    \node[below=5pt, font=\small, inner sep=0.5pt] at (EL) {$L-2\delta$};
    \node[below=5pt, font=\small, inner sep=0.5pt] at (-\d,0) {$L-\delta$};
    \node[below=5pt, font=\small, inner sep=0.5pt] at ( \d,0) {$L+\delta$};
    \node[below=5pt, font=\small, inner sep=0.5pt] at (ER) {$L+2\delta$};
    \node[align=center, font=\small, inner sep=1.5pt] (TowardExtractionLabel) at (-4.1*\d-0.85,2*\d) {backwards\\ light-cone of $\mreg$\\towards extraction\\ zone at target\\ time $t$};
    \draw[->, blue!55, shorten <=1pt, shorten >=1pt]
      (TowardExtractionLabel.east) -- (-2.5*\d,-0.18*\z);

    \fill[green!25]
      (0,0) --
      (-\d,\d) --
      (\d,\d) -- cycle;
    \node at (0,0.15*\z) {$\cal{A}_M^1$};

    \fill[red!25]
      (-\d,\d) --
      (-\z,\z) --
      (-\z,\z+\d+\p) --
      (\p,\d) -- cycle;
    \node at (-0.43*\z,0.60*\z) {$\cal{A}_L^0$};

    \fill[blue!25]
      (\d,\d) --
      (\z,\z) --
      (\z,\z+\d-\p) --
      (\p,\d) -- cycle;
    \node at (0.5*\z,0.55*\z) {$\cal{A}_R^0$};

    \draw[very thick, red!60] (\p,\d) -- (-\z,\z+\d+\p)
      node[midway,above,color=black] {$\cal{A}_L^1$};

    \draw[very thick, blue!60] (\p,\d) -- (\z,\z+\d-\p)
      node[midway,above=4pt,color=black] {$\cal{A}_R^1$};

    \draw[dashed] (\yb,\yb) -- (1.5*\z,1.5*\z)
      node[pos=0.8, right, font=\scriptsize] {$\substack{t-p=\\i^\star\delta-L}$};
    \draw[dashed] (-\yb,\yb) -- (-1.5*\z,1.5*\z)
      node[pos=0.8, left, font=\scriptsize] {$\substack{t+p=\\i^\star\delta+L}$};

    \draw[dotted,thick] (\yb-\d+\p,\yb) -- (1.5*\z,1.5*\z+\d-\p)
      node[pos=1.0, above, font=\scriptsize] {$\substack{t-p=\\(i^\star+1)\delta-L'}$};
    \draw[dotted,thick] (-\yb+\d+\p,\yb) -- (-1.5*\z,1.5*\z+\d+\p)
      node[pos=0.93, left, font=\scriptsize] {$\substack{t+p=\\(i^\star+1)\delta+L'}$};

    \draw[dotted, thick] (-1.5*\z,\d) -- (1.5*\z,\d) node[right] {$t=(i^\star+1)\delta$};
    \draw[line width=1.6pt, blue!65] (0,0) -- (-\d,\d) (0,0) -- (\d,\d);
    \node[align=center, inner sep=1.5pt] (MRegLabelTrajectory) at (0.08*\z,0.72*\z) {$\mreg$};
    \draw[->, blue!65, shorten <=1pt, shorten >=1pt]
      (MRegLabelTrajectory.south west) -- (-0.5*\d,0.5*\d);
    \fill[black] (\p,\d) circle (2.2pt)
      node[above=3pt] {$\cal{A}_M^2$};

    \draw[dotted, thick] (-1.5*\z,\z) -- (1.5*\z,\z) node[right] {$t=i^\star\delta+1$};

    \node at (0.86*\z,-0.27*\z) {$\cal{A}_M^0$};
    \draw[->] (-1.5*\z,0) -- (1.5*\z,0) node[right] {$t=i^\star\delta$};
    \draw[-] (0,\yb) node[below] {$p=L$} -- (0,0);
    \draw[-] (-\z,2*\z) -- (-\z,\yb) node[below] {$p=-1$};
    \draw[-] (\z,2*\z) -- (\z,\yb) node[below] {$p=1$};

    \end{tikzpicture}

    \caption{Spacetime decomposition of the reduction $\cal A$ for trajectory verification, where the prover is moving between checkpoints. In this picture, $L=0$ and $L'=0.16$.}
    \label{fig:trajectory-time-position-B}
    \end{figure}
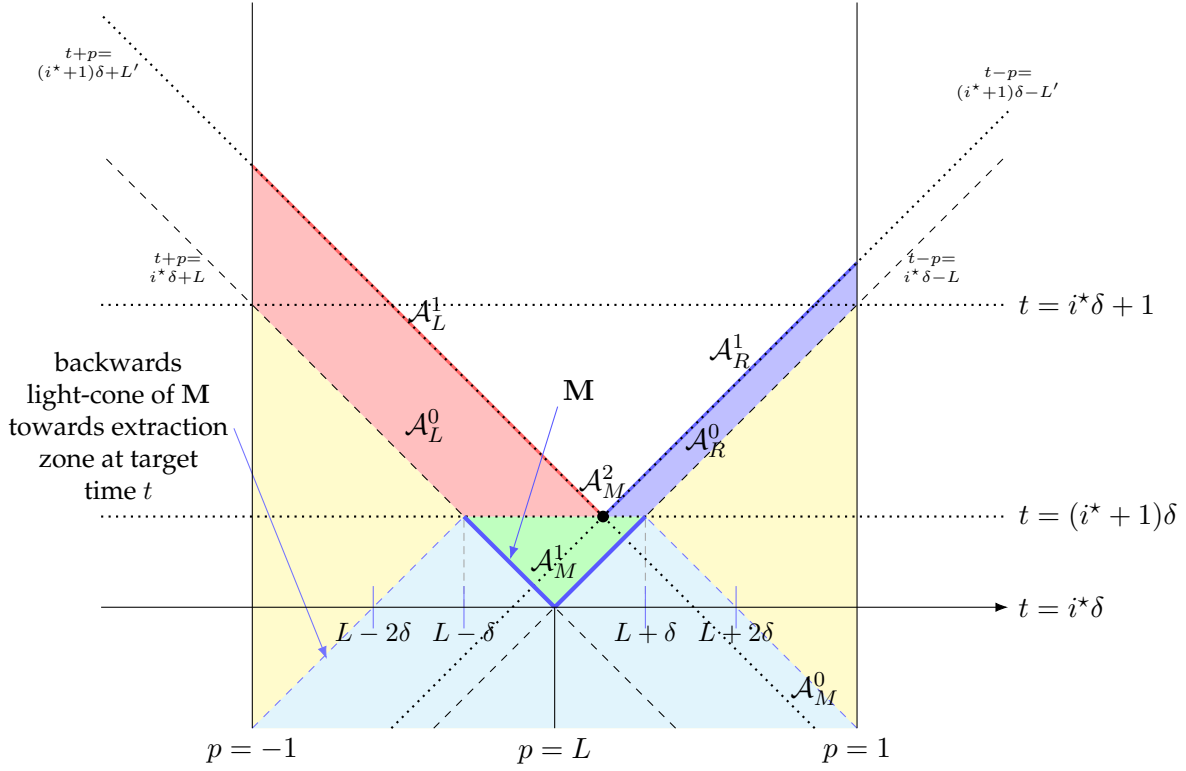

    As usual, the reduction simulates $\cal{O}_\chal$ and the hidden random oracle $\cal O_{\sf random}$ by a $2r(\secp)$-wise independent function $F_k$, where $r$ is the number of oracle queries $\cal P^*$ makes. Importantly, due to the repetition, the first-stage reduction component must simulate the $\Verify$ experiment from the beginning of time up through checkpoint $i^\star-1$, thus ensuring that it is non-aborting iff the simulated prover $\cal P^*$ makes it to checkpoint $i^\star$. To do this, it samples challenges $(x_{L,i},x_{R,i})$ for $i<i^*$ on its own and simulates the verifiers' interactions with the prover as in $\Verify$. If the prover responds with $y_{L,i},y_{R,i}$ such that $y_{L,i}\neq y_{R,i}$ or if $y_{L,i}\neq F_k(\theta_i,x_{L,i}\oplus x_{R,i})$, the reduction aborts. The oracle $\cal O$ is reprogrammed only at the two challenge points $(\theta_{i^\star},x_{L,i^\star}\oplus x_{R,i^\star})$ and $(\theta_{i^\star+1},x_{L,i^\star+1}\oplus x_{R,i^\star+1})$. The basis oracle $\cal{O}_\chal$ is also reprogrammed at $x_{L,i^\star}\oplus x_{R,i^\star}$ and $x_{L,i^\star+1}\oplus x_{R,i^\star+1}$. The position-verification reprogramming lemma (\Cref{lem:oracle_reprogram}) applies exactly as in the localization proof because the two additive shares of the challenge become jointly available only at the claimed spacetime point for each checkpoint.

    Hence, the constructed monogamy adversary does not abort with probability at least $\prod_{r<i^\star}q_r\geq\eta$, and conditioned on not aborting it wins with probability at least $\eta^{2/\gamma}-\negl(\secp)$. Applying the extractor from \Cref{thm:extracting_EPR_pairs}, exactly as in \Cref{thm:NDEL_overall}, there exists an extractor which with probability at least $0.99\eta$ outputs a state having the following fidelity with $\ket{\epr}^{\otimes n}$:
    \[
        1-O\!\left(\left(1-\eta^{2/\gamma}\right)^{1/4}\right)-\negl(\secp)
        =
        1-O\!\left(\left(\frac{\log(1/\eta)}{\gamma}\right)^{1/4}\right)-\negl(\secp).
    \]
    It remains only to translate the checkpoint extractor into the extraction statement at the target time. The input to the monogamy extractor is the middle register on the two lower boundaries of the green triangle in \Cref{fig:trajectory-time-position-B}; this register can be recovered by forward simulation from $\register[\widetilde L(i^\star\delta)_{2\delta}~@~i^\star\delta]({\sf Exp})$, where ${\sf Exp}=\Verify(\cal P^*(\breg)\rightleftharpoons\cal V(\sparam))(\param,L)$. By construction, $0\leq i^\star\delta-t\leq\gamma\delta$. Since the valid trajectory has speed at most one, the backwards light cone of $\register[\widetilde L(i^\star\delta)_{2\delta}~@~i^\star\delta]$ at time $t$ is contained in $\register[L(t)_{(2\gamma+2)\delta}~@~t]=\register[L(t)_\Delta~@~t]$. Thus the final extractor acts on $\register[L(t)_\Delta~@~t]$, forward-simulates to the appropriate future checkpoint neighborhood, and then runs the monogamy extractor. This gives $(\beta,\beta',\Delta)$-soundness with
    \[
        \beta(\secp,\eta)=1-O\!\left(\left(\frac{\log(1/\eta)}{\gamma}\right)^{1/4}\right)-\negl(\secp),
        \qquad
        \beta'(\secp,\eta)=0.99\eta(\secp).
    \]
\end{proof}

\subsection{State Localization}\label{subsec:state}
Note that in the entanglement localization scheme, if we ignore register $\areg$, then we are localizing the reduced state on register $\breg$. This idea extends to the concept of state localization.
\begin{definition}[State Localization]\label{def:state-localization}
    A state localization scheme is parameterized by the security parameter $\secp$ and a state family $\{\ket{\psi_i}\}_{i \in [N]}$.\footnote{Technically, we have one state family per security parameter $\secp$ -- see \Cref{footnote:secp}} It consists of the following syntax.
    \begin{itemize}
        \item $\Setup(1^\secp) \to \param,\sparam,\ek,i,\breg$: The Setup algorithm takes the security parameter $1^\secp$ and outputs the public parameter $\param$, the secret parameter $\sparam$, the extraction key $\ek$, an index $i$, and a state on register $\breg$.
        \item $\Localize\left(\cal{P}(\breg) \rightleftharpoons \cal{V}(\sparam)\right)(\param,L,t) \to \breg',b$: This is a positional protocol between a QPT prover holding a quantum register $\breg$ and a PPT verifier with input $\sparam$,  further specified by the following public inputs: public parameters $\param$, spatial location $L \in [-1,1]$, and time $t$. The prover's output is $\rho'$ and the verifier's output is an acceptance indicator $b\in\{\top,\bot\}$.
    \end{itemize}

    It should satisfy the following properties --- parameterized by completeness probability $\alpha(\secp)$, extraction fidelity $\beta(\secp,\eta)$, extraction probability $\beta'(\secp,\eta)$, uniqueness error $\zeta(\secp)$, and localization parameter $\Delta(\secp)$ --- for all $\secp\in\N$.
    \begin{itemize}
        \item \textbf{$\alpha$-Completeness}: For any location $L$ and time $t$, there is a QPT prover $\cal P$ at $(L,t)$ such that
        \[\Pr\left[b = \top : \begin{array}{r} \param,\sparam,\ek,i,\breg \gets \Setup(1^\secp) \\ (\breg',b) \gets \Localize\big( \cal{P}(\breg) \rightleftharpoons \cal{V}(\sparam)\big)(\param,L,t)\end{array}\right] \geq \alpha(\secp).\]
        \item \textbf{$(\beta,\beta',\Delta)$-Extraction Soundness}: For any QPT prover $\cal{P}^*$, spatial location $L$, and time $t$, let
        \[\eta(\secp) := \Pr\left[b = \top : \begin{array}{r} \param,\sparam,\ek,i,\breg \gets \Setup(1^\secp) \\ (\breg',b) \gets \Localize\big( \cal{P}^*(\breg) \rightleftharpoons \cal{V}(\sparam)\big)(\param,L,t)\end{array}\right]\]
        be $\cal P^*$'s success probability. Assume $\eta(\secp)\ge 1/p(\secp)$ for some polynomial $p$. Then there exists a (potentially unbounded, see \Cref{remark:efficient-extractor}) local extractor\footnote{It is crucial that this extractor does not know $i$.} $\cE$ acting on registers near location $L$ at time $t$ such that the probability of extracting $\ket{\psi_i}$ with fidelity at least $\beta(\secp,\eta)$ is at least $\beta'(\secp,\eta)$:
        \[
        \hspace{-4em}
        \prob{
            \bra{\psi_i}\rho_{\breg^*}\ket{\psi_i}\ge \beta(\secp,\eta)
        :
        \begin{array}{r}
            \param,\sparam,\ek,i,\breg\gets \Setup(1^\secp)\\
            \breg^*\gets\register[L_\Delta~@~t]\brackets{\Localize\brackets{\cal{P}^*(\breg)\rightleftharpoons \cal{V}(\sparam)}(\param,L,t)}\\
            \breg^*\gets\cE(\param,\ek,\breg^*)
        \end{array}
        }
        \geq \beta'(\secp,\eta)~,
        \]
        where $\rho_{\breg^*}$ is the state on $\breg^*$ after running the extractor, and $L_\Delta := [L-\Delta(\secp),L+\Delta(\secp)]$.

        \item \textbf{$\zeta$-Uniqueness}: For any (potentially unbounded) cloner $\cal{A}$, 

        \[\hspace{-4em}\E\left[\bra{\psi_i}\bra{\psi_i}\rho_{\breg^*_0\breg^*_1}\ket{\psi_i}\ket{\psi_i}: \begin{array}{r}\param,\sparam,\ek,i,\breg \gets \Setup(1^\secp) \\ \breg^*_0,\breg^*_1 \gets \cal{A}\brackets{\param,\sparam,\ek,\breg} \end{array}\right] \leq \zeta(\secp)~,\]
        where $\rho_{\breg^*_0\breg^*_1}$ is the state on $\breg^*_0\breg^*_1$ after running $\cal{A}$.
    \end{itemize}

    We say that the state localization scheme is \emph{non-destructive} if it additionally satisfies the following property.
    \begin{itemize}
        \item $\textbf{Non-destructive}$: For any spatial location $L$ and time $t$,
        \[
        \expect{\TraceDist{\rho_{\breg}}{\rho_{\breg'}}:\begin{array}{r}\param,\sparam,\ek,i,\breg\gets \Setup(1^\secp)\\ \breg',b \gets \Localize\brackets{\cal{P}(\breg)\rightleftharpoons \cal{V}(\sparam)}(\param,L,t)\end{array}} \leq \negl(\secp)~,
        \]
        where $\rho_{\breg}$ is the state on $\breg$ after $\Setup$ and $\rho_{\breg'}$ is the state on $\breg'$ later after $\Localize$.
    \end{itemize}

\end{definition}

\begin{construction}[Non-Destructive State Localization Scheme]\label{con:non-destructive_state_localization}
    Let $\secp\in\mathbb{N}$ be the security parameter and $n(\secp),m(\secp)\geq\secp$ be polynomials. Take some localization parameter $\Delta(\secp)\leq 1$ and a repetition factor $\gamma(\secp)$. Take a state family $\{\ket{\psi_i}\}_{i \in [N]}$. We consider the following non-destructive state localization scheme:
    \begin{itemize}
        \item There are two verifiers: $\cal{V}_L$ at $-1$ and $\cal{V}_R$ at $1$.
        \item $\Setup(1^\secp)$:
        \begin{enumerate}
            \item Generate $\brackets{\sk,\ket{\chain^\sk}_{\areg\breg}}\gets\genchain(1^\secp)$, and let $U_{\sf shift},S,T,u,v$ be the basis-matrix, subspaces, and shifts determined by $\sk$.
            \item Sample $i \gets [N]$ and define the isometry \[\Enc_{U_{\sf shift},u,v}: \ket{x} \to \frac{1}{\sqrt{2^n}}\sum_{s \in S}(-1)^{s \cdot u}\ket{s + U_{\sf shift}(0^n \times x \times 0^n) + v},\] and compute \[\rho = \Enc_{U_{\sf shift},u,v}\ketbra{\psi_i}{\psi_i}\Enc_{U_{\sf shift},u,v}^\dagger.\]
            \item Sample the public challenge oracle $\cal{O}_{\chal}:\mathbb{F}_2^\secp\rightarrow\set{0,1}$.
            \item Sample the public oracle $\cal{O}:\{0,1\}\times\mathbb{F}_2^\secp \times\mathbb{F}_2^{3n}\rightarrow \mathbb{F}_2^m\cup\set{\bot}$ based on random oracle $\cal{O}_{\sf random}:\{0,1\}\times\mathbb{F}_2^\secp\rightarrow \mathbb{F}_2^m$ that is \textbf{not} publicly available.
            \begin{equation*}
                \cal{O}(\theta,x,z)=\left\{
                \begin{aligned}
                    &\cal{O}_{\sf random}(\theta,x) \quad &\text{if }\theta=0\text{ and }z\in T+v,\\
                    &\cal{O}_{\sf random}(\theta,x) \quad &\text{if }\theta=1\text{ and }z\in S^\perp+u,\\
                    &\bot \quad &\text{otherwise}
                \end{aligned}
                \right.
            \end{equation*}
        \end{enumerate}
        Then $\Setup$ outputs:
        \begin{itemize}
            \item the public parameter $\param=(\cal{O}_\chal,\cal{O})$.
            \item The secret parameter $\sparam=(u,v)$
            \item The extraction key $\ek=(U_{\sf shift},u,v)$.
            \item The sampled state and the index $(i,\rho)$.
        \end{itemize}
        \item $\Localize\left(\cal{P}(\breg) \rightleftharpoons \cal{V}(\sparam)\right)(\param,L,t)$: Let $\delta=\Delta/(\gamma+2)$ and define checkpoints $t_i=t+i\delta$ for $i=0,1,\ldots,\gamma$.
        \begin{itemize}
            \item For $i=0,1,\cdots,\gamma$, $\cal{V}_L$ samples $x_{L,i}\rand\mathbb{F}_2^\secp$ and broadcasts it at time $t_i-(L+1)$. It expects responses $y_{L,i}$ at time $t_i+(L+1)$.
            \item For $i=0,1,\cdots,\gamma$, $\cal{V}_R$ samples $x_{R,i}\rand\mathbb{F}_2^\secp$ and broadcasts it at time $t_i-(1-L)$. It expects responses $y_{R,i}$ at time $t_i+(1-L)$.
            \item If any of these responses is missing, the verifier outputs $\bot$. Otherwise, the verifier outputs $\top$ if and only if for each $i=0,1,\cdots,\gamma$, one of the following holds:
            \begin{itemize}
                \item $\theta_i:=\cal{O}_{\chal}(x_{L,i}\oplus x_{R,i})=0$ and $y_{L,i}=y_{R,i}=\cal{O}(0,x_{L,i}\oplus x_{R,i},v)$.
                \item $\theta_i:=\cal{O}_{\chal}(x_{L,i}\oplus x_{R,i})=1$ and $y_{L,i}=y_{R,i}=\cal{O}(1,x_{L,i}\oplus x_{R,i},u)$.
            \end{itemize}
            \item The honest prover consists of a single party always holding the $\breg$ register at location $L$. For each $i=0,1,\ldots,\gamma$, it receives $x_{L,i},x_{R,i}$ at time $t_i$. If $\theta_i=0$, it coherently evaluates $\cal{O}(0,x_{L,i}\oplus x_{R,i},\cdot)$ on the $\breg$ register in the standard basis and broadcasts the result. If $\theta_i=1$, it coherently evaluates $\cal{O}(1,x_{L,i}\oplus x_{R,i},\cdot)$ on the $\breg$ register in the Hadamard basis and broadcasts the result.
        \end{itemize}
    \end{itemize}
\end{construction}

\begin{theorem}\label{thm:state-localization}
    Let $\cal{S} = \{\ket{\psi_i}\}_{i \in [N]}$ be any set of $n$-qubit states that consist of a union of orthonormal bases. That is $\cal{S} = \bigcup_j \cal{B}_j$, where each $\cal{B}_j$ consists of $2^n$ orthonormal states.

   Then the construction in \Cref{con:non-destructive_state_localization} is a non-destructive state localization scheme for $\cal{S}$ in the classical oracle model with $\alpha(\secp) = 1$, $\brackets{\beta,\beta',\Delta}$-\textbf{Extraction Soundness} for
   \[
        \beta(\secp,\eta)=1-O\!\left(\left(\frac{\log(1/\eta)}{\gamma}\right)^{1/4}\right)-\negl(\secp),
        \qquad
        \beta'(\secp,\eta)=0.98\eta(\secp),
   \]
   and uniqueness parameter
   \[
        \zeta(\secp) := \sup_{\cal{A}}\left\{\E\left[\bra{\psi_i}\bra{\psi_i}\rho\ket{\psi_i}\ket{\psi_i} : \begin{array}{r} \ket{\psi_i} \gets \cal{S} \\ \rho \gets \cal{A}(\ket{\psi_i})\end{array}\right]\right\}.
   \]
   In particular, since we only require extraction soundness against adversaries for which $1/\eta$ is at most polynomial in $\secp$, for any polynomial $t(\secp)$, we can set $\gamma=c\secp t^4$ for some constant $c$ to achieve
    \[\beta(\secp,\eta)\geq1-1/t(\secp)~.\]
\end{theorem}

Before coming to the proof, we give two examples of state families that are covered by our theorem. Note that we can write any $\cal{S}$ from the above theorem statement as \[\cal{S} = \bigcup_{U \in \cal{U}}\left\{U\ket{x}\right\}_{x \in \{0,1\}^n}\] for some set of unitaries $\cal{U}$.

\begin{itemize}
    \item \textbf{BB84 states}: $\cal{U}$ is the set of $n$-qubit unitaries $H^{\theta_1} \otimes \dots \otimes H^{\theta_n}$ for each $\theta \in \{0,1\}^n$.
    \item \textbf{Coset states}: $\cal{U}$ is the set of unitaries that first map $\ket{x} \to \ket{Bx}$ for some full rank $B \in \mathbb{F}_2^{n \times n}$ and then apply $H$ to the final $n/2$ qubits.
\end{itemize}
\begin{proof}
    $1$-completeness and non-destructiveness follow by observation, and $\zeta$-uniqueness follows directly from the definition of $\zeta$ given in the theorem statement.

    Thus, it remains to argue extraction soundness. To do so, we use exactly the same extractor $\cE$ as in the proof of \Cref{thm:entanglement_localization}. Now, note that our extraction experiment is equivalent to the following purified version.

    \begin{itemize}
        \item In $\Setup$, rather sampling $i \gets [N]$ and applying $\Enc_{U_{\sf shift},u,v}$ to $\ket{\psi_i}$, instead prepare \[\ket{\mathsf{EPR}}^{\otimes n} = \frac{1}{\sqrt{2^n}}\sum_{x \in \{0,1\}^n}\ket{x}_\areg\ket{x}_\breg,\] and let $\rho_\breg$ be the result of applying $\Enc_{U_{\sf shift},u,v}$ to register $\breg$.
        \item After $\cE$ outputs $\rho_{\breg^*}$, sample $U \gets \cal{U}$ (where $\cal{U}$ is defined based on $\cal{S}$ as above), apply $U^t$ to register $\areg$, and then measure $\areg$ in the standard basis to obtain $x$. Let $\ket{\psi_{i}} := \ket{\psi_{U,x}} = U\ket{x}$ be the state\footnote{In this purified version, the family of states is indexed by $U,x$. There is a one-to-one correspondence between $U,x$ and, $i$ which we used in the original definition.} that is used in the expression $\bra{\psi_i}\rho_{\breg^*}\ket{\psi_i}$.
    \end{itemize}

    For any choice of $U$, the expected success probability of extracting the state is
    \begin{align*}
    \E_x\left[\bra{\psi_{U,x}}\rho_{\breg^*}\ket{\psi_{U,x}}\right] &= \Tr\left[\left(\sum_{x}\ketbra{x}{x}_\areg \otimes \ketbra{\psi_{U,x}}{\psi_{U,x}}_{\breg^*}\right)U^t_\areg\rho_{\areg\breg^*}U^*_\areg\right] \\
        &= \Tr\left[\left(\sum_{x}U^*\ketbra{x}{x}U^t_\areg \otimes U\ketbra{x}{x}U^\dagger_{\breg^*}\right)\rho_{\areg\breg^*}\right] \\
        &= \Tr\left[\left(\left(U^* \otimes U\right)\left(\sum_{x}\ketbra{x}{x} \otimes \ketbra{x}{x}\right)\left(U^t \otimes U^\dagger\right)\right)\rho_{\areg\breg^*}\right] \\
        &\geq \Tr\left[\left(\left(U^* \otimes U\right)\left(\ketbra{\mathsf{EPR}}{\mathsf{EPR}}^{\otimes n}\right)\left(U^t \otimes U^\dagger\right)\right)\rho_{\areg\breg^*}\right] \\
        &= \Tr\left[\left(\ketbra{\mathsf{EPR}}{\mathsf{EPR}}^{\otimes n}\right)\rho_{\areg\breg^*}\right] \\
        &\geq \beta(\secp,\eta),
    \end{align*}

    where the last inequality holds with probability at least $0.99\eta(\secp)$ by \Cref{thm:entanglement_localization}. By an additional Markov inequality on expectation over $x$ we get $\beta'(\secp,\eta)=0.98\eta(\secp)$. This gives the claimed extraction soundness parameters.
\end{proof}

\begin{remark}
    One could strengthen the completeness definition of State Localization to demand that $\rho = \ketbra{\psi_i}{\psi_i}$. That is, the prover obtains the state $\ket{\psi_i}$ itself rather than an encoded version of it.

    We can achieve this stronger completeness guarantee for the set of $3n$-qubit coset states with subspaces of dimension $3n/2$. Indeed, one can sample a uniformly random $3n$-qubit coset state by first sampling an $n$-qubit coset state $\ket{\psi}$ with subspace of dimension $n/2$, and then outputting the $3n$-qubit state $\Enc_{U_{\sf shift},u,v}\ket{\psi}$, where $U_{\sf shift},u,v$ are sampled by the $\Setup$ algorithm given above. The resulting $3n$-qubit state is still unclonable even given $U_{\sf shift},u,v$, and thus uniqueness continues to hold.
\end{remark}

\subsection{Functionality Localization}\label{subsec:functionality}
We first recall the notion of quantum copy-protection from \cite{10.1109/CCC.2009.42}. Our definition of functionality localization, stated later, will be a strengthening of copy-protection.
\newcommand{\simver}{\ensuremath{{\sf SimVer}}\xspace}
\begin{definition}[Copy-protectable functionality]
    A family of functions $\cF = \{f\}_f$ is copy-protectable with security $\zeta = \zeta(\secp)$ if there exists a distribution $\cD$ on its domain and  algorithms
    \begin{itemize}
        \item $\Protect(1^\secp,f) \to \rho_f$
        \item $\Eval(\rho_f,x) \to y$
    \end{itemize}
    such that the following properties hold.
    \begin{itemize}
        \item \textbf{Correctness}: For any $f \in \cF$ and $x \in \cal{X}$, \[\Pr[\Eval(\rho_f,x) = f(x) : \rho_f \gets \Protect(1^\secp,f)] = 1-\negl(\secp).\]
        \item \textbf{Security}: For any QPT tri-partite adversary $(\cA,\cB,\cC)$,
        \[\Pr\left[y_B = f(x_B) \wedge y_C = f(x_C) : \begin{array}{r} f \gets \cF \\ \rho_f \gets \Protect(1^\secp,f)\\ \rho_{B,C} \gets \cA(\rho_f) \\ x_B,x_C \gets \cD \\ y_B \gets \cB(x_B,B), y_C \gets \cC(x_C,C)\end{array}\right] \leq \zeta(\secp).\]
    \end{itemize}
\end{definition}
\newcommand{\QObf}{\ensuremath{{\sf QObf}}\xspace}
\newcommand{\QEval}{\ensuremath{{\sf QEval}}\xspace}
We also recall quantum state obfuscation, which will be used later on for our construction of functionality localization. The following definition is adapted from \cite{BBV24}, but for simplicity we have significantly weakened the definition to work only for classical functionalities.
\begin{definition}[Quantum State Obfuscation (\cite{CG24},\cite{BBV24})]\label{def:obf}
    Let $\cal{F}=\{f\}_f$ be a classical function family, where for each $f\in\mathcal{F}$ there exists a quantum program $(C_f,\ket{\psi_f})$ such that $C_f(x,\ket{\psi}_f)=f(x)$ deterministically.
    A quantum state obfuscator for $\mathcal{F}$ is a pair of QPT algorithms $(\QObf,\QEval)$ with the following syntax.
    \begin{itemize}
        \item $\QObf\left(1^\secp,\ket{\psi},C\right) \to \ket{\widetilde{\psi}}$: The obfuscator takes as input the security parameter $1^\secp$ and a quantum program $(C,\ket{\psi})$, and outputs an obfuscated state $\ket{\widetilde{\psi}}$.
        \item $\QEval\left(x,\ket{\widetilde{\psi}}\right) \to y$: The evaluation algorithm takes an input $x \in \{0,1\}^{m(\secp)}$ and an obfuscated state $\ket{\widetilde{\psi}}$, and outputs $y \in \{0,1\}^{m'(\secp)}$.
    \end{itemize}
    \textbf{Correctness} is defined as follows for any $f\in\mathcal{F}$.
    \[\forall x \in \{0,1\}^{m(\secp)}, \Pr\left[\QEval\left(x,\ket{\widetilde{\psi}}\right) = C_f(x,\ket{\psi_f} ): \ket{\widetilde{\psi}} \gets \QObf\left(1^\secp,\ket{\psi_f},C_f\right)\right] = 1-\negl(\secp).\]

    \textbf{Ideal Obfuscation}: For any QPT adversary $\cA$, there exists a QPT simulator $\Sim$ such that for any $f\in\mathcal{F}$ and QPT distinguisher $\cD$,

    \begin{align*}&\bigg|\Pr\left[1 \gets \cD\left(\cA\left(\QObf\left(1^\secp,\ket{\psi_f},C_f\right)\right)\right)\right]\\ & ~~~~~ - \Pr\left[1 \gets \cD\left(\Sim^f\left(1^\secp\right)\right)\right]\bigg| = \negl(\secp).\end{align*}
\end{definition}
\begin{theorem}[Quantum State Ideal Obfuscation in the Oracle Model, \cite{BBV24}]
    There exists a quantum state ideal obfuscator for any classical functionality, in the classical oracle model.
\end{theorem}
Now we state the main definition of this section, and present our construction of it from quantum state obfuscation and signatures.
\begin{definition}[Functionality localization]\label{def:functionality_localization_revised}
    A functionality localization scheme is parameterized by a function family $\cF = \{f:\cal{X}\to\cal{Y}\}_f$\footnote{Technically, we have one function family per security parameter $\secp$ -- see \Cref{footnote:secp}}, and a distribution $\cD$ on its domain. It consists of the following syntax.
    \begin{itemize}
        \item $\Setup(1^\secp) \to \param,\sparam,\ek,f,\breg$: The setup algorithm takes the security parameter $1^\secp$ and outputs the public parameter $\param$, the secret parameter $\sparam$, the extraction key $\ek$, a function $f\in\cal{F}$, and a state on register $\breg$.
        \item $\Localize\big( \cal{P}(\breg) \rightleftharpoons \cal{V}(\sparam,f)\big)(\param,L,t) \to \breg',b$: The Localize protocol is an interaction between a QPT prover whose input is register $\breg$, and a PPT verifier with inputs $\sparam$ and $f$, further specified by the following public inputs: public parameters $\param$, spatial location $L \in [-1,1]$, and time $t$. The prover's output is $\breg'$ and the verifier's output is an acceptance indicator $b\in\{\top,\bot\}$.
    \end{itemize}
    It should satisfy the following properties --- parameterized by completeness probability $\alpha(\secp)$, extraction success $\beta(\secp,\eta)$, uniqueness error $\zeta(\secp)$, and localization parameter $\Delta(\secp)$ --- for all $\secp\in\N$.
    \begin{itemize}
        \item \textbf{Functionality Preservation}: There exists a procedure $\Eval$ such that for any $x\in\cal{X}$,
		\[
			\Pr\left[\Eval(\breg,x)=f(x) : \param,\sparam,\ek,f,\breg \gets \Setup(1^\secp)\right]=1-\negl(\secp)
		\]
        \item $\alpha$-\textbf{Completeness}: For any location $L$ and time $t$, there is a QPT prover $\cal P$ at $(L,t)$ such that
        \[
            \Pr\left[b = \top : \begin{array}{r}
                \param,\sparam,\ek,f,\breg\gets\Setup(1^\secp)\\
                (\breg',b) \gets \Localize\big( \cal{P}(\breg) \rightleftharpoons \cal{V}(\sparam,f)\big)(\param,L,t)
            \end{array}\right] \geq \alpha(\secp).
        \]
        \item \textbf{$(\beta,\Delta)$-Extraction Soundness}: For any QPT prover $\cal{P}^*$, spatial location $L$, and time $t$, let
        \[\eta(\secp) := \Pr\left[b = \top : \begin{array}{r} \param,\sparam,\ek,f,\breg \gets \Setup(1^\secp) \\ (\breg',b) \gets \Localize\big( \cal{P}^*(\breg) \rightleftharpoons \cal{V}(\sparam,f)\big)(\param,L,t)\end{array}\right]\]
        be $\cal P^*$'s success probability. Assume $\eta(\secp)\ge 1/p(\secp)$ for some polynomial $p$, and let spatial region $L_\Delta := [L-\Delta(\secp),L+\Delta(\secp)]$. Then there exists a QPT extractor $\cE$ acting on registers near location $L$ at time $t$ such that
        \[\hspace{-2em}
        \Pr\left[y = f(x) : \begin{array}{r} \param,\sparam,\ek,f,\breg \gets \Setup(1^\secp) \\ \breg^* \gets \register[L_\Delta~@~t]\left(\Localize(\cal{P}^*(\breg) \rightleftharpoons \cal{V}(\sparam,f)\right)(\param,L,t)) \\ x \gets \cD \\ y \gets \cE(\param,\ek,x,\breg^*)\end{array}\right] \geq \beta(\secp,\eta).
        \]
        \item \textbf{$\zeta$-Uniqueness}: For any QPT cloners $(\cal{A},\cal{B},\cal{C})$,

        \[\hspace{-3em}\Pr\left[\begin{array}{cc}
             y_0 = f(x_0), \\
             y_1 = f(x_1)
        \end{array}: \begin{array}{r}\param,\sparam,\ek,f,\breg \gets \Setup(1^\secp) \\ \breg^*_0,\breg^*_1 \gets \cal{A}(\param,\sparam,\ek,\breg) \\ x_0 \gets \cD, x_1 \gets \cD \\ y_0 \gets \cal{B}(x_0,\breg^*_0), y_1 \gets \cal{C}(x_1,\breg^*_1)\end{array}\right] \leq \zeta(\secp).\]
    \end{itemize}

    We say that the functionality localization scheme is \emph{non-destructive} if it additionally satisfies the following property.
    \begin{itemize}
        \item $\textbf{Non-destructive}$: For any spatial location $L$ and time $t$, \[\E\left[\TD\left(\rho_\breg,\rho_{\breg'}\right) : \begin{array}{r} \param,\sparam,\ek,f,\breg \gets \Setup(1^\secp) \\ (\breg',b) \gets \Localize\big( \cal{P}(\breg) \rightleftharpoons \cal{V}(\sparam,f)\big)(\param,L,t)\end{array}\right] \leq \negl(\secp).\]
    \end{itemize}
    where $\rho_{\breg}$ is the state on $\breg$ after $\Setup$ and $\rho_{\breg'}$ is the state on $\breg'$ later after $\Localize$.
\end{definition}

\begin{remark}
    Note that in contrast to entanglement and state localization, here we require the extractor $\cE$ to be efficient (QPT). This is because the uniqueness property may only hold against computationally-bounded cloners.
\end{remark}

\begin{construction}[Non-Destructive Functionality Localization Scheme]\label{con:function_localization_from_qsio}
    Let $\cal{F}$ be a function family. Let $n(\secp),m(\secp)\ge \secp$ be polynomials. Take a localization parameter $\Delta(\secp)\leq 1$ and a repetition factor $\gamma(\secp)$. Let $\QSIO$ be a quantum state obfuscator, and let $(\keygen,\Sign,\Ver)$ be a digital signature scheme, with message space $\{0,1\}^*$ and signature space $\cal{S}$. We define the following functionality localization scheme.
    \begin{itemize}
        \item $\Setup(1^\secp)$:
        \begin{enumerate}
            \item Sample $f\gets\cF$.
            \item Sample a uniformly random subspace $S\le \F_2^{n}$ of dimension $n/2$, sample $v\rand \CS(S)$ and $u\rand \CS(S^\perp)$, and prepare the coset state $X^v Z^u\ket{S}$.
            \item Sample $\sk,\vk\gets\keygen(1^\secp)$.
            \item Define the signed functionality
            \begin{equation}
                \widehat f_{\sk}(x):=\big(f(x),\Sign(\sk,(x,f(x)))\big),
            \end{equation}
            and obfuscate it using the ideal obfuscator:
            \[
                \rho_{\QSIO}\gets \QSIO.\QObf(1^\secp,\widehat f_{\sk}).
            \]
            \item Sample a public random oracle $\cO_{\sf chal}:\F_2^m\to\F_2^m$.
            \item Let $C_0:=S+v$ and $C_1:=S^\perp+u$. Sample a random oracle $\cO_{\sf random}:\F_2^{m}\times\mathcal X_\secp\times \mathcal Y_\secp\times\F_2\to \F_2^{m}$ and define the public oracle $\cal{O}_{S,T,u,v,\vk}:\F_2^m\times\cal{X}\times\cal{Y}\times\cal{S}\times\{0,1\}\times\F_2^n\to\F_2^m\cup\{\bot\}$ as follows.
            \[
                \cO_{S,T,u,v,\vk}(h,x,y,\sigma,c,z)=
                \begin{cases}
                    \cO_{\sf random}(h,x,y,c) & \text{if }\Ver(\vk,(x,y),\sigma)=\top\text{ and }z\in C_c,\\
                    \bot & \text{otherwise}
                \end{cases}
            \]
            Note that $\cO_{\sf random}$ is \textbf{not} made public.
        \end{enumerate}
        Then $\Setup$ outputs:
        \begin{itemize}
            \item The public parameter $\param=(\cO,\cO_{\sf chal},\vk)$.
            \item The secret parameter $\sparam=(u,v,\sk)$.
            \item The extraction key $\ek=\emptyset$.
            \item The functionality $f$.
            \item The quantum state $\rho=(\rho_{\QSIO},X^vZ^u\ket{S})$.
        \end{itemize}
        \item $\Localize\left(\cal{P}(\breg)\rightleftharpoons \cal{V}(\sparam,f)\right)(\param,L,t)$:
        Let $\delta=\Delta/(2\gamma+2)$ and define checkpoint times $t_{2i}=t+2i\delta$ and $t_{2i+1}=t+(2i+1)\delta$ for $i=0,1,\ldots,\gamma-1$.
        \begin{enumerate}
            \item For each $i\in\{0,\ldots,\gamma-1\}$, the verifiers sample $x_i\gets\cD$, then sample the challenges
            \[
                \hat{x}_{L,i},\hat{x}_{R,i},h_{L,2i},h_{R,2i},h_{L,2i+1},h_{R,2i+1}\gets\{0,1\}^m,
            \]
            for rounds $2i$ and $2i+1$. Let $\bar{x}_i:=x_i\oplus\cO_{\sf chal}(\hat{x}_{L,i}\oplus\hat{x}_{R,i})$. $\cal V_L$ and $\cal V_R$ send $(\bar{x}_i,\hat{x}_{L,i},h_{L,2i})$ and $(\hat{x}_{R,i},h_{R,2i})$, respectively, so that they arrive at location $L$ at time $t_{2i}$. They send $h_{L,2i+1}$ and $h_{R,2i+1}$ so that they arrive at $L$ at time $t_{2i+1}$.
            \item For each $i$, upon receiving the messages for checkpoint $2i$ at location $L$, the honest prover constructs the real challenges $x_i$ and $h_{2i}$ as follows:
            \begin{align*}
                x_i&:=\bar{x}_i\oplus\cO_{\sf chal}(\hat{x}_{L,i}\oplus\hat{x}_{R,i}) \\
                h_{2i}&:=h_{L,2i}\oplus h_{R,2i}
            \end{align*}
            It then evaluates
            \[
                (y_i,\sigma_i)\gets \QSIO.\QEval(x_i,\rho_{\QSIO}),
            \]
            and coherently queries $\cO(h_{2i},x_i,y_i,\sigma_i,0,\cdot)$ on the coset state in the standard basis, obtaining an answer $a_{2i}$, which it immediately broadcasts.
            \item For each $i$, upon receiving the messages for checkpoint $2i+1$ at location $L$, the honest prover reconstructs $h_{2i+1}:=h_{L,2i+1}\oplus h_{R,2i+1}$ and coherently queries $\cO(h_{2i+1},x_i,y_i,\sigma_i,1,\cdot)$ on the same coset state in the Hadamard basis, obtaining an answer $a_{2i+1}$, which it immediately broadcasts.
            \item If any response is missing, the verifiers output $\bot$. Otherwise, for each $i$, the verifiers set $h_{2i}:=h_{L,2i}\oplus h_{R,2i}$ and $h_{2i+1}:=h_{L,2i+1}\oplus h_{R,2i+1}$, compute
            \[
                \sigma_{x_i}:=\Sign(\sk,(x_i,f(x_i)))
            \]
            and accept iff for every $i\in\{0,\ldots,\gamma-1\}$,
            \[
                a_{2i}=\cO(h_{2i},x_i,f(x_i),\sigma_{x_i},0,v)
                \qquad\text{and}\qquad
                a_{2i+1}=\cO(h_{2i+1},x_i,f(x_i),\sigma_{x_i},1,u).
            \]
        \end{enumerate}
    \end{itemize}
\end{construction}
\begin{remark}
    Note that unlike the constructions of trajectory verification, entanglement localization, and state localization, \Cref{con:function_localization_from_qsio} does not randomize the basis of each challenge. Instead it uses a fixed ordering --- standard, then Hadamard. This simplification is possible because here we are no longer relying on the EPR extraction machinery of \cite{vidick2021classicalproofsquantumknowledge}, and instead all we need is to show that the prover must have queried at the designated position. This is already a property we establish for the standard-then-Hadamard monogamy game in \Cref{def:multi-Stage_monogamy-of-entanglement}, and by reducing to this game instead of \MultiStageIndependentMonogamy we get a quadratically better loss in $\beta(\secp,\eta)$ than for the other primitives.
\end{remark}
\begin{lemma}[Functionality Preservation]
    \Cref{con:function_localization_from_qsio} satisfies \textbf{Functionality Preservation}.
\end{lemma}
\begin{proof}
    We define $\Eval(\breg,x)$ as follows. Take the state on $\breg$, $\rho_\breg=(\rho_\QSIO,X^vZ^u\ket{S})$, run $(y,\sigma)\gets\QSIO.\QEval(\rho_\QSIO,x)$, and output $y$. The property follows immediately from correctness of $\QSIO$.
\end{proof}
\begin{lemma}[Completeness, non-destructiveness]\label{prop:function_localization_completeness}
    \Cref{con:function_localization_from_qsio} satisfies $(1-\negl(\secp))$-\textbf{Completeness}. Furthermore, it is {non-destructive}.
\end{lemma}

\begin{proof}
    By correctness of $\QSIO$, in every checkpoint pair the honest prover obtains $(f(x_i),\sigma_{x_i})$ except with negligible probability, and the state $\rho_{\QSIO}$ is preserved by gentle measurement. For each $i$, on checkpoint $2i$, $\cO(h_{2i},x_i,f(x_i),\sigma_{x_i},0,\cdot)$ is constant on the coset $S+v$, so coherently querying it on $X^vZ^u\ket{S}$ in the standard basis returns $\cO(h_{2i},x_i,f(x_i),\sigma_{x_i},0,v)$ while leaving the coset state unchanged. Similarly, $\cO(h_{2i+1},x_i,f(x_i),\sigma_{x_i},1,\cdot)$ is constant on the coset $S^\perp+u$, so coherently querying it in the Hadamard basis returns $\cO(h_{2i+1},x_i,f(x_i),\sigma_{x_i},1,u)$ while again leaving the coset state unchanged. A union bound over the polynomially many checkpoints gives acceptance probability $1-\negl(\secp)$ and final prover state negligibly close to the initial state.
\end{proof}

\begin{lemma}[Extraction Soundness]\label{prop:function_localization_existence}
    If $(\Sign,\Ver)$ is a secure MAC scheme and $\QSIO$ is an ideal obfuscator, then \Cref{con:function_localization_from_qsio} satisfies $\brackets{\beta,\Delta}$-\textbf{Extraction soundness}, where
    \[
        \beta(\secp,\eta)=\eta(\secp)\left(1-O\!\left(\sqrt{\frac{\log(1/\eta)}{\gamma}}\right)-\negl(\secp)\right).
    \]
\end{lemma}

\begin{proof} The proof proceeds in several steps, which we outline below.
    \paragraph{Removing the Signature} To analyze the construction, we first consider modifying the public oracle to reject any queries not output by the signed functionality. Define
    \begin{equation}\label{eq:oracle_f}
        \cO_f(h,x,y,\sigma,c,z)=
        \begin{cases}
            \cO(h,x,y,\sigma,c,z) & \text{if } y=f(x),\\
            \bot & \text{otherwise.}
        \end{cases}
    \end{equation}
    \begin{claim}\label{claim:oracle_switch_function_localization}
        No QPT distinguisher has more than negligible advantage in telling apart a setup which outputs $\cO$, from a setup which outputs $\cO_f$.
    \end{claim}

    \begin{proof}
        The two worlds differ only on queries $(h,x,y,\sigma,c,z)$ such that $\Ver(\vk,(x,y),\sigma)=\top$ and $y\neq f(x)$. Any such query contains a valid signature on a message $(x,y)$ that is never signed by the functionality $\widehat f_{\sk}$, since the only signatures returned by $\widehat f_{\sk}$ are on messages of the form $(x,f(x))$. If there were a distinguisher between the two worlds, then by the security of $\QSIO$ we could replace the obfuscated state by black-box access to the functionality $\widehat f_{\sk}$ without losing more than negligible advantage. A standard BBBV-type search-to-distinguishing reduction would then find, with non-negligible probability, a query on which $\cO$ and $\cO_f$ differ, yielding a valid forgery for the MAC scheme.
    \end{proof}

    We may therefore analyze the experiment with $\cO_f$ in place of $\cO$, losing only a negligible term.

    \paragraph{Finding a Successful Checkpoint Pair}
    Write $\delta=\Delta/(2\gamma+2)$ as in the construction.
    Fix $L$, $t$, and a QPT prover $\cal{P}^*$ succeeding in $\Localize$ with probability $\eta$. Let $q_i$ denote the probability that $\cal P^*$ passes the $i$-th checkpoint pairs condition on passing all previous checkpoints. There exists $i^\star\in\{0,\ldots,\gamma-1\}$ such that the checkpoint pair $2i^\star,2i^\star+1$ is passed with probability
    \[
        q_{i^\star}\ge \eta^{1/\gamma}
    \]
    conditioned on passing checkpoint pairs $0,\ldots,i^\star-1$. Let time $t^\star=t+2i^\star\delta$ be the first challenge point of the pair, and denote the failure probability by $\varepsilon:=1-q_{i^\star}$ for convenience.

    \paragraph{Extracting the Functionality.} We now define the implementation extractor $\cE$. In addition to a sampled input $x$, the extractor takes as input the public setup parameters $\param$, and the local register $\breg^*$ from which it will extract. For our construction, $\ek=\emptyset$. We will split $\cE$ into two parts $\cal{E}_0$ and $\cal{E}_1$, which also act on different parts of the input, for convenience later on in the proof.

    $\cal{E}_0(\param,\ek,\breg^*)$: The extractor first forward-simulates $\breg^*=\register[L_\Delta~@~t]$ to obtain $\register[L_{2\delta}~@~t^\star]$, which we can see is possible by checking the backwards light-cone: $t^\star-t\le 2\gamma\delta=\Delta-2\delta$. Note that when $\cE_0$ is doing this forward-simulation, it is not sampling challenges or answering oracle queries on its own. Instead, it is using all information contained in $\breg^*$ (including inflight verifier messages, etc) along with the code of $\cal{P}^*$, and access to the oracles in $\param$, to simulate the prover perfectly.
    
    Next, in the same way, $\cE_0$ simulates $\cal{P}^*$ inside the spacetime region for checkpoint $2i^\star$, essentially corresponding to the light blue-shaded region of \Cref{fig:time-position-B}). Let $R_{\vee}$ be the vee-shaped top border of this region (bottom border of the green triangle in \Cref{fig:time-position-B}), and note that its backwards light-cone is $L_{2\delta}$ at time $t^\star$. The output of $\cE_0$ is \[\mreg:=\register[R_{\vee}](\text{Forward time-evolution of $\breg^*$})\]

    $\cal{E}_1(\param,\ek,\mreg,x)$: Continuing from $\cal{E}_0$, and now having additional input $x$, the extractor now recovers $\cal{V}$'s messages for checkpoints $2i^\star,2i^\star+1$,
    \[
        \bar{x}_{i^\star},\hat{x}_{L,i^\star},\hat{x}_{R,i^\star},h_{L,2i^\star},h_{R,2i^\star},h_{L,2i^\star+1},h_{R,2i^\star+1}
    \]
    from the transcript contained in the forward simulation. Let $\hat{x}_{i^\star}:=\hat{x}_{L,i^\star}\oplus\hat{x}_{R,i^\star}$, $h_{2i^\star}:=h_{L,2i^\star}\oplus h_{R,2i^\star}$, and $h_{2i^\star+1}:=h_{L,2i^\star+1}\oplus h_{R,2i^\star+1}$.

    It then simulates $\cal{P}^*$ inside the spacetime triangle $\nabla$, corresponding to the green triangle in \Cref{fig:time-position-B}. For convenience, we use $\cal{P}^*_\nabla$ to refer to the portion of $\cal{P}^*$'s circuits contained within the region $\nabla$. During its simulation of $\cal{P}^*_\nabla$, $\cE$ intercepts all oracle queries and responds in the following way.
    \begin{itemize}
        \item On query $\cO_{\sf chal}(x')$:
        \begin{itemize}
            \item If $x'=\hat{x}_{i^\star}$, respond with $x\oplus\bar{x}_{i^\star}$;
            \item otherwise, forward to the true oracle $\cO_{\sf chal}$.
        \end{itemize}
        \item On query $\cO(h,x',y,\sigma,c,z)$:
        \begin{itemize}
            \item If $(h,x',c)=(h_{2i^\star},x,0)$, first compute $w=\cO(h_{2i^\star},x,y,\sigma,c,z)$.
            \begin{itemize}
                \item If $w\neq\bot$, respond with $H(y)$, where $H$ is a compressed oracle;
                \item otherwise, respond with $\bot$.
            \end{itemize}
            \item otherwise, forward to the true oracle $\cO$;
        \end{itemize}
    \end{itemize}
    At the end of the simulation, $\cE$ measures the database of the compressed oracle $H$. If there is a unique non-$\bot$ entry $y$, it outputs $y$; otherwise it outputs $\bot$.

    We now analyze $\cE$, which proceeds very similarly to the analysis in \Cref{thm:NDEL_overall}.
    \begin{claim}
        In the extraction soundness experiment, $\cE$ outputs $y=f(x)$ with probability at least $\eta(1-2\sqrt{2\varepsilon})-\negl(\secp)$.
    \end{claim}
    \begin{proof}
    We will prove this by first constructing a reduction from the $\Localize$ experiment to the $\MultiStageSearchMonogamy$ game of \Cref{def:multi-Stage_monogamy-of-entanglement}, and then showing how to view our extractor $\cE$ as a component in that reduction.

    The reduction will be a seven-part adversary
    \[
        \cal A=
        \big(
        \underbrace{\cal A_M^0}_{\text{stage 0}},
        \underbrace{\cal A_M^1,\cal A_L^0,\cal A_R^0}_{\text{stage 1}},
        \underbrace{\cal A_M^2,\cal A_L^1,\cal A_R^1}_{\text{stage 2}}
        \big)
    \]
    for ${\sf MultiStageSearchMonogamy}$, where $S=T$ because the parameter $n_e$ is set to 0. The challenger samples $(S,u,v)$ and provides the initial coset-state register $X^v Z^u\ket{S}$ to $\cal A_M^0$. In stage~1 it gives oracle access to $\cO_{S+v}^{b_0}$ for a uniformly random hidden answer $b_0\in\F_2^m$, and in stage~2 it additionally gives oracle access to $\cO_{S^\perp+u}^{b_1}$ for an independent uniformly random hidden answer $b_1\in\F_2^m$.

    The adversary $\cal A$ is defined just as in \Cref{thm:NDEL_overall}, excepting a few differences specific to the functionality localization setting which we highlight below. Let $r(\secp)$ be an upper bound on the number of oracle queries made by $\cal{P}^*$.

    \begin{enumerate}
        \item Once $\cal A_M^0$ receives the coset-state $X^vZ^u\ket{S}$ along with membership oracles $\cO_{S+v}$ and $\cO_{S^\perp+u}$ from the challenger, it additionally samples the rest of $\Setup$:
        \[
            f\gets\cF,\qquad \sk,\vk\gets\keygen(1^\secp),\qquad \rho_{\QSIO}\gets \QSIO.\QObf(1^\secp,\widehat f_{\sk}).
        \]
        Letting $\breg$ be a register with state
        \(
            \rho=(\rho_{\QSIO},X^vZ^u\ket{S}),
        \)
        the simulated prover is initialized with input $\breg$.
        $\cal A_M^0$ also initializes two $2r(\secp)$-wise independent functions $F_k$ and $G_{k'}$ (for simulating $\cO_{\sf random}$ and $\cO_{\sf chal}$, respectively), and samples the following challenges for each $i\in\{0,\dots,i^\star\}$
        \[
            x_i,\bar{x}_i,\hat{x}_{L,i},\hat{x}_{R,i},h_{L,2i},h_{R,2i},h_{L,2i+1},h_{R,2i+1}
        \]
        where $x_i\gets\cD$, and the rest are uniform subject to $x_i=\bar{x}_i\oplus G_{k'}(\hat{x}_{L,i}\oplus \hat{x}_{R,i})$. Let $h_{2i^\star}:=h_{L,2i^\star}\oplus h_{R,2i^\star}$ and $h_{2i^\star+1}:=h_{L,2i^\star+1}\oplus h_{R,2i^\star+1}$.
        The first $i^\star$ of these challenge tuples are used, along with the hash functions $F_k$ and $G_{k'}$, to simulate the prover through rounds $0,\dots,2i^\star-1$ of the $\Localize$ experiment. If the prover responds incorrectly in any of these rounds, $\cal A_M^0$ aborts. The challenges for rounds $2i^\star$ and $2i^\star+1$ are placed in the $\info$ register which is copied along between the reduction components. In particular, $\info$ contains $x_{i^\star}$, $h_{2i^\star}$, $h_{2i^\star+1}$, and $f$. Like in \Cref{thm:TV_overall}, access to the oracle $\cal{O}$ is simulated using the sampled components of $\Setup$, along with $F_k$ for the random oracle and $\cO_{S+v}$ and $\cO_{S^\perp+u}$ for checking coset membership. Access to $\cO_{\sf chal}$ is similarly simulated with $G_{k'}$.
        \item For the simulation of all parts of $\cal A$ in stages 1 and 2, we reprogram the oracle as follows: all queries of the form $(h_{2i^\star},x_{i^\star},f(x_{i^\star}),\sigma,0,z)$ are remapped to $\cO_{S+v}^{b_0}(z)$. In other words, $\cO_{\sf random}$ is reprogrammed at the lone point $(h_{2i^\star},x_{i^\star},f(x_{i^\star}),0)$.

        \item For the simulation of stage 2, we additionally reprogram the oracle as follows: all queries of the form $(h_{2i^\star+1},x_{i^\star},f(x_{i^\star}),\sigma,1,z)$ are remapped to $\cO_{S^\perp+u}^{b_1}(z)$. In other words, $\cO_{\sf random}$ is reprogrammed at the lone point $(h_{2i^\star+1},x_{i^\star},f(x_{i^\star}),1)$.
\end{enumerate}

    \begin{claim}\label{claim:functionality-reduction}
        \begin{align*}
            \Pr\brackets{{\sf MultiStageSearchMonogamy}(\cal A,1^\secp)\neq\bot}
            &\ge \eta,\text{ and }\\
            \frac{\Pr\brackets{{\sf MultiStageSearchMonogamy}(\cal A,1^\secp)=1}}{\Pr\brackets{{\sf MultiStageSearchMonogamy}(\cal A,1^\secp)\neq\bot}}
            &\ge 1-\varepsilon-\negl(\secp).
        \end{align*}
    \end{claim}
    \begin{proof}
    The component $\cal A_M^0$ aborts exactly if the simulated prover fails one of checkpoint pairs $0,\ldots,i^\star-1$, giving the first bound. Since the tuple $(h_{2i^\star},x_{i^\star})$ is not jointly available before spacetime point $(L,t^\star)$, and the tuple $(h_{2i^\star+1},x_{i^\star})$ is not jointly available before $(L,t^\star+\delta)$, the same oracle-reprogramming argument used in \Cref{thm:TV_overall} (invoking \Cref{lem:oracle_reprogram}) implies the simulation of checkpoints $2i^\star,2i^\star+1$ is negligibly close to the real-world execution of adversary $\cal{P}^*$ in the \Localize experiment (when the oracle is $\cO_f$ instead of $\cO$), conditioned on reaching checkpoint pair $i^\star$. Then, the second inequality follows from \Cref{claim:oracle_switch_function_localization}.
    \end{proof}
We now introduce a sequence of hybrids, where the first corresponds to the \MultiStageSearchMonogamy reduction, and the last corresponds to the extraction experiment. Throughout these hybrids, probabilities are conditioned on $\cal A_M^0$ not aborting, corresponding to the event where the real execution of $\Localize$ reaches checkpoint pair $i^\star$.

\noindent\textbf{Hybrid 1:} This denotes the following experiment.
\begin{enumerate}
    \item Run $\cal{A}^0_M$ from the reduction above, and take its middle output $\mreg'=(\info,\mreg)$. Receive $\cO_{S+v}^{b_0}$ from the $\MultiStageSearchMonogamy$ challenger. Take $x_{i^\star},h_{2i^\star},h_{2i^\star+1},f\gets\info$.
    \item Simulate $\cal{P}^*_{\nabla}$ (defined above, in the description of $\cE$) on input $\mreg$. In this simulation, oracle queries to $\cO$ are answered with $\cO'$, which depends on $x_{i^\star},h_{2i^\star},f$ and is described below. Here, $H$ is a compressed oracle on $\cal{Y}_\secp\to\{0,1\}^m$.
    \[\hspace{-3em}
        \cO'(h,x',y,\sigma,c,z)=
        \begin{cases}
            H(f(x')) & \text{if } c=0\land x'=x_{i^\star}\land h=h_{2i^\star}\land y=f(x')\land\cO(h,x',y,\sigma,c,z)\neq\bot,\\
            \cO(h,x',y,\sigma,c,z) & \text{otherwise.}
        \end{cases}
    \]
    \item After simulating $\cal{P}^*_{\nabla}$, measure the compressed oracle database of $H$.
    \item If there is exactly one non-$\bot$ entry $y$, output $y$. Otherwise, output $\bot$.
\end{enumerate}
\begin{claim}\label{claim:functionality-hyb2}
    Let $p_1$ denote the probability that the output of \textbf{Hybrid 1} is exactly $f(x_{i^\star})$. Then,
    \[p_1\geq1-2\sqrt{2\varepsilon}-\negl(\secp)\]
\end{claim}
\begin{proof}
    Consider the measurement of the $H$ database. Clearly there can be at most one non-$\bot$ entry, since $H$ can only be accessed at $f(x_{i^\star})$. Thus we must now argue that the probability that the measurement of the $H(f(x_{i^\star}))$ database entry is $\bot$ is at most $2\sqrt{2\varepsilon}+\negl(\secp)$. We will do this as follows:

    First, we notice that steps 1 and 2 of \textbf{Hybrid 1} are the same experiment as running $\cal A^0_M$ and then $\cal A^1_M$ in the \MultiStageSearchMonogamy game, except that the queries $\cO(h_{2i^\star},x_{i^\star},f(x_{i^\star}),\sigma,0,z)$ which $\cal A^1_M$ would reprogram to $\cO^{b_0}_{S+v}(z)$ are now instead reprogrammed to $H(f(x_{i^\star}))$. In other words, $\cO_{\sf random}(h_{2i^\star},x_{i^\star},f(x_{i^\star}),0)$ is now reprogrammed to $H(f(x_{i^\star}))$ instead of $b_0$. These two experiments are identical, since $b_0$ and $H(f(x_{i^\star}))$ are both uniformly random. Then, by \Cref{claim:functionality-reduction}, if we run the \MultiStageSearchMonogamy reduction but replace $\cal A_M^0$ and $\cal A_M^1$ with steps 1 and 2 of \textbf{Hybrid 1}, it would succeed conditioned on non-abort with probability $1-\varepsilon-\negl(\secp)$.

    Finally, let $\breg$ denote the $f(x_{i^\star})^{th}$ entry in the compressed oracle database of $H$, which will initially be in the uniform superposition over $m$-bit output strings. We notice that the simulated prover has the following access to $\breg$:
    \[\cO(h_{2i^\star},x_{i^\star},f(x_{i^\star}),\sigma,0,\cdot)=\sum_{b_0\in\mathbb{F}_2^{m_0}}\ket{b_0}_{\breg}\bra{b_0}\otimes\cal{O}_{S+v}^{b_0}.\]
    and that there is no other dependence that the prover can have on $\breg$. This is the same as the interface given to the adversary in \Cref{def:multi-Stage_monogamy-of-entanglement}. Now, given the last paragraph, \Cref{lem:multi-stage_search_monogamy-of-entanglement} applies: if we write the purified state of the entire \textbf{Hybrid 1} experiment, after step 2, as
    \[\ket{\psi}_{\areg\breg}=\sum_{\sk}\sum_{b_0\in\F_2^m}\ket{\psi^{\sk}_{b_0}}_{\areg}\ket{b_0}_{\breg}~,\]
    where $\sk=(S,u,v)$, $\areg$ is all of $\cal{P}^*$'s registers and messages at the top edge of the region $\nabla$, and $\breg$ is the purified choice of $b_0$ from \Cref{def:multi-Stage_monogamy-of-entanglement}, then,
    \begin{equation*}
        \trace{\ket{+_\breg}_\breg\bra{+_\breg}\cdot\ket{\psi}_{\areg\breg}\bra{\psi}}\leq2\sqrt{2\varepsilon}+\negl(\secp).
    \end{equation*}
    Thus, if we measure this register in the Hadamard basis, it will be $\bot$ with probability $\leq2\sqrt{2\varepsilon}+\negl(\secp)$. This proves the claim.
\end{proof}
\noindent\textbf{Hybrid 2:} This denotes the following experiment:
\begin{enumerate}
    \item Run $\cal{A}^0_M$ from the reduction above, and take its middle output $\mreg'=(\info,\mreg)$. Receive $\cO_{S+v}^{b_0}$ from the $\MultiStageSearchMonogamy$ challenger. Take \textcolor{red}{$x_{i^\star},h_{2i^\star},h_{2i^\star+1}\gets\info$ (discarding $f$)}.
    \item Simulate $\cal{P}^*_{\nabla}$ on input $\mreg$. In this simulation, oracle queries to $\cO$ are answered with \textcolor{red}{$\cO''$}, which depends on $x_{i^\star},h_{2i^\star}$ and is described below. Here, $H$ is a compressed oracle on $\cal{Y}_\secp\to\{0,1\}^m$.
    \[
        \cO''(h,x',y,\sigma,c,z)=
        \begin{cases}
            \textcolor{red}{H(y)} & \textcolor{red}{\text{if } c=0\land x'=x_{i^\star}\land h=h_{2i^\star}\land\cO(h,x',y,\sigma,c,z)\neq\bot},\\
            \cO(h,x',y,\sigma,c,z) & \text{otherwise.}
        \end{cases}
    \]
    \item After simulating $\cal{P}^*_{\nabla}$, measure the compressed oracle database of $H$.
    \item If there is exactly one non-$\bot$ entry $y$, output $y$. Otherwise, output $\bot$.
\end{enumerate}
\begin{claim}
    Let $p_2$ denote the probability that the output of \textbf{Hybrid 2} is exactly $f(x_{i^\star})$. Then,
    \[|p_2-p_1|\leq\negl(\secp).\]
\end{claim}
\begin{proof}
    We can first swap $\cO$ in \textbf{Hybrid 1} for $\cO_f$, which is distinguishable with only negligible probability by \Cref{claim:oracle_switch_function_localization}. Then, the check for $y=f(x_{i^\star})$ becomes redundant, since it already exists in $\cO_f(h,x',y,\sigma,c,z)\neq\bot$, so we can remove it from the definition of $\cO'$. Then, we switch back to $\cO$ by another application of \Cref{claim:oracle_switch_function_localization}, and we get the definition of $\cO''$. Since the oracle is now independent of $f$, we can discard $f$ from $\info$.
\end{proof}
\noindent\textbf{Hybrid 3:} This denotes the following experiment.
\begin{enumerate}
    \item Run $\param,\sparam,\ek,f,\rho\gets\Setup(1^\secp)$.
    \item Start the procedure $\Localize\left(\cal{P}^*(\breg)\rightleftharpoons \cal{V}(\sparam,f)\right)(\param,L,t)$, but freeze all circuits and messages at time $t$. Let $\breg^*:=\register[L_\Delta~@~t]\big(\Localize\left(\cal{P}^*(\breg)\rightleftharpoons \cal{V}(\sparam,f)\right)(\param,L,t)\big)$.
    \item Use the code of $\cal{P}^*$, along with the oracles $(\cO,\cO_{\sf chal})=\param$, to evolve $\breg^*$ forward to $\register[L_{2\delta}~@~t^\star]$ and then find the state on all registers and messages in the region described by the following:
    \begin{enumerate}
        \item $s-p=t^\star-L$;
        \item $s+p=t^\star+L$;
        \item $p\in[L-2\delta,L+2\delta]$,
    \end{enumerate}
    which we call $R_{\vee}$. The output of this step is
    \[\mreg:=\register[R_{\vee}](\text{Forward time-evolution of $\breg^*$})\]
    \item Compute $\info=(x_{i^\star},h_{2i^\star},h_{2i^\star+1})$ for checkpoint pair $i^\star$ from the transcript of verifier messages in the forward simulation.
    \item Run steps 2-4 of \textbf{Hybrid 2} with inputs $(\info,\mreg)$.
    \begin{itemize}
        \item Simulate $\cal{P}^*_{\nabla}$ on input $\mreg$. In this simulation, oracle queries to $\cO$ are answered with {$\cO''$}, which depends on $x_{i^\star},h_{2i^\star}$ and is described in \textbf{Hybrid 2}.
        \item After simulating $\cal{P}^*_{\nabla}$, measure the compressed oracle database of $H$.
        \item If there is exactly one non-$\bot$ entry $y$, output $y$. Otherwise, output $\bot$.
    \end{itemize}
\end{enumerate}
\begin{claim}
    Let $p_3$ denote the probability that the output of \textbf{Hybrid 3} is exactly $f(x_{i^\star})$. Then,
    \[p_3=p_2~.\]
\end{claim}
\begin{proof}
    By examining the definition of $\cal A_M^0$, one can see that the output $$(\info=(x_{i^\star},h_{2i^\star},h_{2i^\star+1}),\mreg)\gets\cal A_M^0()$$ is equal to the output of the simulation in steps 1-4 above. In other words, we argue that since the reduction $\cal A$ is specified in such a way as to perfectly simulate the $\Localize$ experiment, \textbf{Hybrid 2} and \textbf{Hybrid 3} are the same experiment up to rewriting of steps.
\end{proof}
\noindent\textbf{Hybrid 4:} This denotes the following experiment.
\begin{enumerate}
    \item Run $\param,\sparam,\ek,f,\breg\gets\Setup(1^\secp)$, \textcolor{red}{and sample $x\gets\cD$.}
    \item Start the procedure $\Localize\left(\cal{P}^*(\breg)\rightleftharpoons \cal{V}(\sparam,f)\right)(\param,L,t)$, but freeze all circuits and messages at time $t$. Let $\breg^*:=\register[L_\Delta~@~t]\big(\Localize\left(\cal{P}^*(\breg)\rightleftharpoons \cal{V}(\sparam,f)\right)(\param,L,t)\big)$.
    \item \textcolor{red}{Run $\mreg\gets\cal{E}_0(\param,\ek,\breg^*)$.}
    \item \textcolor{red}{Run $y\gets\cal{E}_1(\param,\ek,\mreg,x)$, and output $y$.}
\end{enumerate}
\begin{claim}\label{claim:functionality-hyb4}
    Let $p_4$ denote the probability that the output of \textbf{Hybrid 4} is exactly $f(x)$. Then,
    \[|p_4-p_3|\leq\negl(\secp)\]
\end{claim}
\begin{proof}
    We notice that step 3 of \textbf{Hybrid 3} is definitionally equivalent to $\mreg\gets\cE_0(\param,\breg^*)$. Also, the fresh input $x$ in \textbf{Hybrid 4} has the same distribution as $x_{i^\star}$ in \textbf{Hybrid 3}. Thus, the remaining steps are exactly the same procedure as $y\gets\cE_1(\param,\mreg,x)$, except that $\cE_1$ reprograms $\cO_{\sf chal}$ at $\hat{x}_{i^\star}$ so the selected checkpoint challenge is this fresh $x$ rather than the original $x_{i^\star}$. By \Cref{lem:oracle_reprogram}, this reprogramming can be detected only with negligible probability.
\end{proof}
Now, all that is left to notice is that \textbf{Hybrid 4} is precisely the extraction experiment for checkpoint pair $i^\star$. Also, $\cal{E}$ is efficient. Thus, as a consequence of \Cref{claim:functionality-hyb2} through \Cref{claim:functionality-hyb4}, the extractor succeeds with probability $1-2\sqrt{2\varepsilon}-\negl(\secp)$ conditioned on reaching checkpoint pair $i^\star$. The overall extraction experiment succeeds with probability at least
\[
    \eta(\secp)\left(1-2\sqrt{2\varepsilon}-\negl(\secp)\right).
\]
\end{proof}
Finally, $\varepsilon=1-q_{i^\star}\le 1-\eta^{1/\gamma}\le \log(1/\eta)/\gamma$. The loss term $2\sqrt{2\varepsilon}$ is
\[
    O\!\left(\sqrt{\frac{\log(1/\eta)}{\gamma}}\right).
\]
\end{proof}
\begin{lemma}[Uniqueness]\label{prop:function_localization_uniqueness}
    Assume that $\cal F$ is a copy-protectable functionality with security $\zeta=\zeta(\secp)$, and that $\QSIO$ is an ideal obfuscator. Then, \Cref{con:function_localization_from_qsio} satisfies $\zeta$-\textbf{Uniqueness}.
\end{lemma}
\begin{proof}
    As noted in \cite{CG24}, $\QSIO$ is a best-possible copy-protector. In other words, if $\cF$ is copy-protectable with security $\alpha$, then $(\QSIO.{\sf QObf},\QSIO.{\sf QEval})$ is a copy-protection scheme with security at least $\alpha$.

\end{proof}

\begin{theorem}
    Let $\cal{F}$ be a classical functionality which is copy-protectable with security $\zeta(\secp)$. Then there exists a {non-destructive} functionality localization scheme for $\cal{F}$ with $(1-\negl(\secp))$-\textbf{Completeness}, $\brackets{\beta,\Delta}$-\textbf{Extraction Soundness} for
    \[
        \beta(\secp,\eta)=\eta(\secp)\left(1-O\!\left(\sqrt{\frac{\log(1/\eta)}{\gamma}}\right)-\negl(\secp)\right)
    \]
    and $\zeta(\secp)$-\textbf{Uniqueness} in the classical oracle model. In particular, since we only consider adversaries for which $1/\eta$ is at most polynomial in $\secp$, for any polynomial $t(\secp)$, we can set $\gamma=c\secp t^2$ for some constant $c$ to achieve
    \[\beta(\secp,\eta)=\eta(\secp)\left(1-1/t(\secp)\right)~.\]
\end{theorem}
\begin{proof}
    Instantiating \Cref{con:function_localization_from_qsio} in the classical oracle model gives such a scheme as claimed, with the properties following from \Cref{prop:function_localization_completeness}, \Cref{prop:function_localization_existence}, and \Cref{prop:function_localization_uniqueness}, respectively.
\end{proof}

%% file: applications.tex
\section{Discussion}
\subsection{Publicly-Verifiable Localization}\label{subsec:public_verifiability}
All of our constructions in \Cref{sec:localizing} utilize some secret parameters $\sparam$ generated by $\Setup$ and then held privately by the verifiers. The purpose of $\sparam$ is, generally speaking, to serve as a verification key for the protocol -- i.e. the eventual accept/reject decision is a deterministic predicate of two things, 1. the (classical) protocol transcript and 2. the secret parameters $\sparam$. It is natural to see if these can be made public, a la publicly-verifiable digital signatures. However unfortunately, in our constructions $\sparam$ contains the coset descriptors $u,v$ -- meaning that a prover knowing $\sparam$, even if he completely discards $\rho$, could respond to all challenges perfectly, nullifying any extraction guarantees.

An optimistic reader might hope for alternative solutions, and therefore ask:

\noindent\emph{Are there sound localization schemes in which the verifiers hold no secret information? I.e., schemes where \textbf{anyone} can play the role of the verifiers, and achieve an extraction guarantee?}

In this section, we will argue that we can indeed achieve this for entanglement localization, state localization, and trajectory verification, by making a slight modification to our existing constructions. Firstly, we introduce the formal notion of public verifiability.

\begin{definition}
     We say that a localization scheme (e.g. one of Definitions \ref{def:NDEL},\ref{def:trajectory_verification},\ref{def:state-localization}, or \ref{def:functionality_localization_revised}) is \textbf{publicly verifiable} if the $\Localize$ protocol takes no private verifier input $\sparam$. In other words, if the interface for the $\Localize$ protocol can be written as $\Localize\brackets{\cal{P}(\breg) \rightleftharpoons \cal{V}}(\param,\dots)\rightarrow \breg',b$.
\end{definition}

Next, we will show an idea for how to augment our constructions of each aforementioned localization scheme, so that the verifiers can check the protocol transcript without knowledge of $u,v$. In order to give the general idea and intuition without getting too deep into the implementation details for each specific protocol, we will keep this explanation relatively high-level.

\begin{theorem}[informal]
    Each of \Cref{con:entanglement_localization}, \Cref{con:non-destructive_state_localization}, and \Cref{con:trajectory_verification} can be made publicly verifiable\footnote{\label{footnote:unknown_public_verifiable_functionality_localization}We leave out functionality localization from this statement, because it is unclear how public verifiability should be defined with respect to the private functionality $f$. If $f$ itself has a natural notion of public verification (e.g. if $f$ is a copy-protectable signing algorithm for a public-key signature scheme), then we could potentially hope to extend the techniques in this section to allow such a functionality $f$ to be localized in a publicly verifiable way. }.
\end{theorem}
\begin{proof}
    To do this, we will apply the same, simple technique to each construction. As an initial (flawed) attempt, consider adding a public verification oracle like the following:
    \[\cO_{\sf ver}(\{\text{challenge}_i\}_i,\{\text{response}_i\}_i)=\begin{cases}
        \top&\text{if }\text{response}_i=\cO_{\sf random}(\text{challenge}_i)~~\forall i,\\
        \bot&\text{otherwise.}
    \end{cases}\]
    This is essentially an oracle for the verifiers' accept/reject decision: it takes as input the entire protocol transcript, and one can check that its acceptance predicate is logically equivalent to the one computed by the verifiers in each of our constructions (under the proper interpretations of ``challenge'' and ``response'').
    
    However, we cannot add such a public oracle -- in our reduction to the monogamy game, the simulated prover might play a malicious trick on us: before sending a response, check if it's correct. If yes, send the response, and if not, instead send garbage. By doing so, a prover has not changed its success probability whatsoever. However, the monogamy game reduction is required to answer queries to this verification oracle -- without knowing what the right answers are! So a prover can succeed with very high probability in the real experiment, but might (maliciously) fail to produce correct answers in any reduction where the verification oracle cannot be faithfully simulated.

    To get around this, we use the following idea: at the start of $\Localize$, the verifier will generate a random ``key'' which somehow gates access to the public verification oracle on a certain set of inputs, and release this key at the very end of the protocol. Challenges are then generated with respect to this key, meaning that the prover is unable to use the public verification oracle on the active protocol transcript until after it has produced all of its responses, circumventing the malicious behavior described above.

    One such gating mechanism can be made with a simple hashing trick. Let $H$ be a new public random oracle, such that $H$ and $\cal{O}_{\sf ver}$ are now both generated during setup and output as part of the public parameters $\param$. The verifier will generate the key and challenges jointly at the start of $\Localize$ as follows:
    \[\text{key}\gets\set{0,1}^n~,\]
    \[\text{challenge}_n=H(\text{key}),~~\text{challenge}_{n-1}=H(\text{challenge}_n)=H^{(2)}(\text{key}),~~\cdots~~,\text{challenge}_1=H^{(n)}(\text{key})~.\]
    These challenges will then be directly used in place of the randomly-generated challenges from our original, privately-verifiable protocols. Additionally, we will now swap out the faulty verification oracle from above, and use the following new one which will be generated in $\Setup$ jointly with $H$:
    \[\cO_{\sf ver}(\text{key},\{\text{challenge}_i\}_i,\{\text{response}_i\}_i)=\begin{cases}
        \top&\text{if }H^{(n-i+1)}(\text{key})=\text{challenge}_i~~\forall i\\
        &\land~~\text{response}_i=\cO_{\sf random}(\text{challenge}_i)~~\forall i,\\
        \bot&\text{otherwise.}
    \end{cases}\]
    Finally, the verifier broadcasts the key in a secret-shared manner (like all of the challenges) after all prover responses have been received, thus enabling public verification of the transcript.

    To argue that this transformation preserves extraction soundness, we will describe the reduction from a given round $i$ of the publicly-verifiable protocol to the \MultiStageIndependentMonogamy game --- an updated version of the reduction $\cal{A}$ defined in \Cref{thm:NDEL_overall}. The reduction will work in the same way as for the original schemes, except that it is now tasked with simulating the new oracles $\cal{O}_{\sf ver}$ and $H$, as well as generating challenges and the key in the updated way. It uses a fresh $2q$-wise independent hash function to simulate $H$ in the same way as was done for $\cal{O}_{\sf random}$, and generates simulated verifier challenges and key according to $H$ just as in the modified construction. It remains to show that $\cal{O}_{\sf ver}$ can be simulated. On a query $({\sf key}',\{\text{challenge}'_i\}_i,\{\text{response}'_i\}_i)$, $\cal{A}$ first checks if ${\sf challenge}'_i=(\theta_0,x_{L,0}\oplus x_{R,0})$ or ${\sf challenge}'_i=(\theta_1,x_{L,1}\oplus x_{R,1})$ for any $i$. If yes, then it outputs $\bot$, otherwise it checks the predicate according to its simulated $\cal{O}_{\sf random}$ and $H$ and outputs appropriately. We then use a standard one-way to hiding argument to show that this simulation is computationally indistinguishable from the real experiment, since it is clearly hard to find any value ${\sf key}'$ such that $x_{L,0}\oplus x_{R,0}$ or $x_{L,1}\oplus x_{R,1}$ is an iterate of ${\sf key}'$ under $H$. Thus, the reduction still wins with only negligible loss, and extraction proceeds as usual.
\end{proof}
\subsubsection{Applications to Black-Box Trajectory Verification}
Aside from being an interesting and desirable property on its own, public verifiability of localization protocols has a secondary benefit: allowing trajectory verification to be constructed in a relatively black-box way from a state or entanglement localization protocol. The construction is simple: run $\Setup$ as normal, and then execute back-to-back instances of $\Localize$ for each checkpoint along the trajectory. The details are somewhat tedious, so we will argue about this construction at a very high level.

For completeness to hold, we must be careful that completeness of the localization protocol still holds when the prover is moving along the trajectory rather than staying in place (a la \Cref{fig:trajectory-time-position-B}), and that non-destructiveness of $\Localize$ applies over a short enough timescale such that the prover's state $\rho$ can be plugged into the next instance right away. 

As for soundness, public verifiability allows for a simple black-box reduction: given an adversary for trajectory verification, simulate rounds 1 through $i$ via the public verification algorithm. Then, forward the real $\Localize$ challenges to the adversary for the given round, which it should win with high probability.

\subsection{Position-Based Cryptography}\label{subsec:additional}

Position-based authentication (PBA) is a position-based cryptographic primitive that naturally extends plain position verification, and was initially introduced together with the original construction of QPV, in Buhrman et al. (\cite{buhrman2014position}). In a PBA protocol, the honest prover decides on a message $m$ to relay to the verifiers, and the security guarantee is twofold: 1. When the honest prover is positioned in the secure location, no coalition of malicious parties located elsewhere can successfully tamper with the authenticated message, and 2. A coalition of parties with nobody in the secure location cannot pass the authentication protocol, regardless of the message. The applications of such a scheme are quite numerous and natural, e.g. authenticating military orders from a secure command center, arguably even more so than for plain position verification. As for building PBA, \cite{buhrman2014position} gave a construction based on any black-box quantum position verification scheme, and Unruh showed that their QROM $f$-BB84 position verification scheme could be straightforwardly extended to a PBA scheme secure against all QPT attackers in the QROM (\cite{unruh2014pvqrom}).

However, we will now point out that the security notion described above, under which all PBA schemes from prior literature are proven secure, seems unsatisfying when put under scrutiny. Points 1 and 2 above rule out attacks where nobody is in the secure location, or where the honest prover algorithm runs in the presence of malicious tampering parties, but \emph{does not} rule out attacks where there might be a distributed attack involving a (potentially unwitting) party in the secure location. In particular, the distributed attack on $f$-BB84 described in \Cref{sec:non-local-fbb84} (and depicted in \Cref{fig:attack}) is an example of such an attack on the PBA scheme of Buhrman et al., when their scheme is instantiated with the $f$-BB84 QPV scheme. 
\subsubsection{An explicit ``attack'' scenario}
\newcommand{\Chal}{\ensuremath{{\sf Chal}}\xspace}
\newcommand{\Resp}{\ensuremath{{\sf Resp}}\xspace}
\renewcommand{\Ver}{\ensuremath{{\sf Ver}}\xspace}
\renewcommand{\dec}{\ensuremath{{\sf Dec}}\xspace}
\newcommand{\Auth}{\ensuremath{{\sf Auth}}\xspace}
We briefly describe an adversarial strategy for the PBA construction of \cite{buhrman2014position}, which does not constitute an attack according to the standard security definition, but is nonetheless a concerning scenario for any real-world application of PBA. To describe this attack we'll first give a summary of their construction.
\begin{construction}[\cite{buhrman2014position}]\label{con:buhrman-pba}
    Take a QPV protocol whose routines are described as $(\Chal,\Resp,\Ver)$: \Chal generates verifier challenges, and \Resp computes the prover's response, which is finally checked by \Ver (including timings). Now, suppose the prover wishes to authenticate a bit $c\in\{0,1\}$.
    \begin{enumerate}
        \item First, the prover computes $m=\enc(c)$, according to some carefully chosen multi-bit encoding $\enc$.
        \item Next, for each bit $m_i$, the prover and verifier run the following protocol:
        \begin{itemize}
            \item Verifiers send challenges $c_L,c_R\gets\Chal(1^\secp)$.
            \item If $m_i=1$, the prover responds to both verifiers with $r_i\gets\Resp(c_L,c_R)$. If $m_i=0$, however, the prover sends $r_i=\bot$ to both verifiers with some probability $q$, and otherwise responds honestly.
        \end{itemize}
        \item The verifiers compute $b\gets\dec(\{i:\Ver(r_i)=\top\})$.
    \end{enumerate}
\end{construction}
Intuitively, this protocol is secure in their model because if nobody is at the purported location, then indeed there is nothing a malicious prover can do to inject 1s into the encoding of the message. Injecting 0s is easier, since anyone can send $\bot$, but the encoding prevents this from ever succeeding to tamper with the true message bit.

We now get to describing our ``attack''. Imagine that there are three provers, $P_L,P_M,P_R$, where only $P_M$ is at the true location. They participate in the scheme of \Cref{con:buhrman-pba} where the QPV scheme is $f$-BB84, and have the following strategy. Say, without loss of generality, that the BB84 qubit comes from the left.
\begin{itemize}
    \item $P_L$ and $P_R$ pre-determine a message bit $c\in\{0,1\}$. They compute the encoding $\enc(c)$, and execute the following strategy for each index $i$ of the encoded string:
    \begin{enumerate}
        \item When $P_L$ receives the qubit $H^\theta\ket{b}$, it applies a quantum one-time pad $X^rZ^s$, where the one-time pad key $(r,s)$ belongs to private shared randomness of $P_L$ and $P_R$. Importantly, $P_M$ never sees this key. $P_L$ forwards the scrambled qubit to $P_M$.
        \item When $P_M$ receives the qubit, it measures in the correct basis determined by the $f$-BB84 challenges. It sends the measurement outcome $b'$ to $P_L$ and $P_R$.
        \item $P_L$ and $P_R$ receive this measurement outcome, and by this time they also have full knowledge of the $f$-BB84 challenges. Namely, they know which basis the middle prover measured in, and thus which correction to apply from the quantum one-time pad in order to find the intended measurement outcome $b$.
        \item If $m_i=1$, then $P_L$ and $P_R$ send the outcome $b$ to their respective nearest provers. Otherwise, if $m_i=0$, then with probability $q$ they instead send $\bot$ to each prover.
    \end{enumerate}
\end{itemize}
It is clear that this attack would not violate the typical security definition, since it involves active participation from a prover $P_M$ who is in the secure location. However, there is a very concerning property of this scheme -- namely, that $P_M$ is \emph{entirely unaware} of the message being authenticated. In other words, there might be a prover in the right location who was simply tricked into participating, and the authenticated message in no way originates from the secure location.
\subsubsection{Towards stronger security notions}
To address this attack, we give some informal directions for tackling this weakness in the security definition of PBA. We propose that the localization techniques developed in this paper, for example that of functionality localization, could be used to rule out this class of cheating strategies where the middle prover is acting blindly, without knowledge of the message.

A concrete idea is the following (informal) PBA scheme:
\begin{itemize}
    \item Let $\cal{F}=\{f_k\}_k$ be a PRF. Take $(\Setup,\Localize,\Eval)$ from the functionality localization scheme in \Cref{con:function_localization_from_qsio} instantiated for $\cal{F}$. Let $H$ be a public random oracle.
        \item The verifiers run $\param,\sparam,\ek,f_k,\rho\gets\Setup(1^\lambda)$, and send $\param,\rho$ to the prover. Next, they sample two random challenge strings $c_L,c_R\in\{0,1\}^\secp$.
        \item Two routines are executed simultaneously:
        \begin{itemize}
            \item The prover and verifiers execute
            \[
            \Localize\big( P(\rho) \rightleftharpoons V(\sparam,f)\big)(\param,L,t)
            \]
            \item The left and right verifiers send $c_L$ and $c_R$, respectively, towards $(L,t)$. The prover receives these, computes $c:=c_L\oplus c_R$ and runs $y\gets\Eval(\rho,H(m||c))$. The prover sends $(m,y)$ to both verifiers simultaneously.
        \end{itemize}
    \item The verifiers accept iff $\Localize$ accepted and $y=f_k(H(m||(c_L\oplus c_R)))$.
\end{itemize}
Intuitively, the point of this construction is to authenticate the prover's message with the copy-protected PRF, which is simultaneously being localized to a single position. If a group of provers pass this protocol with high probability, then soundness of functionality localization guarantees that we can extract the ability to sign (say, uniformly random) messages from the circuits at the secure location. Uniqueness also guarantees that this cannot be happening anywhere else at time $t$. Thus, when the verifiers receive the signed message $m,f_k(H(m||c))$, we know that this evaluation must have taken place locally at the secure location, which additionally tells us that the prover queried the oracle value $H(m||c)$ at this location at time $t$. The techniques from our localization constructions would then allow this query to be recorded and ``extracted'' from the prover's circuitry at that position.

We leave further exploration of this direction to future work, e.g. formalize a stronger notion of security for position-based authentication and analyze the above construction in terms of said notion.

%% file: ref.bib
@article{girish2026private,
  title={Private proofs of when and where},
  author={Girish, Uma and Gluch, Greg and Goldwasser, Shafi and Malkin, Tal and Orshansky, Leo and Yuen, Henry},
  journal={arXiv preprint arXiv:2601.18961},
  year={2026}
}

@inproceedings{chandran2009position,
  title={Position based cryptography},
  author={Chandran, Nishanth and Goyal, Vipul and Moriarty, Ryan and Ostrovsky, Rafail},
  booktitle={Annual International Cryptology Conference},
  pages={391--407},
  year={2009},
  organization={Springer}
}

@inproceedings{zhandry2019record,
  title={How to record quantum queries, and applications to quantum indifferentiability},
  author={Zhandry, Mark},
  booktitle={Annual International Cryptology Conference},
  pages={239--268},
  year={2019},
  organization={Springer}
}

@inproceedings{C:Zhandry12,
  author    = {Mark Zhandry},
  editor    = {Reihaneh Safavi-Naini and
  Ran Canetti},
  title     = {Secure Identity-Based Encryption in the Quantum Random Oracle Model},
  booktitle = {Advances in Cryptology - {CRYPTO} 2012 - 32nd Annual Cryptology Conference,
  Santa Barbara, CA, USA, August 19-23, 2012. Proceedings},
  series    = {Lecture Notes in Computer Science},
  volume    = {7417},
  pages     = {758--775},
  publisher = {Springer},
  year      = {2012},
  doi       = {10.1007/978-3-642-32009-5_44},
  bibsource = {dblp computer science bibliography, https://dblp.org}
}

@article{buhrman2014position,
  author    = {Harry Buhrman and
               Nishanth Chandran and
               Serge Fehr and
               Ran Gelles and
               Vipul Goyal and
               Rafail Ostrovsky and
               Christian Schaffner},
  title     = {Position-Based Quantum Cryptography: Impossibility and Constructions},
  journal   = {{SIAM} J. Comput.},
  volume    = {43},
  number    = {1},
  pages     = {150--178},
  year      = {2014},
  doi       = {10.1137/130913687},
  bibsource = {dblp computer science bibliography, https://dblp.org}
}

@article{Tomamichel_2013,
	doi = {10.1088/1367-2630/15/10/103002},
	year = 2013,
	publisher = {{IOP} Publishing},
	volume = {15},
	number = {10},
	pages = {103002},
	author = {Marco Tomamichel and Serge Fehr and J{\k{e}}drzej Kaniewski and Stephanie Wehner},
	title = {A monogamy-of-entanglement game with applications to device-independent quantum cryptography},
	journal = {New Journal of Physics}
}

@InProceedings{unruh2014pvqrom,
  author    = {Dominique Unruh},
  editor    = {Juan A. Garay and
               Rosario Gennaro},
  title     = {Quantum Position Verification in the Random Oracle Model},
  booktitle = {Advances in Cryptology - {CRYPTO} 2014 - 34th Annual Cryptology Conference,
               Santa Barbara, CA, USA, August 17-21, 2014, Proceedings, Part {II}},
  series    = {Lecture Notes in Computer Science},
  volume    = {8617},
  pages     = {1--18},
  publisher = {Springer},
  year      = {2014},
  doi       = {10.1007/978-3-662-44381-1_1},
  bibsource = {dblp computer science bibliography, https://dblp.org}
}

@article{liu2021beating,
  title={Beating classical impossibility of position verification},
  author={Liu, Jiahui and Liu, Qipeng and Qian, Luowen},
  journal={arXiv preprint arXiv:2109.07517},
  year={2021}
}

@misc{bluhm2021positionbased,
      title={Position-based cryptography: Single-qubit protocol secure against multi-qubit attacks}, 
      author={Andreas Bluhm and Matthias Christandl and Florian Speelman},
      year={2021},
      eprint={2104.06301v2},
      archivePrefix={arXiv},
      primaryClass={quant-ph}
}

@article{KMS11qubit-routing,
  title = {Quantum tagging: Authenticating location via quantum information and relativistic signaling constraints},
  author = {Kent, Adrian and Munro, William J. and Spiller, Timothy P.},
  journal = {Phys. Rev. A},
  volume = {84},
  issue = {1},
  pages = {012326},
  numpages = {7},
  year = {2011},
  publisher = {American Physical Society},
  doi = {10.1103/PhysRevA.84.012326}
}

@article{CGMO09PBC,
  author    = {Nishanth Chandran and
               Vipul Goyal and
               Ryan Moriarty and
               Rafail Ostrovsky},
  title     = {Position-Based Cryptography},
  journal   = {{SIAM} J. Comput.},
  volume    = {43},
  number    = {4},
  pages     = {1291--1341},
  year      = {2014},
  doi       = {10.1137/100805005},
  bibsource = {dblp computer science bibliography, https://dblp.org}
}

@misc{vidick2021classicalproofsquantumknowledge,
      title={Classical proofs of quantum knowledge}, 
      author={Thomas Vidick and Tina Zhang},
      year={2021},
      eprint={2005.01691},
      archivePrefix={arXiv},
      primaryClass={quant-ph},
      url={https://arxiv.org/abs/2005.01691}, 
}

@misc{cryptoeprint:2025/1219,
      author = {Fuyuki Kitagawa and Takashi Yamakawa},
      title = {Foundations of Single-Decryptor Encryption},
      howpublished = {Cryptology {ePrint} Archive, Paper 2025/1219},
      year = {2025},
      url = {https://eprint.iacr.org/2025/1219}
}

@inproceedings{CG24,
author = {Coladangelo, Andrea and Gunn, Sam},
title = {How to Use Quantum Indistinguishability Obfuscation},
year = {2024},
isbn = {9798400703836},
publisher = {Association for Computing Machinery},
address = {New York, NY, USA},
url = {https://doi.org/10.1145/3618260.3649779},
doi = {10.1145/3618260.3649779},
booktitle = {Proceedings of the 56th Annual ACM Symposium on Theory of Computing},
pages = {1003–1008},
numpages = {6},
location = {Vancouver, BC, Canada},
series = {STOC 2024}
}

@inproceedings{BBV24,
author = {Bartusek, James and Brakerski, Zvika and Vaikuntanathan, Vinod},
title = {Quantum State Obfuscation from Classical Oracles},
year = {2024},
isbn = {9798400703836},
publisher = {Association for Computing Machinery},
address = {New York, NY, USA},
url = {https://doi.org/10.1145/3618260.3649673},
doi = {10.1145/3618260.3649673},
booktitle = {Proceedings of the 56th Annual ACM Symposium on Theory of Computing},
pages = {1009–1017},
numpages = {9},
location = {Vancouver, BC, Canada},
series = {STOC 2024}
}

@misc{hao2026needquantummemoryshort,
      title={On the Need for (Quantum) Memory with Short Outputs}, 
      author={Zihan Hao and Zikuan Huang and Qipeng Liu},
      year={2026},
      eprint={2602.23763},
      archivePrefix={arXiv},
      primaryClass={cs.CC},
      url={https://arxiv.org/abs/2602.23763}, 
}

@inproceedings{10.1109/CCC.2009.42,
author = {Aaronson, Scott},
title = {Quantum Copy-Protection and Quantum Money},
year = {2009},
isbn = {9780769537177},
publisher = {IEEE Computer Society},
address = {USA},
url = {https://doi.org/10.1109/CCC.2009.42},
doi = {10.1109/CCC.2009.42},
booktitle = {Proceedings of the 2009 24th Annual IEEE Conference on Computational Complexity},
pages = {229–242},
numpages = {14},
series = {CCC '09}
}

@inproceedings{10.1145/2213977.2213983,
author = {Aaronson, Scott and Christiano, Paul},
title = {Quantum money from hidden subspaces},
year = {2012},
isbn = {9781450312455},
publisher = {Association for Computing Machinery},
address = {New York, NY, USA},
url = {https://doi.org/10.1145/2213977.2213983},
doi = {10.1145/2213977.2213983},
booktitle = {Proceedings of the Forty-Fourth Annual ACM Symposium on Theory of Computing},
pages = {41–60},
numpages = {20},
location = {New York, New York, USA},
series = {STOC '12}
}

@article{BenDavid2023quantumtokens,
  doi = {10.22331/q-2023-01-19-901},
  url = {https://doi.org/10.22331/q-2023-01-19-901},
  title = {Quantum {T}okens for {D}igital {S}ignatures},
  author = {Ben-David, Shalev and Sattath, Or},
  journal = {{Quantum}},
  issn = {2521-327X},
  publisher = {{Verein zur F{\"{o}}rderung des Open Access Publizierens in den Quantenwissenschaften}},
  volume = {7},
  pages = {901},
  month = jan,
  year = {2023}
}

@inproceedings{10.1007/978-3-030-84242-0_19,
author = {Aaronson, Scott and Liu, Jiahui and Liu, Qipeng and Zhandry, Mark and Zhang, Ruizhe},
title = {New Approaches for Quantum Copy-Protection},
year = {2021},
isbn = {978-3-030-84241-3},
publisher = {Springer-Verlag},
address = {Berlin, Heidelberg},
url = {https://doi.org/10.1007/978-3-030-84242-0_19},
doi = {10.1007/978-3-030-84242-0_19},
booktitle = {Advances in Cryptology – CRYPTO 2021: 41st Annual International Cryptology Conference, CRYPTO 2021, Virtual Event, August 16–20, 2021, Proceedings, Part I},
pages = {526–555},
numpages = {30}
}

@InProceedings{bartusek_et_al:LIPIcs.ITCS.2022.15,
  author =	{Bartusek, James and Malavolta, Giulio},
  title =	{{Indistinguishability Obfuscation of Null Quantum Circuits and Applications}},
  booktitle =	{13th Innovations in Theoretical Computer Science Conference (ITCS 2022)},
  pages =	{15:1--15:13},
  series =	{Leibniz International Proceedings in Informatics (LIPIcs)},
  ISBN =	{978-3-95977-217-4},
  ISSN =	{1868-8969},
  year =	{2022},
  volume =	{215},
  editor =	{Braverman, Mark},
  publisher =	{Schloss Dagstuhl -- Leibniz-Zentrum f{\"u}r Informatik},
  address =	{Dagstuhl, Germany},
  URL =		{https://drops.dagstuhl.de/entities/document/10.4230/LIPIcs.ITCS.2022.15},
  URN =		{urn:nbn:de:0030-drops-156115},
  doi =		{10.4230/LIPIcs.ITCS.2022.15}
}

@inproceedings{10.1145/3564246.3585179,
author = {Bartusek, James and Kitagawa, Fuyuki and Nishimaki, Ryo and Yamakawa, Takashi},
title = {Obfuscation of Pseudo-Deterministic Quantum Circuits},
year = {2023},
isbn = {9781450399135},
publisher = {Association for Computing Machinery},
address = {New York, NY, USA},
url = {https://doi.org/10.1145/3564246.3585179},
doi = {10.1145/3564246.3585179},
booktitle = {Proceedings of the 55th Annual ACM Symposium on Theory of Computing},
pages = {1567–1578},
numpages = {12},
location = {Orlando, FL, USA},
series = {STOC 2023}
}

@INPROCEEDINGS {11369076,
author = { Huang, Mi-Ying Miryam and Tang, Er-Cheng },
booktitle = { 2025 IEEE 66th Annual Symposium on Foundations of Computer Science (FOCS) },
title = {{ Obfuscation of Unitary Quantum Programs }},
year = {2025},
volume = {},
ISSN = {},
pages = {1665-1671},
doi = {10.1109/FOCS63196.2025.00088},
url = {https://doi.ieeecomputersociety.org/10.1109/FOCS63196.2025.00088},
publisher = {IEEE Computer Society},
address = {Los Alamitos, CA, USA},
month =Dec}

@misc{huang2026obfuscationarbitraryquantumcircuits,
      title={Obfuscation of Arbitrary Quantum Circuits}, 
      author={Miryam Mi-Ying Huang and Er-Cheng Tang},
      year={2026},
      eprint={2601.08969},
      archivePrefix={arXiv},
      primaryClass={quant-ph},
      url={https://arxiv.org/abs/2601.08969}, 
}

@article{JACM:BGIRSVY12,
  author    = {Boaz Barak and
               Oded Goldreich and
               Russell Impagliazzo and
               Steven Rudich and
               Amit Sahai and
               Salil P. Vadhan and
               Ke Yang},
  title     = {On the (im)possibility of obfuscating programs},
  journal   = {J. {ACM}},
  volume    = {59},
  number    = {2},
  pages     = {6:1--6:48},
  year      = {2012},
}

@misc{bartusek2025newapproachargumentsquantum,
      title={A New Approach to Arguments of Quantum Knowledge}, 
      author={James Bartusek and Ruta Jawale and Justin Raizes and Kabir Tomer},
      year={2025},
      eprint={2510.05316},
      archivePrefix={arXiv},
      primaryClass={quant-ph},
      url={https://arxiv.org/abs/2510.05316}, 
}

@InProceedings{10.1007/978-3-031-38551-3_8,
author="Jain, Aayush
and Lin, Huijia
and Luo, Ji
and Wichs, Daniel",
editor="Handschuh, Helena
and Lysyanskaya, Anna",
title="The Pseudorandom Oracle Model and Ideal Obfuscation",
booktitle="Advances in Cryptology -- CRYPTO 2023",
year="2023",
publisher="Springer Nature Switzerland",
address="Cham",
pages="233--262",
isbn="978-3-031-38551-3"
}

@article{10.1145/3450745,
author = {Zhandry, Mark},
title = {How to Construct Quantum Random Functions},
year = {2021},
issue_date = {October 2021},
publisher = {Association for Computing Machinery},
address = {New York, NY, USA},
volume = {68},
number = {5},
issn = {0004-5411},
url = {https://doi.org/10.1145/3450745},
doi = {10.1145/3450745},
journal = {J. ACM},
month = aug,
articleno = {33},
numpages = {43}
}

@inproceedings{10.1007/978-3-030-84242-0_20,
author = {Coladangelo, Andrea and Liu, Jiahui and Liu, Qipeng and Zhandry, Mark},
title = {Hidden Cosets and Applications to Unclonable Cryptography},
year = {2021},
isbn = {978-3-030-84241-3},
publisher = {Springer-Verlag},
address = {Berlin, Heidelberg},
url = {https://doi.org/10.1007/978-3-030-84242-0_20},
doi = {10.1007/978-3-030-84242-0_20},
booktitle = {Advances in Cryptology – CRYPTO 2021: 41st Annual International Cryptology Conference, CRYPTO 2021, Virtual Event, August 16–20, 2021, Proceedings, Part I},
pages = {556–584},
numpages = {29}
}

@article{Coladangelo2024quantumcopy,
  doi = {10.22331/q-2024-05-02-1330},
  url = {https://doi.org/10.22331/q-2024-05-02-1330},
  title = {Quantum copy-protection of compute-and-compare programs in the quantum random oracle model},
  author = {Coladangelo, Andrea and Majenz, Christian and Poremba, Alexander},
  journal = {{Quantum}},
  issn = {2521-327X},
  publisher = {{Verein zur F{\"{o}}rderung des Open Access Publizierens in den Quantenwissenschaften}},
  volume = {8},
  pages = {1330},
  month = may,
  year = {2024}
}

@inproceedings{10.1007/978-3-031-22318-1_11,
author = {Liu, Jiahui and Liu, Qipeng and Qian, Luowen and Zhandry, Mark},
title = {Collusion Resistant Copy-Protection for Watermarkable Functionalities},
year = {2022},
isbn = {978-3-031-22317-4},
publisher = {Springer-Verlag},
address = {Berlin, Heidelberg},
url = {https://doi.org/10.1007/978-3-031-22318-1_11},
doi = {10.1007/978-3-031-22318-1_11},
booktitle = {Theory of Cryptography: 20th International Conference, TCC 2022, Chicago, IL, USA, November 7–10, 2022, Proceedings, Part I},
pages = {294–323},
numpages = {30},
location = {Chicago, IL, USA}
}

@inproceedings{10.1007/978-3-031-68394-7_1,
author = {Ananth, Prabhanjan and Behera, Amit},
title = {A Modular Approach to Unclonable Cryptography},
year = {2024},
isbn = {978-3-031-68393-0},
publisher = {Springer-Verlag},
address = {Berlin, Heidelberg},
url = {https://doi.org/10.1007/978-3-031-68394-7_1},
doi = {10.1007/978-3-031-68394-7_1},
booktitle = {Advances in Cryptology – CRYPTO 2024: 44th Annual International Cryptology Conference, Santa Barbara, CA, USA, August 18–22, 2024, Proceedings, Part VII},
pages = {3–37},
numpages = {35},
location = {Santa Barbara, CA, USA}
}

@misc{kaleoglu2025equivalenceclassicalpositionverification,
      title={On the Equivalence between Classical Position Verification and Certified Randomness}, 
      author={Fatih Kaleoglu and Minzhao Liu and Kaushik Chakraborty and David Cui and Omar Amer and Marco Pistoia and Charles Lim},
      year={2025},
      eprint={2410.03982},
      archivePrefix={arXiv},
      primaryClass={quant-ph},
      url={https://arxiv.org/abs/2410.03982}, 
}

@misc{kalai2026classicallyverifyquantumcat,
      title={How to Classically Verify a Quantum Cat without Killing It}, 
      author={Yael Tauman Kalai and Dakshita Khurana and Justin Raizes},
      year={2026},
      eprint={2602.09282},
      archivePrefix={arXiv},
      primaryClass={quant-ph},
      url={https://arxiv.org/abs/2602.09282}, 
}

@InProceedings{10.1007/978-3-032-25291-3_19,
author="Ananth, Prabhanjan
and Behera, Amit
and Huang, Zikuan
and Kitagawa, Fuyuki
and Yamakawa, Takashi",
editor="Daemen, Joan
and Thom{\'e}, Emmanuel",
title="Copy-Protection from Unclonable Puncturable Obfuscation, Revisited",
booktitle="Advances in Cryptology -- EUROCRYPT 2026",
year="2026",
publisher="Springer Nature Switzerland",
address="Cham",
pages="541--570",
isbn="978-3-032-25291-3"
}

@InProceedings{10.1007/978-3-032-25291-3_17,
author="{\c{C}}akan, Alper
and Goyal, Vipul",
editor="Daemen, Joan
and Thom{\'e}, Emmanuel",
title="How to Copy-Protect Malleable-Puncturable Cryptographic Functionalities Under Arbitrary Challenge Distributions: A Unified Solution to Quantum Protection",
booktitle="Advances in Cryptology -- EUROCRYPT 2026",
year="2026",
publisher="Springer Nature Switzerland",
address="Cham",
pages="481--509",
isbn="978-3-032-25291-3"
}

@misc{asadi2025ranklowerboundsnonlocal,
      title={Rank lower bounds on non-local quantum computation}, 
      author={Vahid R. Asadi and Eric Culf and Alex May},
      year={2025},
      eprint={2402.18647},
      archivePrefix={arXiv},
      primaryClass={quant-ph},
      url={https://arxiv.org/abs/2402.18647}, 
}

@article{Asadi_2025,
   title={Linear gate bounds against natural functions for position-verification},
   volume={9},
   ISSN={2521-327X},
   url={http://dx.doi.org/10.22331/q-2025-01-21-1604},
   DOI={10.22331/q-2025-01-21-1604},
   journal={Quantum},
   publisher={Verein zur Forderung des Open Access Publizierens in den Quantenwissenschaften},
   author={Asadi, Vahid and Cleve, Richard and Culf, Eric and May, Alex},
   year={2025},
   month=Jan, pages={1604} }

@misc{kavuri2026quantumpositionverificationremote,
      title={Quantum Position Verification with Remote Untrusted Devices}, 
      author={Gautam A. Kavuri and Yanbao Zhang and Abigail R. Gookin and Soumyadip Patra and Joshua C. Bienfang and Honghao Fu and Yusuf Alnawakhtha and Dileep V. Reddy and Michael D. Mazurek and Carlos Abellán and Waldimar Amaya and Morgan W. Mitchell and Sae Woo Nam and Carl A. Miller and Richard P. Mirin and Martin J. Stevens and Scott Glancy and Emanuel Knill and Lynden K. Shalm},
      year={2026},
      eprint={2601.16892},
      archivePrefix={arXiv},
      primaryClass={quant-ph},
      url={https://arxiv.org/abs/2601.16892}, 
}

@article{fanyuan2026relativistic,
  title     = {Relativistic position verification with coherent states},
  author    = {Fan-Yuan, Guan-Jie and Shan, Yang-Guang and Zhang, Cong and
               Wang, Yu-Long and Fan, Yu-Xuan and Xie, Wei-Xin and
               He, De-Yong and Wang, Shuang and Yin, Zhen-Qiang and
               Chen, Wei and Fu, Song-Nian and Guo, Guang-Can and
               Han, Zheng-Fu},
  journal   = {Nature Physics},
  year      = {2026},
  doi       = {10.1038/s41567-026-03439-5},
  url       = {https://doi.org/10.1038/s41567-026-03439-5}
}

@article{kanneworff2025experimental,
  title     = {Towards experimental demonstration of quantum position verification using single photons},
  author    = {Kanneworff, Kirsten and Poortvliet, Mio and Bouwmeester, Dirk and
               Allerstorfer, Rene and Verduyn Lunel, Philip and Speelman, Florian and
               Buhrman, Harry and Steindl, Petr and L{\"o}ffler, Wolfgang},
  journal   = {Quantum Science and Technology},
  volume    = {10},
  number    = {4},
  pages     = {045004},
  year      = {2025},
  publisher = {IOP Publishing},
  doi       = {10.1088/2058-9565/adf2da},
  url       = {https://doi.org/10.1088/2058-9565/adf2da}
}
